\RequirePackage[british]{babel}
\documentclass[reqno,a4paper,12pt]{amsart}
\usepackage[latin1]{inputenc}
\usepackage[a4paper,hmargin=2cm,tmargin=3cm,bmargin=3cm]{geometry}
\usepackage{amsmath,amssymb,amstext,amsthm,amscd}
\usepackage{mathrsfs,eucal,graphicx}
\usepackage{color}
\usepackage{verbatim}
\usepackage{array}
\usepackage[nosort]{cite}
\usepackage[vcentermath]{youngtab}
\Yboxdim{4pt}
\usepackage{hyperref}
\hypersetup{%
  pdftitle   = {n-ary algebras: a review with applications},
  pdfkeywords = {n-algebras, higher order-Lie algebras, Filippov (or $n$-Lie) algebras,
              Generalized Poisson structures, Bagger-Lambert-Gustavsson model},
  pdfauthor  = {J. A. de Azc\'arraga, J. M. Izquierdo},
  pdfcreator = {\LaTeX\ with package \flqq hyperref\frqq}
      }
\newcommand{\fg}{\mathfrak{g}}

\newcommand{\fG}{\mathfrak{G}}

\newcommand{\fH}{\mathfrak{H}}

\newcommand{\fL}{\mathfrak{L}}

\newcommand{\fN}{\mathfrak{N}}

\newcommand{\fX}{\mathfrak{X}}
\newcommand{\ff}{\mathfrak{f}}
\newcommand{\fgl}{\mathfrak{gl}}
\newcommand{\fh}{\mathfrak{h}}

\newcommand{\fT}{\mathfrak{T}}

\newcommand{\be}{\boldsymbol{e}}

\DeclareMathOperator{\End}{End}

\DeclareMathOperator{\ad}{ad}

\theoremstyle{plain}
\newtheorem{lemma}{Lemma}
\newtheorem{proposition}[lemma]{Proposition}
\newtheorem{theorem}[lemma]{Theorem}
\newtheorem{corollary}[lemma]{Corollary}
\theoremstyle{definition}
\newtheorem{definition}[lemma]{Definition}
\newtheorem{remark}[lemma]{Remark}
\newtheorem{example}[lemma]{Example}
\allowdisplaybreaks

\begin{document}
\title{$n$-ary algebras: a review with applications}
\author[de Azc\'arraga]{Jos\'e A. de Azc\'arraga}
\author[Izquierdo]{Jos\'e M. Izquierdo}
\address{Department of Theoretical Physics \& IFIC (CSIC-UVEG),
Faculty of Physics, Valencia University, 46100-Burjassot (Valencia), Spain.
e-mail: j.a.de.azcarraga@ific.uv.es}
\address{Department of Theoretical Physics, Faculty of Sciences,
Valladolid University, 47011-Valladolid, Spain.
e-mail: izquierd@fta.uva.es}
\date{February, 28; v2: minor changes and refs. added. Published in
J. Phys. A: Math. Theor.~{\bf 43}, 293001-1-117 (2010).}

\begin{abstract}
This paper reviews the properties and applications of certain
$n$-ary generalizations of Lie algebras in a self-contained and
unified way. These generalizations are algebraic structures in which
the two entries Lie bracket has been replaced by a bracket with $n$
entries. Each type of $n$-ary bracket satisfies a specific {\it
characteristic identity} which plays the r\^ole of the Jacobi
identity for Lie algebras. Particular attention will be paid to {\it
generalized Lie algebras}, which are defined by even multibrackets
obtained by antisymmetrizing the associative products of its $n$
components and that satisfy the {\it generalized Jacobi identity},
and to {\it Filippov (or $n$-Lie) algebras}, which are defined by
fully antisymmetric $n$-brackets that satisfy the {\it Filippov
identity}. Three-Lie algebras have surfaced recently in multi-brane
theory in the context of the Bagger-Lambert-Gustavsson model.
Because of this, Filippov algebras will be discussed at length,
including the cohomology complexes that govern their central
extensions and their deformations (it turns out that Whitehead's
lemma extends to all semisimple $n$-Lie algebras). When the
skewsymmetry of the Lie or $n$-Lie algebra bracket is relaxed, one
is led to a more general type of $n$-algebras, the $n$-Leibniz
algebras. These will be discussed as well, since they underlie the
cohomological properties of $n$-Lie algebras.

The standard Poisson structure may also be extended to the $n$-ary
case. We shall review here the even {\it generalized Poisson
structures}, whose generalized Jacobi identity reproduces the
pattern of the generalized Lie algebras, and the {\it Nambu-Poisson
structures}, which satisfy the Filippov identity and determine
Filippov algebras. Finally, the recent work of Bagger-Lambert and
Gustavsson on superconformal Chern-Simons theory will be briefly
discussed. Emphasis will be made on the appearance of the 3-Lie
algebra structure and on why the $A_4$ model may be formulated in
terms of an ordinary Lie algebra, and on its Nambu bracket
generalization.
\end{abstract}

\maketitle
\tableofcontents

\section{Introduction and overview} \label{sec:intro}

In the last decades there has been an increasing interest in the
applications of various $n$-ary generalizations of the ordinary Lie algebra
structure to theoretical physics problems, which has peaked in the
last three years. $n$-ary algebraic operations are, however, very
old: ternary operations appeared for the first time associated
with the cubic matrices that had been introduced by A. Cayley in
the middle of the XIXth century and that were also considered by
J.J. Sylvester some forty years later, still in that century.
In spite of this, the  modern mathematical work on
general multioperator rings and algebras (not
necessarily associative) begins much later with a series of papers
by Kurosh (see \cite{Kurosh} and earlier refs. therein). In particular, the
{\it linear $\Omega$-algebras} are given by a vector space on which
certain multilinear operations are defined; they are reviewed in
\cite{Bara-Bur:75}. This general
class of algebras is, however,  larger than the two main
generalizations of Lie algebras to be discussed in this review,
which will be denoted generically as $n$-ary algebras. In these, the
standard Lie bracket is replaced by a linear $n$-ary bracket with
$n>2$ entries, the algebra structure being defined by the
characteristic identity satisfied by the $n$-ary bracket. There are
two (main) ways of achieving this, depending on how the Jacobi
identity (JI) of the ordinary Lie algebras is looked at. The JI can
be viewed as the statement that (a) a nested double Lie bracket
gives zero when antisymmetrized with respect to its three entries or
that (b) the Lie bracket is a derivation of itself. Both (a) and (b)
are equivalent for ordinary Lie algebras where the JI is indeed a
necessary identity that follows from the associativity of the
composition of the Lie algebra elements in the Lie bracket.
\medskip

  Let now $\fH$ be a generic $n$-ary algebra, $n>2$, endowed
with a skewsymmetric $n$-linear bracket $\wedge^n\fH \rightarrow
\fH$, $(X_1,\dots,X_n) \mapsto [X_1,\dots,X_n]$. When $\fH$ is
defined using the {\it characteristic identity} that extends
property (a) to the $n$-ary bracket, one is led to the {\it higher
order Lie algebra} or {\it generalized Lie algebra} structure (GLA)
\cite{AzBu:96,AzPePB:96a,AzPePB:96b},
\cite{Han-Wac:95,Gne:95,Gne:97} (see also \cite{JLL:95,Mic-Vin:96}),
denoted $\fH=\mathcal{G}$, and the characteristic identity satisfied
by its multibracket is called {\it generalized Jacobi identity}
(GJI); GLAs will be discussed in Sec.~\ref{sec:GLA}. These algebras
are a case of the more general antisymmetric linear
$\Omega$-algebras of Kurosh \cite{Kurosh}, to whom the earlier
generalizations of Lie algebras may be traced; see
\cite{Bara-Bur:75} for a review of the russian school contributions
to the subject and further references.
The GLA generalization is natural for $n$ even (for
$n$ odd, the $r.h.s.$ of the GJI, rather than being zero, is a
larger bracket with ($2n-1$) entries). GLAs may also be considered
as a particular case (in which there is no violation of the GJI
\cite{RACSAM:98}) of the strongly homotopy algebras of Stasheff
\cite{Lad.Sta:93,Lad.Mar:95, Sta:97, Be-La:09}. When possibility (b)
is used as the guiding principle, then one obtains the {\it Filippov
identity} \cite{Filippov} (FI) as characteristic identity and,
correspondingly, the {$n$-Lie} or {\it Filippov algebras} (FA)
\cite{Filippov}; both terms, Filippov and $n$-Lie, will be used
indistinctly. FAs will be denoted by
$\fH=\mathfrak{G}$ and reviewed in Sec.~\ref{sec:filippov}.
The characteristic GJI and FI still admit further generalizations
that essentially preserve their original structure, but these
(see \cite{Ata-Ma-Sil:08}) will not be considered here.
\medskip

As with metric Lie algebras (Sec. \ref{sec:metric-Lie}), the
$n$-bracket alone is often insufficient for applications and the
existence of an inner product on the underlying algebra vector space
is required. This leads to notion of metric Filippov $n$-algebras to
be discussed in Sec.~\ref{sec:metric-n-FA}. When $n=2$, both algebra
structures $\mathcal{G}$ and $\fG$ coincide and determine ordinary
Lie algebras $\mathfrak{g}$. For $n\geq 3$, the GJI ($n$ even) and
the FI become different characteristic identities and define,
respectively, GLAs $\mathcal{G}$ and $n$-Lie or FAs $\mathfrak{G}$.
There is also the possibility of relaxing the antisymmetry of the
$n$-bracket: this leads to ordinary ($n=2$) Loday's (or Leibniz)
algebras \cite{Lod:93,Lod-Pir:93,Lod-Pir:96} and to their
$n$-Leibniz algebra generalizations \cite{Da-Tak:97,
Cas-Lod-Pir:02}, which differ from their Lie and $n$-Lie
counterparts by having brackets that do not require
anticommutativity. These algebras will be denoted by $\mathscr{L}$
($n=2$) and $\mathfrak{L}$ ($n\geq 3$) respectively, and will be
considered in Secs.~\ref{sec:Loday},~\ref{sec:n-Leibniz}.
\medskip

Filippov algebras \cite{Filippov, Kas:87, Kas:95a, Ling:93}
 have recently been found useful in the search for an
effective action describing the low energy dynamics of coincident
M2-branes or, more specifically, in the Bagger-Lambert-Gustavsson
(BLG) type models
\cite{Ba-La:06,Ba-La:07a,Gustav:08,Raam:08,Ba-La:08,Cher-Sa:08,Go-Mi-Ru:08}.
The fact that there is a unique simple\footnote{That there is
only a simple $n$-Lie algebra for $n>2$ had been found in \cite{Ling:93, Filippov}
(see Th.~\ref{th:FAsimple}). For $n=3$ this was also
rediscovered in \cite{Fried:08} (unaware of all earlier work), plus
other results. This reference contains a new proposal for the use
of FAs in the context of orbifold singularities in M-theory compactifications.}
euclidean 3-Lie algebra ($A_4$) was actually rediscovered in the
context of the first BLG model, where it follows \cite{Pap:08,Ga-Gu:08}
by assuming that the metric needed for the BLG action has to
be positive definite, a condition that may be relaxed. We shall
discuss the BLG and related models in Secs.~\ref{sec:BLG},
\ref{sec:BLG+NB}; we just mention now that the original BLG
action was subsequently reformulated \cite{Be-Ta-Tho:08, Raam:08}
without using a 3-Lie algebra, and that other models for low energy
multiple M2-brane dynamics have appeared (albeit with
$\mathcal{N}$=6 rather than $\mathcal{N}$=8 manifest
supersymmetries) that do not use a FA structure
\cite{Aha-Be-Ja-Mal:08}. Some modifications of the original BLG
model based on $A_4$ have been considered in the literature using
non-fully skewsymmetric 3-brackets. These define, in fact, various
3-Leibniz algebras, and some of these will be discussed (along with
Lie-triple systems $\mathfrak{T}$, a particular case of 3-Leibniz
algebras) in Sec.~\ref{sec:triple+3Leib}.
\medskip

  Nambu algebras are a particular, infinite-dimensional case of $n$-Lie
algebras. Their $n$-bracket is provided by the Jacobian determinant
of $n$ functions or Nambu bracket \cite{Nambu:73}, although Nambu
did not write the characteristic identity satisfied by his
bracket, which is none other than the FI. This was discussed in
\cite{Filippov,Sah-Val:92,Sa-Va:93, Tak:93,Fil:98}, and {\it
Nambu-Poisson structures} (N-P) have been much studied since Nambu's
original paper \cite{Nambu:73} (mostly devoted to
$n=3$) and Takhtajan general study \cite{Tak:93}
for arbitrary $n$, see \cite{Mu-Sud:76,Sa-Va:93,
Cha:95, Cha-Tak:95, Ale.Guh:96, Hie:97, Da-Tak:97,
Mar-Vi-Vin:97,Vai:99, Mi-Va:00, Cu-Za:02} (ref.~\cite{Da-Tak:97}
also considers Nambu superalgebras). In fact, since the earlier
considerations of $p$-branes as gauge theories of volume preserving
diffeomorphisms \cite{dW-Ho-Ni:88,Be-Se-Ta-To:90} (see also
\cite{Hop:88}), the infinite-dimensional FAs given by Nambu brackets
have also appeared in applications to brane theory \cite{Hoppe:96}
and, in particular, in the Nambu three-bracket realization of the
mentioned BLG model as a gauge theory associated with volume
preserving diffeormorphisms in a three-dimensional space; see, in
particular, refs. \cite{Ho-Hou-Ma:08, Ho-Mat:08, Ba-To:08} (see also
\cite{Ho:09}). The BLG-Nambu bracket model in \cite{Ba-To:08} will
be considered in Sec.~\ref{sec:BLG}.
\medskip

Much in the same way the Nambu-Poisson structures follow the pattern
of FAs, it is also possible to introduce {\it generalized Poisson
structures} (GPS) \cite{AzPePB:96a, AzPePB:96b, AIP-B:97} (see
further \cite{Mi-Va:00,Iba.Leo.Mar.Die:97,Iba.Leo.Mar:97}), the
$n$-even {\it generalized Poisson brackets} (GPB) of which satisfy
the GJI and correspond to the  GLAs earlier mentioned. Note,
however, that besides the two properties that each $n$-ary Poisson
generalization share respectively with the GLAs and FAs
(skewsymmetry of both $n$-ary brackets plus the GJI (FI) for the GP
(N-P) structures, respectively), the $n$-ary brackets of both GPS
and of N-PS satisfy an additional condition, Leibniz's rule. Adding
the appropriate grading factors, it is also possible to define the
{\it graded generalized Poisson structures} (graded GPS)
\cite{AzIzPePB:96}; of course, setting aside the Leibniz rule, the
remaining structural properties of the graded GPS define the graded
multibracket and the graded GJI of the {\it graded generalized Lie
algebras} (graded GLAs). We conclude this paragraph by noting that
there has been an extensive discussion since the papers by Nambu
\cite{Nambu:73} and Takhtajan \cite{Tak:93} about the difficulties
of quantizing the N-P strucures, a question to which we shall come
back in Sec.~\ref{GPSvsN-P}; here we just refer in this connection
to the papers above and, {\it e.g.}, to \cite{Hoppe:96,
AIP-B:97,Stern:98, Cu-Za:03b,Cu-Fa-Ji-Me-Za:09} and further
references therein.
\medskip

  Although the GLA, FA and $n$-Leibniz algebra structures
will be the main subject of this review (besides the triple
systems to be briefly discussed in Sec.~\ref{sec:triple}),
these do not exhaust, of course, the wide class of algebra
and Lie algebra generalizations. Without considering a characteristic
identity, other $n$-algebraic structures have been studied {\it e.g.} in
\cite{Carl:80,Ker-Vain:96}. Other Lie algebra generalizations
have to do with introducing a grading; this includes the
well known case of superalgebras (see {\it e.g.} \cite{Fra-Sci-Sor:00}
for a useful collection of definitions, results and basic references
and \cite{CoNeSt:74} for a very early -perhaps the first- review).
There is, besides, a full array of other algebraic structures. These include
non-associative algebras in general \cite{Schaf:66,Brem-Mur-Shes:07},
Mal\v cev (or Moufang-Lie in Mal\v cev's terminology)
algebras \cite{Sagle:61,Brem-Per:06} and their ternary
\cite{Brem-Per:06} and $n$-ary generalizations \cite{Pozhi:01}, some specific types of ternary
structures \cite{Ker:00,Ba-Bo-Ker:04,Ker:08}), $F$-Lie algebras
\cite{Rau-Slu:99,Go-Tra-Ta:07,Camp-Rau:08}, etc. Some of these algebras
were motivated in part by certain aspects of fractional supersymmetry
and/or the problem of parastatistics (see {\it e.g.}
\cite{deA-MacF:95,Du-MacF-Az-PB:96a,deAz-Du-MacF-PB:96,Du-MacF-Az-PB:96b,Ab-Ker-LeRoy:96,Camp-Raus:09})
but they were also studied by their own interest (see \cite{Rau-Slu:02}
and \cite{RdeTr:08} for a review of higher order generalizations of
the Poincar\'e algebra and possible applications).
   None of these algebraic structures, however, will be considered here;
we refer to the quoted papers for further information and references.
Other types of structures and (anti)brackets/algebras, as the Batalin-Vilkovisky
algebras, the master equation, etc. (see {\it e.g.} \cite{Ber-Dam-Al:96}
and references therein), which are relevant {\it e.g.} in the B-V approach to
quantization, will also be omitted. Nevertheless, a few important constructions such as
Koszul's \cite{Koszul:85}, strongly homotopy algebras \cite{Lad.Sta:93}
(see also \cite{Zwi:92} in the context of closed string theory) and
Gerstenhaber algebras \cite{Ger:63} will be briefly mentioned in
Sec.~\ref{sec:SH} or in connection with Sec.~\ref{sec:higherPoisson};
see also \cite{Kos:96} for a unified review and \cite{Vor:05} for
other aspects.
\medskip

Since both GLAs and FAs reduce to Lie algebras for $n=2$, and Lie
algebras are often a guide for any of their $n>2$ generalizations,
we start by summarizing some Lie algebra properties in
Sec.~\ref{sec:Lie-review}.
\medskip

\subsection{Notation and conventions}
\label{sec:not-con} {\ }

 We will use $\fg$ to
denote standard Lie algebras and the larger case $\mathcal{G}$ and $\fG$,
respectively, for the $n$-ary higher order or {\it Generalized Lie algebras}
(GLAs, $\mathcal{G}$) and the {\it Filippov} or
{\it $n$-Lie algebras} (FAs, $\fG$).
In general, we will use the same symbol
for the different $n$-ary algebras and their underlying vector
spaces. Ordinary $n=2$ and ($n>2$)-Leibniz algebras (LAs) will be
denoted by $\mathscr{L}$ and $\mathfrak{L}$ respectively;
triple systems will be denoted by $\fT$. The
infinite-dimensional FAs generated by the Jacobian of functions,
also referred to as Nambu algebras, will be denoted by $\fN$.

For the sake of distinguishing clearly among the different brackets
considered, we shall refer to the $n$-ary brackets of the GLAs
$\mathcal{G}$ as {\it higher order brackets} or {\it multibrackets},
and to those of the FAs $\fG$ or of the $n$-Leibniz ones $\fL$
simply as $n$-{\it brackets}. The elements of the different algebras
will be frequently denoted by capital letters $X,Y$  etc (and, for
FAs, occasionally by $\be_a$). Chosen a basis, the structure
constants of GLAs $\mathcal{G}$ will be written as $C_{i_1\dots
i_{2s}}{}^j$ ($i=1,\dots,\mathrm{dim}\mathcal{G}$) and those of the
FAs $\fG$ and LAs $\fL$ as $f_{a_1\dots a_n}{}^b$
$\;(a=1,\dots,\mathrm{dim}(\fG,\fL)$). For $n=2s=2$,
$\mathcal{G}=\fG=\fg$ and $\fL$ becomes $\mathscr{L}$.

All algebras in this review will be on $\mathbb{R}$ or $\mathbb{C}$
and, with the exception of Nambu and gauge algebras,
finite-dimensional.
\medskip

{\it A terminological remark}.

 GLAs were called {\it Lie $n$-algebras}
in \cite{Han-Wac:95}. Such a terminology may cause confusion with
the very different $n$-Lie algebras (FAs) which were introduced by
Filippov with this name a dozen of years earlier. Thus, rather than
betting the distinction between $n$-Lie algebras ($\equiv$ FAs) and
GLAs on the precise location of a single letter ($n$-Lie alg. {\it
vs.} Lie $n$-alg.), we shall use $n$-Lie and FA indistinctly for
Filippov algebras and keep our {\it higher order} or {\it
generalized Lie algebras} GLA terminology for the $n$-ary
generalization of Lie algebras reviewed in Sec.~\ref{sec:GLA}.
\medskip

{\it A note on references}.

Besides some references quoted by their historical value, most of
them have been selected to give due credit to the relevant papers
and, further, for their potential usefulness when they are directly
related to the text. We do not consider helpful the recent practice
in some (physics) e-papers of grouping twenty (or far more)
references under a single number; this might result in quoting all
possibly related work, but it is useless to the reader, who could
perform such an indiscriminate search in the arXives if she or he
wished so.

\section{A short summary of Lie algebras}
\label{sec:Lie-review}
 We summarize in this section some ingredients of the
 theory of Lie algebras (see {\it e.g.} \cite{Jac:79})
that may be useful when considering ($n>2$)-ary
generalizations.

\subsection{General properties of Lie algebras}
\label{Lie-defs} {\ }

 A {\it Lie algebra structure} is a
vector space $\fg$ together with a bilinear operation $\fg\times
\fg \rightarrow \fg $, the Lie bracket $[\; ,\;
]:(X,Y)\mapsto [X,Y]$, that satisfies
\begin{equation}
\label{Lie-bra}
 [X,Y]= - [Y,X] \;  ,
\end{equation}
\begin{equation}
\label{Lie-JI}
 [X,[Y,Z]] + [Y,[Z,X]] + [Z,[X,Y]]=0 \;.
\end{equation}
A finite example is the associative algebra $\End V$ of linear
transformations of a finite vector space, which is the general linear Lie
algebra $\fgl(\hbox{dim}V)$; another are the Lie algebras $\fg$ generated by
the (left, say) invariant vector fields associated with the (then, right)
action of a Lie group $G$ on itself. An infinite-dimensional
example is the Lie algebra of the vector fields $\fX (M)$ on a manifold $M$.

 The Jacobi identity (JI) \eqref{Lie-JI} may be looked at
as a necessary consequence of the associativity of the composition
of the bracket elements with $[X,Y]=XY-YX$. A second view of the JI
is obtained by rewriting eq.~\eqref{Lie-JI} as
\begin{equation}
  \label{eq:derivation}
  [X,[Y,Z]] = [[X,Y],Z] + [Y,[X,Z]]~.
\end{equation}
A linear transformation $D: \fg \to \fg$ of the Lie algebra is said to be a
{\it derivation} of the Lie algebra, $D\in \hbox{Der}\,\fg$, if
\begin{equation}
  \label{eq:derivation-def}
  D [Y,Z] = [DY,Z] + [Y,DZ] \quad \forall \;Y,Z\in \fg\quad .
\end{equation}
Thus, eq.~\eqref{eq:derivation} states that, for all $X\in\fg$, $[X,
\;\,]$ is a derivation of $\fg$. This is the adjoint derivative
acting from the left, $ad_X Y := [X,Y]$. In terms of $ad_X$,
eq.~\eqref{eq:derivation} becomes
\begin{equation}
  \label{eq:ad-der}
  ad_X [Y,Z] = [ad_X Y,Z] + [Y,ad_X Z]~.
\end{equation}
Thus, the JI identity also expresses that $ad_X\in \hbox{Der}\,\fg$
is an inner derivation of $\fg$ $\forall\,X\in \fg$.

Viewing the above as the result of a linear
transformation on any $Z\in\fg$ and removing it,
the previous equation yields
\begin{equation}
  \label{eq:adg}
  ad_X ad_Y - ad_Y ad_X = ad_{ad_X Y} \qquad\text{or equivalently}\qquad [ad_X,
  ad_Y] = ad_{[X,Y]}~,
\end{equation}
where the bracket on the $l.h.s.$ of the second equation stands for
the commutator in $\End\fg$ and that of the $r.h.s.$ is the original
bracket in $\fg$.  Thus, the map $ad:\fg \to
\hbox{End}\fg=\fgl(\hbox{dim}\fg)$, $ad: X \mapsto ad_X$, is a Lie
algebras homomorphism and defines the {\it adjoint representation} of $\fg$.
Its image $ad\,\fg$ is a subalgebra of $\fgl(\hbox{dim}\fg)$, the
Lie algebra of {\it inner} derivations of $\fg$ or
$\mathrm{InDer}\,\fg$, inner because they are defined by elements of
$\fg$. The kernel of $ad$ is the centre $Z(\fg)$ of $\fg$ and is an
ideal of $\fg$; thus, $\mathrm{InDer}\,\fg= \fg/Z(\fg)$.
InDer$\,\fg$ is an ideal of the Lie algebra of all derivations
$\hbox{Der}\,\fg$ of $\fg$  since rearranging equation
\eqref{eq:derivation-def} it follows that $[D, ad_Y] = ad_{DY}$ for
any $D\in \hbox{Der}\,\fg\;,\,Y\in\fg$. As a result, the quotient
$\hbox{Der}\,\fg / \hbox{InDer}\,\fg = \hbox{OutDer}\,\fg$ is, by
construction, the Lie algebra of the {\it outer derivations}. If
$\fg$ is semisimple, $\fg \rightarrow ad\,\fg$ is an isomorphism of
Lie algebras and, further, all derivations are inner,
$\mathrm{Der}\,\fg=\mathrm{InDer}\,\fg$. By {\it Ado's theorem},
every finite-dimensional Lie algebra over $\mathbb{R}$ or
$\mathbb{C}$ is isomorphic to a matrix algebra (has a faithful
finite representation). An important result towards the
classification of arbitrary Lie algebras is the {\it Levi-Mal{\v
c}ev decomposition theorem}: a Lie algebra is the semidirect sum
$\fg=\hbox{Rad}(\fg) \;+\!\!\!\!\!\!\supset  \;\fg_L$ of a
semisimple subalgebra ($\fg_L$, its {\it Levi factor}) and its
radical Rad($\fg$). The Lie algebra is called {\it reductive} when
$\hbox{Rad}(\fg)=Z(\fg)$, in which case $\fg=Z(\fg)\oplus \fg_L$
with $\fg_L=[\fg,\fg]$.

Once a basis $\{X_i \}$ of $\fg$ is chosen, a Lie algebra
may be described in terms of the corresponding {\it structure constants}
$C_{ij}{}^k$ of $\fg$, $[X_i, X_j] = C_{ij}{}^k
X_k$. The defining conditions \eqref{Lie-bra}, \eqref{Lie-JI}
for a Lie algebra are then given by the expressions
\begin{equation}
\label{JIascocy}
  C_{ij}{}^k = - C_{ji}{}^k \qquad\text{and}\qquad  C_{[ij}{}^l C_{k]l}{}^s = 0 \quad \hbox{(JI)}
  \quad,\quad i,j,k=1,\dots,r=\hbox{dim}\fg \quad,
\end{equation}
respectively, where the square brackets surrounding the indices
denote total skewsymmetrisation. This will be defined throughout the
paper by
\begin{equation}
\label{defmultibra}
  [a_1\dots a_n] := \sum_{\sigma \in S_n} (-1)^{\pi(\sigma)} a_{\sigma(1)}\dots
    a_{\sigma(n)} \quad ,
\end{equation}
where $\pi(\sigma)=0,1$ is the even or odd parity of the permutation
$\sigma$ in the group $S_n$ of permutations of the indices
$(1,\dots,n)$, without the `weight one' factor $1/n!\,$.
This means that
\begin{equation}
\label{anticycl} [a_1\dots a_n]=\sum_{cycl \,\sigma\in S_n}
(-1)^{\pi(\sigma)} a_{\sigma(1)}[a_{\sigma(2)}\dots a_{\sigma(n)}]
\end{equation}
which, for $n$ odd, does not produce any signs. In terms of the
structure constants, the matrices $ad_X\in \hbox{End}\,\fg$ of the
adjoint representation $ad_{X_i} X_j :=[X_i,X_j]$ are given by
$(ad_{X_i})^k{}_j=C_{ij}{}^k$.

\subsection{A comment on associativity}
\label{asso-FA}
{\ }

The {\it associator} of a triple product,
which is given by
\begin{equation}
\label{associat}
(X,Y,Z):=(XY)Z-X(YZ)\; ,
\end{equation}
accounts for the lack of associativity in
the same way that the commutator $[X,Y]=XY-YX$ measures
the lack of commutativity. Non-associative algebras
(see \cite{Schaf:66,Oku:95,Gur-Tze:96}
for book discussions) under the composition
have non-zero associators. It is possible to
define an antisymmetric associator by
\begin{equation}
\label{3asso} [X,Y,Z]_{\mathrm{ant.assoc.}} :=
(X,Y,Z)+(Y,Z,X)+ (Z,X,Y)-
(Y,X,Z)-(X,Z,Y)-(Z,Y,X) \quad .
\end{equation}
The above expression is equivalent to
$[([X,Y]),Z]+[([Y,Z]),X]+ [([Z,X]),Y]$. Clearly,
if the associator is zero there is
associativity and the ordinary JI is satisfied.

\subsection{Metric Lie algebras}
\label{sec:metric-Lie}{\ }

A Lie algebra is called {\it metric} when it is endowed with an invariant, symmetric
and non-degenerate bilinear form $<\quad,\quad>\,:\, \fg\times\fg \rightarrow
\mathbb{R}$ which defines the scalar product in the $\fg$ vector space.
This means that
\begin{equation}
  \label{metr-Lie}
  X \cdot \left<Y, Z \right> =
  \left<[X\,,\, Y], Z\right> + \left<Y, [X\,,\, Z] \right> = 0 \qquad \forall\; X,Y,Z \in \fg
\end{equation}
{\it i.e.}, the bilinear form $\left<\quad,\quad \right>$ is
`associative' in the sense that $<[Y,X],Z>=<Y,[X,Z]>$.

If $\{ \omega^i \}$ is a basis of the coalgebra $\fg^*$ dual to $\{X_j\}$, the
bilinear form and the invariance condition are expressed as
\begin{equation}
\label{me-Li-Li}
g=g_{ij} \omega^i\otimes\omega^j \quad, \quad C_{li}{}^s\,g_{sj}+ C_{lj}{}^s\,g_{is}=0\quad ,
\end{equation}
where the coordinates of $g$ are $g_{ij}=g(X_i,X_j)\equiv\left<X_i\,,\,X_j\right>$.

Semisimple Lie algebras are metric because their {\it Cartan-Killing metric}
$k(X,Y):=\hbox{Tr}(ad_X ad_Y)$ is non-degenerate and invariant
(trace forms are invariant), and hence defines an inner product;
clearly, $k([Y,X],Z)=k(Y,[X,Z])$ follows from the associativity of
$ad_X \in \hbox{End}\, \fg$. In terms of the structure constants, the
coordinates of the Killing metric $k$ are given by
\begin{equation}
\label{Killing}
k_{ij}:= \textrm{Tr}(ad_{X_i} ad_{X_j})=C_{il}{}^s\,C_{js}{}^l \quad .
\end{equation}
The Killing metric is negative-definite for the compact real form
of a semisimple Lie algebra (when the generators are taken to be
antihermitian). The metricity condition
(\ref{metr-Lie}) is equivalent to the total skewsymmetry
of the structure constants with all indices down, $C_{ijk}:=
C_{ij}{}^l g_{lk}$. Although all Lie algebras are endowed with a
canonical $ad$-invariant inner product given by the Killing form,
the inner product of metric nonsemisimple Lie algebras has to be
imposed as an external structure. For instance, Abelian
Lie algebras are metric relative to any inner product,
since the adjoint action is trivial. For the structure of
metric Lie algebras see
\cite{Medina:85,Med-Rev:85,JMF-Sta:96,Kath-Ul:04,Kath-Ul:03}.
\\

\subsection{Lie algebras and elementary differential geometry on Lie groups}
\label{sec:diffG-on-G}{\ }

Let $G$ be a $r$-dimensional Lie group with elements $g$
parametrized by the their coordinates $g^l,  l=1,\dots,r$ and let
$L_{g'} g =g'g = R_g g'$ ($g',\ g\in G$) be the left and right
actions $G\times G\to G$ with obvious notation. The left (right)
invariant vector fields LIVF $X_i^L(g)$ (RIVF, $X_i^R(g)$) on $G$
reproduce the commutators of the Lie algebra $\fg$:
\begin{equation}
\label{LIVFG} [X_i^L (g), X_j^L (g)] = C_{ij}^k X_k^L (g)
\quad,\quad [X_i^R (g), X_j^R (g)] = - C_{ij}^k X_k^R (g) \quad,
\end{equation}
$i,j,k=1,\dots,r=\hbox{dim}\,\fg$. In terms of the Lie derivative,
the $L$- ($R$-) invariance conditions read\footnote{The superindex
L (R) in the fields refers to the left (right) invariance of them;
LIVF (RIVF) generate right (left) translations.}
\begin{equation}
\label{LRzero} L_{X_j^R (g)} X_i^L (g) = [X_j^R (g), X_i^L (g)] =0
\quad, \quad L_{X_i^L (g)} X_j^R (g) = [X_i^L (g), X_j^R (g)] =0
\quad,
\end{equation}
which simply express that the left and right translations of $G$
commute.

The Lie algebra $\fg$ may be equally described in terms of
invariant forms on the manifold of the associated Lie group $G$.
Let $\omega^{L i}(g)$ be the basis of LI one-forms of $G$ dual to
a basis of $\fg$ given by LIVF ($\omega^{L i}(g)(X^L_j(g)) =
\delta^i_j$). The $\omega^i$ satisfy the Maurer-Cartan (MC)
equations
\begin{equation}
\label{3MC}
 d \omega^{L\,i}(g) = - \frac{1}{2} C^i_{jk}
\omega^{L\,j}(g) \wedge \omega^{L\,k}(g) = -C^i_{jk}\omega^{Lj}(g)\otimes\omega^{Lk}(g) \quad,
\end{equation}
and the JI  $C^\rho_{[i_1 i_2} C_{i_3]\rho}^\sigma = 0$ is now
implied by the nilpotency of $d$, $d^2\equiv 0$.

The exterior derivative $d$ of a $p$-form $\alpha\in \Lambda^p(M)$
on a manifold $M$ is the $(p+1)$-form defined by the functions in
$\mathscr{F}(M)$ obtained by taking arguments on a set of $(p+1)$
arbitrary vector fields $X_1(x),\dots,X_{p+1} (x)$ on $M$. This
expression is given by Palais formula,
\begin{equation}
\label{Palais}
    \begin{aligned}
     (d\alpha)(X_1,&\dots, X_p,X_{p+1}):=\sum^{p+1}_{i=1}(-1)^{i+1}
      X_i\cdot \alpha(X_1,\dots,{\hat X}_i,\dots,X_{p+1})\,+\\
     &\sum_{i<j}(-1)^{i+j}\alpha([X_i,X_j],X_1,\dots,{\hat X}_i,\dots,
     {\hat X}_j,\dots,X_{p+1})\quad .
\end{aligned}
\end{equation}
Thus, if $\alpha$ is a LI form on $G$,
\begin{equation}
\label{3diff}
d\alpha^L (X^L_{i_1},\dots,X^L_{i_{p+1}}) =
\sum_{s<t}(-1)^{s+t}\alpha^L([X^L_{i_s},X^L_{i_t}],
X^L_{i_1},\dots,{\hat X}^L_{i_s},\dots,{\hat
X}^L_{i_t},\dots,X^L_{i_{p+1}}) \quad,
\end{equation}
since $\alpha^L (X^L_1,\dots,{\hat X}^L_i,\dots,X^L_{p+1})$ is
constant and the first term in (\ref{Palais}) is then
zero\footnote{ {}From now on we shall assume that the vector
fields on $G$ generating $\fg$ and their dual one-forms are the left
invariant ones (\emph{i.e.}, $X\in\mathfrak{X}^L(G)$, etc.) and drop
the superindex $L$. Superindices $L,\ R$ will be used to avoid
confusion when both LI and RI objects appear.}.

The MC equations may be written in a more compact way by
introducing the (canonical) $\fg$-valued LI one-form $\theta$ on
$G$, $\theta(g) = \omega^i(g)\circ X_i(g)$, $\theta(X_i)=X_i\,$;
then, the MC equations read
\begin{equation}
\label{3canonMC} d\theta = -\theta \wedge \theta = - \frac{1}{2}
[\theta,\theta] \quad ,
\end{equation}
since the $\fg$-valued bracket of two $\fg$-valued $p,q$-forms
$\alpha,\beta$ is the $\fg$-valued $(p+q)$-form
\begin{equation}
\label{gvalbra}
[\alpha,\beta] := \alpha^i \wedge \beta ^j \circ [X_i,X_j]
= C^k_{ij}\alpha^i \wedge \beta ^j\,\circ X_k \quad ,
\quad [\alpha,\beta]= (-1)^{pq+1} [\beta ,\alpha] \quad,
\end{equation}
so that $[\theta,\theta]^i=2\,C^i_{jk}\omega^{Li}\otimes\omega^{Lj}$.

The transformation properties of the MC forms $\omega^i(g)$ follow
from the action of the Lie derivative on one-forms, $(L_Y
\beta)(X)= Y.\beta(X)-\beta([Y,X])$, which on the LI MC forms
gives
\begin{equation}
 \label{3Lomega}
L_{X_i(g)} \omega^j(g) = -C^j_{ik} \omega^k(g)
\end{equation}
(if $\omega$ were RI, the $r.h.s$ above and those in
eq.~\eqref{3canonMC} would have plus signs). A general LI $p$-form
on $G$ is a linear combination of exterior products of MC forms,
$\displaystyle \alpha(g)=\frac{1}{p!}\alpha_{i_1\dots i_p}
\omega^{i_1}(g)\wedge\dots\wedge\omega^{i_p}(g)$, where
$\alpha_{i_1\dots i_p}$ are constants; for it
\begin{equation}
L_{X_i(g)} \alpha(g) = -\sum_{s=1}^p \frac{1}{p!} C^{i_{s}}_{i k}
\alpha_{i_1 \dots i_p} \omega^{i_1}(g)\wedge\dots\wedge
\widehat{\omega^{i_s}(g)} \wedge \omega^k(g) \wedge\dots\wedge
\omega^{i_p}(g) \quad. \label{3Lalpha}
\end{equation}
In terms of the Lie derivative $L_X$, the invariance of the scalar
product is simply written as $L_X\,g=0$ where now the $\omega^i(g)$
in eq.~\eqref{me-Li-Li} are the MC forms.

\section{Lie algebra cohomology, central extensions and deformations}
\label{sec4.1}

We provide here a summary of the basic Lie algebra cohomology
notions and expressions that will be useful later on when
we consider the cohomology of $n$-Lie algebras (see
\cite{Che-Eil:48,Koszul:50} and {\it e.g.}, \cite{CUP} and
references therein).

\subsection{Main definitions}{\ }

Let $\fg$ be a Lie algebra, $\fg^*$ the dual of the $\fg$ vector
space and $V$ a vector space which is a left $\rho(\fg)$-module {\it
i.e.}, $V$ carries a representation $\rho$ of the Lie algebra $\fg$
on $V$, $\rho(X_i)^A_{.C} \rho (X_j)^C_{.B} - \rho(X_j)^A_{.C} \rho
(X_i)^C_{.B} = \rho([X_i,X_j])^A_{.B}\;$,
$A,B=1,\dots,\hbox{dim}\,V$.

\begin{definition}
\label{alg-coch}
 ({\it $V$-valued $p$-dimensional cochains on a Lie algebra
 $\fg$}) {\ }

  A $V$-valued $p$-cochain $\Omega^p $ on $\fg$ is a skewsymmetric
$p$-linear mapping
\begin{equation}
\label{4cochain}
 \Omega^p: \fg \times \mathop{\dots}\limits^p
\times \fg \to V \quad, \quad \Omega^A = {1 \over p!} \Omega_{i_1
\dots i_p}^A \omega^{i_1} \wedge \dots \wedge \omega^{i_p} \quad,
\end{equation}
where $i_1,\dots,i_p=1,\dots, r= \textrm{dim}\,\fg\,,\,
A=1,\dots,\textrm{dim} V$ and $\{\omega^i \}$ is a basis of $\fg^*$
dual to the one \{$X_j$\} of $\fg\ $; the constants $\Omega_{i_1
\dots i_p}^A$ are the coordinates of the $p$-cochain in the
(non-minimal) basis $\omega^{i_1} \wedge \dots \wedge \omega^{i_p}$.
At this stage, the $\{\omega^i\}$ are simply elements of $\fg^*$
{\it i.e.} linear maps, $\omega^i: \fg \rightarrow \mathbb{R}$, say.
\end{definition}

Under the natural addition law of $V$-valued skewsymmetric covariant
$p$-tensors on $\fg$, the cochains $\Omega^p \in C^p$ of a given
order $p$ form an abelian group which is denoted by $C^{p} (\fg, V
)$. The action of the coboundary operator is given by

\begin{definition} ({\it Coboundary operator} (for the left action
$\rho$ of $\fg$ on $V$))
\label{coho-for-repr}){\ }

 The coboundary operator is the map
$s: \Omega^p\in C^{p} (\fg , V) \mapsto (s\Omega^p) \in C^{p+1}
(\fg, V )$ given by
\begin{equation}
\label{Lie-cob}
\begin{aligned}
(s \Omega^p )^A \, ( X_{1} ,..., X_{p{+}1} ) & :=
\sum_{i{=}1}^{p{+}1} (-)^{i + 1} \rho ( X_i )^A_{.B} \,(
\Omega^{pB} ( X_{1} ,..., {\hat X}_{i},..., X_{p{+}1} ) )\\
 &+ \sum_{{j,k=1 \atop j < k}}^{p{+}1} (-)^{j{+}k}
 \Omega^{pA} ( [X_{j} , X_{k} ] , \ X_{1} ,..., {\hat X}_{j} ,...,
  {\hat X}_{k},..., X_{p{+}1} ) \quad;
\end{aligned}
\end{equation}
$(s\Omega^p)=0$ for $p\geq$ dim$\,\fg$. This defines the {\it Lie
algebra cohomology complex} ($\,C^\bullet(\fg,V) , s)$ for the
representation $\rho$ of $\fg$.
\end{definition}

The structure of this formula is analogous to that for the exterior
derivative $d$ acting on forms $\alpha(x)$ on a finite-dimensional
manifold $M$, eq.~(\ref{Palais}). The only difference is that in
(\ref{Lie-cob}) the $\Omega$ is a $p$-antisymmetric covariant tensor
defined on the vector space of $\fg$, the $X_i$ are vectors and
$\Omega(X_1,\dots,X_p)$ is a number, while in eq.~\eqref{Palais}
$\alpha(x)$ is a $p$-skewsymmetric covariant tensor {\it field} on
the manifold $M$ and the $X_i(x)$ are vector {\it fields} on $M$. Of
course, $s^2\equiv 0$ on any $\Omega^p$, as will be shown in Prop.
\ref{prop:nilp}.

   We give below, for later use, a few examples of the action of
the coboundary operator on cochains of the lower orders:
\begin{equation}
\label{Lie-coh-ex}
\begin{aligned}[b]
s\Omega^0(X_1)=& \rho(X_1)\Omega^0 \; ,\\
s\Omega^1(X_1,X_2) =&\rho(X_1)\Omega^1(X_2)-\rho(X_2)\Omega^1(X_1)
-\Omega^1([X_1,X_2]) \; ;\\
s\Omega^2(X_1,X_2,X_3) =&\rho(X_1)\Omega^2(X_2,X_3)
-\rho(X_2)\Omega^2(X_1,X_3)+\rho(X_3)\Omega^2(X_1,X_2)\\
-&\Omega^2([X_1,X_2],X_3)+\Omega^2([X_1,X_3],X_2)-\Omega^2([X_2,X_3],X_1)\;\\
 s\Omega^3 (X_1,X_2,X_3,X_4) =& \rho(X_1)\Omega^3(X_2,X_3,X_4)-
   \rho(X_2)\Omega^3(X_1,X_3,X_4)  \\
   +& \rho(X_3)\Omega^3(X_1,X_2,X_4)- \rho(X_4)\Omega^3(X_1,X_2,X_3)  \\
   - & \Omega^3([X_1,X_2],X_3,X_4) + \Omega^3([X_1,X_3],X_2,X_4)
   -\Omega^3([X_1,X_4],X_2,X_3)\\
   -&  \Omega^3([X_2,X_3],X_1,X_4)  +  \Omega^3([X_2,X_4],X_1,X_3)
   -\Omega^3([X_3,X_4],X_1,X_2) \;  .
\end{aligned}
\end{equation}

For later generalization to $n$-ary algebras, it is convenient to
rewrite the action of the coboundary operator for $\rho=0$ using the
adjoint derivative. In terms of $ad_X$, it obviously reads
\begin{equation}
\label{Lie-cob-ad}
\begin{aligned}
 (s \Omega^p )\, ( X_{1} ,..., X_{p{+}1} ) = &
\sum_{{1\leq j<k}}^{p+1} (-1)^{j+k}
 \Omega^p ( ad_{X_{j}} X_{k} , \ X_{1} ,..., {\hat X}_{j} ,...,{\hat
X}_{k},..., X_{p{+}1} ) \\
= & \sum_{{1\leq j<k}}^{p+1}  (-1)^j \Omega^p (X_1 ,..., {\hat
X}_{j} ,...,ad_{X_j}X_{k},..., X_{p+1} ) \quad .
\end{aligned}
\end{equation}
Clearly, $s^2=0$ follows from the fact that
\begin{equation}
\label{adJI}
 ad_X  ad_Y  Z - ad_Y  ad_X  Z = ad_{[X,Y]} Z \quad ,
\end{equation}
eq.~(\ref{eq:ad-der}). At present, this is
a mere change of notation, but it will prove
useful later to define the $n$-Lie algebra cohomology since, as
stated, FAs constitute a generalization of Lie algebras based on
extending the adjoint derivative to the ($n>2$)-bracket case.

\begin{definition}
({\it Cocycles, coboundaries and the $p$-th cohomology group})
\label{coho-groups}
{\ }

 A $V$-valued $p$-{\it cochain} $\Omega^p$ (or $\Omega^{pA}$,
making explicit the coordinate index of the $\rho(\fg)$-module $V$)
is called a $p$-{\it cocycle} when
 $s \Omega^p =0$. If a cocycle $\Omega^p$ may be written as
$\Omega^p = s \Omega^{p-1}$ where $\Omega^{p-1}$ is a
$(p-1)$-cochain, the $p$-cocycle is trivial and $\Omega^p$ is called
a $p$-{\it coboundary}. The spaces of $p$-cocycles and
$p$-coboundaries are labelled, respectively, by $Z^p_\rho(\fg,V)$
and $B^p_\rho(\fg,V)$. The $p$-th {\it Lie algebra cohomology group}
$H^p_\rho(\fg,V)$, with values in $V$ for the representation $\rho$,
is defined by the quotient group
\begin{equation}
\label{coho-gr}
H^p_\rho(\fg,V) = Z^p_\rho(\fg,V) /
B^p_\rho(\fg,V)\quad.
\end{equation}
\end{definition}
\noindent
Cohomology groups measure the lack of exactness of the
sequence $\cdots C^{p-1} \mathop{\rightarrow}\limits^s C^p
\mathop{\rightarrow}\limits^s  C^{p+1} \cdots $
 {\it i.e.}, how much $s C^{p-1}\equiv B^p \subset C^p$
differs from the kernel $Z^p$ of $s$ acting on $C^p$; if $H^p=0$,
$s C^{p-1}\equiv B^p=Z^p \equiv \textrm{ker}\,s \subset C^p$.\\

%
%
%

\begin{remark} ({\it Lie algebra cohomology and the MC equations})
\label{re:MCcoho}
{\ }

Let $\omega$ be the $\fg$-valued one-cochain on $\fg$ defined by
$\omega(X_i)=X_i$ (when defined on the group manifold, it is the
canonical one-form $\theta$ on $G$). Then, the expression of the
action of the coboundary operator in  eq.~\eqref{Lie-coh-ex} for the
trivial representation gives $s\omega(X_i,X_j)=-\omega([X_i,X_j])$,
which may be written as
$s\omega(X_i,X_j)=-(\omega\wedge\omega)(X_i,X_j)$ {\it i.e.}, as
$s\omega=-\omega\wedge\omega$. This is the MC equation
\eqref{3canonMC}, which appears here as the expression of the action
of the coboundary operator on the one-cochain $\omega$. In
Sec.~\ref{sec:CEco}, $s$ will be identified with the exterior
derivative, the $\omega^i$ with the left invariant MC forms on $G$
and $d^2\equiv 0$ will correspond to $s^2\equiv 0$ above, which
holds as a consequence of the JI.
\end{remark}

\subsection{Central extensions of a Lie algebra}
\label{ex:central-Lie}{\ }

It is easy to see that the possible central extensions of a Lie
algebra $\fg$ (see {\it e.g.} \cite{CUP}) are characterized by
non-trivial two-cocycles for the trivial representation. Chosen a
basis, the commutators of a central extension $\widetilde{\fg}$ of a
Lie algebra $\fg$ are generically given by
\begin{equation}
\label{c-ext}
 [\tilde{X}_i,\tilde{X}_j]=
 C_{ij}{}^k \ \tilde{X}_k + \Omega^2(X_i,X_j) \Xi\; ,
\end{equation}
where $X_i\in \fg$, $\tilde{X_i}\in \widetilde{\fg}$, $\Xi$ is the
central ($[\tilde{X_i},\Xi]=0$) generator in $\widetilde{\fg}$,
$\Omega^2(X_i,X_j)=-\Omega^2(X_j,X_i)$ and $[X_i,X_j]= C_{ij}{}^k
X_k$ are the commutators of the original, unextended algebra
$\fg$. If (\ref{c-ext}) has to be a Lie algebra, the Lie bracket
in $\widetilde{\fg}$ has to satisfy the JI, and this forces the antisymmetric bilinear
map $\Omega^2$ in (\ref{c-ext}) to be a two-cocycle, as it follows
from the third equality for $s \Omega(X,Y,Z)$ in
(\ref{Lie-coh-ex}) for $\rho=0$.

If $\Omega^2$ is a trivial two-cocycle (a two-coboundary,
$\Omega^2=s\Omega^1$) there exists a basis of $\widetilde{\fg}$ in
which the $\Xi$ can be removed from the $r.h.s$ of (\ref{c-ext}).
This means that $\widetilde{\fg}$ is the direct sum $\widetilde{\fg} =\fg
\oplus \mathbb{R}$ and that the central extension is actually
trivial. Indeed, in this case (see the second equality in
\eqref{Lie-coh-ex}) $\Omega^2(X_i,X_j)= (s\Omega^1)(X_i,X_j)
=-\Omega^1([X_i,X_j])=-C_{ij}{}^k\Omega(X_k)$. It is then
sufficient to define new generators $\tilde{X}'$ of $\widetilde{\fg}$
by the linear combination
$\tilde{X}'_k=\tilde{X}_k-\Omega^1(X_k)\,\Xi$ to obtain
$[\tilde{X}'_i,\tilde{X}'_j]=C_{ij}{}^k (\tilde{X}_k -
\Omega^1(X_k)) = C_{ij}{}^k \tilde{X}'_k\ $,
$\,[\tilde{X}'_i,\Xi]=0$, which shows that
$\widetilde\fg= \fg \oplus \mathbb{R}$ in an explicit manner.
Of course, such a $\widetilde\fg$ is a trivial central extension
irrespective of the basis used to express its commutators.

Thus, the central extensions of a Lie algebra $\fg$ are governed
by the second cohomology group for the trivial representation,
$H^2_0(\fg)$. When $H^2_0(\fg)=0$, all central extensions of $\fg$
are trivial. This is the case for semisimple algebras by
virtue of the Whitehead's Lemma \ref{lem:White} below.

\subsection{Deformations of Lie algebras}
\label{sec:def-Lie}{\ }

 Deformations of algebras were studied long ago by
Gerstenhaber \cite{Gers:63} and  by Nijenhuis and Richardson
\cite{Nij-Rich:67} specifically for the Lie algebra case (see also
\cite{Her:70} for an overview and further early references). The
idea is to find a new Lie algebra `close'
to the original one. This leads to a cohomology
problem and to the notion of {\it stability} or {\it rigidity};
algebras that cannot be deformed are called rigid. The idea of
stability has a clear physical meaning: it is associated
with theories that are not deformable {\it i.e.}, that they are stable
in the sense that they do not change in a qualitative
({\it i.e.}, structural) manner by smoothly changing a
parameter. For instance, since the Poincar\'e
algebra is a deformation of the Galilei one, Einsteinian mechanics
may be looked as a stabilization of Newtonian mechanics (albeit a
partial one, since the Poincar\'e algebra itself may still be
deformed into either of the stable simple de Sitter and anti-de
Sitter algebras, $so(1,4)$ and $so(2,3)$). The deformation process
may also be applied to superalgebras (see \cite{Bine:86}); for
instance, $osp(1|4)$ is a deformation of the $N$=1, $D$=4
superPoincar\'e algebra.
\medskip

Since we shall consider in Sec.~\ref{coho-for-def} deformations of
Filippov algebras, let us review briefly here the problem of
deforming Lie algebras. This provides an example where the
relevant cohomology is the Lie algebra cohomology for a
representation, $\rho=ad$. The aim here is to obtain a deformation
of the original Lie bracket $[X,Y]$, depending of a parameter $t$,
in a way that still defines a Lie algebra. Thus, one looks for a
new Lie bracket $[X,Y]_t$ depending on $t$ ($[X,Y]_{t=0}=[X,Y]$),
defined by
\begin{equation}
\label{def-Lie}
 [X\,,\,Y]_t := [X\,,\,Y] + \sum_{n=1}^{\infty}t^n \alpha_n (X,Y)
 \quad ,
\end{equation}
where the $\alpha_n$ are necessarily bilinear and  skewsymmetric
$\fg$-valued maps, $\alpha_n :\wedge^2 \fg \rightarrow \fg$,
$\alpha_n\in \hbox{Hom}(\wedge^2 \fg, \fg)$,  that must satisfy
the conditions that the JI imposes on the deformed bracket $[X\,,\,Y]_t
\,$, which has to be a Lie algebra bracket. Thus, the dimension of
the deformed Lie algebra is the same as that of the original one;
only its structure is deformed.  The first
order deformation, $[X\,,\,Y]_t := [X\,,\,Y]+t\alpha(X,Y)$, is the
{\it infinitesimal deformation}; not every infinitesimal
deformation is the first-order term of a full deformation. When it
is, the deformation is called {\it integrable}.

\subsubsection{Infinitesimal deformations of a Lie algebra $\fg$}{\ }

   To see how cohomology enters, consider an infinitesimal
deformation {\it i.e.}, eq.~\eqref{def-Lie} neglecting terms of
order $t^2$ and higher. This means that, to find the conditions
that the $\fg$-valued $\alpha_1\equiv\alpha$ has to satisfy, only
$t$-linear terms have to be kept in
\begin{equation}
\label{def-Lie-JI}
[X\,,\,[Y\,,\,Z]_t\,]_t+[Y\,,\,[Z\,,\,X]_t\,]_t+[Z\,,\,[X\,,\,Y]_t\,]_t=0
\quad ,
\end{equation}
which is the JI for the deformed Lie algebra. Using the JI
for the original one $\fg$, the remaining (order $t$) terms give rise to the
condition
\begin{equation}
\label{def-Lie-2-coc}
 ad_X \alpha(Y,Z) + ad_Y \alpha(Z,X) + ad_Z \alpha(X,Y)+
 \alpha (X,[Y,Z])+ \alpha (Y,[Z,X])+ \alpha (Z,[X,Y]) = 0 \; .
\end{equation}
Comparing with the third equality in eq.~\eqref{Lie-coh-ex}, we
see that the skewsymmetric map $\alpha$ has to be a
$\fg$-cohomology two-cocycle for the action $\rho=ad$,
since eq.~\eqref{def-Lie-2-coc} simply reads $s\alpha(X,Y,Z)=0$.
Alternatively, we might take the $l.h.s.$ of \eqref{def-Lie-2-coc} to {\it define}
$s\alpha(X,Y,Z)$ and thus the three-linear skewsymmetric map
$s\alpha$; this would lead us to {\it derive} the cohomology
relevant for the deformation problem (see below).

We may now ask ourselves whether this infinitesimal deformation is
a true one {\it i.e.}, not removable by redefining the basis of
the algebra, in which case the deformation would be trivial. It
turns out that this will be so if the two-cocycle defining the
infinitesimal deformation is a two-coboundary, hence obtained from
a $\fg$-valued one-cochain $\beta$ {\it i.e.}, if
$\alpha(X,Y)=(\delta\beta)(X,Y)=
ad_X\beta(Y)-ad_Y\beta(X)-\beta([X,Y])$ (see the second equality
in eq.~\eqref{Lie-coh-ex}). To check this, it is sufficent to
redefine new generators $X'=X-t\beta(X)$ to find that, with
$[X,Y]=Z$ and neglecting terms of order higher than $t$,
\begin{equation}
\begin{aligned}
& [X',Y']_t=  [X\,,\,Y] + t\alpha(X,Y)-t\,ad_X \beta(Y) +
ad_Y \beta(X) + (t \beta([X,Y])-t\beta([X,Y])\,) \\
&\qquad \qquad  = Z -t  \beta(Z) +
t(\alpha(X,Y)-(\delta\beta)(X,Y)\,)= Z' \quad ,
\end{aligned}
\end{equation}
since the $t$-term in the last equality is zero for a two-coboundary
$\alpha=\delta\beta$. Note that, again, enforcing that the
redefinition manifestly exhibits the undeformed character of the
algebra would show that the two-cocycle $\alpha$ is given by a
specific expression in terms of a $\fg$-valued linear map
(one-cochain) $\beta$, which then would determine the form of the
two-coboundary generated by $\beta$. In this way, by studying the
deformations of Lie algebras (and then by generalizing the action of
the coboundary operator on higher order cochains, etc.) we would
have been led naturally to the cohomology complex ($\,C^\bullet(\fg
,\fg) , s$) of Def. \ref{coho-for-repr} for the representation $ad$
and, by extension, for an arbitrary one $\rho$. The JI is the key
ingredient in the definition of the Lie algebra
cohomology\footnote{The Lie algebra cohomology complex relevant for
the general (not necessarily central) extensions of $\fg$ by an
abelian Lie algebra kernel $\mathscr{A}$ (see {\it e.g.} \cite{CUP})
is also the cohomology for a representation $\rho$. We shall not
look in this section at the cohomology of $\fg$ from this point of
view, which would lead us to the coboundary operator of
Def.~\ref{coho-for-repr}. This approach will be followed in
Sec.~\ref{extLeib}, where the extension problem for Leibniz algebras
$\mathscr{L}$ is discussed. Sec.~\ref{extLeib} can be readily
translated to the Lie algebra case by adding the requirement of
skewsymmetry where appropriate.}.

The outcome of this discussion is that if $H^2_{ad} (\fg,\fg)=0$
all two-cocycles are trivial, $\fg$ cannot be deformed and hence
the Lie algebra is rigid \cite{Gers:63, Nij-Rich:67}; therefore,
by Witehead's Lemma \ref{lem:White} below, all semisimple algebras are
stable. Note, however, that $H^2_{ad}(\fg,\fg)=0$ is a {\it
sufficient} condition for the rigidity of a Lie algebra, but not a
necessary one; there are examples of Lie algebras which do not
satisfy this condition {\it i.e.}, with $H^2_{ad}(\fg,\fg)\not=0\,$
-therefore, not semisimple- that nevertheless are rigid
\cite{Rich:67}.

\subsubsection{Higher order deformations}
\label{sec:high-ord-def}
{\ }

Moving further one encounters an obstruction when $H^3\not=0$, which
prevents the integrability of the infinitesimal deformation. To see
it, let us consider \eqref{def-Lie} up to the $t^2$ term,
\begin{equation}
 \label{deform-2}
[X,Y]_t = [X,Y] + t\alpha_1(X,Y) + t^2\alpha_2 (X,Y) \quad .
\end{equation}
Imposing the condition that $[X,Y]_t$ now satisfies the JI to order
$t^2$, and taking into account that $\alpha_1(X,Y)$ gives an
infinitesimal deformation and hence it is a two-cocycle,
$s\alpha_1(X,Y,Z)=0$, we obtain from eq.~\eqref{deform-2} that the
$t^2$ terms have the form
\begin{equation}
\label{def-sec-or}
\begin{aligned}
\alpha_1(X,\alpha_1(Y,Z))+&\alpha_1(Y,\alpha_1(Z,X))+\alpha_1(Z,\alpha_1(X,Y))
\\
+ad_X\alpha_2(Y,Z) +& ad_Y\alpha_2(Z,X)+ ad_Z\alpha_2(X,Y)\\
+\alpha_2 (X,[Y,Z]) +&\alpha_2 (Y,[Z,X])+ \alpha_2 (Z,[X,Y])\\
\equiv \gamma(X,Y,Z) +& s\alpha_2(X,Y,Z)  \quad,
\end{aligned}
\end{equation}
where the first line above defines a $\fg$-valued three-linear map,
\begin{equation}
\label{def-3-co}
\gamma(X,Y,Z) :=
\alpha_1(X,\alpha_1(Y,Z))+ \hbox{cycl.} (X,Y,Z) \quad ,
\end{equation}
which is fully skewsymmetric and hence a three-cochain, $\gamma\in
C^3(\fg,\fg)$, and where the remaining ($\alpha_2$) terms in
eq.~\eqref{def-sec-or} have
been identified as $s\alpha_2$ using eq.~\eqref{Lie-coh-ex}. It is
then seen that $\gamma\in Z^3_{ad}(\fg,\fg)$: indeed, since
$\alpha_1\in Z^2_{ad}(\fg,\fg)$, $s\gamma(X_1,X_2,X_3,X_4)=0$ with
$s\gamma$ given by \eqref{Lie-coh-ex} for $\rho=ad$ and
$\gamma=\Omega^3$.

Hence, when the cocycle $\gamma$ is actually a three-coboundary,
$\gamma=s\alpha'$ where $\alpha'$ is a two-cochain, it is sufficient
to take $\alpha_2=-\alpha'$ in \eqref{deform-2} to see in
\eqref{def-sec-or} that the JI is fulfilled up to second order. We
can now continue in the same way up to terms of order $t^3$ to find
at this stage that again a three-cocycle appears potentially
obstructing the deformation to that order, and so on. Thus, all the
obstructions that prevent expanding an infinitesimal deformation to
a one parameter family of deformations are elements of
$H^3_{ad}(\fg,\fg)$; as all three-cocyles have to be trivial, it
follows that only if $H^3_{ad}(\fg,\fg)=0$ all obstructions vanish
and every infinitesimal deformation $\alpha_1\in
Z^2_{ad}(\fg,\fg)$ is integrable. \\

\subsection{Coordinates expression of the coboundary
operator action for the trivial representation}{\ }

  It is convenient to have the action of the coboundary operator
of eq.~\eqref{Lie-cob-ad} expressed in terms of the coordinates
 $\Omega_{i_1 \dots i_p}=\Omega(X_{i_1},\dots,X_{i_p})$
 of the $p$-cochain on which it acts. Eq.~\eqref{Lie-cob-ad} gives
\begin{equation}
\nonumber
\begin{aligned}
(s\Omega)_{i_1 \dots i_{p+1} } &=
 \sum_{s,t=1 \atop s<t}^{p+1} (-1)^{s+t}C_{i_s i_t}{}^k \,\Omega_{k
 i_1\dots {\hat i_s}\dots {\hat i_t}\dots i_{p+1}} \cr
 &= \frac{1}{2}\sum_{s,t=1 \atop s<t}^{p+1}(-1)^{s+t}\epsilon_{i_s
 i_t}^{j_1 j_2} C_{j_1 j_2}{}^k \frac{1}{(p-1)!}\epsilon_{i_1
 \dots {\hat i_s}\dots {\hat i_t}\dots i_{p+1}}^{j_3\dots j_{p+1}}
 \Omega_{k j_3\dots j_{p+1}} \cr
 &= -\frac{1}{2}\frac{1}{(p-1)!} C_{j_1 j_2}{}^k \, \Omega_{k
 j_3\dots j_{p+1}} \sum_{s,t=1 \atop s<t}^{p+1} (-1)^{s+t+1}\epsilon_{i_s
 i_t}^{j_1 j_2} \epsilon_{i_1 \dots {\hat i_s}\dots {\hat i_t}\dots i_{p+1}}^{j_3\dots j_{p+1}}
\end{aligned}
\end{equation}
which, using
\begin{equation}
\label{eps-relat1}
\sum_{s,t=1 \atop s<t}^{p+1}
(-1)^{s+t+1}\epsilon_{i_s
 i_t}^{j_1 j_2} \epsilon_{i_1 \dots {\hat i_s}\dots {\hat i_t}\dots i_{p+1}}^{j_3\dots j_{p+1}}
 = \epsilon^{j_1 \dots j_{p+1}}_{i_1 \dots i_{p+1}} \quad ,
\end{equation}
which follows by developing the determinant that defines the
antisymmetric Kronecker symbol
\begin{equation}
\epsilon^{i_1\ldots i_{p}}_{j_1\ldots j_{p}}= {\rm det}\left(
\begin{aligned}
\label{defep}
\delta^{i_1}_{j_1} & \cdots & \delta^{i_1}_{j_p}
\\ \vdots & & \vdots \\ \delta^{i_p}_{j_1} & \cdots &
\delta^{i_p}_{j_p}
\end{aligned}
\right) \quad ,
\end{equation}
gives the coordinates of the ($p+1$)-cochain (in fact, coboundary)
\begin{equation}
\label{cob-in-coord}
  (s\Omega)_{i_1 \dots i_{p+1}}=
-\frac{1}{2}\frac{1}{(p-1)!}\epsilon^{j_1 \dots j_{p+1}}_{i_1
\dots i_{p+1}} C_{j_1 j_2}{}^k \, \Omega_{k j_3\dots,j_{p+1}}
\quad.
\end{equation}

\subsection{Chevalley-Eilenberg formulation of Lie algebra cohomology}
\label{sec:CEco}{\ }

  The Chevalley-Eilenberg (CE) formulation \cite{Che-Eil:48}
makes use of the `localisation' process which allows us to obtain
invariant tensor {\it fields} on the group manifold $G$ by
left-translating the appropriate algebraic objects at the unit
element $e\in G$ to an arbitrary group element $g$. In this way, the
expression (\ref{Lie-cob}) for the Lie algebra coboundary operator
and eq. (\ref{Palais}) for the exterior derivative become equivalent
if we we take $M$ as the manifold of the group $G$ associated to
$\fg$ and convert the multilinear applications into invariant tensor
fields on the group manifold. This is done by identifying $\fg$ with
$T_e(G)$, the vector tangent space at the unit element, and by
moving from $e$ to an arbitrary group element $g\in G$ by the left
translation $L_g$ generated by $g$. In this way, the basis elements
$X_i$ of the vector space $\fg$ become LI vector fields $X_i(g)$ on
the group manifold $G$ satisfying the Lie algebra commutation
relations, $[X_i(g),X_j(g)]=C^i{}_{jk}X^k(g)$, the $\omega^i \in
\fg^*$ become the LI Maurer-Cartan (MC) one-forms $\omega^i(g)$ on
$G$ which characterize the Lie algebra from the MC equations dual
point of view (eqs. (\ref{3MC})), and the $p$-cochains become LI
$p$-forms on the group manifold $G$.

 Let $V=\mathbb{R}$, so that $\rho$ is trivial. Then, the first
term in (\ref{Lie-cob}) is not present and, on LI one-forms, $s$
and $d$ act in the same manner. Since there is a one-to-one
correspondence between $p$-antisymmetric maps on $\fg$ and LI
$p$-forms on $G$, a $p$-cochain in $C^p(\fg,\mathbb{R})$ is given
by a LI $p$-form on $G$, which in terms of the MC forms may be
written as
\begin{equation}
\label{p-CE-coc}
\Omega^p (g)= {1\over p!} \Omega_{i_1\dots i_p}
\omega^{i_1}(g) \wedge \dots \wedge \omega^{i_p}(g) \; ,
\end{equation}
with constant coordinates $\Omega_{i_1\dots i_p}$. The Lie algebra
cohomology coboundary operator $s$ for the trivial representation
$\rho=0$ is thus given by the exterior
differential $d$ acting on LI $p$-forms on the group manifold $G$,
which are the $p$-cochains (the explicit dependence of the forms
$\Omega(g)$, $\omega^i(g)$ on $g$ will be omitted henceforth).
\medskip

{\it Coordinates expression of the cocycle condition, $\rho=0$}.

A $p$-cochain $\Omega$ is a $p$-{\it cocycle} for the trivial representation
of $\fg$ if $(s\Omega)_{i_1 \dots i_{p+1} } =0$ {\it i.e.},
when its coordinates satisfy  (see eq.~\eqref{cob-in-coord}),
\begin{equation}
\label{cocy-coord}
 C_{[j_1 j_2}{}^k \Omega_{i_1\dots i_{p-1} ] k} =0 \; .
\end{equation}
This condition also follows from eq.~\eqref{p-CE-coc} by imposing $d\Omega=0$
and using the MC eqs.~\eqref{3MC}.
The nilpotency of $s$ can be easily checked using the CE formulation
of the Lie algebra cohomology:

\begin{proposition}
({\it Nilpotency of the coboundary operator $s$})
\label{prop:nilp}
{\ }

The Lie algebra cohomology operator $s$ is nilpotent, $s^2=0$.
\end{proposition}

\begin{proof}
First, we notice that a $V$-valued $p$-cocycle may be written as
$\Omega^{pA} (g)= {1\over p!} \Omega^A_{i_1\dots i_p}
\omega^{i_1}(g) \wedge \dots \wedge \omega^{i_p}(g)$. Thus,
looking at the definition of the coboundary operator in
(\ref{Lie-cob}) and taking into account (\ref{3diff}), we see that
$s$ may be written as
\begin{equation}
(s)_{.B}^A = \rho(X_i)_{.B}^A \omega^i + \delta_B ^A d   \quad,
\quad (s= d+ \rho(X_i) \omega^i) \quad.
\end{equation}
Then, since $d^2$=0, the proposition follows from the fact that
\begin{equation}
\begin{array}{rl}
s^2 = & (\rho(X_i) \omega^i + d) (\rho(X_j) \omega^j + d) =
\rho(X_i) \rho(X_j) \omega^i \wedge \omega^j + \rho(X_i) \omega^i
d + \rho(X_j) d (\omega^j) + d^2
\\[0.25cm]
= & \displaystyle - {1\over 2} \rho(X_j) C^j_{l k} \omega^l \wedge
\omega^k + {1\over 2} [\rho(X_i), \rho(X_j) ]  \omega^i \wedge
\omega^j = 0 \quad,
\end{array}
\end{equation}
where use has been made of the fact that $\rho$ is a
representation of $\fg$, $[\rho(X_i), \rho(X_j) ]$
 = $C_{ij}{}^k\rho(X_k)$.
\end{proof}

In spite of $s$ being given by $d$, the Lie algebra CE cohomology is
in general different from the de Rham cohomology: a closed LI
$p$-form $\alpha$ on $G\,$ -{\it i.e.}, a $p$-cocycle- may be de
Rham exact and hence de Rham trivial without being CE-trivial. This
is because a de Rham exact form or de Rham coboundary,
$\alpha=d\beta$, will not be a CE coboundary if the potential
($p-1$)-form $\beta$ of $\alpha$ is not a CE cochain {\it i.e.}, is
not a LI form\footnote{This is  the case {\it e.g.}, for certain
forms on superspace group manifolds (`rigid' superspaces) which
appear in M- and super-$p$-brane theory (see \cite{AzTo:89}). This
is not surprising due to the absence of global considerations in the
fermionic, Grassmann odd sector of supersymmetry. The Lie algebra
cohomology notions may be, in fact, extended to superalgebras (see
{\it e.g.} \cite{Lei:75,Sch-Zha:97,Al-Mi-Ru:01} and refs.
therein).}.
 Nevertheless, for $G$ compact the
Lie and de Rham cohomologies coincide, $H_{DR}(G) =
H_0(\fg,\mathbb{R})$, as stated by the following

\begin{proposition} ({\it de Rham vs. CE cohomology})
\cite{Che-Eil:48}
\label{prop:CEcom}

 Let $G$ be a compact and connected Lie group. Every de Rham
 cohomology class on $G$ contains one and only one bi-invariant form.
 The bi-invariant forms span a ring isomorphic to $H_{DR}(G)$.
\end{proposition}

\begin{example}
Let $\fg$ be the abelian two-dimensional algebra. The
corresponding Lie group is $\mathbb{R}^2$, hence de Rham
trivial. However, the translation algebra $\mathbb{R}^2$ has a
non-trivial Lie algebra second cohomology group; it admits a
one-parameter family of non-trivial central extensions, all
isomorphic to the three-dimensional Heisenberg-Weyl algebra.
\end{example}

\subsection{Whitehead's lemma for vector valued cohomology }

\begin{lemma}({\it Whitehead's lemma} \cite{Jac:79})
\label{lem:White} {\ }\\
Let $\fg$ be a finite-dimensional
semisimple Lie algebra over a field of characteristic zero and let
$V$ be a finite-dimensional irreducible $\rho (\fg)$-module such
that $\rho ( \fg) V \not = 0$ ($\rho $ \emph{non-trivial}). Then,
\begin{equation}
H_{\rho }^{q} ( {\fg} , V ) = 0 \quad \forall\, q \geq 0  \quad .
\end{equation}
If $q = 0$, the non-triviality of $\rho $ and the irreducibility
imply that $\rho ( \fg ) \cdot v = 0 \ (v \in V)$ holds only for
$v = 0$.
\end{lemma}

\begin{proof}
Since $\fg$ is semi-simple, the Cartan-Killing metric $k_{ij}$ is
invertible, $k^{ij} k_{jk} = \delta^i_k$. Let ${\tau}$ be the
operator on the space of $p$-cochains ${\tau}:C^p(\fg,V)\to
C^{p-1}(\fg,V)$ defined by
\begin{equation}
( {\tau} \Omega )_{i_1 \dots i_{p-1}}^A = k^{ij} \rho(X_i)^A_{.B}
\Omega_{j i_1 \dots i_{p-1}}^B \quad. \label{whitehead1}
\end{equation}
It is not difficult to check that on cochains the Laplacian-like
operator $(s\tau + \tau s)$ gives\footnote{ For instance, for a
two-cochain eq. (\ref{whitehead2}) reads
\begin{equation}
\begin{array}{rl}
[(s\tau + \tau s)\Omega]^A_{ij} = & g^{kl} \rho(X_i)^A_{.B}
\rho(X_k)^B_{.C} \Omega^C_{l j}
  - g^{kl} \rho(X_j)^A_{.B} \rho(X_k)^B_{.C} \Omega^C_{l i}
  - g^{kl} \rho(X_k)^A_{.B} C_{ij}^m \Omega^B_{l m}
\\
+ & g^{kl} \rho(X_{k})^A_{.B} \rho(X_{l})^B_{.C} \Omega^C_{i j} +
g^{kl} \rho(X_{k})^A_{.B} \rho(X_{i})^B_{.C} \Omega^C_{j l} +
g^{kl} \rho(X_{k})^A_{.B} \rho(X_{j})^B_{.C} \Omega^C_{l i}
\\
- & g^{kl} \rho(X_{k})^A_{.B} C_{i j}^m \Omega^B_{m l} - g^{kl}
\rho(X_{k})^A_{.B} C_{l i}^m \Omega^B_{m j} - g^{kl}
\rho(X_{k})^A_{.B} C_{j l}^m \Omega^B_{m i}
\\
= & g^{kl} [\rho(X_i),\rho(X_k)]^A_{.B}\Omega^B_{l j} - g^{kl}
[\rho(X_j),\rho(X_k)]^A_{.B}\Omega^B_{l i} + I_2(\rho)^A_{.B}
\Omega^B_{i j}
\\
- & g^{kl} \rho(X_{k})^A_{.B} C_{l i}^m \Omega^B_{m j} -g^{kl}
\rho(X_{k})^A_{.B} C_{j l}^m \Omega^B_{m i} = I_2(\rho)^A_{.B}
\Omega^B_{i j}\quad.
\end{array}
\end{equation}
}
\begin{equation}
[(s {\tau} + {\tau} s) \Omega] ^A _{i_1 \dots i_p} = \Omega_{i_1
\dots i_p} ^B I_2(\rho)^A_{.B} \quad, \label{whitehead2}
\end{equation}
where $I_2(\rho)^A_{.B}= k^{ij}(\rho(X_i) \rho(X_j))^A_{.B}$ is
the quadratic Casimir in the representation $\rho$. By Schur's
lemma it is proportional to the unit matrix. Hence, applying
(\ref{whitehead2}) to a cocycle $\Omega\in Z^p_\rho(\fg,V)$ we
find
\begin{equation}
s {\tau} \Omega = \Omega I_2(\rho) \; \Rightarrow\; s( {\tau}
\Omega I_2(\rho)^{-1} ) = \Omega\quad. \label{whitehead3}
\end{equation}
Thus, $\Omega$ is the coboundary generated by the cochain ${\tau}
\Omega I_2(\rho)^{-1} \in C^{p-1}_\rho(\fg,V)$.
\end{proof}

For semisimple algebras and $\rho=0$ we also have $H_0^1=0$ and
$H_0^2=0$, but already $H_0^3\ne 0$; the three-cocycle
$\Omega_{i_2 i_2 i_3}$ is given by the fully antisymmetric
structure constants $C_{i_1 i_2 i_3}$.

\subsection{Simple Lie algebras, invariant polynomials, and cohomology}
\label{sec6.1}{\ }

Let $\fg$ be simple.
By virtue of Whitehead's Lemma, only the $\rho=0$ case is interesting
in the simple case since, if $\fg$ is simple, $H^p_\rho(\fg,V)=0$
for $\rho$ non-trivial. The non-trivial cohomology groups are
related to the primitive ({\it i.e.}, not reducible to products,
see Def. \ref{def:primitive}) symmetric invariant tensors
\cite{Rac:50,Gel:50,Kle:63,Gru.Rai:64,Bie:63,Per.Pop:68,Oku-Pat:83,Oku:82}
on $\fg$, which in turn determine Casimir elements in the
universal enveloping algebra ${\mathcal U}({\fg})$.

\begin{definition} ({\it Symmetric and invariant polynomials on}
$\fg$) \label{def:sympol} {\ }

A symmetric polynomial on $\fg$ is given by a covariant symmetric
LI tensor.
\end{definition}

In terms of the MC forms on the group manifold $G$, a symmetric
invariant polynomial is given by a LI covariant tensor field on
$G$, $k=k_{i_1...i_m}\omega^{i_1}\otimes...\otimes\omega^{i_m}$
with symmetric constant coordinates $k_{i_1...i_m}$.

The polynomial $k$ is said to be a {\it bi-invariant} (or
$ad$-invariant) symmetric polynomial if it is also
right-invariant, {\it i.e.}\ if $L_{X_l}k=0$ $\forall\,
X_l\in\mathfrak{X}^L(G)$. Using (\ref{3Lalpha}) we find
that
\begin{equation}
L_{X_l} k=0\ \Rightarrow\
C_{li_1}{}^s \,k_{si_2...i_m}+C_{li_2}{}^s \,k_{i_1s...i_m}+\dots+
C_{li_m}{}^s \,k_{i_1...i_{m-1}s} =0\quad . \label{kinva}
\end{equation}
Since the coordinates of $k$ are given by $k_{i_1\dots i_m}=
k(X_{i_1},\dots,X_{i_m})$, eq.~\eqref{kinva} is equivalent to
stating that $k$ is $ad$-invariant\footnote{For a LI $p$-form
$\alpha$, the equivalent to \eqref{adinfi} and \eqref{3diff} show
that a bi-invariant form on $G$ is closed.}, {\it i.e.},
\begin{equation}
 \label{adinfi}
  k([X_l,X_{i_1}],\dots,X_{i_m})+k(X_{i_1},[X_l,X_{i_2}],\dots,X_{i_m})
 +\dots + k(X_{i_1},\dots,[X_l,X_{i_m}])=0
\end{equation}
or, equivalently,
\begin{equation}
  \label{adfini}
 k(Ad\,g\, X_{i_1},\dots,Ad\,g\, X_{i_m})=k(X_{i_1},\dots,X_{i_m})\quad,
\end{equation}
from which eq. (\ref{adinfi}) follows by taking the derivative
$\partial /\partial g ^l$ in $g=e\in G$.

The invariant symmetric polynomials just described can be used to
construct Casimir elements of the enveloping algebra ${\mathcal
U}(\mathcal{\fg})$ in the following way

\begin{proposition}
\label{prop6.1} ({\it Higher order Casimirs})

Let $k$ be a symmetric invariant tensor. Then
$k^{i_1\dots i_m}X_{i_1}\dots X_{i_m}$ (coordinate indices of $k$
raised using the Killing metric), is a Casimir of $\fg$ of order
$m$, $[k^{i_1\dots i_m}X_{i_1}\dots X_{i_m}, Y]=0$ $\forall\, Y\in
\fg $.
\end{proposition}
\begin{proof}
\begin{eqnarray}
[k^{i_1\dots i_m}X_{i_1}\dots X_{i_m},X_s] &=&
\sum^m_{j=1}k^{i_1\dots i_m}X_{i_1}\dots [X_{i_j},X_s]\dots
X_{i_m}
\nonumber\\
&=& \sum^m_{j=1}k^{i_1\dots i_m}X_{i_1}\dots C^t_{i_j s} X_t\dots
X_{i_m}=0 \label{casimir}
\end{eqnarray}
by (\ref{kinva}).
\end{proof}

An easy way of obtaining symmetric ($ad$-)invariant
polynomials (used \emph{e.g.}, in the construction of
characteristic classes) is given by

\begin{proposition}
\label{prop6.2}

Let $X_i$ denote now a representation of $\fg$.
Then, the symmetrized trace
\begin{equation}
     k_{i_1\dots i_m}=\hbox{sTr}(X_{i_1}\dots X_{i_m})   \label{symtra}
\end{equation}
defines a symmetric invariant polynomial.
\end{proposition}

\begin{proof}
$k$ is symmetric by construction and the $ad$-invariance is
obvious since $Adg\,X:=gXg^{-1}$.
\end{proof}
The simplest illustration of a trace-invariant form is the
non-singular Killing metric (eq.~\eqref{Killing}) for a simple Lie algebra $\fg$;
its associated Casimir is the (second order) Casimir $I_2$.

\begin{example}
\label{example6.2}

Let $\fg=su(n)$, $n\ge2$, and let $X_i$ be
(hermitian) matrices in the defining representation. Then
\begin{equation}
     \hbox{sTr}(X_iX_jX_k)\propto 2\hbox{Tr}(\{ X_i,X_j\} X_k)=d_{ijk}\quad,
\label{deij}
\end{equation}
using that, for the $su(n)$ algebra, $\hbox{Tr}(X_k)=0$ and with
generators normalized to $\hbox{Tr}(X_iX_j)=\frac{1}{2}
\delta_{ij}\,$, the anticommutator is given \cite{MF-S-W:68} by
 $\{X_i,X_j\}=c\delta_{ij} +d_{ijl}X_l$ with $c=1/n$. The symmetric
polynomial $d_{ijk}$ leads to the third order Casimir $I_3$; for
$su(2)$ only the Killing metric $k_{ij}=\delta_{ij}$ and the
quadratic Casimir $I_2$ exist.
\end{example}

\begin{example}
\label{example6.3}

In the case $\fg=su(n)$, $n\ge 4$, we have a
fourth order polynomial
\begin{equation}
 \hbox{sTr}(X_{i_1}X_{i_2}X_{i_3}X_{i_4})\propto d_{(i_1i_2l}d_{li_3)i_4}
+2 c\delta_{(i_1i_2}\delta_{i_3)i_4} \quad,      \label{kcuatro}
\end{equation}
where $(\ )$ indicates symmetrization. The first term
$d_{(i_1i_2l}d_{li_3)i_4}$ gives the fourth order Casimir $I_4$.
It generalizes easily to higher $n$ by nesting more $d$'s, leading
to the Klein \cite{Kle:63} form of the $su(n)$ Casimirs.  The last
part of \eqref{kcuatro} (see \cite{Azc.Mac.Mou.Bue:97}) is clearly
the symmetrized product of two copies of the order two Casimir
$I_2 $ and thus it is not {\it primitive}.
\end{example}

\begin{definition}
(\emph{Primitive symmetric invariant polynomials})
\label{def:primitive}{\ }

A symmetric invariant polynomial $k_{i_1\dots i_m}$ on
$\fg$ is called primitive if it is not of the form
\begin{equation}
   k_{i_1\dots i_m}=k^{(p)}_{(i_1\dots i_p}k^{(q)}_{i_{p+1}\dots i_m)}
\ ,\quad p+q=m\quad, \label{primitive}
\end{equation}
where $k^{(p)}$ and $k^{(q)}$ are two lower order symmetric
invariant polynomials.
\end{definition}

Of course, we could also have considered eq.~\eqref{kcuatro} for
$su(3)$, but then it would not have led to a fourth-order
primitive polynomial, since $su(3)$ is a rank 2 algebra and has
only two primitive invariant polynomials (of ranks 2 and 3).
Indeed, $d_{(i_1i_2l}d_{li_3)i_4}$ is not primitive for $su(3)$
and can be written in terms of $\delta_{i_1i_2}$ as in
(\ref{primitive}) (see, \emph{e.g.}, \cite{Sud:90}; see also
\cite{Azc.Mac.Mou.Bue:97} and refs. therein). In general, for
a compact simple algebra of rank $l$ there are $l$ invariant primitive
polynomials and Casimirs
\cite{Rac:50,Gel:50,Kle:63,Gru.Rai:64,Bie:63,Per.Pop:68,Oku-Pat:83,Oku:82}
and, as shown below, $l$ primitive Lie algebra cohomology
cocycles (see Table in Sec.~\ref{sec:compGtable}).
{\ }

\begin{lemma}($p$-cocycles as skewsymmetric invariant polynomials)
\label{le:inv-cocycle}
{\ }

Let $\Omega^p$ (eq.~\eqref{p-CE-coc}) be a $p$-cocycle. Then, it defines
a skewsymmetric invariant polynomial of rank $p$.
\begin{proof}
The statement follows since invariance means $L_{X_i}\Omega^p=0$ {\it i.e.},
\begin{equation}
\label{inv-cocycle}
\sum^{p}_{k=1} C_{ij_s}{}^k\Omega_{j_1\dots,j_{s-1} k j_{s+1}\dots j_p}=0
\qquad \mathrm{or} \qquad C_{i[j_1}{}^k\,\Omega_{j_2\dots j_p]k}=0 \quad,
\end{equation}
which is satisfied by any $p$-cocycle.
\end{proof}
\end{lemma}

\subsection{Cocycles from invariant polynomials}
\label{sec6.2}{\ }

To make explicit the connection between the invariant polynomials
and the non-trivial cocycles of a simple Lie algebra $\fg$ let us
use the particular case $\fg=su(n)$ as a guide. On the manifold of
the  $SU(n)$ group one can construct the {\it odd} $p$-form
\begin{equation}
\label{HOCa}
    \Omega^p=\frac{1}{p!} \hbox{Tr}(\theta\wedge
    \mathop{\cdots}\limits^p
     \wedge\theta)
     \quad ,
\end{equation}
where again $\theta=\omega^i X_i$ is the canonical form and we
take $\{X_i\}$ in the defining representation; $p$ has to be odd
since otherwise $\Omega$ would be zero by virtue of the cyclic
property of the trace and the anticommutativity of one-forms.

\begin{proposition}
\label{prop6.3}

The odd LI form $\Omega^p$ on $G$ in (\ref{HOCa})
is a non-trivial (CE) Lie algebra cohomology $p$-cocycle.
\end{proposition}

\begin{proof}
Since $\Omega^p$ is LI by construction, it is sufficient to show
that $\Omega^p$ is closed and that it is not the differential of a
LI ($p-1$)-form {\it i.e.}, that it is not a coboundary. By using
(\ref{3canonMC}) we get
\begin{equation}
       d\Omega^p \sim \hbox{Tr}(\theta\wedge
\mathop{\cdots}\limits^{p+1}
       \wedge\theta)=0\quad,            \label{HOCb}
\end{equation}
since $p+1$ is even. Suppose now that $\Omega^p=d\Omega^{p-1}$,
with $\Omega^{p-1}$ LI. Then $\Omega^{p-1}$ would be of the form
(\ref{HOCa}) and hence zero because ($p-1$) is also even.
\end{proof}

All non-trivial $p$-cocycles in $H_0^p(su(n),\mathbb{R})$ are of
the form (\ref{HOCa}). The fact that these forms are closed and
de Rham non-exact
($SU(n)$ is compact) allows us to use them to construct
Wess-Zumino-Witten \cite{Wit:84,Wit:83} terms on the group
manifold (see also \cite{Azc.Izq.Mac:90}).

Let us set $p=2m-1$. Since $\theta=\omega^i\circ X_i$, the form
$\Omega^p$ expressed in coordinates is
\begin{eqnarray}
\Omega^p&=&\frac{1}{q!}\hbox{Tr}(X_{i_1}\dots X_{i_{2m-1}})
\omega^{i_1}\wedge\dots\wedge\omega^{i_{2m-1}} \nonumber\\ & &
\propto \hbox{Tr}([X_{i_1},X_{i_2}][X_{i_3},X_{i_4}]\dots
[X_{i_{2m-3}},X_{i_{2m-2}}]X_{i_{2m-1}})
\omega^{i_1}\wedge\dots\wedge\omega^{i_{2m-1}} \nonumber\\ & &
=\hbox{Tr}(X_{l_1}\dots X_{l_{m-1}}X_\sigma)C^{l_1}_{i_1i_2} \dots
C^{l_{m-1}}_{i_{2m-3}i_{2m-2}}
\omega^{i_1}\wedge\dots\wedge\omega^{i_{2m-2}}\wedge\omega^\sigma
      \quad .                              \label{HOCc}
\end{eqnarray}
We see here how the order $m$ symmetric invariant polynomial
$\hbox{Tr}(X_{l-1}\dots X_{l_{m-1}}X_\sigma)$ appears in this
context. There is symmetry in ${l_1}\dots l_{m-1}$ (and hence
$\hbox{Tr}(X_{l_1}\dots X_{l_{m-1}}X_\sigma)$ is fully symmetric)
because there is antisymmetry on the $i$ indices due to the
$\omega^{i}$'s.

   Conversely, the following statement holds:

\begin{proposition} ({\it Cocycles associated with invariant
polynomials})
\label{cocy-polyn} \\
Let $k_{i_1\dots i_m}$ be a
symmetric invariant polynomial. Then, the polynomial
\begin{equation}
\label{HOCe}
    \Omega_{\rho i_2\dots i_{2m-2}\sigma}=C^{l_1}_{j_2j_3}\dots
    C^{l_{m-1}}_{j_{2m-2}\sigma}k_{\rho l_1\dots l_{m-1}}
     \epsilon^{j_2\dots j_{2m-2}}_{i_2\dots i_{2m-2}}
\end{equation}
is skewsymmetric and defines \cite{Azc.Mac.Mou.Bue:97} the closed form ((2m-1)-cocycle)
\begin{equation}
\Omega^{2m-1}=\frac{1}{(2m-1)!}\Omega_{\rho i_2\dots
i_{2m-2}\sigma}\omega^\rho\wedge
    \omega^{i_2}\wedge\dots\wedge\omega^{i_{2m-2}}\wedge\omega^\sigma
\quad. \label{HOCd}
\end{equation}
\end{proposition}

\begin{proof}
To check the complete skew-symmetry of $\Omega_{\rho i_2\dots
i_{2m-2}\sigma}$ in (\ref{HOCe}) it is sufficient, due to the
Levi-Civita symbol, to show the antisymmetry in $\rho$ and
$\sigma$. This is done by using the invariance of $k$
(\ref{kinva}) and the symmetry properties of $k$ and $\epsilon$ to
rewrite $\Omega_{\rho i_2\dots i_{2m-2}\sigma}$ as the sum of two
terms. The first one,
\begin{eqnarray}
       \sum^{m-2}_{s=1}\epsilon^{j_2\dots
j_{2s}j_{2s+1} j_{2m-2}j_{2s+2}\dots j_{2m-3}}_{i_2\dots i_{2m-2}}
  k_{\rho l_1\dots l_{s-1}l_{m-1}l_s\dots
l_{m-2}\sigma} \; \qquad \qquad &&  \nonumber\\
  \qquad\qquad  \times \;   C^{l_1}_{j_2j_3}\dots C^{l_s}_{j_{2s}j_{2s+1}}
     C^{l_{m-1}}_{l_sj_{2m-2}}C^{l_{s+1}}_{j_{2s+2}j_{2s+3}}
    \dots  C^{l_{m-2}}_{j_{2m-4}j_{2m-3}} &&         \label{HOCf}
\end{eqnarray}
vanishes due to the standard JI, and the second one is
\begin{equation}
    \Omega_{\rho i_2\dots i_{2m-2}\sigma}=
     -\epsilon^{j_2\dots j_{2m-2}}_{i_2\dots i_{2m-2}}
k_{\sigma l_1\dots
 l_{m-1}} C^{l_1}_{j_2j_3}\dots C^{l_{m-1}}_{j_{2m-2}\rho}
    =-\Omega_{\sigma i_2\dots i_{2m-2}\rho}\quad .
\label{HOCg}
\end{equation}
To show that $d\Omega=0$ we make use of the fact that any
bi-invariant form is closed. Since $\Omega$ is LI by construction,
we only need to prove its right-invariance, but
\begin{equation}
\Omega\propto \mbox{Tr}(\theta \wedge
\mathop{\cdots}\limits^{2m-1} \wedge \theta) \label{HOCh}
\end{equation}
is obviously RI since, under a right translation, the canonical
one-form transforms by $R_g^*\theta = Ad g^{-1} \theta$ (see Prop.
\ref{prop6.2}).
\end{proof}

Without discussing the origin of the invariant polynomials for the
different compact simple Lie algebras
\cite{Rac:50,Gel:50,Kle:63,Gru.Rai:64,Bie:63,Per.Pop:68,Oku-Pat:83,Oku:82,
Azc.Mac.Mou.Bue:97}, we may conclude that to each symmetric primitive invariant
polynomial of order $m$ we can associate a non-trivial Lie
algebra cohomology $(2m-1)$-cocycle (see \cite{Azc.Mac.Mou.Bue:97}
for practical details). In fact, this one-to-one correspondence
is a famous result of Chevalley \cite{Che:50,Che-Bor:55} after a conjecture
of A. Weil\footnote{The correspondence between invariant polynomials
and cocycles reflects its transgression character \cite{Cart:50}. See
further \cite{Kost:97} in the context of the Hopf-Koszul-Samelson theorem;
we thank G. Pinczon for pointing out this reference to us.}.

    A question that might immediately arise (from the above
explicit construction) is whether it could be extended further
since, from the $l=\hbox{rank}\,\fg$ primitive invariant polynomials
that exist for a simple $\fg$, we may obtain an arbitrary number of
non-primitive ones (see eq.~(\ref{primitive})) by taking
symmetrized products of primitive polynomials and applying the
above formulae. This question is
answered negatively by Prop.~\ref{prop:25} and
Cor.~\ref{cor:26} below \cite{Azc.Mac.Mou.Bue:97}.

\begin{proposition}

\label{prop:25} Let $k_{i_1\dots i_m}$ be a symmetric
$G$-invariant polynomial. Then,
\begin{equation}
\epsilon^{j_1 \dots j_{2m}}_{i_1\dots i_{2m}} C^{l_1}_{j_1
j_2}\ldots C^{l_m}_{j_{2m-1} j_{2m}} k_{l_1\dots l_m}=0 \quad.
\label{simplelemma}
\end{equation}
\end{proposition}

\begin{proof}
By replacing $C^{l_m}_{j_{2m-1} j_{2m}}k_{l_1\dots lm}$ in the
$l.h.s$ of (\ref{simplelemma}) by the other terms in (\ref{kinva})
we get
\begin{equation}
\epsilon^{j_1 \dots j_{2m}}_{i_1 \dots i_{2m}} C^{l_1}_{j_1
j_2}\ldots C^{l_{m-1}}_{j_{2m-3} j_{2m-2}}
(\sum_{s=1}^{m-1}C^{k}_{j_{2m-1} l_s} k_{l_1\ldots l_{s-1} k
l_{s+1}\ldots l_{m-1}\,j_{2m}}) \quad,
\end{equation}
which is zero due to the JI (note that \eqref{simplelemma} also
follows from \eqref{HOCb}).
\end{proof}

\begin{corollary} (Primitive invariant polynomials {\it vs.}
cocycles) \label{cor:26} {\ }

Let $k$ be a non-primitive symmetric invariant polynomial
(\ref{primitive}). Then the $(2m-1)$-cocycle $\Omega$
associated to it by eq.~\eqref{HOCe} is zero \cite{Azc.Mac.Mou.Bue:97}.
\end{corollary}

Thus, to a {\it primitive} symmetric $m$-polynomial it is possible
to associate uniquely a Lie algebra $(2m-1)$-cocycle. Conversely,
we also have the following

\begin{proposition} ({\it Invariant polynomials from cocycles})
\label{prop:27} {\ }

Let $\Omega^{(2m-1)}$ be a primitive cocycle.
The $l$ polynomials $t^{(m)}$ given by
\begin{equation}
{t}^{i_1 \dots i_{m}}=[\Omega^{(2m-1)}]^{j_1 \dots j_{2m-2} i_m}
C^{i_1}_{j_1 j_2} \dots C^{i_{m-1}}_{j_{2m-3} j_{2m-2}}
\label{seIIxii}
\end{equation}
are invariant, symmetric and primitive (see
\cite[Lemma 3.2]{Azc.Mac.Mou.Bue:97}).
\end{proposition}

This converse proposition relates the cocycles of the Lie algebra
cohomology to Casimirs in the enveloping algebra ${\mathcal
U}(\fg)$. The polynomials in (\ref{seIIxii}) have certain
advantages (for instance, they have all traces equal to zero)
\cite{Azc.Mac.Mou.Bue:97} over other more conventional ones such
as \emph{e.g.}, those obtained in (\ref{symtra}). A list of the
ranks of the invariant polynomials and associated cocycles for
the simple algebras is given for convenience in Table 1
below.

\subsection{Simple compact algebras: cocycles and Casimir operators}
\label{sec:compGtable}{\ }

We have seen that the Lie algebra cocycles may be expressed in
terms of LI forms on the group manifold $G$ (Sec.~\ref{sec:CEco}).
The equivalence of the Lie algebra (CE) cohomology and the de Rham
cohomology in Prop.~\ref{prop:CEcom} in the simple compact case is
specially interesting because, since all primitive cocycles are
odd, compact groups behave as products of odd spheres from the
point of view of real homology. This leads to a number of simple
and elegant formulae concerning the Poincar\'e polynomials, Betti
numbers, the primitive invariant tensors and the non-trivial Lie
algebra cocycles, etc. The following table summarizes many of
these results. Details on the topological properties of Lie groups
may be found in
\cite{Car:36,Pon:35,Hop:41,Che:50,Che-Bor:55,Sam:52,Cole:58,Bor:65,Bot:79,Boy:93};
for book references see \cite{Wey:46,Hod:41,Gre.Hal.Van:76,CUP}.

\begin{equation*}
\nonumber
\begin{array}{cccc}
\fg & \mbox{dim}\,\fg & \mbox{Orders $m_i$ of\ invariants\ and\
Casimirs} & \mbox{Orders\ $p=2m_i-1$ of $\fg$-cocycles}
\\
\hline \hline A_l & (l+1)^2 -1\ [l\geq 1] & 2,3,\dots,l+1         &
3,5,\dots,2l+1
\\
B_l & l(2l+1)\ [l\geq 2]    & 2,4,\dots,2l          & 3,7,\dots,4l-1
\\
C_l & l(2l+1)\ [l\geq 3]    & 2,4,\dots,2l          & 3,7,\dots,4l-1
\\
D_l & l(2l-1)\ [l\geq 4]    & 2,4,\dots,2l-2,l      &
3,7,\dots,4l-5,2l-1
\\
G_2 & 14                & 2,6                   & 3,11
\\
F_4 & 52                & 2,6,8,12              & 3,11,15,23
\\
E_6 & 78                & 2,5,6,8,9,12          & 3,9,11,15,17,23
\\
E_7 & 133               & 2,6,8,10,12,14,18     &
3,11,15,19,23,27,35
\\
E_8 & 248               & 2,8,12,14,18,20,24,30 &
3,15,23,27,35,39,47,59
\end{array}
\end{equation*}
\centerline{Table 1.  {\it Orders (ranks) of the
primitive invariant tensors}} \centerline{{\it and associated
cocycles for the compact simple Lie algebras} ($i=1\dots,l$)}\\
\medskip

The structure constants always define a non-trivial
three-cocycle, the one associated with the non-degenerate Cartan-Killing
form; this is why in the table above there is always a 2 (3) in the
third (last) column. It is worth noticing that the sum of the orders
($2m_i-1$) of the different cocycles in each entry in the last
column is the dimension of the corresponding algebra {\it i.e.},
\begin{equation}
\nonumber \sum_{i=1}^{\hbox{rank}\; l} (2m_i-1)  =\hbox{dim}\;\fg
\quad.
\end{equation}
This is so because the sum of the dimensions of the odd spheres has
to be equal to the dimension of the group manifold.

The cohomology ring of the compact simple $\fg$ is generated
by the $l$ primitive cocycles of order $(2m_i-1)$.

\section{Leibniz (Loday's) algebras and cohomology}
\label{sec:Loday}

\subsection{Main definitions}
\label{sec:Leib-defs}{\ }

Loday's algebras, or Leibniz algebras \cite{Lod:93,Lod-Pir:93,Lod-Pir:96}
in Loday's terminology (see also \cite{Al-Ay-Om:05, Cuv:94}),
are a non-skewsymmetric ($[X,Y]\not=-[Y,X]$) version of Lie algebras.
More specifically, a Leibniz algebra $\mathscr{L}$
is given by the following

\begin{definition} (({\it Left}) {\it Leibniz algebra}){\,}

A (left) Leibniz algebra (LA) is a vector space $\mathscr{L}$ endowed
bith a bilinear operation $\mathscr{L}\times \mathscr{L}\rightarrow
\mathscr{L}$ that satisfies the relation ({\it left Leibniz identity})
\begin{equation}
\label{left-Lei-alg}
 [X,[Y,Z]]=[[X,Y],Z]+[Y,[X,Z]]   \qquad \forall X,Y,Z \in \mathscr{L} \quad ,
\end{equation}
\end{definition}
\noindent which is no longer equivalent to the
$[X,[Y,Z]]+[Y,[Z,X]]+[Z,[X,Y]]=0$ Lie algebra JI.

  Although $[X,Y]\not= -[Y,X]$ for a Leibniz algebra, some anticommutativity
is left in the double bracket, since eq.~\eqref{left-Lei-alg} implies
\begin{equation}
\label{lLanti}
[[X,Y],Z]= - [[Y,X],Z] \quad .
\end{equation}
A Leibniz algebra that satisfies $[X,X]=0\;\forall X\in \mathscr{L}$
has an anticommutative bracket and therefore is a Lie algebra since
then eq.~\eqref{left-Lei-alg} is the JI. Obviously, any Lie algebra
$\fg$ is also a Leibniz algebra. Many Lie algebra notions such as
those of subalgebra, quotient by a two-sided ideal, etc. extend
trivially to the Leibniz algebra case (other notions, such as that
of representation, require more care since the Leibniz bracket is
not skewsymmetric, see Sec. \ref{extLeib} below). For instance, a
{\it homomorphism $\phi$ of Leibniz algebras} is a homorphism of the
underlying vector spaces such that $\phi([X,Y])=[\phi(X),\phi(Y)]$.
The elements of the form $[X,X]$  (and those of the form
$[X,Y]+[Y,X]$) are central in $\mathscr{L}$. They generate a
two-sided ideal $\mathscr{I}$ of $\mathscr{L}$, and the quotient
$\mathscr{L}/\mathscr{I}$ is a Lie algebra.

Since the Leibniz bracket is not antisymmetric, one
has to distinguish between left (above) and {\it right LA}s, for which
the left derivation property of eq.~\eqref{left-Lei-alg} is
replaced by the right one,
\begin{equation}
\label{right-Lei-alg}
[[X,Y],Z]=[[X,Z],Y]+[X,[Y,Z]]  \qquad \forall X,Y,Z \in \mathscr{L}{}^{r}
\end{equation}
or {\it right Leibniz identity} (which again becomes the JI in the
anticommutative case), and eq.~\eqref{lLanti} by
\begin{equation*}
[X,[Y,Z]]= - [X,[Z,Y]] \quad.
\end{equation*}

The images of the left and right {\it adjoint maps}, $ad$ and
$ad{}^r$, are derivations of the corresponding Leibniz algebras
since, in terms of them, the above two defining equations read
\begin{equation}
\label{l/r-der-Lei}
ad_X [Y,Z]=[ ad_X\,Y,Z]+[Y, ad_X\,Z] \quad, \quad [X,Y] ad{}^r_Z= [X ad{}^r_Z,Y]+[X,Y ad{}^r_Z] \quad ,
\end{equation}
where here we have added the superscript $r$ and further
located $ad^r$ at the right to emphasize its right action character.
  Thus, the left and the right adjoint derivatives, which give
rise to the left and right Leibniz identities, are essentially
different, as they are in general any left and right actions
$\rho(X)$. In contrast with the Lie algebra case, taking the
opposite $X\mapsto -X$ is not an antiautomorphism of $\mathscr{L}$.
Nevertheless, left and right Leibniz algebras are still related in
the following sense: if the bracket $[X,Y]$ satisfies
\eqref{left-Lei-alg} and hence defines a left Leibniz algebra, the
bracket $[X,Y]'=[Y,X]$ satisfies eq.~\eqref{right-Lei-alg} and
defines a right one. The centre $Z(\mathscr{L})$ of a Leibniz
algebra may be defined as the kernel of $ad$.
\medskip

An interesting question is whether there is some analogue of Lie's
third theorem for Leibniz algebras {\it i.e.}, whether there
is some kind of generalization of the notion of Lie group
(some kind of `Leibniz group') so that Leibniz algebras are
its corresponding tangent structures. Such an object has been dubbed by Loday
\cite{Lod:93} as a {\it coquecigrue}\footnote{After Rabelais
imaginary animal in his {\it Gargantua} [{\it coquecigrue} = coq (cock) +
cigu\"e (hemlock)  + grue (crane)], an embodiment of absolute
absurdity.}. However, the problem of integrating Leibniz algebras in
general remains open, although progress has been made with
the introduction of {\it Lie racks} \cite{Kin:07}. There is, however,
a geometrical interpretation of certain {\it three}-Leibniz
algebras, the Lie-triple systems (see Sec.~\ref{sec:triple}),
as tangent spaces: see \cite{Yama:57}.
\medskip

The notion of LA may be extended to include the $Z_2$-graded or
Leibniz superalgebra case \cite{Al-Ay-Om:05}, although this will not be treated here.
Further, we shall only consider {\it left} LAs from now on.\\

\subsection{Extensions of a Leibniz algebra $\mathscr{L}$ by an abelian one
 $\mathscr{A}$ }
\label{extLeib}{\ }

By definition (see, {\it e.g.} \cite{CUP} for the Lie algebra case),
$\widetilde{\mathscr{L}}$ is said to be an extension a Leibniz
algebra $\mathscr{L}$ by an abelian one $\mathscr{A}$ if
$\mathscr{A}$ is a two-sided ideal of $\widetilde{\mathscr{L}}$ and
$\widetilde{\mathscr{L}}/\mathscr{A}=\mathscr{L}\;$ {\it i.e.},
there is an exact homomorphisms sequence
\begin{equation*}
0\longrightarrow \mathscr{A} \longrightarrow \widetilde{\mathscr{L}}
\longrightarrow \mathscr{L}\longrightarrow 0 \quad .
\end{equation*}
To determine a solution $\widetilde{\mathscr{L}}$, we need the data
of the extension problem, namely $\mathscr{L}$ and a
$\mathscr{L}$-module $\mathscr{A}$ which is the abelian LA, so that
left {\it and} right actions $\rho: \mathscr{L}\mapsto
\hbox{End}\mathscr{A}$ of $\mathscr{L}$ on $\mathscr{A}$ are given
(we shall not write $\rho, \rho^r$ hereafter to distinguish the left
or right actions since they will be distinguished by their
location). In contrast with the Lie algebra case, both left and
right actions on $\mathscr{A}$ are needed as the Leibniz bracket is
not anticommutative in general.

  Let us assume that $\widetilde{\mathscr{L}}$ is a solution to the extension
problem  and let us characterize its elements by $(A,\tau(X))$, where
$\tau$ is a trivializing section injecting $\mathscr{L}$ into
$\widetilde{\mathscr{L}}$. Since $\mathscr{A}$ is abelian, it is clear that
the left and right actions given by
\begin{equation}
\label{actLonA}
\rho(X)A:=[\tau(X),A]=[\pi^{-1}(X),A]   \quad, \quad A\rho(X):=[A,\tau(X)]=[A, \pi^{-1}(X)]\quad,
\end{equation}
where $\pi^{-1}(X)$ is the fibre over $X$ ($\pi$ is the projection of
$\widetilde{\mathscr{L}}$ onto $\mathscr{L}$), are well defined; indeed
all the elements $\pi^{-1}(X)\in \widetilde{\mathscr{L}}$ ({\it i.e.} those in the class
of the element $\tau(X)$ of $\widetilde{\mathscr{L}}$ that defines the
element $X\in\mathscr{L}$) give rise to the same (left or right) action
as $\tau(X)$. Let us write $\tau(X)=(0,X)$
and then denote the elements of $\widetilde{\mathscr{L}}$ by $(A,X)$. Then,
the bracket in $\widetilde{\mathscr{L}}$ is defined by
\begin{equation}
\label{compoext}
[(A_1,X_1)\,,\,(A_2,X_2)]=(\,\rho(X_1)A_2+A_1\rho(X_2)+\omega^2(X_1,X_2)\;,\;
[X_1,X_2]\;)
\end{equation}
where, in contrast with the Lie algebra case, the antisymmetry of
the two-cochain $\omega^2(X_1,X_2)$
is not required. The presence of $\omega^2(X_1,X_2)\,$,
$\;\omega^2: \mathscr{L}\otimes \mathscr{L}\rightarrow \mathscr{A}$,
indicates that $\mathscr{L}$ is not necessarily a subalgebra of $\widetilde{\mathscr{L}}$
since, in general, $\tau$ is not a homomorphism of Leibniz algebras,
\begin{equation}
\label{dos-Lei-co}
[\tau(X_2),\tau(X_2)]-\tau([X_1,X_2])=\omega^2(X_1,X_2) \quad ;
\end{equation}
in fact,
\begin{equation*}
[\tau(X_1),\tau(X_2)] = [(0,X_1), (0,X_2)]=
(\omega^2(X_1,X_2)\,,\,[X_1,X_2])\not= (0, [X_1,X_2])=\tau([X_1,X_2]) \quad .
\end{equation*}

Our objective is to determine the conditions that the bracket in
eq.~\eqref{compoext} must satisfy for $\widetilde{\mathscr{L}}$ to
be a LA. These follow by imposing the Leibniz identity
(eq.~\eqref{left-Lei-alg}), namely
\begin{equation}
\label{Le-le-ext}
\begin{aligned}
&[(A_1,X_1)\,,\,[(A_2,X_2)\,,\,(A_3,X_3)]\,] = [\,[(A_1,X_1)\,,\,(A_2,X_2)]\,,\,(A_3,X_3)]\\
& \qquad \qquad \qquad \qquad\qquad+ [(A_2,X_2)\,,\,[(A_1,X_1)\,,\,(A_3,X_3)]\,] \quad .
\end{aligned}
\end{equation}
A simple calculation shows that eq.~\eqref{Le-le-ext} implies the relations
\begin{equation}
\begin{aligned}
\label{reprcond-a}
&\rho(X_1)(\rho(X_2)A_3)=\rho([X_1,X_2])A_3+\rho(X_2)(\rho(X_1)A_3) \quad, \\
&\rho(X_1)(A_2\rho(X_3))=(\rho(X_1)A_2)\rho(X_3)+A_2\rho([X_1,X_3]) \quad,  \\
&A_1\rho([X_2,X_3])=(A_1\rho(X_2))\rho(X_3)+\rho(X_2)(A_1\rho(X_3)) \quad,
\end{aligned}
\end{equation}
plus the condition
\begin{equation}
\begin{aligned}
\label{Leib-2-cocy}
(s\omega^2)(X_1,X_2,X_3):= &
\rho(X_1)\omega^2(X_2.X_3)-\rho(X_2)\omega^2(X_1,X_3)-\omega^2(X_1,X_2)\rho(X_3)+\\
 - & \omega^2([X_1,X_2],X_3])-\omega^2(X_2,[X_1,X_3])+\omega^2(X_1,[X_2,X_3])=0 \quad ,
\end{aligned}
\end{equation}
where the position of the $\rho$'s indicates the left or
right action and the first equality defines $s\omega^2$; then,
$s\omega^2=0$ characterizes $\omega^2$ as a two-cocycle. Later we shall
write expressions such as the above one in the form
\begin{equation}
\begin{aligned}
\label{Leib-2-cocy-b}
(s\omega^2)(X_1,X_2,X_3):=&X_1\cdot\omega^2(X_2.X_3)-X_2\cdot\omega^2(X_1,X_3)-\omega^2(X_1,X_2)\cdot X_3+\\
&-\omega^2([X_1,X_2],X_3])-\omega^2(X_2,[X_1,X_3])+\omega^2(X_1,[X_2,X_3])=0 \quad ,
\end{aligned}
\end{equation}
where $X\cdot\,$ and $\,\cdot X$ stand, respectively, for the left
and right action of $\rho(X)$.

The left and right actions in eqs.~\eqref{reprcond-a}, viewed in
$\widetilde{\mathscr{L}}$, correspond to adjoint ones. They read,
in the same order,
\begin{equation}
\begin{aligned}
\label{reprcond-b}
[X,[Y,A]]&=[[X,Y],A]+[Y,[X,A]] \quad,  \\
[X, [A,Y]]&= [[X,A],Y]+[A,[X,Y]] \quad,  \\
[A, [X,Y]]&=[[A,X],Y]+[X,[A,Y]] \quad,
\end{aligned}
\end{equation}
where {\it e.g.}, we have written $[X,A]$ for the left action of
$\tau(X)$ on the ideal $A$ of $\widetilde{\mathscr{L}}$.
The above equations are the statement that the actions $[X,A]$,
$[A,X]$ define a Leibniz {\it representation} \cite{Lod-Pir:93,Lod-Pir:96}
{\it of $\mathscr{L}$ on $\mathscr{A}$}; note that both the left
and right actions intervene in the definition. Eqs.~\eqref{reprcond-b}
are the analogue of the left module for an associative algebra (there is
a corresponding set of equations for the
analogue of a right module \cite{Lod:93}).

   The sum of the last two equations in \eqref{reprcond-b} gives ({\it cf.} \eqref{lLanti})
\begin{equation}
\label{lLanti-b}
[[X,A],Y]=-[[A,X],Y] \quad .
\end{equation}
A representation such that
\begin{equation}
\label{Leib-sym}
[X,A]=-[A,X] \qquad \forall X\in \mathscr{L}\;,\,\;A\in\mathscr{A}
\end{equation}
is satisfied is called a {\it symmetric representation} of $\mathscr{L}$. When
the representation is symmetric, all eqs.~\eqref{reprcond-b}
are equivalent among themselves. In particular, a representation
of a Lie algebra is a symmetric representation in the Leibniz algebra
sense\footnote{Notice that if we had $A\rho(X)=-\rho(X)A$ and
$\omega^2$ were skewsymmetric, the first equality in
eq.~\eqref{Leib-2-cocy} would coincide with the third
equation in \eqref{Lie-coh-ex} which defines the action of the
Lie algebra coboundary operator on a two-cochain.}.

 Different sections $\tau$ define images $\tau(X)\in \widetilde{\mathscr{L}}$
of $X\in \mathscr{L}$ that differ in an element of the abelian ideal $\mathscr{A}$.
Thus, if $\tau'$ is another trivializing section it follows that
$\tau'(X)=\tau(X)+\omega^1(X)$, where $\omega^1$ is a linear
map $\omega^1:\mathscr{L}\rightarrow \mathscr{A}$. In analogy with
\eqref{dos-Lei-co}, we may write $[\tau'(X_1),\tau'(X_2)]=\tau'([X_1,X_2])+
{\omega'}^2(X_1,X_2)$ and, comparing now with the result of
computing $[\tau(X_1)+\omega^1(X_1)\,,\,\tau(X_2)+\omega^1(X_2)]$,
we immediately obtain that the two bilinear maps ${\omega'}^2, \omega^2$ are
related by
\begin{equation}
\label{Leib-two-cob}
{\omega'}^2(X_1,X_2)-\omega^2(X_1,X_2)=
\rho(X_1)\omega^1(X_2)+\omega^1(X_2)\rho(X_1)-\omega^1(X_1,X_2)
= (s\omega^1)(X_1,X_2) \; ,
 \end{equation}
where the last equality defines $(s\omega^1)(X_1,X_2)$. Thus, the
extensions are characterized by actions that satisfy the {\it
representation conditions} (eqs.~\eqref{reprcond-b}), which
characterize $\mathscr{A}$ as a (left) $\mathscr{L}$-module, and by
the bilinear maps that satisfy eq.~\eqref{Leib-2-cocy} {\it i.e.},
by the {\it two-cocycles} $\omega^2 \in
Z^2_\rho(\mathscr{L},\mathscr{A})$. Using different trivializing
sections to characterize the same extension
$\widetilde{\mathscr{L}}$ corresponds to taking two-cocycles that
are equivalent according to eq.~\eqref{Leib-two-cob} {\it i.e.},
that differ in the two-coboundary $\omega^2_{cob}= s\omega^1\in
B^2_\rho(\mathscr{L},\mathscr{A})$ generated by the one-cochain
$\omega^1$.

Therefore, the inequivalent extensions $\widetilde{\mathscr{L}}$ of a
LA $\mathscr{L}$ by an abelian one $\mathscr{A}$ which is a
$\mathscr{L}$-module for the representation \eqref{reprcond-b} are
classified \cite{Lod-Pir:93,Lod:93} by the second cohomology group
\begin{equation*}
H^2_\rho(\mathscr{L},\mathscr{A})=
Z^2_\rho(\mathscr{L},\mathscr{A})/B^2_\rho(\mathscr{L},\mathscr{A}) \quad.
\end{equation*}
As with Lie algebras, this general situation has two important special subcases:
\begin{itemize}
\item
When $\omega^2=0$ (or it is equivalent to zero, $\omega^2=s\omega^1$),
eq.~\eqref{dos-Lei-co} shows that
$\tau:\mathscr{L}\rightarrow \widetilde{\mathscr{L}}$ is a
homomorphism of Leibniz algebras; then, the extension {\it splits} (this is the case
that for Lie algebras corresponds to the semidirect sum). This solution to
the extension problem always exists, since only requires the definition
of the actions $\rho$ ({\it i.e.}, eqs.~\eqref{reprcond-b}), which are known
as they are necessary data for the extension problem.
\item
When $\omega^2$ is non-trivial but $\rho=0$, all eqs.~\eqref{reprcond-b} are
trivial and then the second cohomology group
$H_0^2(\mathscr{L},\mathscr{A})=
Z_0^2(\mathscr{L},\mathscr{A})/B_0^2(\mathscr{L},\mathscr{A})$ characterizes
the possible central extensions of $\mathscr{L}$ by $\mathscr{A}$.
\end{itemize}
Clearly, the case where both $\rho$ and  $\omega^2$ are trivial corresponds to
the direct sum $\widetilde{\mathscr{L}}=\mathscr{A}\oplus\mathscr{L}$ of
Leibniz algebras, which does not contain any structure beyond that in the $\mathscr{L}$
and $\mathscr{A}$ LA summands themselves.

\subsection{Leibniz algebra cohomology}
\label{Leib-coho}{\ }

We are now in a position to generalize the previous results to define
higher order cochains and the {\it Leibniz algebra cohomology
complex} ($C^\bullet(\mathscr{L},\mathscr{A}),s$)

\begin{definition}({\it Leibniz $p$-cochains}){\ }

An $\mathscr{A}$-valued $p$-cochain is a $p$-linear map
$\omega^p: \otimes^p \mathscr{L} \rightarrow \mathscr{A}$. The space of $p$-cochains
will be denoted $C^p(\mathscr{L},\mathscr{A})$.
\end{definition}

\begin{definition} ({\it Coboundary operator for Leibniz algebra cohomology})
\label{2-Leib-cob}{\ }

The coboundary operator $s$ is the map $s: C^p(\mathscr{L},\mathscr{A})\rightarrow
C^{p+1}(\mathscr{L},\mathscr{A})$ defined by
\begin{equation}
\label{Leib-s}
\begin{aligned}
(s \omega^p)(X_1,\dots,X_{p+1}):= &\sum_{i=1}^p (-1)^{i+1}\rho(X_i)\omega^p(X_1,\dots,{\widehat{X}}_i,\dots,X_{p+1})\\
+ & \sum_{{i,j \atop 1\leq i < j}}^{p+1} (-1)^{i}\omega^p(X_1,\dots,{\widehat{X}}_i,\dots,[X_i,X_j],\dots,X_{p+1})\\
+ & (-1)^{p+1}\omega^p(X_1,\dots,X_p)\rho(X_{p+1}) \quad.
\end{aligned}
\end{equation}
\end{definition}
We see that both the left and the right actions of $\mathscr{L}$ on
the $\mathscr{A}$-valued cochains intervene in the definition of $s$.
When $\omega^p=\omega^2$, eq.~\eqref{Leib-s}
reproduces eq.~\eqref{Leib-2-cocy}.
Eq.~\eqref{Leib-s} is the expression of the coboundary operator
for the Leibniz algebra cohomology \cite{Lod-Pir:93,Lod:93,Cas-Lod-Pir:02}
(there given for a right LA) and \cite{Da-Tak:97}.
It is proved in \cite{Lod-Pir:93} that $s^2=0$, so that
($C^\bullet(\mathscr{L},\mathscr{A}),s$) is indeed a LA
cohomology complex.

When the representation is a {\it symmetric} one (eq.~\eqref{Leib-sym}),
the last term in eq.~\eqref{Leib-s} may be included in the first one as one more
contribution to the sum.  Then, eq.~\eqref{Leib-s} adopts same form as the
action of the Lie algebra coboundary operator, namely
\begin{equation}
\begin{aligned}
\label{Leib-s-symm}
(s \omega^p)(X_1,\dots,X_n):= &\sum_{i=1}^{p+1} (-1)^{i+1}\, \rho(X_i)\omega^p(X_1,\dots,{\widehat{X}}_i,\dots,X_{p+1})\\
 + &  \sum_{i,j \atop 1\leq i < j}^{p+1} (-1)^i \, \omega^p(X_1,\dots,{\widehat{X}}_i,\dots,[X_i,X_j],\dots,X_{p+1}) \quad ,
\end{aligned}
\end{equation}
which formally coincides with eq.~\eqref{Lie-cob} (see eq.~\eqref{Lie-cob-ad});
here, the proof that $s^2=0$ follows by analogy to the Lie algebra cohomology
case.

The above two expressions also reproduce those in \cite{Gau:96}
(when particularized to the $n=2$ case) and in \cite{Da-Tak:97}.

\subsection{Deformations of Leibniz algebras}
{\ }

As for Lie algebras, one may ask the question of applying the
cohomology complex in Def.~\ref{2-Leib-cob} to the problem of
deforming Leibniz algebras. It is clear that the first order
deformations will be classified by $H^2(\mathscr{L},\mathscr{L})$.
We shall not discuss this here since it will be done later directly
for the $n$-ary generalization of $\mathscr{L}$, the $n$-Leibniz
algebras $\mathfrak{L}$ (Secs.~\ref{n-Leib-coho} and
\ref{coho-for-def}) from which the $\mathscr{L}$ case follows for
$n=2$. It is also possible to consider Leibniz deformations of Lie
and $n$-Lie algebras, since these are in particular Leibniz and
$n$-Leibniz; see \cite{Fia-Man:08} and \cite{Az-Iz:10}.
\medskip

\section{$n$-ary algebras}
\label{n-ary-alg}
{\ }

   To extend the ordinary Lie algebra $\fg$ structure to the case of
brackets with $n>2$ entries, we have to define first the $n$-ary
brackets and then generalize the JI. The Lie bracket is naturally
extended to an $n$-ary bracket by requiring it to be a
multilinear application
 \begin{equation*}
\label{n-ary-bra}
 [\;,\;,\;\mathop{\cdots}\limits^n \;,\;,\;]:\,
 \fH\times\mathop{\cdots}\limits^n \times \fH\rightarrow \fH \; ,
\end{equation*}
where $\fH$ is a generic $n$-ary algebra. The next step is
specifying the consistency condition to be satisfied by this $n$-ary
bracket. As described in the Introduction, there are two natural
interpretations of the JI; these lead to two main $n>2$ different
generalizations of the Lie algebra structure obtained, respectively,
by extending to the $n$-ary brackets the antisymmetrization of the
nested Lie brackets [\;\,,\;[\;\,,\,\;]] of the $n=2$ case, or the
derivation character of the adjoint map (intermediate possibilities
between these two exist: see \cite{Gau:96b, Vin2:98}).
\medskip

\section{Higher order or generalized Lie algebras (GLAs)}
\label{sec:GLA}{\ }

GLAs make emphasis on the associativity of the composition of the
elements in their multibracket. These algebras were introduced
independently in \cite{AzBu:96,AzPePB:96a, AzPePB:96b} and
\cite{JLL:95,Han-Wac:95,Gne:95,Gne:97}; we refer to
Sec.~\ref{sec:not-con} for the terminology. General linear
antisymmetric `$\Omega$-algebras' were considered earlier
\cite{Kurosh,Bara-Bur:75}.

An obviously skewsymmetric higher order multilinear bracket is
provided by the following

\begin{definition} ({\it Higher order generalized Lie bracket} or {\it multibracket})
\label{def:multibracket}{\ }

 Let $X_i$ be arbitrary associative
operators, $i=1,\dots,r$. A {\it multibracket} of order $n$
\cite{AzBu:96} is defined by the fully antisymmetrized product of
its entries
\begin{equation}
\label{multic}
     [X_{i_1},\dots,X_{i_n}]: =\sum_{\sigma\in S_n}(-1)^{\pi(\sigma)}
  X_{i_{\sigma(1)}}\dots X_{i_{\sigma(n)}}     \; ,
\end{equation}
which obviously reduces to the ordinary Lie bracket
$[X_{i_1},X_{i_2}]=\epsilon^{j_1j_2}_{i_1i_2}X_{j_1}X_{j_2}$ for
$n=2$. Using the Levi-Civita symbol (eq.~\eqref{defep}) , the
above definition is obviously equivalent to
\begin{equation}
\label{multic2}
 [X_{i_1},\dots,X_{i_n}]=\epsilon^{j_1\dots j_n}_{i_1 \dots i_n}
 X_{j_1}\cdots X_{j_n} \quad.
\end{equation}
\end{definition}

   Since
\begin{equation}
\label{eps-1}
\epsilon^{i_1\dots i_n}_{j_1\dots j_n}=
\sum^{n}_{s=1}(-1)^{s+1}\delta^{i_1}_{j_s}
 \epsilon^{i_2\dots i_n}_{j_1 \dots{\hat j_s}\dots j_n}=
 \sum^{n}_{s=1}(-1)^{s+n}
 \epsilon^{i_1\dots i_{n-1}}_{j_1 \dots{\hat j_s}\dots j_n}\delta^{i_n}_{j_s}  \quad ,
\end{equation}
it is clear that a multibracket of order $n$ may be expressed in
terms of multibrackets of increasingly lower orders by using
\begin{equation}
\begin{aligned}
\label{n-menos1} [X_1, X_2, \dots, X_n] &=
 \sum^{n}_{s=1}(-1)^{s+1} X_s \, [ X_2, X_3, \dots, \hat{X_s},\dots, X_n] \cr
 & = \sum^{n}_{s=1}(-1)^{s+n} [ X_1, X_2, \dots, \hat{X_s},\dots, X_n]\, X_s\quad .
\end{aligned}
\end{equation}
For instance, for the order three and four multibrackets we find
\begin{equation}
\label{NQB-3}
\begin{aligned}
&[X_1,X_2,X_3] = X_1[X_2,X_3] -X_2[X_1,X_3] + X_3[X_1,X_2] \cr
    &\qquad\qquad\quad= [X_2,X_3]X_1 - [X_1,X_3] X_2 + [X_1,X_2] X_3 \;  , \
\end{aligned}
\end{equation}
\begin{equation}
\begin{aligned}
&[X_1,X_2,X_3,X_4]  =X_1 [X_2,X_3,X_4] - X_2 [X_1,X_3,X_4] +
 X_3 [X_1,X_2,X_4] - X_4 [X_1,X_2,X_3] \cr
 &\qquad\qquad\qquad\quad = - [X_2,X_3,X_4] X_1 + [X_1,X_3,X_4] X_2
 -[X_1,X_2,X_4] X_3 + [X_1,X_2,X_3] X_4\; .
\end{aligned}
\end{equation}

 The associativity of the product of the entries in \eqref{multic}
implies that the multibracket necessarily satisfies an identity,
which has a {\it different} structure depending on whether $n$ is
{\it even} or {\it odd}, according to the following
\begin{lemma} (Generalized Jacobi identity (GJI))
\label{prop1}
 {\ }\\ For $n$ even, the higher order bracket (\ref{multic})
satisfies the following identity
\begin{equation}
  \label{GJI}
    \sum_{\sigma\in S_{2n-1}}(-1)^{\pi(\sigma)}\left[ [X_{\sigma(1)},\dots,
     X_{\sigma(n)}],X_{\sigma(n+1)},\dots ,X_{\sigma(2n-1)}\right] =0
\quad.
\end{equation}
We shall refer to this identity (\ref{GJI}) satisfied by the $n$ even
multibracket as the {\it generalized Jacobi identity} (GJI).

For $n$ odd, the identity is structurally different: the $r.h.s.$ of
the above expression is proportional to the larger bracket
$[X_1,\dots , X_{2n-1}]$ rather than zero.
\end{lemma}

\begin{proof}
In terms of the Levi-Civita symbol, the {\it l.h.s.} of
(\ref{GJI}) reads
\begin{equation}
\label{multie} \epsilon^{j_1\dots j_{2n-1}}_{i_1\dots i_{2n-1}}
\epsilon^{l_1\dots l_n}_{j_1\dots j_n} [X_{l_1}\cdots
X_{l_n},X_{j_{n+1}}, \dots ,X_{j_{2n-1}}] \quad.
\end{equation}
Since the $n$ entries in this bracket are also  antisymmetrized,
eq. (\ref{multie}) is equal to
\begin{equation}
\label{multiep}
\begin{aligned}[b]
 n! & \epsilon^{l_1\dots l_n
j_{n+1}\dots j_{2n-1}}_{i_1\dots\dots\dots i_{2n-1}}
   \epsilon^{l_{n+1}\dots l_{2n-1}}_{j_{n+1}\dots j_{2n-1}}
   \sum^{n-1}_{s=0}(-1)^s X_{l_{n+1}}\cdots X_{l_{n+s}}X_{l_1}\cdots X_{l_n}
     X_{l_{n+1+s}}\cdots X_{l_{2n-1}}
\\ = & n!(n-1)!\epsilon^{l_1\dots l_{2n-1}}_{i_1\dots i_{2n-1}}
    X_{l_1}\cdots X_{l_{2n-1}}\sum^{n-1}_{s=0}(-1)^s(-1)^{ns}
\\ = &
    n!(n-1)! [X_{i_1},\dots ,X_{i_{2n-1}}]\sum^{n-1}_{s=0}(-1)^{s(n+1)}\quad,
\end{aligned}
\end{equation}
where we have used the skewsymmetry of $\epsilon$ to relocate the
block $X_{l_1}\cdots X_{l_n}$ in the second equality. Thus, the
{\it l.h.s.} of (\ref{multiep}) is proportional to a multibracket
of order $(2n-1)$ times a sum, which is zero for {\it even} $n$
and for $n$ {\it odd} is equal to $n$ \cite{AzBu:96, RACSAM:98}.
\end{proof}

  The GJI \eqref{GJI} contains, in all, $\frac{(2n-1)!}{n!(n-1)!}$=
$2n-1\choose{n}$ independent terms. For $n$=2 it reduces to the
ordinary JI; for $n$=4, for instance, it has already 35 terms. It is
easy to find (see \cite{AzBu:96,Cu-Za:02}) that a multibracket with
$n$ even entries is reducible to sums of
 products of $n/2$ ordinary two-brackets. All one has to do is to
iterate the identity easily obtained from \eqref{eps-1}
\begin{equation}
\label{eps-2}
 \epsilon^{i_1\dots i_n}_{j_1 \dots j_n}=\sum_{t>s=1}^{n}
 (-1)^{s+t+1}\epsilon^{i_1 i_2}_{j_s j_t}\,
  \epsilon^{i_3 \dots\dots i_n}_{j_1 \dots \hat{j_s}\dots
  \hat{j_t}\dots j_n} \quad.
\end{equation}
For instance, for $n$=4 one immediately obtains
\begin{equation}
\label{n=4-resol}
\begin{aligned}
& [X_1,X_2,X_3,X_4] =
[X_1,X_2][X_3,X_4]-[X_1,X_3][X_2,X_4]+[X_1,X_4][X_2,X_3]\,+ \cr &
\qquad \qquad \qquad \qquad
[X_2,X_3][X_1,X_4]-[X_2,X_4][X_1,X_3]+[X_3,X_4][X_1,X_2] \; .
\end{aligned}
\end{equation}
For $n>4$ ($n$ even or odd), other decompositions of the
multibracket in terms of lower order brackets are possible by using
various resolutions of the Levi-Civita symbol into products of lower
order ones.

 For $n$ even, the GJI makes it natural
to define a higher order Lie algebra by means of

\begin{definition} ({\it Higher order} or {\it generalized Lie
algebra} (GLA)) \label{sec:GLAn} {\ }

 An order $n=2p$ higher order generalized Lie algebra
\cite{AzBu:96,AzPePB:96a,AzPePB:96b},\cite{Han-Wac:95,Gne:95} is a
vector space $\mathcal{G}$ endowed with a fully skewsymmetric
bracket $\mathcal{G}\times \mathop{\cdots}\limits^n \times
\mathcal{G} \rightarrow \mathcal{G}$, $(X_1,\dots ,X_n)\mapsto
[X_1,\dots ,X_n] \in \mathcal{G} $, such that the GJI (\ref{GJI}) is
fulfilled.
\end{definition}

Consequently, given a basis  $\{ X_i\}$ of $\mathcal{G}$
($i=1,\dots,r=\mathrm{dim}\mathcal{G}$), a finite-dimensional GLA of
order $n=2p$ is defined by an expression of the form
\begin{equation}
\label{multig}
 [X_{i_1},\dots ,X_{i_{2p}}]={C_{i_1\dots i_{2p}}}^j X_j \; ,
\end{equation}
where the multibracket is defined by eq.~\eqref{multic} and
the constants ${C_{i_1\dots i_{2p}}}^j$ are the {\it higher
order algebra structure constants}. The defining properties
of the higher order algebra, eqs.~\eqref{multic} and \eqref{GJI}, now
translate into the antisymmetry of the ${C_{i_1\dots i_{2p}}}^j$ in
$i_1,\dots,i_{2p}$ and in that these structure constants
satisfy the GJI of the $n=2p$ GLA,
\begin{equation}
\label{GJIcoord}
\epsilon^{j_1\dots j_{4p-1}}_{i_1\dots i_{4p-1}}
{C_{j_1\dots j_{2p}}}^l C_{j_{2p+1}\dots j_{4p-1}l}{}^s=0 \quad
\hbox{or}\quad  {C_{[j_1\dots j_{2p}}}^l C_{j_{2p+1}\dots j_{4p-1}]l}{}^s=0 \quad .
\end{equation}
Clearly, for $n$=2 this reduces to the JI for a Lie algebra $\fg$,
$C_{[ij}{}^l\,C_{k]l}{}^s=0$.

 We now include, for the sake of completeness, the identities that
are obtained when two multibrackets of different orders $n,m$ are
nested. These are given by the following

\begin{proposition}
\label{propMGJI}{(Mixed order generalized Jacobi identities, \textbf{}MGJI)} \\
Let $m,n$ be even. The mixed order generalized Jacobi identity for
even order multibrackets reads
\begin{equation}
\label{MGJI}
       \epsilon^{j_1\dots j_{n+m-1}}\left[ [X_{j_1},\dots ,X_{j_n}],
   \dots , X_{j_{n+m-1}}\right]=0 \quad .
\end{equation}
\end{proposition}

\begin{proof}
Following the same reasoning as in Lemma~\ref{prop1},
\begin{equation}
\label{multifnew}
  \begin{aligned}
&  \epsilon^{j_1\dots j_{n+m-1}}_{i_1\dots i_{n+m-1}}
  \epsilon^{l_1\dots l_n}_{j_1\dots j_n} [X_{l_1}\cdots X_{l_n},X_{j_{n+1}},
   \dots ,X_{j_{n+m-1}}]
\\ &\quad = n! \epsilon^{l_1\dots l_n j_{n+1}\dots
j_{n+m-1}}_{i_1\dots\dots\dots i_{n+m-1}} \epsilon^{l_{n+1}\dots
l_{n+m-1}}_{j_{n+1}\dots j_{n+m-1}} \sum^{m-1}_{s=0}(-1)^s
X_{l_{n+1}}\cdots X_{l_{n+s}} X_{l_1}\cdots X_{l_n}
     X_{l_{n+1+s}}\cdots X_{l_{n+m-1}}
\\ &\quad =n!(m-1)!\epsilon^{l_1\dots l_{n+m-1}}_{i_1\dots
i_{n+m-1}}
    X_{l_1}\cdots X_{l_{n+m-1}}\sum^{m-1}_{s=0}(-1)^s(-1)^{ns}
\\ &\quad = n!(m-1)! [X_{i_1},\dots,X_{i_{n+m-1}}]
\sum^{m-1}_{s=0}(-1)^{(n+1)s}\quad,
\end{aligned}
\end{equation}
which is zero for $n$ and $m$ even as stated.

In contrast, if $n$ {\it and/or} $m$ {\it are odd} the sum $\displaystyle
\sum^{m-1}_{s=0}(-1)^{(n+1)s}$ is different from zero ($m$ if $n$
is odd and 1 if $n$ is even). In this case, the {\it l.h.s}. of
(\ref{MGJI}) is proportional to the $(n+m-1)$-commutator
$[X_{i_1},\dots,X_{i_{n+m-1}}]$. This MGJI, found independently
in \cite{RACSAM:98}, had been defined as a multiplication
of skewsymmetric multilinear maps in \cite{Dzhuma:88}.
\end{proof}

It is simple to find the expression of the MGJI in terms of the
structure constants. If $n$ and $m$ are the even orders of the nested
brackets in eq. (\ref{MGJI}), and assuming that expressions
corresponding to eq.~\eqref{multig} exist (this will be the
case of Th.~\ref{th:HOsimple} below), the MGJI leads to
\begin{equation}
\label{multij}
  \epsilon^{i_1\dots i_{n+m-1}} C_{i_1\dots i_n}{}^l \,
       C_{i_{n+1}\dots i_{n+m-1}l}{}^s=0 \quad ;
\end{equation}
eq.~\eqref{GJIcoord} corresponds to $n$=$m$=$2p$. We shall not
discuss any further identities following from associativity, but the
above are not the only ones if one allows for more nested brackets.
This leads to the Bremner identities, initially proposed for the
ternary commutator \cite{Brem:98,Brem-Per:06} and later extended
to more general cases (see \cite{Brem:97} and
\cite{Curt-Jin-Mez:09,Dev-Nuy-Wei:09}).
\medskip

Having defined $n$-GLAs, one faces the question of providing
some examples. As we shall see, eq.~\eqref{multij} for
$n=2$ (eq.~\eqref{cocy-coord}) provides a hint for a wide
class of them: a way of finding examples of these higher order
algebras is to look at ordinary Lie algebra cohomology
since, as we saw in Sec.~\ref{sec:compGtable}, the primitive invariant
polynomials of the simple compact algebras and their
associated non-trivial cocycles are in one-to-one correspondence.

\subsection{Higher order simple algebras associated with a compact
simple Lie algebra $\fg$}
\label{sec:sec8}{\ }

We present here a construction of GLAs for which the previous
cohomology notions play a crucial r\^ole, namely the definition of
the higher order Lie algebras associated with a compact simple
algebra $\fg$. If the Lie algebra is simple $\Omega_{ij\rho} \equiv
C_{ij\rho}=k_{\rho\sigma}C^\sigma_{ij}$ is, by (\ref{HOCe}), the
non-trivial three-cocycle associated with the Cartan-Killing metric.
Since $\Omega_{ij\rho}$ is given by the antisymmetric form of the
structure constants of $\fg$, there always exists a three-cocycle.
The question arises as to whether higher order cocycles (and
therefore Casimirs of orders higher than two) can be used to define
the structure constants of a higher order bracket. Given the odd
dimensionality  of the cocycles, these multibrackets will involve an
even number of entries.

By way of an example, let us consider the case of $su(n)$, $n>2$,
and the four-bracket. Let $X_i$ be matrices of the defining
representation. The four-bracket is defined by
\begin{equation}
 [X_{i_1},X_{i_2},X_{i_3},X_{i_4}]:=\epsilon^{j_1j_2j_3j_4}_{i_1i_2i_3i_4}
  X_{j_1}X_{j_2}X_{j_3}X_{j_4}   \quad .                    \label{multia}
\end{equation}
Using the skew-symmetry in $j_1\dots j_4$, we may rewrite
(\ref{multia}) in terms of commutators as
\begin{eqnarray}
   [X_{i_1},X_{i_2},X_{i_3},X_{i_4}]&=& \frac{1}{2^2}
 \epsilon^{j_1j_2j_3j_4}_{i_1i_2i_3i_4}[X_{j_1},X_{j_2}]
[X_{j_3},X_{j_4}]
 =\frac{1}{2^2}\epsilon^{j_1j_2j_3j_4}_{i_1i_2i_3i_4}
   C^{l_1}_{j_1j_2}C^{l_2}_{j_3j_4}X_{l_1}X_{l_2}\nonumber\\
    &=&\frac{1}{2^2}\epsilon^{j_1j_2j_3j_4}_{i_1i_2i_3i_4}
        C^{l_1}_{j_1j_2}C^{l_2}_{j_3j_4}\frac{1}{2}
     ({d_{l_1l_2}}^\sigma_{.}
X_\sigma+c\delta_{l_1l_2})\nonumber\\
  &=&\frac{1}{2^3}\epsilon^{j_1j_2j_3j_4}_{i_1i_2i_3i_4}
     C^{l_1}_{j_1j_2}C^{l_2}_{j_3j_4}{d_{l_1l_2}}^\sigma_{.}
   X_\sigma
  \equiv{\Omega_{i_1\dots i_4}}^\sigma_{.} X_\sigma\quad ,
\label{multib}
\end{eqnarray}
where in going from the first line to the second we have used that
the factor multiplying $X_{l_1}X_{l_2}$ is symmetric in $l_1,l_2$,
so that we can replace $X_{l_1}X_{l_2}$ by
$\frac{1}{2}\{X_{l_1},X_{l_2}\}$ which is written in terms of the
$d$'s. The contribution of the term proportional to $c$ vanishes due
to the JI. Thus, the structure constants $C_{i_1\dots i_4}{}^\sigma$
of the four-bracket are given by the five-cocycle $\Omega_{i_1\dots
i_4 \sigma}$ associated with the primitive symmetric polynomial
$d_{ijk}$.
\medskip

The above result is, in fact, general. A ($2p$+1)-cocycle $\Omega$
for the Lie algebra cohomology of $\fg$ defines a higher order
$2p$ algebra by
\begin{equation}
\label{multigcocy}
 [X_{i_1},\dots ,X_{i_{2p}}]={\Omega_{i_1\dots i_{2p}}}^j X_j\quad ,
\end{equation}
where the structure constants satisfy the GJI (eqs.~\eqref{GJI}).
One may check directly that
$\Omega_{i_1\dots i_{2p}}{}^j$ defines a ($2p+1$)-cocycle
for $\fg$, since the MGJI for the ordinary Lie bracket of $\fg$ and
the above one give
\begin{equation}
\label{multija}
  \epsilon^{i_1\dots i_{m+1}} C_{i_1 i_2}{}^l \,
       \Omega_{i_{3}\dots i_{m+1}l}{}^j=0
\end{equation}
by eq.~\eqref{multij}. Lowering the index $j$ with the invariant
Killing metric we see that the fully antisymmetric structure constants
$\Omega_{i_1 \dots i_{m+1}}$ above are, in fact, those
of a cocycle for the Lie algebra cohomology of $\fg$ since
the above equation implies
\begin{equation}
  \epsilon^{i_1\dots i_{m+2}} C_{i_1 i_2}{}^l \,
       \Omega_{i_{3}\dots i_{m+1} i_{m+2}l}=0 \quad ,
\label{multijb}
\end{equation}
which is the ($2p+1$)-cocycle condition \eqref{cocy-coord}. Thus,
the following theorem follows \cite{AzBu:96}:

\begin{theorem}
({\it Higher order simple Lie algebras associated with a compact simple
algebra $\fg$}) \label{th:HOsimple}

Given a simple algebra $\fg$ of rank $l$, there are $(l-1)$
$(2m_i-2)$-higher order simple Lie algebras associated with $\fg$.
They are given by the $(l-1)$ Lie algebra cocycles of order
$(2m_i-1)>3$ which are associated with the $(l-1)$ symmetric
invariant polynomials on $\fg$ of order $m_i>m_1=2$. The $m_1=2$
case (for which the invariant polynomial is the Killing metric)
reproduces the original simple Lie algebra $\fg$; for the remaining
$(l-1)$ cases, the skewsymmetric $(2m_i-2)$-commutators define an
element of $\fg$ by means of the $(2m_i-1)$-cocycles, $(2m_i-1)>3$.
These higher order structure constants (as the ordinary structure
constants with all the indices written down) are fully antisymmetric
cocycles and satisfy the GJI.
\end{theorem}

\subsection{Multibrackets, higher order coderivatives and exterior
derivatives}
\label{sec9} {\ }

Higher order brackets can be used to
generalize the ordinary coderivation of multivectors.

\begin{definition} ({\it Exterior coderivative}) { \ }

\label{def9.1} Let $\{ X_i\}$ be a basis of $\fg$ given in terms of
LIVF on $G$ and $\wedge^* \fg$ the exterior algebra of multivectors
generated by them ($X_1\wedge\dots\wedge X_{q}\equiv\epsilon_{1\dots
q}^{i_1\dots i_q} X_{i_1}\otimes\dots\otimes X_{i_q}$). The exterior
coderivation is the map of degree $-1$,
$\partial:\wedge^q\fg\rightarrow\wedge^{q-1}\fg$, defined by
\begin{equation}
\label{codea}
       \partial(X_1\wedge\dots\wedge
X_q)=\sum^q_{\underset{l<k}{l=1}}(-1)^{l+k+1}
   [X_l,X_k]\wedge X_1\wedge\dots\wedge{\hat X_l}\wedge\dots
\wedge{\hat X_k}\wedge\dots\wedge X_q \quad ,
\end{equation}
where $\partial(X_1\,,\,X_2)=[X_1\,,\,X_2]$.

This definition is analogous to that of the exterior derivative $d$,
as given by (\ref{Palais}) with its first term missing when one
considers left-invariant forms (eq. (\ref{3diff})). As the exterior
derivative $d$, $\partial$ is nilpotent, $\partial^2=0$, due to the
JI for the commutator.
\end{definition}

In order to generalize (\ref{codea}), let us note that
$\partial(X_1\wedge X_2)=[X_1,X_2]$, so that (\ref{codea}) can be
interpreted as a formula that gives the action of $\partial$ on a
$q$-vector in terms of that on a bivector. For this reason we may
write $\partial_2$ for $\partial$ above. It is then natural to
introduce an operator $\partial_s$ that on a $s$-vector gives the
multicommutator of order $s$. On an $n$-multivector its action is
given by

\begin{definition} ({\it Coderivation $\partial_s$})
\label{def:cods}{\ }

 The general coderivation $\partial_s$ of degree
$-(s-1)$, $s$ even, is the map $\partial_s:\wedge^q \fg\rightarrow
\wedge^{q-(s-1)} \fg$ defined by
\begin{eqnarray}
& &
\partial_s(X_1\wedge\dots\wedge X_q):=\frac{1}{s!}\frac{1}{(q-s)!}
\epsilon^{i_1\dots i_q}_{1\dots
q}\partial_s(X_{i_1}\wedge\dots\wedge X_{i_s}) \wedge
X_{i_{s+1}}\wedge\dots\wedge X_{i_q}\quad ,
\nonumber\\
& &
\partial_s \wedge^q \fg=0\quad \hbox{for}\ s>q\quad ,
\nonumber\\
& &
\partial_s(X_{1}\wedge\dots\wedge X_{s})=[X_{1},\dots ,X_{s}]
\quad . \label{codeb}
\end{eqnarray}
\end{definition}
\noindent Note that, using \eqref{eps-relat1}, this expression
reproduces \eqref{codea} for $s$=2.
\begin{proposition}
\label{prop:codnil}
The coderivation (\ref{codeb}) is nilpotent,
\emph{i.e.}, $\partial^2_s\equiv 0$.
\end{proposition}

\begin{proof}
Let $q$ and $s$ be such that $q-(s-1)\geq s$ (otherwise the
statement is trivial). Then,
\begin{eqnarray}
     & & \partial_s\partial_s(X_1\wedge\dots\wedge X_q) =
\nonumber\\
& &\frac{1}{s!}\frac{1}{(q-s)!}\frac{1}{s!} \frac{1}{(q-2s+1)!}
    \epsilon^{i_1\dots i_q}_{1\dots q}
   \epsilon^{j_{s+1}\dots j_q}_{i_{s+1}\dots i_q}
\big\{s \left[\left[X_{i_1},\dots,X_{i_s}\right],
X_{j_{s+1}},\dots,X_{j_{2s-1}}\right] \wedge X_{j_{2s}}\wedge\dots
X_{j_q}
\nonumber\\
& &
-(q-2s+1)[X_{i_1},\dots,X_{i_s}]\wedge[X_{j_{s+1}},\dots,X_{j_{2s}}]
 \wedge X_{j_{2s+1}}\wedge\dots\wedge X_{j_q}\big\}=0\quad .
\label{codec}
\end{eqnarray}
The sum over $i_{s+1}\dots i_q$ produces an overall antisymmetrization
over the $i,j$ indices. As a result, the  first term above
vanishes because, since $s$ is even, the double bracket
is the GJI. Similarly, the second one is also zero because the wedge product
of the two $s$-brackets is antisymmetric while the resulting
$\epsilon$ symbol is symmetric under the interchange $(i_1,\dots
i_s)\leftrightarrow (j_{s+1}, \dots ,j_{2s})$.
\end{proof}

Let us now see how the nilpotency condition (or equivalently the
GJI) looks like in the simplest cases.

\begin{example} ({\it The coderivation $\partial_2$})
\label{ex:coder}{\ }

Consider $\partial\equiv\partial_2$. Then we have
\begin{equation}
     \partial(X_1\wedge X_2\wedge X_3)=
[X_1,X_2]\wedge X_3 - [X_1,X_3]\wedge X_2 + [X_2,X_3]\wedge X_1
\label{coded}
\end{equation}
and
\begin{equation}
    \partial^2(X_1\wedge X_2\wedge X_3)=
[[X_1,X_2],X_3] - [[X_1,X_3],X_2] + [[X_2,X_3],X_1] =0
\label{codee}
\end{equation}
by the JI.

 When we move to $\partial\equiv\partial_4$, the number
of terms grows very rapidly. The explicit expression for
$\partial^2(X_{i_1}\wedge\dots\wedge X_{i_7})=0$ (which, as we know,
is equivalent to the GJI) is given in full in \cite{AIPB:97} and
contains ${\binom 73}=35$ terms (note that the tenth term in
eq.~(32) there should read
 $[[X_{i_1},X_{i_2},X_{i_6},X_{i_7}],X_{i_3},X_{i_4},X_{i_5}]\,$).
  In general, the GJI which follows from
$\partial^2_{2m-2}(X_1\wedge\dots\wedge X_{4m-5})=0$ $(s=2m-2)$
contains $\binom{4m-5}{2m-2}$ different terms, as for the GJI
\eqref{GJI}.
\end{example}

\subsection{Higher order exterior derivative and generalized MC equations}
\label{sec10}{\ }

We now generalize the Lie algebra MC equations (eqs.~\ref{3canonMC}) to
the case of the GLAs of Th.~\ref{th:HOsimple} and write them in a
BRST-like form. The result is a higher order BRST-type operator
that contains all the information on the $l$ possible GLAs
$\mathcal{G}$ associated with a given simple Lie algebra $\fg$ of
rank $l$.

Let us first note that eq.~\eqref{3canonMC} gives $d^2\theta=
\frac{1}{2}\left[ [\theta,\theta],\theta\right] =0$ so that the
JI reads
\begin{equation}
\label{canon-JI}
\left[ [\theta,\theta],\theta\right] =0 \quad .
\end{equation}
In Sec.~\ref{sec9} we considered higher order coderivations which
also had the property $\partial^2_s=0$ as a result of the GJI. We
may now introduce the corresponding dual higher order exterior
derivatives ${\tilde d}_s$ to provide a generalization of the
MC eq.~\eqref{3MC}. Since $\partial_s$ was defined on multivectors
that are product of left-invariant vector fields, the dual
${\tilde d}_s$ will be given for left-invariant forms.

It is easy to introduce dual bases in $\wedge^q \fg^*$ and $\wedge^q
\fg$. With $\omega^i (X_j) =\delta_j^i$, these are given by
$\omega^{I_1}\wedge\dots\wedge\omega^{I_q}$, $\frac{1}{q!}
X_{I_1}\wedge\dots\wedge X_{I_q}\,$, $I_1< \dots <I_q\,$, since
$(\epsilon_{j_1\dots j_q}^{i_1\dots i_q}
\omega^{j_1}\otimes\dots\otimes\omega^{j_q}) (\frac{1}{q!}
\epsilon_{l_1\dots l_q}^{k_1\dots k_q} X_{k_1}\otimes\dots\otimes
X_{k_q})= \epsilon_{l_1\dots l_q}^{i_1\dots i_q} $ and
$\epsilon_{L_1\dots L_q}^{I_1\dots I_q}$ is 1 if all indices
coincide and 0 otherwise. Nevertheless it is customary to use the
non-minimal `basis' $\omega^{i_1}\wedge\dots\wedge\omega^{i_q}$ to
write $\alpha=\frac{1}{q!}\alpha_{i_1\dots i_q}
\omega^{i_1}\wedge\dots\wedge\omega^{i_q}$ with $\alpha_{i_1\dots
i_q}\equiv \alpha(X_{i_1},\dots,X_{i_q})= \frac{1}{q!}
\alpha(X_{i_1}\wedge\dots\wedge X_{i_q})$ since
$(\omega^{i_1}\wedge\dots\wedge\omega^{i_q}) (X_{j_1},\dots,X_{j_q})
=\epsilon_{j_1\dots j_q}^{i_1\dots i_q}$.

\begin{definition} ({\it Higher order exterior derivative})
{\ }

The action of ${\tilde d}_m:\wedge^q \fg^*\rightarrow
\wedge^{q+(2m-3)}\fg^*$ (recall that $s=2m-2$) is given by (cf.
(\ref{3diffa}))
\begin{equation}
\label{completeb}
\begin{array}{r@{}l}
\displaystyle & ({\tilde d}_m\alpha)(X_{i_1},\dots,
X_{i_{q+2m-3}}):=
\\
&\qquad \displaystyle \frac{1}{(2m-2)!} \frac{1}{(q-1)!}
     \epsilon_{i_1\dots i_{q+2m-3}}^{j_1\dots j_{q+2m-3}}
    \alpha([X_{j_1},\dots ,X_{j_{2m-2}}], X_{j_{2m-1}}, \dots,
      X_{j_{q+2m-3}})
\quad,
\\[0.3cm]
& \displaystyle ({\tilde d}_m\alpha)_{i_1\dots i_{q+2m-3}}
=\frac{1}{(2m-2)!} \frac{1}{(q-1)!} \epsilon_{i_1\dots
i_{q+2m-3}}^{j_1\dots j_{q+2m-3}} {\Omega_{j_1\dots
j_{2m-2}}}^\rho_{\cdot} \alpha_{\rho j_{2m-1}\dots j_{q+2m-3}}
\quad,
\end{array}
\end{equation}
where the first (second) factorial in the denominator is the
number of arguments inside (outside) the multibracket.
\end{definition}
For $m=2$, $\widetilde d_2$ gives  eq. (\ref{3diff}) with $p=q$,
\begin{equation}
\label{3diffa}
 \widetilde d_2\alpha(X_{i_1},\dots,X_{i_{q+1}}) = \frac{1}{(2\cdot 2
-2)!}\frac{1}{(q-1)!} \epsilon_{i_1\dots i_{q+1}}^{j_1\dots
j_{q+1}} \alpha([X_{j_1},X_{j_2}],X_{j_3},\dots,X_{i_{q+1}})\quad
\end{equation}
with the identification $d\equiv -\widetilde d_2$.

\begin{proposition} ({\it Higher order coderivative})
\label{prop10.1}

$\widetilde d_m:\wedge^p\fg^*\rightarrow
\wedge^{p+(2m-3)}\fg^*$ is dual to the coderivation
$\partial_{2m-2}:\wedge^p\fg \to \wedge^{p-(2m-3)}\fg$ {\i.e.}, on
a generic $p$-form,
\begin{equation}
\label{partduality} {\tilde d}_m \, \alpha \;\propto \; \alpha \,
\partial_{2m-2} \quad .
\end{equation}
\end{proposition}

\begin{proof}
 If $\alpha$ is a $p$-form, $\tilde{d}_m\alpha$ is a ($p+2m-3$)-form
 and, since $\partial_{2m-2}:\wedge^{p+2m-3}\fg \to \wedge^p\fg$,
 eq. (\ref{codeb}) tells us that
\begin{eqnarray}
& & \alpha\left( \partial_{2m-2}(X_{i_1}\wedge\dots\wedge
X_{i_{p+2m-3}})\right) =\frac{1}{(2m-2)!} \frac{1}{(p+2m-3-2m+2)!}
\times
\nonumber\\
& & \quad   \times \epsilon_{i_1\dots i_{p+2m-3}}^{j_1\dots
j_{p+2m-3}}
   \alpha([X_{j_1},\dots ,X_{j_{2m-2}}]\wedge
X_{j_{2m-1}}\wedge \dots \wedge
      X_{j_{p+2m-3}})\quad ,
\label{completec}
\end{eqnarray}
which is proportional\footnote{ One finds ${\tilde d}_m\alpha=
\frac{(p+2m-3)!}{p!}\alpha\partial_{2m-2}$, where $p$ is the order
of the form $\alpha$. The factor appears as a consequence of using
the same definition (antisymmetrization with no weight factor) for
the $\wedge$ product of forms and vectors.} to $({\tilde
d}_m\alpha)(X_{i_1}\wedge\dots\wedge X_{i_{p+2m-3}})$.
\end{proof}

\begin{proposition}
\label{prop:dmLeib} The operator ${\tilde d}_m$ satisfies
Leibniz's rule,
\begin{equation}
\label{delta-Lieb} {\tilde d}_m(\alpha\wedge\beta)={\tilde
d}_m\alpha \wedge\beta+(-1)^p\alpha\wedge{\tilde d}_m\beta \quad .
\end{equation}
\end{proposition}
\begin{proof}
If $\alpha$ and $\beta$ are $p$ and $q$ forms, respectively, we
get using (\ref{completeb})
\begin{equation}
\begin{array}{r@{}l}
{\tilde d}_m (\alpha\wedge\beta) & _{i_1\dots i_{p+q+2m-3}} =
\displaystyle \frac{1}{(2m-2)!}\frac{1}{(p+q-1)}
\epsilon_{i_1\dots i_{p+q+2m-3}}^{j_1\dots j_{p+q+2m-3}}
{\Omega_{j_1\dots j_{2m-2}}}^\rho_\cdot
\\[0.3cm]
& \displaystyle \quad\qquad \cdot \Big( \frac{1}{p!q!}
\epsilon^{k_1\dots\dots\dots k_{p+q}}_{\rho j_{2m-1}\dots
j_{p+q+2m-3}} \alpha_{k_1\dots k_p}\beta_{k_{p+1}\dots k_{p+q}}
\Big)
\\[0.3cm]
& \displaystyle = \frac{1}{(2m-2)!} \frac{1}{p!q!}
\epsilon_{i_1\dots i_{p+q+2m-3}}^{j_1\dots j_{p+q+2m-3}}
{\Omega_{j_1\dots j_{2m-2}}}^\rho_\cdot \Big( p \alpha_{\rho
j_{2m-1} \dots j_{p+2m-3}} \beta_{j_{p+2m-2} \dots j_{p+q+2m-3}}
\\[0.3cm]
& \displaystyle \quad\qquad + (-1)^p q \alpha_{j_{2m-1} \dots
j_{p+2m-2}} \beta_{\rho j_{p+2m-1}\dots j_{p+q+2m-3}} \Big)
\\[0.3cm]
& \displaystyle = \epsilon_{i_1\dots i_{p+q+2m-3}}^{j_1\dots
j_{p+q+2m-3}} \Big( \frac{1}{q! (p+2m-3)!} ({\tilde
d}_m\alpha)_{j_1\dots j_{p+2m-3}} \beta_{j_{p+2m-2}\dots
j_{p+q+2m-3}}
\\[0.3cm]
& \displaystyle \quad\qquad + (-1)^p \frac{1}{p! (q+2m-3)!}
\alpha_{j_{2m-1} \dots j_{p+2m-2}} ({\tilde d}_m\beta)_{j_1\dots
j_{2m-2} j_{p+2m-1}\dots j_{p+q+2m-3}} \Big)
\\[0.3cm]
& \displaystyle = \Big(({\tilde d}_m\alpha)\wedge\beta + (-1)^p
\alpha \wedge ({\tilde d}_m\beta)\Big)_{i_1\dots
i_{p+q+2m-3}}\quad.
\end{array}
\end{equation}
Thus, ${\tilde d}_m$ is odd and satisifies Leibniz's rule.
\end{proof}

The coordinates of ${\tilde d}_m\omega^\sigma$, where
$\omega^\sigma$ is a MC form, are given by
\begin{eqnarray}
({\tilde d}_m\omega^\sigma)(X_{i_1}, \dots, X_{i_{2m-2}})&=&
\frac{1}{(2m-2)!}\epsilon_{i_1\dots i_{2m-2}}^{j_1\dots j_{2m-2}}
\omega^\sigma([X_{j_1},\dots,X_{j_{2m-2}}])
\nonumber\\
&=& \omega^\sigma([X_{i_1},\dots,X_{i_{2m-2}}])
=\omega^\sigma({\Omega_{i_1\dots i_{2m-2}}}^\rho_\cdot X_\rho)=
{\Omega_{i_1\dots i_{2m-2}}}^\sigma_\cdot \label{completed}
\end{eqnarray}
from which we conclude that
\begin{equation}
      {\tilde d}_m\omega^\sigma=\frac{1}{(2m-2)!}
      {\Omega_{i_1\dots i_{2m-2}}}^\sigma_\cdot
   \omega^{i_1}\wedge\dots\wedge\omega^{i_{2m-2}}\quad .  \label{completee}
\end{equation}
For $m=2,\ \tilde d_2=-d$, eqs. (\ref{completee}) reproduce the MC
eqs. (\ref{3canonMC}). In the compact notation that uses the
canonical one-form $\theta$ on $G$, this leads to

\begin{proposition} ({\it Generalized Maurer-Cartan equations})
{\  }

The action of ${\tilde d}_m$ on the canonical form $\theta$ is
given by
\begin{equation}
\label{completef}
{\tilde d}_m\theta=\frac{1}{(2m-2)!} \left[
\theta,\mathop{\cdots}\limits^{2m-2},\theta \right] \quad,
\end{equation}
where the multibracket of forms is defined by $\left[
\theta,\mathop{\cdots}\limits^{2m-2},\theta \right]=
\omega^{i_1}\wedge\dots\wedge\omega^{i_{2m-2}}[X_{i_1},\dots,X_{i_{2m-2}}]$.
\end{proposition}

Using Leibniz's rule for ${\tilde d}_m$ we arrive at
${\tilde d}_m^2\theta= -\frac{1}{(2m-2)!}\frac{1}{(2m-3)!}
[\theta,\mathop{\cdots}\limits^{2m-3} ,\theta,
[\theta,\mathop{\cdots}\limits^{2m-2} ,\theta ]\,] =0$,
which expressses the GJI as
\begin{equation}
\label{completeg}
[\theta,\mathop{\cdots}\limits^{2m-3} ,\theta,
[\theta,\mathop{\cdots}\limits^{2m-2} ,\theta ]\,] =0 \quad ,
\end{equation}
which recovers the JI equation \eqref{canon-JI} for $m=2$.\\
\medskip

Each generalized Maurer-Cartan equation (\ref{completeg}) can be
expressed in terms of ghost (Grasmmann odd) variables
$c^i,\; c^i c^j=-c^j c^i\,,\, c^i{}^2=0$, by means of a
`generalized BRST operator',
\begin{equation}
    {s}_{2m-2}=-\frac{1}{(2m-2)!}c^{i_1}\dots c^{i_{2m-2}}
     {\Omega_{i_1\dots i_{2m-2}}}^\sigma_\cdot \frac{\partial}{\partial
c^\sigma}
               \quad .                                    \label{completeh}
\end{equation}
By adding together all the $l$ generalized BRST operators, the
complete BRST operator is obtained. Then we have the
following \cite{AzBu:96}

\begin{theorem} ({\it Complete BRST operator}){\  }\\
\label{th:BRSTfull} Let $\fg$ be a simple Lie algebra. Then, there
exists a nilpotent associated operator, the complete BRST operator
associated with $\fg$, given by the odd vector field
\begin{eqnarray}
      s&=&-\frac{1}{2}c^{j_1}c^{j_2}
{\Omega_{j_1j_2}}^\sigma_\cdot
    \frac{\partial}{\partial c^\sigma}-\dots-
      \frac{1}{(2m_i-2)!}c^{j_1}\dots c^{j_{2m_i-2}}
     {\Omega_{j_1\dots j_{2m_i-2}}}^\sigma_\cdot \frac{\partial}{\partial
c^\sigma}
    -\dots\nonumber\\
    &-&\frac{1}{(2m_l-2)!}c^{j_1}\dots c^{j_{2m_l-2}}
     {\Omega_{j_1\dots j_{2m_l-2}}}^\sigma_\cdot \frac{\partial}{\partial
c^\sigma}
    \equiv s_2+\dots +s_{2m_i-2}+\dots +s_{2m_l-2}\quad ,
                                                          \label{completei}
\end{eqnarray}
where $i=1,\dots,l$, ${\Omega_{j_1j_2}}^\sigma_\cdot\equiv
{C_{j_1j_2}}^\sigma_\cdot$ and ${\Omega_{j_1\dots
j_{2m_i-2}}}^\sigma_\cdot$ are the corresponding $l$ higher order
cocycles, which encodes all the multialgebras $\mathcal{G}$
associated with a simple Lie algebra $\fg$.
\end{theorem}

\begin{proof}
We have to show that $\{ s_{2m_i-2},s_{2m_j-2}\}=0$ $\forall\,
i,j$. To prove it, let us write the anti-commutator explicitly:
\begin{eqnarray}
& & \{ s_{2m_i-2},s_{2m_j-2}\}=
\frac{1}{(2m_i-2)!}\frac{1}{(2m_j-2)!} \times
\nonumber\\
& &\quad \times \{ (2m_j-2)c^{l_1}\dots c^{l_{2m_i-2}}
{\Omega_{l_1\dots l_{2m_i-2}}}^\rho_\cdot c^{r_2}\dots
c^{r_{2m_j-2}} {\Omega_{\rho r_2\dots
r_{2m_j-2}}}^\sigma_\cdot\frac{\partial}{\partial c^\sigma}
+i\leftrightarrow j
\nonumber\\
& &\quad +(c^{l_1}\dots c^{l_{2m_i-2}}c^{r_1}\dots c^{r_{2m_j-2}}
{\Omega_{l_1\dots l_{2m_i-2}}}^\rho_\cdot {\Omega_{ r_1\dots
r_{2m_j-2}}}^\sigma_\cdot+i\leftrightarrow j)
\frac{\partial}{\partial c^\rho}\frac{\partial}{\partial
c^\sigma}\}
\nonumber\\
& &\quad =\frac{1}{(2m_i-2)!}\frac{1}{(2m_j-3)!}c^{l_1}\dots
c^{l_{2m_i-2}} c^{r_2}\dots c^{r_{2m_j-2}}{\Omega_{l_1\dots
l_{2m_i-2}}}^\rho_\cdot {\Omega_{\rho r_2\dots
r_{2m_j-2}}}^\sigma_\cdot\frac{\partial}{\partial c^\sigma}
\nonumber\\
& &\quad\quad +i\leftrightarrow j\ , \label{completej}
\end{eqnarray}
where we have used the fact that $\frac{\partial}{\partial
c^\rho}\frac{\partial}{\partial c^\sigma}$ is antisymmetric in
$\rho,\sigma$ while the parenthesis multiplying it is symmetric.
The term proportional to a single $\frac{\partial}{\partial
c^\sigma}$ also vanishes as a consequence of equation
(\ref{multij}).
\end{proof}

The antisymmetric coefficients of $\partial / \partial c^\sigma$,
etc. in $s_{2m_i-2}$ can be viewed, in dual terms, as (even)
multivectors of the type
\begin{equation}
\Lambda=\frac{1}{(2m-2)!}{\Omega_{i_1\dots i_{2m-2}}}^\sigma_\cdot
x_\sigma
\partial^{i_1}\wedge\dots \wedge \partial^{i_{2m-2}}\quad ,
\label{completek}
\end{equation}
by replacing the inherent skewsymmetry associated with the odd
character of the $c^i$ by that introduced by the wedge product of
the derivatives with respect to the even variables $x_i$. The
resulting multivectors $\Lambda$ have the property of having zero
Schouten-Nijenhuis bracket \cite{Sch:40,Nij:55} (see Appendix 2)
among themselves by virtue of the GJI (\ref{GJIcoord}). As a result,
they have precisely the property required to define the {\it
(linear) generalized Poisson structures} that will be discussed in
Sec.~\ref{sec:GPS-linear}.

For other aspects of operators with similar structure
see \cite{Zwi:93,AzBu:96} and further references therein.

\subsection{GLAs and strongly homotopy (SH) Lie algebras}
\label{sec:SH}
{\ }

The above higher order Lie algebras turn out to be a special example
of the strongly homotopy (SH) Lie algebras
\cite{Lad.Sta:93,Lad.Mar:95,Jon:90,Be-La:09} which we briefly
mention below for completeness. These SH algebras allow for
`controlled' violations of the GJI, which are obviously absent for a
GLA.

\begin{definition} ({\it SH algebras} \cite{Lad.Sta:93})
{ \ }

A {\it SH Lie structure} on a vector space $V$ is a collection of
skewsymmetric linear maps $l_s:V\otimes \mathop{\cdots}\limits^s
\otimes V\to V$ such that
\begin{equation}
\label{shalgebra}
\sum_{i+j=s+1} \sum_{\sigma \in S_s} \frac{1}{(i-1)!} \frac{1}{j!}
(-1)^{\pi(\sigma)} (-1)^{i(j-1)} \,l_i (l_j (v_{\sigma (1)}\otimes
\dots \otimes v_{\sigma (j)}) \otimes v_{\sigma(j+1)} \otimes
\dots\otimes v_{\sigma (s)})=0 \quad.
\end{equation}
For a general treatment of SH Lie algebras including $v$ gradings
see \cite{Lad.Sta:93,Lad.Mar:95,Jon:90} and references therein.
Note that $\displaystyle \frac{1}{(i-1)!} \frac{1}{j!}
\sum_{\sigma \in S_s}$ is equivalent to the sum over the
`unshuffles', \emph{i.e.}, over the permutations $\sigma\in S_s$
such that $\sigma(1)<\dots <\sigma(j)$ and $\sigma(j+1)<\dots
<\sigma(s)$.
\end{definition}

\begin{example}
For $s=1$, eq. (\ref{shalgebra}) just says that $l_1^2=0$ ($l_1$
is a differential). For $s=2$, eq. (\ref{shalgebra}) gives
\begin{equation}
-\frac{1}{2} l_1( l_2(v_1\otimes v_2) - l_2(v_2\otimes v_1))
+l_2(l_1(v_1)\otimes v_2 - l_1(v_2)\otimes v_1)=0
\end{equation}
\emph{i.e.}, with $l_2(v_1\otimes v_2) = [v_1, v_2]$,
\begin{equation}
l_1[v_1,v_2]=[l_1 v_1,v_2] + [v_1,l_1 v_2] \quad .
\end{equation}

For $s=3$, we have three maps $l_1\,,l_2\,,l_3$, and eq.
(\ref{shalgebra}) reduces to
\begin{equation}
\begin{array}{c}
\left[l_2 ( l_2 (v_1\otimes v_2)\otimes  v_3) + l_2 ( l_2
(v_2\otimes  v_3)\otimes  v_1) + l_2 ( l_2 (v_3\otimes v_1)\otimes
v_2) \right] +\left[ l_1 ( l_3 (v_1\otimes  v_2 \otimes  v_3) )
\right]
\\[0.3cm]
+ \left[ l_3 (l_1 (v_1)\otimes  v_2 \otimes  v_3) + l_3 (l_1
(v_2)\otimes  v_3 \otimes  v_1) + l_3 (l_1 (v_3)\otimes  v_1
\otimes  v_2) \right] = 0 \quad,
\end{array}
\end{equation}
{\it i.e.}, adopting the convention that
$l_s(v_1\otimes\dots\otimes v_s)= [v_1,\dots,v_s]$,
\begin{equation}
\begin{array}{l}
[[v_1, v_2], v_3] + [[v_2, v_3], v_1] + [[v_3, v_1], v_2]
\\[0.3cm]
\qquad\qquad = -l_1 [v_1, v_2 , v_3] -[l_1 (v_1), v_2 , v_3] -
[v_1, l_1 (v_2), v_3 ] - [v_1 , v_2, l_1 (v_3)] \quad.
\end{array}
\label{shexample}
\end{equation}
The $r.h.s$ in (\ref{shexample}) shows the violation of the
(standard) Jacobi identity appearing in the $l.h.s.$.
\end{example}
In the particular case in which a unique $l_s$ ($s$ even) is
defined, we recover the GLA case since, for $i=j=s$, eq.
(\ref{shalgebra}) reproduces the GJI \eqref{GJI} in the form
\begin{equation}
\sum_{\sigma \in S_{2s-1}} \frac{1}{s!} \frac{1}{(s-1)!}
(-1)^{\pi(\sigma)} l_s ( l_s (v_{\sigma (1)}\otimes \dots \otimes
v_{\sigma (s)}) \otimes v_{\sigma(s+1)} \otimes \dots\otimes
v_{\sigma (2s-1)})=0 \quad.
\end{equation}
Thus, the higher order Lie algebras correspond to a particular case
of SH Lie algebras, the one that appears when the GJI is satisfied.

\section{Filippov or $n$-Lie algebras}
\label{sec:filippov}{\ }

This section reviews some basic properties of Filippov algebras (FAs)
\cite{Filippov}, \cite{Fil:98, Kas:87, Kas:95a, Ling:93}; other
questions, including the existence of the enveloping algebras
of Filippov algebras are discussed in \cite{Pozhi:98, Poji:03}.
We begin by discussing the crucial ingredient of $n$-Lie algebras,
the derivation property that determines their characteristic
identity, the {\it Filippov identity}, which distinguishes FAs from
the higher order GLAs of the previous section. After discussing
the general properties of the FAs and, in particular, their
associated inner derivations Lie algebra, we look at
examples of finite FAs (and, specially, the simple ones). The
infinite dimensional case will be exemplified by the
Nambu or Jacobian FAs (these are not the only examples of
infinite-dimensional FAs: ternary Kac-Moody- and Virasoro-Witt-like algebras
have been considered in \cite{Lin:08} and \cite{Curt-Fair-Zac:08},
respectively).

\subsection{Derivations of an $n$-bracket, the Filippov identity and $n$-Lie algebras}{\ }

Let $\fG$ be a vector space endowed with a fully antisymmetric,
multilinear application $ [\;,\;,\;\mathop{\cdots}\limits^n
\;,\;,\;]:\,
 \fG\times\mathop{\cdots}\limits^n \times \fG\rightarrow \fG \, $,
 $[X_1,X_2,\dots,X_n]\in \fG$, called $n$-{\it bracket}.
Let $ad$ be the map $ad: \wedge^{n-1}\fG \rightarrow
\textrm{End}\,\fG \; $ that to each element of  $\wedge^{n-1}\fG$
associates the {\it left multiplication} of $\fG$ given by
\begin{equation}
\label{n-ad}
 ad_{X_1 X_2 \dots X_{n-1}}: Z \rightarrow
[X_1,X_2,\dots,X_{n-1},Z] \quad , \quad \forall\,X_i\,,\,Z\in \fG
\quad ,
\end{equation}
which for $n$=2 reproduces the action of $ad_X$ on a Lie algebra;
note that the skewsymmetry of the $n$-bracket implies that
of the $n-1$ arguments of $ad$. In similarity with the Lie algebra
case, we now require that $ad_{X_1, X_2, \dots,X_{n-1}}$  is a {\it
derivation} of the $n$-bracket {\it i.e.}, that the following
property holds:

\begin{definition} ({\it Inner derivations of the $n$-bracket})
\label{ad-is-der}

$ad_{X_1, X_2, \dots,X_{n-1}}$ is an inner ({\it left}) derivation of the
$n$-bracket
 {\it i.e.},
\begin{equation}
\begin{aligned}
\label{n-ad-der}
    ad_{X_1, X_2, \dots,X_{n-1}} [ Y_1, Y_2, \dots, Y_n]  & =
\sum_{i=1}^{n} [Y_1,\dots, ad_{X_1 \dots X_{n-1}}Y_i,\dots,Y_n] \\
& = \sum_{i=1}^n [ Y_1, \dots, Y_{i-1},
 [X_1,\dots, X_{n-1}, Y_i], Y_{i+1}, \dots, Y_n] \; .
\end{aligned}
\end{equation}

As for Lie algebras, the derivations $ad_{X_1,\dots,X_{n-1}}$
are called {\it inner} because they are characterized by
($n-1$) elements of the $n$-Lie algebra $\fG$ itself. The
above reads, in full detail,
\begin{equation}
\label{n-der}
\begin{aligned}
    & \qquad [X_1, X_2, \dots,X_{n-1}, [ Y_1, Y_2, \dots, Y_n]]    =
   [[X_1,X_2,\dots, X_{n-1}, Y_1],Y_2,\dots,Y_n]    + \\ &
 [Y_1,[X_1,\dots,X_{n-1}, Y_2], Y_3 \dots, Y_n ]  +  \dots +
[Y_1,\dots, Y_{n-1},[X_1,X_2,\dots,X_{n-1},Y_n]] \quad .
\end{aligned}
\end{equation}
\end{definition}
\noindent Eq.~\eqref{n-der}, which contains ($n+1$) terms, is the
{\it Filippov identity} (FI). It reduces to the JI for $n=2$ and
motivates the following

\begin{definition} ($n$-{\it Lie or Filippov algebra} (FA) \cite{Filippov})
\label{def:n-alg}

 An $n$-Lie algebra $\fG$ is a vector space  together with a
 multilinear fully skewsymmetric application
$ [\;,\;,\;\mathop{\dots}\limits^n \;,\;,\;]:\,
\fG\times\mathop{\dots}\limits^n \times \fG\rightarrow \fG$,
the $n$-bracket, such that the Filippov identity (\ref{n-der}) is satisfied. For
$n$=2, the FA $\fG$ reduces to an ordinary Lie algebra $\fg$.
\end{definition}

Thus, the FI \eqref{n-der} that characterizes a FA just reflects
that $ad_{X_1,\dots,X_{n-1}}$ is an inner derivation of the $n$-Lie
algebra. Note that for $n > 2$ $\,ad_{X_1 X_2 \dots X_{n-1}}$ is a
derivation of the $n$-bracket but not a representation of the
elements of $\fG$ themselves (but see Sec.~\ref{sec:repFA} below),
since the skewsymmetric map
$ad: \fG\times\mathop{\cdots}\limits^{n-1} \times \fG \rightarrow
\textrm{End}\,\fG$ is not defined on $\fG$ itself unless $n$=2.
It is only for $n=2$ that $ad$ is both a representation of $\fg$
and a derivation of the Lie algebra.

  The skewsymmetric sets $(X_1,\dots,X_{n-1})$ that determine
inner derivations $ad_{X_1,\dots,X_{n-1}}$ of the FA $\fG$
appear very frequently in the theory of FAs and it is convenient to
denote them by a symbol, $\mathscr{X}\equiv (X_1,\dots,X_{n-1})$,
and to give them a name. They will be called {\it fundamental
objects} of the FA $\fG$, $\mathscr{X}\in \wedge^{n-1}\fG$; their
properties will be discussed in Sec.~\ref{sec:fund-obj}. As for
Lie algebras, the $ad$ map is not injective in general, the extreme
case being that of an abelian $\fG$, for which
$\ad_\mathscr{X}=0\;\,\forall \mathscr{X}$. The classes of
fundamental objects  $\mathscr{X}$ obtained by taking the quotient
by ker$\,ad$ are the {\it inner derivations} of $\fG$. They define
an ordinary Lie algebra InDer$\,\fG\equiv\,$Lie$\,\fG$ or {\it Lie
algebra associated with} $\fG$, as will be discussed in
Sec.~\ref{sec:nLie-to-Lie}.
\medskip

{\it Observation}. Filippov uses \cite{Filippov} {\it right}
multiplications, $R(X_1,\dots,X_{n-1}):Y\mapsto
[Y,X_1\dots,X_{n-1}]$. Such a right inner derivation leads to
\begin{equation}
\begin{aligned}
\label{n-der-right}
&\qquad [[Y_1,\dots,Y_n],X_1,\dots,X_{n-1}]=[[Y_1,X_1,\dots,X_{n-1}],Y_2,\dots,Y_n] + \\&
[Y_1,[Y_2,X_1,\dots,X_{n-1}],Y_3,\dots,Y_n]+\dots +[Y_1,\dots,Y_{n-1},[Y_n, X_1,\dots,X_{n-1}]] \; .
\end{aligned}
\end{equation}
However, due to the full skewsymmetry of the Filippov $n$-bracket,
all the terms in the `left' (eq. \eqref{n-der}) and in the `right'
identity above differ in a $(-1)^{n-1}$ sign, and therefore both
equations define one and the same expression (as it is of course the
case of the JI, which may be written either as $
[X,[Y,Z]]=[[X,Y],Z]+[Y,[X,Z]]$ or as
$[[Y,Z],X]=[[Y,X],Z]+[Y,[Z,X]]\,$).
\medskip

There may be derivations of a FA that are not defined through
elements of $\fG$ since, in general,

\begin{definition}({\it Derivations of a FA})

A derivation of a FA is an element $D\in\hbox{End}\fG$
that satisfies the derivation property,
\begin{equation}
\label{gral-der}
D [ X_1, X_2, \dots, X_n]  =
\sum_{i=1}^{n} [X_1,\dots, D\,X_i,\dots,X_n] \quad .
\end{equation}
\end{definition}
{\ }\\
As for InDer$\,\fG$ above, the space $\hbox{Der}\,\fG$ of all {\it derivations}
$D$ of $\fG$ defined by eq.~\eqref{gral-der} generate
a Lie algebra. It is checked that $\hbox{InDer}\,\fg$ is an ideal of Der$\,\fG$
(see, {\it e.g.} \cite{Ling:93}). Therefore, the quotient
Der$(\fG)/\hbox{InDer}\,\fG$=OutDer$\,\fG$ is, by definition, the Lie
algebra of {\it outer derivations} of the FA $\fG$, thus called because
they cannot be realized in terms of fundamental elements of $\fG$.
\medskip

\subsubsection{Other forms of the Filippov identity}
{\ }

    The FI may be rewritten in various forms which is useful to have at hand.
Using the skewsymmetry of the $n$-bracket,
the FI may also be written as
\begin{equation}
\label{eq:FI}
 [X_1, X_2, \dots, X_{n-1}, [ Y_1, Y_2, \dots, Y_n]]  =
 \sum_{i=1}^n \,(-1)^{n-i} [Y_1, \dots, \widehat{Y_i},\dots,Y_n ,
 [X_1,\dots, X_{n-1}, Y_i]]
\end{equation}
or, equivalently,
\begin{equation}
\label{FI-per}
[X_1,X_2,\dots,X_{n-1}, [Y_1,\dots,Y_n]]=
\sum_{cycl. perm.} (-1)^{\pi(\sigma)} \,
[[X_1,\dots,X_{n-1},Y_{\sigma (1)}], Y_{\sigma(2)},\dots,
Y_{\sigma(n)}]   \quad ,
\end{equation}
where the sum is extended to the $n$ cicular permutations $\sigma$
of the $n$ indices. For $n$ odd, all circular permutations are
even, and no signs appear; for $n$ even, plus and minus signs
alternate in (\ref{FI-per}).

  Another useful form of the FI is provided by
\begin{equation}
\label{FIshort}
[[X_{[a_1},\dots,X_{a_n}] ,X_{b_1]},\dots,X_{b_{n-1}}]=0 \quad ,
\end{equation}
to be compared with eq.~\eqref{GJI} for a GLA.

Finally, we give one more way of writing the FI. If we introduce a set of ($2n-1$)
anticommuting, `ghost' variables and set $B=b^aX_a\,,\, C=c^a X_a$,
it is easy to rewrite
eq.~\eqref{n-der} in the compact form \cite{Ba-To:08}
\begin{equation}
\label{FIultrashort-a}
[B,\mathop{\cdots}\limits^{n-1},B,[C,\mathop{\cdots}\limits^n,C]]=
n [[B,\mathop{\cdots}\limits^{n-1}, B,C],C,\mathop{\cdots}\limits^{n-1},C] \; .
\end{equation}
The proof is an immediate generalization of the $n=2$ Lie agebra
case, for which $[B,[C,C]]=2[[B,C],C]$ reproduces the JI once the
ghost variables are factored out.
\medskip

The above anticommuting ghost variables can be used to prove the equivalence of
eqs.~\eqref{eq:FI} and \eqref{FIshort}, as we show below for the simplest
$n=3$ case. Assume that \eqref{FIshort} is true. It may be
written as
\begin{eqnarray}
\label{prueba1}
& &  [X_{b_1},X_{a_1},[X_{a_2},X_{a_3},X_{a_4}]] =
[X_{b_1},X_{a_2},[X_{a_1},X_{a_3},X_{a_4}]] \nonumber\\
& &\quad\quad\quad \quad + [X_{b_1},X_{a_3},[X_{a_2},X_{a_1},X_{a_4}]]
+ [X_{b_1},X_{a_4},[X_{a_2},X_{a_3},X_{a_1}]] \ .
\end{eqnarray}
Contracting with $b^{b_1}b^{a_1}c^{a_2}c^{a_3}c^{a_4}$
leads to
\begin{equation}
\label{prueba2}
  [B,B,[C,C,C]] = -3[B,C,[B,C,C]] \ .
\end{equation}
If we now contract with $c^{b_1}b^{a_1}b^{a_2}c^{a_3}c^{a_4}$ we obtain
\begin{equation}
\label{prueba3}
  -[B,C,[B,C,C]] = [[B,B,C],C,C] \;,
\end{equation}
and combining the last two equations the FI in the form
\eqref{FIultrashort-a} follows. Conversely, start from the conventional
form \eqref{eq:FI}. To show that this implies
eq.~\eqref{FIshort}, it is sufficient to prove the equivalent
expression $[B,C,[C,C,C]]=0$. This follows by applying
the FI to $[C,C,[B,C,C]]$:
\begin{eqnarray}
\label{prueba4}
     [C,C,[B,C,C]] & =& [[C,C,B],C,C] + [B,[C,C,C],C] + [B,C,[C,C,C]]
  \nonumber\\
    & = & [C,C,[B,C,C]] + 2[B,C,[C,C,C]] \ ,
 \end{eqnarray}
which gives $[B,C,[C,C,C]]=0$. The proof may be extended to any $n$,
and it is relegated to Appendix 1.

\begin{proposition}
\label{n-1 from n}
($(n-1)$-Lie algebras from $n$-Lie algebras \cite{Filippov})
{\ }

Let $\fG$ be an arbitrary $n$-Lie algebra and define, by fixing an
element $A\in \fG$ in its $n$-bracket, the obviously ($n-1$)-linear
and fully antisymmetric ($n-1$)-bracket by
\begin{equation}
\label{reduc}
  [X_1,X_2,\dots,X_{n-1}]_{n-1} :=
  [A,X_1,X_2,\dots,X_{n-1}]_n \quad .
\end{equation}
Then, the ($n-1$)-bracket above satisfies the FI.
\begin{proof}
Clearly, with $A$ fixed, the FI for the $n$-bracket implies the
equality
\begin{equation}
\label{n-der-2}
\begin{aligned}
    & \qquad [A, X_1, X_2, \dots,X_{n-2}, [A, Y_1, Y_2, \dots, Y_{n-1}]]    =
   [[A, X_1,X_2,\dots, X_{n-2}, A],Y_1,\dots,Y_{n-1}]    + \\
 & [A,[A, X_1,X_2,\dots,X_{n-2}, Y_1], Y_2 \dots, Y_{n-1} ]  +  \dots
+ [A, Y_1,\dots, Y_{n-2},[A, X_1,X_2,\dots,X_{n-2},Y_{n-1}]]\\
 & \qquad\quad = \sum_{i=1}^{n-1} [A, Y_1, \dots, Y_{i-1},
 [A, X_1,X_2,\dots,X_{n-2}, Y_i], Y_{i+1}, \dots, Y_{n-1}]
\end{aligned}
\end{equation}
which, in turn, implies that the ($n-1$)-bracket defined by
(\ref{reduc}) satisfies the $n$ terms FI and hence defines
a ($n-1$)-Lie algebra on the same vector space of the $n$-Lie
algebra.
\end{proof}
\end{proposition}
Thus, given an $n$-Lie algebra, one may obtain by the above procedure
a subordinated chain of $m$-Lie algebras of increasingly lower
orders (see also \cite{Rot:03} for further details and references).
In particular, a 3-Lie algebra $\fG$ defines a family of Lie algebras
$\fg$ characterized by the (fixed) elements of $\fG$, since $[X,Y]:=[A,X,Y]
\,\;\forall X,Y\in\fG$ satisfies the JI in $\fg$.
\\

\subsection{Structure constants of the $n$-Lie algebras and the FI}{\ }

 Chosen a basis $\{X_a \}$ of $\fG$, $a=1,\dots,\hbox{dim}\fG$, the FA
 bracket may be defined by the
 $n$-Lie algebra structure constants,
\begin{equation}
\label{FAstruct}
 [X_{a_1} \dots X_{a_n}]= f_{a_1 \dots a_n}{}^d \,X_d \quad.
\end{equation}
The  $f_{a_1 \dots a_n}{}^d $ are fully skewsymmetric in the $a_i$
indices and satisfy the condition
\begin{equation}
\label{FIstrconst}
 f_{b_1 \dots b_n}{}^l \; f_{a_1 \dots a_{n-1}l}{\ }^s \,=
 \sum_{k=1}^{n}\,f_{a_1 \dots a_{n-1}b_k}{}^l \; f_{b_1 \dots b_{k-1} l b_{k+1}\dots
 b_n}{}^s \quad ,
\end{equation}
which expresses the FI
\begin{equation}
\nonumber
 [X_{a_1},\dots, X_{a_{n-1}},[X_{b_1},\dots,X_{b_n}]] =
 \sum_{k} [X_{b_1},\dots,X_{b_{k-1}}
[X_{a_1},\dots,X_{a_{n-1}},X_{b_k}],X_{b_{k+1}},\dots, X_{b_n}]
\end{equation}
in terms of the structure constants  of $\fG$. For later convenience, we write
the coordinates expression of the $n=3$ FI explicitly:
\begin{equation}
\label{coord3FI}
f_{b_1 b_2 b_3}{}^l\,f_{a_1 a_2 l}{}^s=f_{a_1 a_2 b_1}{}^l\;f_{l b_2 b_3}{}^s
+f_{a_1 a_2 b_2}{}^l\;f_{b_1 l b_3}{}^s+f_{a_1 a_2 b_3}{}^l\;f_{b_1 b_2 l}{}^s \quad .
\end{equation}

Similarly, the form \eqref{FIshort} of the FI leads in coordinates
to
\begin{equation}
\label{FIshort-coor}
f_{[a_1 \dots a_n}{}^l \; f_{b_1]b_2 \dots b_{n-1}l}{\ }^s = 0 \quad,
\end{equation}
to be compared in the $n$ even case with the coordinates expression of the
GJI for a GLA in eq.~\eqref{GJIcoord}. Thus, every even FA defines a generalized
Lie algebra, a fact that will find an analogue in Lemma \ref{le:NPimpGPS}
for the even $n$-ary generalizations of the Poisson structures
to be discussed in Sec.~\ref{sec:higherPoisson}.

To conclude, we give two more forms for the FI in coordinates. From
expression \eqref{n-der} or \eqref{eq:FI} it follows that
\begin{equation}
f_{c_1 \dots c_n}{\ }^l \; f_{b_1 \dots b_{n-1} l}{\ }^s =\sum_{i=1}^{n}
(-1)^{n-i}\,f_{b_1\dots b_{n-1}c_i }{\ }^l\, f_{c_1\dots\hat{c}_i \dots c_n l}{\ }^s
\quad,
\end{equation}
which is equivalent to
\begin{equation}
\label{FIultrashort-b}
f_{c_1 \dots c_n}{\ }^l \; f_{b_1 \dots
b_{n-1} l}{\ }^s = \frac{(-1)^{n-1}}{(n-1)!} f_{b_1\dots b_{n-1}
[c_1 }{\ }^l\, f_{c_2\dots c_n ] l}{\ }^s
\end{equation}
(the $r.h.s.$ would become $(-1)^{n-1}n\,f_{b_1\dots b_{n-1} [c_1 }
 {\ }^l\, f_{c_2\dots c_n ] l}{\ }^s$ with unit weight
antisymmetrization of the $n$ indices $c\;$).
Eq.~\eqref{FIultrashort-b} also follows by writing $B=b^a X_a$ etc
and extracting the ghosts from eq.~\eqref{FIultrashort-a}.

\subsection{Structural properties of Filippov $n$-Lie algebras}
\label{sec:structure}{\ }

Let us go briefly through some basic properties of FAs following
the pattern of the Lie algebras in Sec.~\ref{sec:Lie-review} which,
after all, constitute the $n$=$2$ FA case. Many Lie algebra results
may be straightforwardly translated to general FA, although there
are important differences, the main ones being that the rich variety
of simple Lie algebras is drastically reduced when moving to $n>2$
FAs and that {\it e.g.},  some Lie algebra concepts (as solvability)
allow for more than one possible definition when extended to $n\geq
3$ $n$-Lie algebras. The structure of the FAs (for a review of $n=3$
FAs, see \cite{JMF:08}) was already developed in the original paper
of Filippov \cite{Filippov} and in later work of Kasymov \cite
{Kas:87}, where the notion of representation and the analogues of the
Cartan subalgebra and Killing metric for Lie algebras were introduced
for FAs. Further developments can be found in \cite{Kas:95a,Kas:98}
(see also \cite{Kas:91,Kas:95b}) and in the very complete Ph. D.
thesis of Ling \cite{Ling:93}. This last paper proved the analogue
of the Levi decomposition for finite-dimensional $n$-Lie algebras
and showed that for $n>2$ all ($n+1$)-dimensional simple $n$-Lie
algebras are of {\it one} type up to isomorphisms, the one given by
Filippov \cite{Filippov}, thus classifying the simple FAs.
\medskip

\subsubsection{Basic definitions, properties and results}
\label{FAbasics}
{\ }

Let $\fG$ be a FA (Def.~\ref{def:n-alg}). Then,
\medskip

A FA is {\it abelian} if $[X_1,\dots,X_n]=0$ for all $X\in\fG$.
\medskip

A subspace $\fh$ of a FA $\fG$ is a {\it Filippov subalgebra} when
it is closed under under the $n$-bracket,
\begin{equation*}
\fh \quad \hbox{subalgebra} \; \Leftrightarrow \; [Y_1,\dots,Y_n]
\subset \fh \quad \forall Y\in \fh \quad .
\end{equation*}

A subspace $I\subset \fG$ is an {\it ideal} $\fG$ if
\begin{equation*}
[X_1,\dots,X_{n-1},Y]  \subset I \quad  \forall X\in \fG\,,\,\forall
Y\in I \quad.
\end{equation*}

A FA is {\it simple} if $[\fG,\dots,\fG]\not=\{0\}$ and has no
ideals different from the trivial ones, $\{0\}$ and $\fG$.
\medskip

It is easy to check that if $I,J$ are ideals of $\fG$,
$I+J=\{X+Y|X\in I\,,\,Y\in J\}$ and $I\cap J$ are also ideals of
$\fG$. Thus, since the sum (resp. intersection) of ideals of $\fG$
is an ideal of $\fG$, a FA $\fG$ has maximal (resp. minimal) ideals.
An ideal $J$ is called {\it maximal} if the only ideals containing
$J$ are $\fG$ and $J$. An ideal $I$ of $\fG$ is called {\it minimal}
if the only ideals of $\fG$ contained in $I$ are $0$ and $I$.
Further, in analogy with the Lie algebras particular case (see {\it
e.g.} \cite{Hum:72}) the following Lemma \cite{Filippov,Ling:93}
holds:

\begin{lemma} {\ }

Let $I,J$ be ideals of $\fG$. Then, $(I+J)/I \sim J/(I \cap J)$.
Futher, if $I\subset J$, $J/I$ is an ideal of $\fG/I$ and
$(\fG/I)/(J/I) \sim \fG/J$ {\it i.e.}, one can remove the `common
factor' $I$.
\end{lemma}

\begin{definition} ({\it Centre, centralizer, normalizer})
\label{var-defs}
{\ }

 The {\it centre $Z(\fG)$ of a FA} is given by
$Z(\fG)=\{Z\in\fG\,|\, [X_1,\dots,X_{n-1},Z]=0\; \forall X\in \fG
\}$. It is an abelian ideal of $\fG$. More generally, the {\it
centralizer} $C(\fh)$ of a subset $\fh\subset \fG$ may be defined by
the condition $[C(\fh),\fh, \fG, \dots, \fG]=0$ (thus, and as for
Lie algebras, $C(\fG)=Z(\fG)$). Using the FI \eqref{n-der}, it is
seen that $C(\fh)$ is a subalgebra of $\fG$. Similarly, the {\it
normalizer} $N(\fh)$ of a subalgebra $\fh$ of $\fG$ is defined by
the condition $[N(\fh),\fh, \fG,\dots,\fG]\subset \fh$. Again, the
FI shows that $N(\fh)$ is a subalgebra of $\fG$; clearly, $\fh$ is
an ideal of $N(\fh)$.
\end{definition}

   In analogy with Lie algebras, we have the following
(see further Th.~\ref{th:simple-SO} below)
\begin{theorem}
\label{simple-inder}{\ }

All the derivations of a simple FA are inner \cite{Filippov}.
\end{theorem}

\begin{definition}({\it Homomorphisms of FAs})
\label{def:homFA}

A vector space homomorphism $\phi : \fG \rightarrow \fG'$,
$\phi:X\in \fG\rightarrow \phi(X)\in \fG'$, is  a {\it homomorphism
of Filippov algebras} when the image of the $n$-bracket in $\fG$ is
the $n$-bracket of the images in $\fG'\,$,
\begin{equation*}
\phi ([X_1,\dots,X_n]) =  [\phi(X_1),\dots,\phi(X_n)] \quad.
\end{equation*}
\end{definition}

The {\it kernel of an homomorphism} $\phi$ is the vector space
$\hbox{ker}\,\phi$=$\{Y\in \fG_1 \, | \,\phi(Y)=0\}$;
$\hbox{ker}\phi$ is an ideal of $\fG$ since, if $\phi(Y)=0$,
$\phi([X_1,\dots,X_{n-1},Y])=[\phi(X_1),\dots,\phi(X_{n-1}),\phi(Y)]=0$
and therefore $[X_1,\dots,X_{n-1},Y]\in \hbox{ker}\,\phi$. When
$\hbox{ker}\,\phi=0$, $\phi$ is an {\it isomorphism} of FAs.
\medskip

The quotient space of a FA $\fG$ by an ideal $I$ is also an $n$-Lie
algebra, since the $n$-bracket of any $n$ elements from each of the
classes $X_1$+$I$,...,$X_{n+1}$+$I\,$, is an element in the class
$[X_1,\dots,X_n]$+$I$ because $I$ is an ideal. Therefore, given an
homomorphism $\phi$ as above, there is an exact sequence of
homomorphisms or {\it canonical decomposition of $\phi:\fG
\rightarrow \fG'$},
\begin{equation*}
0 \rightarrow \ker{\phi} \rightarrow \fG \rightarrow
\fG/\hbox{ker}\phi=\hbox{Im}\fG \rightarrow 0 \quad.
\end{equation*}
The first and second arrows are obviously {\it injective},
the third one is the {\it canonical projection} and the forth
homomorphism is also {\it surjective} (and trivially so).
Clearly, if $\phi$ is surjective,
$\fG / \hbox{ker}\,\phi$ and $\fG'$ are isomorphic FAs.
\medskip

\subsubsection{Solvable Filippov algebras, radical. Semisimple FAs.}
\label{sec:solv-FA}{\ }

Let $\fG$ be a FA $\fG$, and define the {\it derived series of
ideals} inductively by\footnote{The solvability notion for
Lie algebras allows for various extensions when
moving to FAs, $n>2$, because the $n$-bracket has
more than two entries. For an $n$-Lie algebra the notion of
$k$-solvability was introduced by Kasymov \cite{Kas:87} by taking
$\fG^{(0,k)}=\fG\;,\;
\fG^{(m,k)}=[\fG^{(m-1,k)},\dots,\fG^{(m-1,k)},\fG,\dots,\fG]$,
where there are $k$ entries $\fG^{(m-1,k)}$ at the beginning of the
$n$-bracket. Filippov's solvability \cite{Filippov}, used above,
corresponds to $k$-solvability for $k=n$; $k$-solvability implies
$n$-solvability for all $k$ \cite{Ling:93}.}

\begin{equation*}
\fG^{(0)}=\fG \quad, \quad
\fG^{(1)}=[\fG^{(0)},\dots,\fG^{(0)}] \quad ,\; \dots \; , \quad
\fG^{m}=[\fG^{(m-1)},\dots,\fG^{(m-1)}] \quad ;
\end{equation*}
if $\fG^{(1)}=0$, $\fG$ is abelian.

\begin{definition}\cite{Filippov}{\ }

A FA $\fG$ is {\it solvable} if $\fG^{(m)}=0$ for some $m$; then,
$\fG^{(m-1)}$ is an abelian ideal.
\end{definition}

The same definition of solvability applies to ideals $\fh,\fh'$ of $\fG$; the sum
$\fh+\fh'$ of two solvable ideals is also a solvable ideal of $\fG$.
This means that any finite-dimensional $\fG$ admits a maximal
solvable ideal $\hbox{Rad}(\fG)$ and, further, that
$\fG/\hbox{Rad}(\fG)$ does not contain non-zero solvable ideals. \\

The {\it radical} $\hbox{Rad}(\fG)$ of a FA $\fG$ is the
maximal solvable ideal  of $\fG$. \\

A FA is called  {\it semisimple} when $\hbox{Rad}(\fG)=0$.\\

The following three theorems \cite{Ling:93} hold :

\begin{theorem}(Semisimple FAs)
\label{ssFAandInDev} {\  }

A finite-dimensional $n$-Lie algebra $\fG$ is semisimple iff it is
the direct sum of simple ideals,
\begin{equation}
\label{sum-simple}
     \mathfrak{G} = \bigoplus_{\mathfrak{s}=1}^k
     \mathfrak{G}_{(\mathfrak{s})}
     = \mathfrak{G}_{(1)}\oplus \dots \oplus
     \mathfrak{G}_{(k)} \ .
\end{equation}
where each ideal $\fG_{(\mathfrak{s})}$ is simple as an $n$-Lie algebra. Then,
Der$\,\fG$ is also semisimple and all derivations of $\fG$ are
inner, $\hbox{Der}\,\fG$=$\hbox{InDer}\,\fG$.
\end{theorem}
The above statements directly extend to FAs the familiar Lie algebra ones
{\it e.g.}, that $ad\,\fg=Der\,\fg$ when  $\fg$ is semisimple.

\begin{theorem} (Levi decomposition of an $n$-Lie algebra) \cite{Ling:93}
\label{th:Levi-FAdec}
{\ }

Let $\fG$ be a finite-dimensional $n$-Lie algebra. Then, $\fG$
admits a Levi decomposition,
\begin{equation*}
\fG=\hbox{Rad}(\fG) \;+\!\!\!\!\!\!\supset  \;\fG_L \quad,
 \quad \hbox{Rad}(\fG)\cap \fG_L=0 \quad ,
\end{equation*}
where $\fG_L$ is a semisimple $n$-Lie subalgebra called the Levi
factor of $\fG$. Therefore, $\fG/\hbox{Rad}(\fG)$ is semisimple.
\end{theorem}

\begin{theorem}(Reductive FAs and semisimplicity)\cite{Filippov}
\label{th:redFAs}{\ }

As for ordinary Lie algebras, an $n$-Lie algebra
is {\it reductive} if its radical
$\mathrm{Rad}\,{\fG}$ is equal to its centre $Z(\fG)$; then
$\fG=Z(\fG)\oplus \fG_L$. It then follows
from the previous theorems that $\fG$ is reductive
{\it iff} the Lie algebra $\mathrm{InDer}\,\fG$ of inner derivations
is semisimple.
\end{theorem}

\subsection{Examples of Filippov algebras}
\label{gral-FA-ex}
{\ }

In its original paper \cite{Filippov}, Filippov already provided
many examples of $n$-Lie algebras, solvable and simple (in fact, all
simple ones). The $n=3$ FAs on $\mathbb{R}^4$ were given in
\cite{Gau:96}. We present explicitly here and in the next
subsections a few additional FAs of interest.

\begin{example}({\it Matrix realizations of general FAs})
\label{ex:FA-tr}

  In terms of matrices, an example of a 3-algebra is
provided by \cite{Awa-Li-Mi-Yo:99} (see also \cite{Ar-Ma-Sil:09})
\begin{equation}
\label{3-alg2} [A,B,C]=tr(A)[B,C]+ tr(B)[C,A]+tr(C)[A,B] \;.
\end{equation}
This expression may be generalized to the arbitrary $n$-case
by \cite{Awa-Li-Mi-Yo:99}
\begin{equation}
\label{n-alg-matr}
 [A_1,A_2,\dots,A_n]=
 \Sigma_{i=1}^n (-1)^{i-1} \langle A_i\rangle [A_1,A_2,\cdots,
 {\hat{A}}_i,\cdots,A_n]\quad,
\end{equation}
where the $<A_i>$ are commuting numbers associated to the $A_i$
(`traces') and $\hat{A}_i$ is absent in the ($n$-1)-bracket. If the
($n-1$)-bracket is skewsymmetric and satisfies the corresponding
($n-1$)-order FI, then the $n$-bracket is also skewsymmetric,
satisfies the FI and thus defines a FA.
\end{example}

For other matrix realizations of three-brackets see {\it e.g.},
\cite{Ho-Hou-Ma:08,Sochi:08,AA-SJ-Sim:08} and Th.~\ref{le:FA-Cliff}
below.

\subsection{The simple $n$-Lie algebras}
\label{sec:simple-n-Lie}
{\ }

\subsubsection{The euclidean $A_{n+1}$ algebras} \cite{Filippov}

Let us first consider the simple case of the Filippov three-algebra
$A_4$, which is defined on a four-dimensional real euclidean vector
space $V$. Let $v_1^a,v_2^a,v_3^a$ ($a=1,2,3,4$) be the coordinates
of three vectors $v_1,v_2,v_3\in V$ in a basis $\{\be_i\}$ of $V$.
The 3-bracket of the three vectors is then defined by the `vector
product' of $v_1,v_2,v_3$,

\begin{equation}
\label{3-bra} [v_1,v_2,v_3]=\left|
    \begin{array}{cccc}
    \be_1 & \be_2 & \be_3 & \be_4 \\
    v_1^1 & v_1^2 & v_1^3 & v_1^4 \\
    v_2^1 & v_2^2 & v_2^3 & v_2^4 \\
    v_3^1 & v_3^2 & v_3^3 & v_3^4
\end{array}
\right| \quad ,
\end{equation}
which is obviously skewsymmetric. For three basis vectors, say
$\be_1,\be_3,\be_4$, this trivially gives
\begin{equation}
[\be_1,\be_3,\be_4]=\left|
    \begin{array}{cccc}
    \be_1 & \be_2 & \be_3 & \be_4 \\
    1 & 0 & 0 & 0 \\
    0 & 0 & 1 & 0 \\
    0 & 0 & 0 & 1
\end{array}
\right| = -\be_2\quad  .
\end{equation}
The above three-bracket may also expressed by\footnote{We will use
here $\{\be_a\}$ rather than $\{X_a\}$ for the basis of the $A$ FAs
as in \cite{Filippov}.}
\begin{equation}
\label{3-basis}
[\be_{a_1},\be_{a_3},\be_{a_3}]= -\epsilon_{a_1 a_2 a_3 a_4} \be_{a_4}
\quad \mathrm{or} \quad
 [\be_1,\be_{a-1},\hat{\be}_a, \be_{a+1},
\be_4]=(-1)^{a+1}\be_a \;,\; a=1,2,3,4\;,
\end{equation}
where the hatted $\hat{e}_a$ is absent. To check that the above
three-bracket (and similarly for higher order ones) satisfies the FI
one may use the Schouten identities technique as in
Ex.~\ref{ex:Nam-FA} below. Since all three-brackets (\ref{3-bra})
follow from (\ref{3-basis}) by linearity, the above shows that the
euclidean four-vector space becomes the four-algebra $A_4$.

The $A_4$ case generalizes easily to the $n$-Lie algebra $A_{n+1}$
defined on a ($n+1$)-dimensional euclidean space of ordered basis
$\{\be_a\}, a=1\dots,n+1$. The `vector product' of $n$  vectors
$v_1,v_2,\cdots,v_n$, $v_s=v_s^a \be_a$, is defined by the
determinant
\begin{equation}
\label{n-prod} [v_1,v_2,\cdots, v_n]=\left|
    \begin{array}{ccccc}
    \be_1 & \be_2 & \cdots & \be_n & \be_{n+1} \\
    v_1^1 & v_1^2&\cdots & v_1^n & v_1^{n+1} \\
    v_2^1 & v_2^2 &\cdots & v_2^n & v_2^{n+1} \\
    \cdot  & \cdot &\cdots & \cdot  & \cdot  \\
    v_n^1 & v_n^2&\cdots & v_n^n & v_n^{n+1}
\end{array}
\right| \quad ,
\end{equation}
which defines the $n$-bracket.
In terms of the elements of the basis $\{e_a\}$ of the vector
space, the algebra is defined by
\begin{equation}
\label{basisn-bra}
 [\be_1, \dots, \be_{a-1}, \hat{\be}_a, \be_{a+1},\dots,
\be_{n+1}]= ( -1)^{a+1} \be_a \qquad (a=1,\dots,n+1)\quad ,
\end{equation}
The above expression has a different sign factor than
eq. (2) in \cite{Filippov} and, further, it does not depend on $n$
due to the different determinant arrangement. Equivalently,
\begin{equation}
\label{n-vec-prod}
 [v_1,v_2,\cdots, v_n] = \be_b\,\epsilon^b{}_{a_1\dots i_n}\,v^{a_1}_1\,v^{a_2}_2\,\cdots
 v^{a_n}_n \quad .
\end{equation}

  The euclidean $A_{n+1}$ algebra above is simple \cite{Filippov}.
In fact, all finite simple $n$-Lie algebras are known and, in
contrast with the plethora of simple Lie algebras provided by the
Cartan classification, the $n>2$ FAs are easily characterized. It
was shown by Filippov \cite{Filippov} that every ($n+1$)-dimensional
$n$-Lie algebra is simple and isomorphic  to one of the form in
eq.~\eqref{simp-n-FA} below. Further, Ling showed \cite{Ling:93}
that every finite-dimensional simple $n$-Lie algebra is of dimension
($n+1$), thus completing the classification of simple FAs:

\begin{theorem}(Classification of simple FAs)
\label{th:FAsimple}{\ }

A simple real Filippov $n$-algebra ($n\geq
3$) is isomorphic to one of the ($n+1$)-dimensional Filippov
$n$-algebras defined by
\begin{equation}
\label{simp-n-FA}
 [\be_1\dots \widehat{\be_a} \dots \be_{n+1}] =
(-1)^{a+1} \varepsilon_a \be_a \quad \mathrm{or} \quad
[\be_{a_1}\dots \dots \be_{a_n}] =
(-1)^n \varepsilon_{a_{n+1}}\epsilon_{a_1 \dots a_{n+1}} \be_{a_{n+1}} \quad ,
\end{equation}
where the $\varepsilon_a$ are signs (if the FA is not real, the
signs may be absorbed by redefining its basis).
\end{theorem}
If all $\varepsilon_a=1$, the above algebras are the euclidean
$A_{n+1}$ FAs; if there are minus signs, the ($n+1$)-Lie algebras
are Lorentzian. For instance, $n=3$ and $\varepsilon_a=1$ defines
the algebra $A_4$; if there is one minus sign, eq.~\eqref{simp-n-FA}
characterizes $A_{1,3}$ and, in general and for arbitrary $n$, it
defines the simple pseudoeuclidean Filippov algebras $A_{p,q}$ with
$p+q=n+1$. Note that we might equally well have used in
eq.~\eqref{simp-n-FA} the ${\epsilon_{a_1\dots a_n}}^a$ without the
$\varepsilon_a$ in the $r.h.s.$ by taking ${\epsilon_{a_1\dots
a_n}}^a= \eta^{ab}\epsilon_{a_1\dots a_n b}$ where $\epsilon_{1\dots
n (n+1)}=+1$ and $\eta$ is a $(n+1)\times(n+1)$ diagonal metric with
$+1$ and $-1$ in the places indicated by the $\varepsilon_a$'s. We
shall keep nevertheless the customary $\varepsilon_a$ factors as in
\cite{Filippov}.

It is easy to check, for instance, that if the basis
$\{ \be_a \}$ ($a$=1,2,3,4) defines the real euclidean three-algebra
$A_4$, the (complex) redefinitions $\be_1'=-i\be_1\,,\,
\be_2'=i \be_2\,,\, \be'_3=-\be_3\,,\,\be'_4=\be_4$
define a basis $\{ \be'_a \}$ for the real $A_{1,3}$ Lorentzian
three-algebra. Thus, in general (and since the real orthogonal
and Lorentzian FAs have the same complex form)
the {\it complex} simple $n$-Lie algebras are given by the
$n$-brackets \eqref{simp-n-FA} above without any $\varepsilon_a$ signs.
\medskip

Ling's theorem above shows that the class of finite, simple
$n$-Lie algebras is very restricted since there is essentially one simple
finite-dimensional FA for $n>2$. It has recently been shown \cite{Can-Kac:10}
that the class of simple linearly compact $n$-Lie algebras contains
four types for $n>2$, the $A_{n+1}$ FAs (called $O^n$
in \cite{Can-Kac:10}) plus three infinite-dimensional $n$-Lie algebras.
Further, for $n>2$ there are no simple linearly compact $n$-Lie superalgebras
which are not $n$-Lie algebras.

For the classification of the $(n+1)$- and $(n+2)$-dimensional $n$-Lie
algebras see \cite{Ba-So-Zha:10}.

\begin{remark}
\label{rem:GLAvsFA}
{\ }

 It is worth mentioning at this stage an
important difference between the simple FAs and the GLAs in
Th.~\ref{th:HOsimple}, which are given by Lie algebra cohomology
cocycles and have an underlying {\it Lie} group manifold. Since the
only three-dimensional simple compact Lie algebra is $su(2)$, and
has $\epsilon_{ijk}$ as its structure constants, the above theorem
states that all the euclidean simple FAs are just direct
generalizations to dimension $n+1$, $n\geq 3$, of the $n=2$ $su(2)$
Lie algebra, and that their structure constants are given by the
fully antisymmetric tensor $\epsilon_{a_1\dots a_{n+1}}$. By adding
the corresponding minus signs, all the $n\geq 3$ Lorentzian FAs
algebras are, similarly, direct generalizations of $so(1,2)$.

This simple observation hints at the rigidity of the $n\geq 3$
simple FAs and is behind the proofs \cite{Az-Iz:09} of
Thms.~\ref{th:s-s-FA-ext} and \ref{th:rigFAsim}.
\end{remark}

The simple FAs admit a realization in terms of the matrices of a
Clifford algebra by

\begin{lemma}
(GLA multibracket realization of the simple FAs)
\label{le:FA-Cliff}
{\ }

All simple real $n$-Lie algebras can be realized in terms
of even multibrackets (eq.~\eqref{multic})
involving the Dirac matrices of the Cifford algebra of
an even D-dimensional vector space.
\begin{proof}
{\ }

{\it a) $n$ odd:}

Here $D=n+1$. It will be sufficient to consider the $n$ odd euclidean
simple FA
\begin{equation}
\label{n-FA-eucl}
[\be_{a_1}\dots \be_{a_n}] = -\epsilon_{a_1\dots a_{n+1}} \be_{a_{n+1}}
\end{equation}
for the $n=3,\, D=4$ case. Let $\{\gamma_a \}$, $a=1,\dots,4$
be the gamma matrices
$\{\gamma_a,\gamma_b\}=2\delta_{ab}$ of the four-dimensional
euclidean space. The $\gamma_5$ matrix is given by
$\gamma_5 \equiv \gamma_1\gamma_2\gamma_3\gamma_4$ and squares
to one, $\gamma_5^2=1$. Then, the three-bracket (see \cite{Ba-La:07a})
\begin{equation}
\label{3FA-Cliff}
[\gamma_a,\gamma_b,\gamma_c] :=\frac{1}{4!}[\gamma_a,\gamma_b,\gamma_c,\gamma_5] \quad,
\end{equation}
defined in terms of the GLA multibracket
of eq.~\eqref{multic}, provides a gamma matrices realization of the
euclidean FA $A_4$.

To see that the above defines the three-bracket of
$A_4$, it is sufficient to notice that the full antisymmetrization of
the product $\gamma_5\gamma_a\gamma_b\gamma_c$ gives
\begin{equation*}
[\gamma_5,\gamma_{a_1},\gamma_{a_2},\gamma_{a_3}]= 4\gamma_5\,[\gamma_{a_1},\gamma_{a_2},\gamma_{a_3}]=
4\gamma_5\,3!\gamma_5\epsilon_{a_1 a_2 a_3 a_4}\gamma_{a_4} = 4!\,\epsilon_{a_1 a_2 a_3 a_4}\gamma_{a_4}\quad,
\end{equation*}
where the second bracket is also given by the full antisymmetrization of
its three entries. Therefore, the three-bracket defined by
\begin{equation}
\label{s-odd-FA-fromGLA}
[\gamma_{a_1},\gamma_{a_2},\gamma_{a_3}] :=[\gamma_{a_1},\gamma_{a_2},\gamma_{a_3},\gamma_5]' =
-\epsilon_{a_1 a_2 a_3 a_4}\gamma_{a_4}   \quad ,
\end{equation}
where the prime indicates that the skewsymmetrization of the
four-dimensional multibracket is taken there with unit weight,
realizes\footnote{For $n=3$, a gamma matrices realization of $A_4$
may be given (see {\it e.g.} \cite{Cher-Sa:08}) in terms of ordinary
two-brackets by means of the double commutator
$[[\gamma_a,\gamma_b]\gamma_5,\gamma_c]$, which is $(a,b,c)$
skewsymmetric; it is equivalent to the one above since
$3![[\gamma_a,\gamma_b]\gamma_5,\gamma_c]= [\gamma_5,
\gamma_a,\gamma_b,\gamma_c]$.} the simple euclidean algebra $A_4$
with basis $\be_a=\gamma_a$, eq.~\eqref{n-FA-eucl}.
Eq. \eqref{3FA-Cliff}, that realizes the 3-bracket of a FA in
terms of a multibracket of order four, uses the type of
bracket that appears in the Basu-Harvey equation \cite{Bas-Har:05}
(see also \cite{Ba-La:06,Ba-La:07b} in the context of the BL model),
to be dicussed in Sec~\ref{sec:BLG+BH}.

 Moving to arbitrary odd $n$ simply requires
an even $D=(n+1)$-dimensional space and a trivial generalization of
eq.~\eqref{s-odd-FA-fromGLA} involving a $D$-dimensional
multibracket $[\gamma_{a_1},\dots,\gamma_{a_n},\gamma_{D+1}]$,
where $\gamma_{D+1}$ is the `gamma five' matrix of the $D$ even
Clifford algebra. The even Lorentzian simple algebras are obtained similarly,
by replacing $\delta_{ab}$ by the Minkowski metric $\eta_{\mu\nu}$,
$\epsilon_{a_1 a_2 a_3 a_4}$ by $\epsilon_{\mu_1 \mu_2 \mu_3\mu_4}$, etc.

It is clear why the above construction does not work for even $n$:
the `$\gamma_5$' of an odd-dimensional space is a scalar matrix and
{\it e.g.} $[\gamma_5,\gamma_a,\gamma_b,\gamma_c]=0$. It is
sufficient in this case, however, to consider even multibrackets
with all the $n$ entries free.

{\it b) $n$ even}:

  For $n$ even, the realization of $A_{n+1}$ is given in terms by
the matrices of the $D=n$ Clifford algebra defined by
$\{\gamma^a,\gamma^b \} = (-1)^{n/2} 2\, \delta^{ab}$ , $a=1,\dots , n$,
plus the matrix $\gamma^{n+1} = \gamma^1 \dots \gamma^n$. Taking the
$n+1$ matrices $\gamma^A= (\gamma^a, \gamma^{n+1})$ as the basis of
the ($n+1$)-dimensional euclidean FA vector space, it follows that
\begin{equation}
\label{FGLA}
  [ \gamma^{A_1}, \dots , \gamma^{A_n} ]' =
 \epsilon^{A_1 \dots A_{n+1}}
\gamma^{A_{n+1}} \quad   (A=1,\dots n+1) \; ,
\end{equation}
where the multibracket is again that of eq.~\eqref{multic} but with
weight one antisymmetrization as before.
Note that by taking the bracket \eqref{FGLA} with its last entry
fixed to be  {\it e.g.} $\gamma^{A_n}= \gamma^{n+1}$ (see
Prop.~\ref{n-1 from n}),  the same even order $n$ multibracket
determines the odd $\tilde{n}$-Lie bracket, $\tilde{n}=n-1$, of the
odd FAs $A_{\tilde{n}+1}$ of the previous case, constructed on even
dimensional $D=\tilde{n}+1$ vector spaces.
\end{proof}
\end{lemma}

\subsection{General metric $n$-Lie algebras}
\label{sec:metric-n-FA}
{\ }

Filippov algebras may be endowed with a scalar product. In the physical
literature, metric 3-Lie algebras have been discussed in the
context of the BLG model of Sec.~\ref{sec:BLG} (see {\it e.g.}
\cite{Ba-La:07b,Ho-Hou-Ma:08,Cher-Sa:08,Go-Mi-Ru:08,deMed-JMF-Men-Rit:08}).
In fact, the Lie algebras metricity condition \eqref{metr-Lie} may be
naturally extended to the FAs. Since it is the fundamental
objects $\mathscr{X}$ that induce derivations, a scalar product on $\fG$
$\left<Y\,,\,Z\right>=g_{ab}Y^a\,Z^b$, where $Y^a,Z^b$ are the coordinates
of $Y,Z\in \fG$, will be invariant when
\begin{equation}
\label{sc-n-inv}
\begin{aligned}
\mathscr{X}\cdot \left< Y\,,\,Z \right> & = \left< \mathscr{X}\cdot Y\,,\,Z \right>+
\left< Y\,,\,\mathscr{X}\cdot Z \right> \\ &=
\left< [X_1,\dots,X_{n-1},Y]\,,\,Z \right> + \left<Y \,,\, [X_1,\dots,X_{n-1},Z] \right> =0 \quad .
\end{aligned}
\end{equation}
This may also be expressed as the condition
\begin{equation*}
(-1)^n<[Y,X_1,\dots,X_{n-1}],Z>=<Y,[X_1,\dots,X_{n-1},Z]> \quad \hbox{or}
\quad (-1)^n<Y\cdot\mathscr{X},Z>=<Y,\mathscr{X}\cdot Z> \quad,
\end{equation*}
which is the analogue of eq.~\eqref{metr-Lie} for the
Lie algebra case. In coordinates, condition \eqref{sc-n-inv} reads
\begin{equation}
\label{sc-n-inv-coor}
f_{a_1 \dots a_{n-1}b}{}^l\,g_{lc}+ f_{a_1 \dots a_{n-1}c}{}^l\,g_{bl} =0 \quad ,
\end{equation}
which reflects that the metric $g_{ab}$ is InDer$\,\fG$-invariant.

An $n$-Lie algebra endowed with a metric $g$ as above is called a
{\it metric Filippov algebra}. If $\fG$ is a metric $n$-Lie algebra,
those obtained by Prop.~\ref{n-1 from n} are metric
($n-1$)-Lie algebras since eq.~\eqref{sc-n-inv} remains
satisfied for fixed $X_1=A$.

\subsubsection{The structure constants of a metric FA as invariant antisymmetric polynomials}
\label{inv-f}
{\ }

Using the non-degenerate metric $g$ to lower indices,
it follows that the structure constants of the FA with all
the indices down are fully skewsymmetric, since they are already
in the first ($n-1$) indices and the above expression gives
\begin{equation*}
\begin{aligned}
f_{a_1 \dots a_{n-1}bc} & +   f_{a_1 \dots a_{n-1}cb} =0 \quad, \\
f_{a_1\dots a_n a_{n+1}}  = \, <[ & X_{a_1},\dots,X_{a_n}]\,, \, X_{a_{n+1}}>  \quad .
\end{aligned}
\end{equation*}

The above may be considered as the coordinates of the ($n+1$)-form
$\ff$ on $\fG$ defined by
\begin{equation*}
\ff(X_{a_1},\dots,X_{a_n},X_{a_{n+1}}) = <
[X_{a_1},\dots,X_{a_n}]\,, \, X_{a_{n+1}} > \;  .
\end{equation*}
Using now the metric $g_{ab}$ to move indices, we see that the FI,
eq.~\eqref{FIstrconst}, may also be written in the form
\begin{equation}
\label{str-cons-inv}
\sum_{i=1}^{n+1} f_{a_1\dots a_{n-1} b_i}{}^l\,f_{b_1\dots b_{i-1} l b_{i+1}\dots b_{n+1}}=0 \qquad
\mathrm{or} \qquad  f_{a_1\dots a_{n-1} [b_1}{}^l\,f_{b_2\dots b_{n+1}]l}=0
\end{equation}
on account of the skewsymmetry. This shows that the fully
antisymmetric structure constants determine an antisymmetric
covariant tensor $\ff$ on $\fG$ of rank ($n+1$) ({\it cf.}
eq.~\eqref{inv-cocycle}) which is Lie$\,\fG$-invariant;
its invariance may be expressed as $L_{\mathscr{X}}.\ff=0$.

\subsection{Nambu algebras}
\label{sec:Nambu-FA}
{\ }

  The Nambu algebra ($n=3$) and, in general,
the  Nambu-Poisson algebras for arbitrary $n$, are
infinite-dimensional FAs that follow closely the pattern of the
simple FAs where the $n$-bracket defined by the determinant `vector
product' of $n$ vectors of a ($n+1$)-dimensional space is replaced
by the Jacobian determinant of $n$ functions on an $n$-dimensional
manifold\footnote{The question of whether another interesting
determinant, the Wronskian, generates an $n$-Lie algebra is discussed
in \cite{Dzhuma:05}; see further \cite{Dzhu:02}.}.
This is why it is possible to give a Nambu
bracket version of the original BLG model, as will be shown in
Sec.~\ref{sec:BLG+NB}.

\begin{example} ({\it $n$-algebras on functions
on $\mathbb{R}^n$ and canonical Nambu bracket})
\label{ex:Nam-FA}
{\ }

 Let $f_1,f_2,\dots,f_n$ be functions on $\mathbb{R}^n$
 with coordinates $\{ x^i \},
i=1,\dots,n $. The  $n$-bracket $ \{f_1,\dots,f_n\}$ or
Nambu-Poisson bracket is defined by the Jacobian determinant
\begin{equation}
\label{Jacobian}
 \{f_1, f_2, \dots, f_n \} :=
\epsilon^{i_1\dots i_n}_{1\,\dots\,n}\;\partial_{i_1}{f^1}\dots\partial_{i_n}{f^n}=
 \left|\frac{\partial (f_1,f_2,\dots,f_n)}{\partial (x^1,x^2,\dots,x^n)}
 \right|\quad, \quad f_i=f_i (x^1,x^2,\dots, x^n) \quad .
\end{equation}
Introducing the multivector $\Lambda_{(n)}=\frac{\partial}{\partial
x^1}\wedge \dots \wedge \frac{\partial}{\partial x^n}$ it follows
that $\{f_1,\dots,f_n\}=\Lambda_{(n)}(df_1,\dots,df_n)$;
$\Lambda_{(n)}$ is the standard Nambu tensor (see further
Sec.~\ref{sec:NP}). For $n$=3, eq.~\eqref{Jacobian} is the
expression of the Nambu bracket $\{f_1,f_2,f_3\}$ studied in
\cite{Nambu:73} although, of course, Nambu also mentioned the
general $n$ case. The fact that the Jacobian defines an $n$-Lie
algebra structure was already noticed by Filippov in his original
paper \cite{Filippov} (see also \cite{Fil:98} and \cite{Dzhuma:06}
for further analysis).

We can check explicitly at this stage that the Jacobian bracket
above satisfies the FI. Consider the simplest $n=3$ case. Then, any
antisymmetrization of more than three indices gives zero, a trick
that leads to the so-called `Schouten identities'. Therefore, as far
as the skewsymmetry among the $k$ and the $j_1,j_2,j_3$ indices is
concerned,
\begin{equation*}
0\equiv \epsilon^{i_1 i_2 [k}_{1\, 2\, 3}\,\epsilon^{j_1 j_2
j_3]}_{1\, 2\, 3} \sim \big( \epsilon^{i_1 i_2
k}_{1\,2\,3}\,\epsilon^{j_1 j_2 j_3}_{1\,2\,3}- \epsilon^{i_1 i_2
j_1}_{1\,2\,3}\,\epsilon^{k j_2 j_3}_{1\,2\,3}- \epsilon^{i_1 i_2
j_2}_{1\,2\,3}\,\epsilon^{j_1 k j_3}_{1\,2\,3}- \epsilon^{i_1 i_2
j_3}_{1\,2\,3}\,\epsilon^{j_1 j_2 k}_{1\,2\,3} \big) \; .
\end{equation*}
This expression leads to the FI because
\begin{eqnarray*}
   & &    \epsilon^{i_1i_2k}_{123}  \epsilon^{j_1j_2j_3}_{123}
 \partial_{i_1} f_1 \partial_{i_2} f_2 \partial_k \left(
 \partial_{j_1} g_1 \partial_{j_2} g_2 \partial_{j_3} g_3
\right) = \cr  & &\qquad \qquad
 \left(\epsilon^{i_1i_2j_1}_{123}
\epsilon^{kj_2j_3}_{123} +\epsilon^{i_1i_2j_2}_{123}
\epsilon^{j_1kj_3}_{123}
 +\epsilon^{i_1i_2j_3}_{123}  \epsilon^{j_1j_2k}_{123}\right)
 \partial_{i_1} f_1 \partial_{i_2} f_2 \partial_k \left(
 \partial_{j_1} g_1 \partial_{j_2} g_2 \partial_{j_3} g_3
\right)
\end{eqnarray*}
implies
\begin{equation}
\label{FI-Sa-Val}
    \{  f_1,f_2, \{ g_1,g_2,g_3 \} \} = \{\{f_1, f_2, g_1\} ,g_2, g_3\} +
   \{g_1, \{f_1, f_2, g_2\} , g_3\} +  \{g_1, g_2,\{f_1, f_2, g_3\}\} \; .
\end{equation}
Actually, there are three more terms in the $r.h.s.$ of the above
equation but these cancel since, using the again a Schouten
identity, they add up to $-  \epsilon^{i_1i_2k}_{123}
\epsilon^{j_1j_2j_3}_{123}
 \partial_k \left(\partial_{i_1} f_1 \partial_{i_2} f_2 \right)
 \partial_{j_1} g_1 \partial_{j_2} g_2 \partial_{j_3} g_3$
 which obviously vanishes.

 The FI \eqref{FI-Sa-Val} is an essential ingredient of the
Nambu-Poisson mechanics \cite{Nambu:73, Tak:93} to be discussed
in Sec~\ref{sec:NP}. The time evolution of a dynamical quantity in
a ($n=3$) Nambu mechanical system is governed \cite{Nambu:73}
by two `hamiltonian functions' $H_1,H_2$ and given by
$\dot{f}=\{f,H_1,H_2\}$ and, in particular, that of
$\{f,g,h\}$ is given by $\{\{f,g,h\},H_1,H_2\}$. Thus,
the FI guarantees consistency since
\begin{equation}
\begin{aligned}
\frac{d}{dt}\{f,g,h\}
&=\{\dot{f},g,h,\}+\{f,\dot{g},h\}+\{f,g,\dot{h}\}\qquad \Leftrightarrow \cr
\{f,g,h,\{H_1,H_2\}\}
&=\{\{f,H_1,H_2\},g,h\}+\{f,\{g,H_1,H_2\},h\}+\{f,g,\{h,H_1,H_2\}\} \;.
\end{aligned}
\end{equation}
The above identity was introduced as a consistency condition for
Nambu mechanics in \cite{Sah-Val:92, Sa-Va:93}. It establishes that
the time evolution (given by the Hamiltonian vector field $X_{H_1
H_2}$) is a derivation of the Nambu bracket.
\end{example}

\begin{example} ({\it Nambu-Poisson bracket on the ring
$\mathscr{F}(M)$ of functions on a compact manifold})
\label{NPcomp}

Let $M$ be a compact, oriented manifold without boundary, with
volume form $\mu$. The previous Nambu $n$-bracket may be defined
(see {\it e.g.} \cite{Gau:96,Vai:99}) for $g_1,\dots,g_n\in
\mathscr{F}(M)$ by adopting
\begin{equation}
\label{Nam-vol}
\{g_1,\dots,g_n\}\,\mu=dg_1 \wedge\dots\wedge dg_n
\end{equation}
(obviously, the Nambu bracket defined by the Jacobian determinant in
Ex.~\ref{ex:Nam-FA} corresponds to taking $\mu =dx^1 \wedge
\dots\wedge d x^n$). Let now $X_{g_1,\dots,g_{n-1}}$ be the vector
field defined by
 $X_{g_1,\dots,g_{n-1}} \cdot g_n= \{g_1,\dots,g_{n-1}, g_n\}=
 dg_n\,(X_{g_1,\dots,g_{n-1}})$ \; ,
 Eq.~\eqref{Nam-vol} clearly implies
 \begin{equation*}
 i_{X_{g_1,\dots,g_{n-1}} } \; \mu = (-1)^{n+1} \,dg_1\wedge\dots\wedge dg_{n-1} \quad,
 \end{equation*}
since there is only a non-zero contribution from its $r.h.s.$ producing
the bracket $\{g_1,\dots,g_{n-1},g_n\}$, which then may be
factored out from both sides.
This also shows that the volume element is invariant under the
action of the {\it hamiltonian vector field} $X_{g_1,\dots,g_{n-1}}$
since $L_{X_{g_1,\dots,g_{n-1}}} \mu \equiv (i_{ X_{g_1,\dots,g_{n-1}}}\,d
+ d\,i_{ X_{g_1,\dots,g_{n-1}}})\,\mu=0$. The application of the
Lie derivative to the two sides of eq.~\eqref{Nam-vol} now gives
\begin{equation*}
\begin{aligned}
L_{X_{f_1,\dots,f_{n-1}}} ( \{g_1,\dots, g_n\} \, \mu) & =
\{f_1,\dots,f_{n-1}, \{g_1,\dots,g_n \}\}\, \mu= \\
\sum_{i=1}^n dg_1\wedge\dots\wedge d(L_{X_{f_1,\dots,f_{n-1}}}\, g_i)\wedge\dots \wedge dg_n &=
\sum_{i=1}^n\{g_1,\dots,\{f_1,\dots,f_{n-1},g_i\}\dots,g_n\} \,\mu \quad ,
\end{aligned}
\end{equation*}
where in the last line $[d,L_X]=0$ has been used on the $r.h.s.$ of
\eqref{Nam-vol}. Then, the FI for the $n$-bracket $\{g_1,\dots,g_n\}$
defined by eq.~\eqref{Nam-vol} follows.
\end{example}

The above defines the Nambu $n$-Lie algebra $\fN$, an example of an
infinite-dimensional FA. An $n=3$ $\fN$ will appear in
Sec.~\ref{sec:BLG+NB} in the context of the BLG-NB model.

\subsubsection{The Nambu bracket of functions
on a compact manifold as a metric $n$-Lie algebra}
\label{sec:NB-metric-al}
{\ }

Consider smooth functions $h,g,f \in \mathscr{F}(M))$ on a compact,
oriented $n$-dimensional manifold $M$ without boundary as in
Ex.~\ref{NPcomp}. The volume $n$-form $\mu$ allows us to define a
scalar product in  $\mathscr{F}(M)$ by
\begin{equation}
\label{Nam-metric}
\left<h,g\right> := \int_M \mu\, hg  \quad .
\end{equation}
Then, it follows that the Nambu algebra is a metric
$n$-FA. Indeed, the equivalent to eq.~\eqref{sc-n-inv} here reads
$X_{f_1,\dots,f_{n-1}}\,\left<h,g\right>=0$. Thus,
\begin{equation*}
\left<\{f_1,\dots,f_{n-1},h\}\,,\,g \right>+\left<h\,,\{f_1,\dots,f_{n-1},g\}\right>=0 \quad ,
\end{equation*}
which is satisfied since, by eq.~\eqref{Nam-vol},
\begin{eqnarray*}
\int_M \mu\, \{f_1,\dots,f_{n-1}, h \}\,g =\int_M df_1\wedge\dots
df_{n-1} \wedge dh \;g \qquad \qquad  \qquad \qquad \\
 = -\int_M h\, df_1\wedge\dots df_{n-1}
\wedge dg = -\int_M \mu \, h \{f_1,\dots,f_{n-1},g\} \; ,
\end{eqnarray*}
which follows by integrating by parts and by using Stokes theorem
for the boundaryless $M$.

\begin{lemma}{(The simple algebras as finite-dimensional subalgebras of Nambu FAs)}
\label{A-from-Nambu}
{\ }

The simple FAs (eq.~\eqref{simp-n-FA}) are finite subalgebras of
the $n$-Nambu algebras $\fN$ of functions on suitable
$n$-dimensional manifolds $M$.

\begin{proof}
Let {\it e.g.} $M_n$ be the unit sphere $S^n\subset \mathbb{R}^{n+1}$
defined by the euclidean metric on $\mathbb{R}^{n+1}$ as $y^a y_a = 1$.
Its points are characterized by $n+1$ Cartesian coordinate functions
subject to the constraint $y^a y_a = 1$,
$y^a=(\textbf{y},y^{n+1}=\sqrt{1-\textbf{y}^2}$). Then, taking
$\frac{d y^n}{\sqrt{1-\textbf{y}^2}}$ as the invariant measure on
$S^n$ it follows that
\begin{equation*}
\{y^{a_1},\dots,y^{a_n}\}= \epsilon^{a_1 \dots a_{n+1}} y^{a_{n+1}} \; ,
\end{equation*}
which defines the $A_{n+1}$ euclidean algebra (up to an irrelevant
global sign with respect to \eqref{simp-n-FA}) with basis elements
($y^1\dots y^{n+1}$). Thus, the $A_{n+1}$ FA is a subalgebra of the
infinite-dimensional Nambu algebra of functions on the compact
sphere $S^n$.

The other simple algebras are obtained by replacing $S^n$ by the
non-compact hyperboloids $\Omega^n$ defined by  suitable
pseudoeuclidean metrics on $\mathbb{R}^{n+1}$, $\mu$ by the
invariant measure on $\Omega^n$, say $\mu=\frac{d^n y}{\sqrt{1+
\textbf{y}^2}}$, to obtain the $n$-bracket of the Lorentzian simple
$A_{1,n}$ ($n+1$)-Lie algebras.
\end{proof}
\end{lemma}

The above lemma extends easily to the case of $n+1$
funcions $\phi^a(y)$ on a manifold $M$, subject to the
constraint
\begin{equation}
\label{func-sph} \phi_1^2+\dots+ \phi_{n+1}^2 =1 \ .
\end{equation}
There is, however, an interesting topological subtlety that we
analyze below\footnote{We thank Paul Townsend for a helpful comment
on this point.}. Consider $M=S^n$; then, the functions $\phi^a(y)$
define a map $\phi: S^n \rightarrow S^n$, where the first sphere is
parametrized by the $n+1$ coordinates $y^a, y^a y_a=1$, and the
second sphere by $\phi^a$. Clearly, the possible maps $\phi$ fall
into disjoint homotopy classes characterized by a winding number and
thus the degree of the map $\phi$ should be reflected in the
definition of the Nambu bracket.

Let $\mu$ be the volume form on the first sphere $S^n$, for which
we take the expression
\begin{equation}
\mu(y)=\label{degree1} \sum^{n+1}_{i=1} (-1)^{n+i-1} y^i dy^1 \wedge \dots
\wedge \widehat{dy^i} \wedge \dots \wedge dy^{n+1} \quad ,\quad y^a y_a=1 \quad ;
\end{equation}
the factors have been chosen to reproduce the measure used above. Similarly,
the volume form on the second sphere is given by
\begin{equation}
\label{vol-2nd-sph}
\mu'(\phi)=\sum^{n+1}_{i=1} (-1)^{n+i-1} \phi_i d\phi_1 \wedge \dots
\wedge \widehat{d\phi_i} \wedge \dots \wedge d\phi_{n+1} \quad ,\quad \phi^a\phi_a=1 \quad .
\end{equation}
By eq.~\eqref{Nam-vol}, the Nambu bracket $\{ \phi_{a_1}(y), \dots , \phi_{a_n}(y) \}$
of functions in $\mathscr{F}(M)$ is defined by
\begin{equation}
\label{degree2} \{ \phi_{a_1}(y), \dots , \phi_{a_n}(y) \} \mu(y) =
d\phi_{a_1}(y) \wedge \dots \wedge d\phi_{a_n}(y) \; ,
\end{equation}
and we would like to see how it is expressed in terms of $\phi_ {a_{n+1}}(y)$.
To this aim, it is sufficient to show that the $r.h.s.$ of \eqref{degree2} is
in fact given by
\begin{equation}
\label{degree3} d\phi_{a_1}(y) \wedge \dots \wedge d\phi_{a_n}(y) =
 \epsilon_{a_1\dots a_{n+1}} \phi_{a_{n+1}}(y) \phi^* (\mu'(\phi)) \ ,
\end{equation}
where $\phi^*(\mu')$ is the pull-back of the volume form $\mu'(\phi)$
to the first sphere, hence given by
\begin{eqnarray}
\label{degree4}
   (\phi^*(\mu'))(y) &=& \sum^{n+1}_{i=1} (-1)^{n+i-1} \phi_i(y) d\phi_1(y) \wedge \dots
\wedge \widehat{d\phi_i(y)} \wedge \dots \wedge d\phi_{n+1}(y)
\nonumber\\
 & & = \frac{(-1)^n}{n!} \epsilon^{b_1\dots b_{n+1}} \phi_{b_1}(y) d\phi_{b_2}(y) \wedge \dots
\wedge d\phi_{b_{n+1}}(y)\ ,
\end{eqnarray}
with $\phi^*(\phi^a\phi_a)=\phi^a(y)\phi_a(y)=1$. Then,
and omitting the $y$ dependence, the $r.h.s.$ of \eqref{degree3}
is equal to
\begin{eqnarray}
\label{degree5}
  & &  \epsilon_{a_1\dots a_{n+1}} \phi_{a_{n+1}} \frac{(-1)^n}{n!}
  \epsilon^{b_1\dots b_{n+1}} \phi_{b_1} d\phi_{b_2} \wedge \dots
\wedge d\phi_{b_{n+1}} \nonumber\\
& & \quad\quad\quad = \frac{(-1)^n}{n!} \phi_{a_{n+1}} \phi_{b_1}
\epsilon_{a_1\dots a_{n+1}}^{b_1\dots b_{n+1}} d\phi_{b_2} \wedge
\dots \wedge d\phi_{b_{n+1}} \nonumber\\
& & \quad\quad\quad = \frac{(-1)^n}{n!} \phi_{a_{n+1}} \phi_{b_1}
\sum^{n+1}_{i=1} (-1)^{n-i+1} \delta^{b_i}_{a_{n+1}}
\epsilon_{a_1\dots \ \dots a_n}^{b_1\dots \hat{b_i} \dots b_{n+1}}
d\phi_{b_2} \wedge \dots \wedge d\phi_{b_{n+1}} \nonumber\\
& & \quad\quad\quad = \frac{1}{n!} \phi_{a_{n+1}} \phi_{b_1}
\delta^{b_1}_{a_{n+1}} \epsilon_{a_1\dots a_n}^{b_2\dots
b_{n+1}} d\phi_{b_2} \wedge \dots \wedge d\phi_{b_{n+1}} \nonumber
\\
& & \quad\quad\quad = d\phi_{a_1}(y) \wedge \dots \wedge d\phi_{a_n}(y)
\; ,
\end{eqnarray}
as we wanted to show. Above we have used eq.~\eqref{eps-1} and,
in the third and fourth lines of the expression, that
$\phi^{b_i} d\phi_{b_i}=0$ (in the terms with $i=2,\dots,n+1$) and
$\phi^{b_1} \phi_{b_1}=1$, respectively.

Now, as is well known, the degree deg$\,\phi$ of the mapping $\phi$ is an integer
that may be expressed by the (Kronecker) integral
\begin{equation}
\label{degree6}
 deg\, \phi = \frac{1}{\mathrm {vol}(S^n)} \int_{S^n} \phi^*(\mu') \quad .
\end{equation}
Thus, since $\int \mu(y)= \mathrm{vol}(S^n)$, the above is tantamount to
$(\phi^*(\mu'))(y) = (\mathrm{deg}\, \phi) \mu(y)$.
Therefore, inserting eq.~\eqref{degree3} into eq.~\eqref{degree2} we
finally obtain
\begin{equation}
\label{degree8}
\{ \phi_{a_1}(y), \dots , \phi_{a_n}(y) \} =
(\mathrm{deg} \phi) \, \epsilon_{a_1\dots a_{n+1}} \phi_{a_{n+1}}(y) \; .
\end{equation}
Thus, the Nambu-brackets of the $S^n$-constrained functions
$\phi^a(y)$ are classified by the Brouwer degree
of the map $\phi$. For instance, deg$\,\phi=1$ for the identity map,
which gives the standard form of the $A_{n+1}$
simple algebras as realized by $S^n$-constrained functions on $S^n$.
\medskip

\subsection{Fundamental objects: definition and properties}
\label{sec:fund-obj}{\ }

The relevance of the fundamental objects $\mathscr{X}$ for an $n$-Lie
algebra $\fG$ stems from the fact that, as already seen, they define
inner derivations $ad_\mathscr{X}\in \mathrm{InDer}\,\fG$. The
$\mathscr{X}$s will also play a crucial r\^ole in FA cohomology. Let
us recall their definition to analyze their properties. {\ }

\definition{({\it Fundamental objects $\mathscr{X}$ for a FA})}
\label{fdtal-objects}

A fundamental object $\mathscr{X}$ of a FA is determined by ($n-1$)
elements $X_1,X_2,\dots,X_{n-1}$ of $\fG$ on which it is
skewsymmetric; thus\footnote{The notation $\mathscr{X}\in
\wedge^{n-1}\mathfrak{G}$ merely reflects that the fundamental
object $\mathscr{X}=(X_1,\dots,X_{n-1})\in \mathfrak{G}\times
\mathop{\dots}\limits^{n-1}\times \mathfrak{G}$ is antisymmetric in
its arguments; it does not imply that $\mathscr{X}$ is a
$(n-1)$-multivector obtained by the associative wedge product of
vectors.}, $\mathscr{X} \in \wedge^{n-1}\fG$. The fundamental
objects act on the elements of the FA by (left) multiplication:
\begin{equation}
\label{xidot} \mathscr{X}\in \wedge^{n-1}\fG \quad, \quad
 ad_{X_1 \dots X_{n-1}}\equiv ad_{\mathscr{X}} \quad,\quad
ad_{\mathscr{X}}\cdot Z \equiv \mathscr{X}\cdot Z :=
[X_1,X_2,\dots,X_{n-1},Z] \; .
\end{equation}

As a result, $ad:\wedge^{n-1}\fG\to \hbox{End}\,\fG\;$ is
characterized by an $n$-bracket having void its last entry,
\begin{equation}
\label{(n-1)-ad-def}
ad:\mathscr{X}\mapsto ad_{\mathscr{X}}\equiv
\mathscr{X}\cdot \equiv [X_1,X_2,\dots,X_{n-1}, \quad] \quad;
\end{equation}
a {\it fundamental object defines an inner derivation} of the FA
$\fG$. For $n$=2, the fundamental objects $\mathscr{X}$ are the
elements of the Lie algebra $\fg$ themselves and $ad_{\mathscr{X}}$
reduces to $ad_X$.

Chosen a basis, we will often use the obvious notation
\begin{equation}
\label{ad-notat}
ad_{X_{a_1} \dots X_{a_{n-1}}} \equiv ad_{a_1 \dots a_{n-1}} \; , \quad
ad_{a_1 \dots a_{n-1}} \cdot X_b := [X_{a_1},\dots,X_{a_{n-1}},X_b]
=f_{a_1\dots a_{n-1} b}{}^l X_l \; .
\end{equation}
Clearly the ($\hbox{dim}\fG \times \hbox{dim}\fG$)-dimensional matrix
$ad_{a_1 \dots a_{n-1}} \in \hbox{End} \,\fG$ is given by
\begin{equation}
\label{FA-adj-matr}
(ad_{a_1 \dots a_{n-1}})^l{}_b = f_{a_1\dots a_{n-1} b}{}^l   \quad ,
\end{equation}
to be compared with the $\fg$ case at the end of Sec.~\ref{Lie-defs}.

\begin{definition} ({\it Composition of fundamental objects}) \cite{Gau:96}
\label{def:dotcompo}{\ }

Given two fundamental objects $\mathscr{X}\,,\,\mathscr{Y} \in
\wedge^{n-1}\fG$, the (non-associative) composition
$\mathscr{X}\cdot\mathscr{Y}\in \wedge^{n-1}\fG$ of the
two is the bi- and $i$-linear map
$\wedge^{n-1}\fG \otimes \wedge^{n-1}\fG \rightarrow \wedge^{n-1}\fG$
given by the sum
\begin{equation}
\label{dot-compo}
\begin{aligned}
\mathscr{X} \cdot \mathscr{Y} := &
\sum_{i=1}^{n-1}(Y_1,\dots,Y_{i-1},\mathscr{X}\cdot Y_i\,,Y_{i+1},\dots,Y_{n-1})
\cr = &
\sum_{i=1}^{n-1}(Y_1,\dots,Y_{i-1},[X_1,X_2,\dots,X_{n-1},Y_i],Y_{i+1},\dots,Y_{n-1})
\;,
\end{aligned}
\end{equation}
which is the natural extension of the adjoint derivative
$ad_\mathscr{X}$ action on $\fG$ to $\wedge^{n-1}\fG$.
\end{definition}
Thus, the dot composition $\mathscr{X} \cdot \mathscr{Y}$ defines
the inner derivation given by $(\mathscr{X} \cdot \mathscr{Y})\cdot
Z= [\mathscr{X} \cdot
\mathscr{Y}\,,\,Z]:=\sum_{i=1}^{n-1}\,[Y_1,\dots,Y_{i-1},\mathscr{X}
\cdot Y_i,\dots,Y_{n-1},Z]$; the composition of a fundamental object
with itself always determines the trivial derivation.

 The following Lemma follows from the FI:
\begin{lemma}(Properties of the composition of fundamental objects)
\label{le:deriv-lemma}{\, }\\

 The dot product of fundamental objects
$\mathscr{X}$ of an $n$-Lie algebra $\fG$ satisfies the relation
\begin{equation}
\label{Lie-setb}
\mathscr{X}\cdot (\mathscr{Y}\cdot \mathscr{Z}) -
\mathscr{Y}\cdot (\mathscr{X} \cdot \mathscr{Z}) =
(\mathscr{X}\cdot\mathscr{Y}) \cdot \mathscr{Z} \qquad \forall
\mathscr{X}, \mathscr{Y}, \mathscr{Z} \in \wedge^{n-1}\fG \; .
\end{equation}
As a result, the images $ad_{\mathscr{X}}$ of the fundamental
objects by the adjoint map $ad:\wedge^{n-1}\fG \rightarrow \hbox{InDer}\,\fG$
determine inner derivations that satisfy
\begin{equation}
\label{Lie-n-Lie}
\begin{aligned}
\mathscr{X}\cdot(\mathscr{Y}\cdot Z)-& \mathscr{Y}\cdot(\mathscr{X}\cdot Z)
=(\mathscr{X}\cdot\mathscr{Y})\cdot Z \quad \\
\mathrm{or} \qquad
 ad_{\mathscr{X}}ad_{\mathscr{Y}} Z - ad_{\mathscr{Y}}  & ad_{\mathscr{X}} Z =
ad_{\mathscr{X}\cdot\mathscr{Y}} Z  \qquad \forall
\mathscr{X}, \mathscr{Y} \in \wedge^{n-1}\fG \,,\,\forall \; Z\in\fG \quad ,
\end{aligned}
\end{equation}
which is equivalent to the FI
for $[X_1,\dots,X_{n-1},[Y_1,\dots,Y_{n-1},Z]]$.
It then follows that the inner derivations $ad_{\mathscr{X}}$ of a FA $\fG$
generate an ordinary Lie algebra, Lie$\,\fG\,\equiv\,$InDer$\,\fG\,\equiv\,$ad$\,\fG$.
\begin{proof}
Let us compute first $\mathscr{X}\cdot (\mathscr{Y}\cdot \mathscr{Z})$.
\begin{equation*}
\begin{aligned}
\mathscr{X}  & \cdot (\mathscr{Y}\cdot \mathscr{Z})=
\sum_{i=1}^{n-1} \mathscr{X} \cdot \left(
Z_1,\dots,Z_{i-1},[Y_1,\dots,Y_{n-1},Z_i],Z_{i+1},\dots, Z_{i-1} \right) \\
  = & \sum_{i=1}^{n-1}  \sum_{j \not=i, \,j=1}^{n-1} \left(
 Z_1,\dots,Z_{j-1},[X_1,\dots,X_{n-1},Z_j],\ Z_{j+1},  \dots,
 Z_{i-1},\ [Y_1,\dots,Y_{n-1},Z_i],Z_{i+1},\dots,Z_{n-1}\right)\\
& \quad + \sum_{i=1}^{n-1} \left(
Z_1,\dots,Z_{i-1},[X_1,\dots,X_{n-1},[Y_1,\dots,Y_{n-1},Z_i]],
 Z_{i+1},\dots, Z_{n-1} \right) \; .
\end{aligned}
\end{equation*}
The first term in the $r.h.s$ is symmetric in $\mathscr{X},
\mathscr{Y}$; hence,
\begin{equation}
\begin{aligned}
\label{n-L-s} & \mathscr{X}\cdot (\mathscr{Y}\cdot \mathscr{Z}) -
\mathscr{Y}\cdot (\mathscr{X} \cdot \mathscr{Z})  =\\
\sum_{j=1}^{n-1} & \left( Z_1, \dots,Z_{j-1}, \{\,
[X_1,\dots,X_{n-1}, [Y_1,\dots,Y_{n-1},Z_j]] - [Y_1,\dots,
Y_{n-1},[X_1,\dots,X_{n-1},Z_j ]]\, \}, Z_{j+1},\dots,Z_{n-1}
\right) \; .
\end{aligned}
\end{equation}
On the other hand, using definition \eqref{dot-compo}, we find
\begin{equation}
\label{ad-nueva1}
(\mathscr{X} \cdot \mathscr{Y} )\cdot \mathscr{Z}=
\sum_{j=1}^{n-1}\sum_{i=1}^{n-1} (Z_1,\dots,Z_{j-1},
[Y_1,\dots,[X_1,\dots,X_{n-1},Y_i],\dots,Y_{n-1},Z_j],Z_{j+1},\dots,Z_{n-1}).
\end{equation}
Now, using the FI for
$[X_1,\dots,X_{n-1},[Y_1,\dots,Y_{n-1},Z_j]]$, we see that the
above expression reproduces \eqref{n-L-s}.

 The above proof obviously carries forward to the simplest case where the
fundamental object $\mathscr{Z}$ in \eqref{Lie-setb} is replaced
by a FA element $Z$ as in \eqref{Lie-n-Lie}. It is sufficient to
note that, by the FI,
\begin{equation}
\begin{aligned}
ad_{\mathscr{X}\cdot\mathscr{Y}} Z &=
\sum_{i=1}^{n-1}[Y_1,\dots,Y_{i-1},[X_1,\dots,X_{n-1},Y_i],Y_{i+1},\dots,Y_{n-1},Z]
\\ & = [X_1,\dots,X_{n-1},[Y_1,\dots,Y_{n-1},Z]] -
[Y_1,\dots,Y_{n-1},[X_1,\dots,X_{n-1},Z]]\\ &=
ad_{\mathscr{X}}ad_{\mathscr{Y}}Z-ad_{\mathscr{Y}}ad_{\mathscr{X}}Z
\quad \forall\,Z\in\fG\quad .
\end{aligned}
\end{equation}
\end{proof}
\end{lemma}
\noindent
This expression also shows, by exchanging the $X$s and the $Y$s, that
$ad_{\mathscr{X}\cdot\mathscr{Y}}Z=-ad_{\mathscr{Y}\cdot\mathscr{X}}Z$
on any $Z\in \fG$, and hence that
\begin{equation}
\label{xy.skew} ad_{\mathscr{X}\cdot\mathscr{Y}} =
-ad_{\mathscr{Y}\cdot\mathscr{X}} \qquad \hbox{or, equivalently,}
\qquad (\mathscr{X}\cdot\mathscr{Y})\,\cdot =
-(\mathscr{Y}\cdot\mathscr{X})\,\cdot \quad .
\end{equation}
Note the dots after the brackets: $\mathscr{X}\cdot\mathscr{Y}
\neq -\mathscr{Y}\cdot\mathscr{X}$  but
$(\mathscr{Y}\cdot \mathscr{X})\cdot Z =
-(\mathscr{Y}\cdot \mathscr{X})\cdot Z$ .
\medskip

As exhibited by eq.~\eqref{Lie-setb} or \eqref{Lie-n-Lie}, the
composition law $\mathscr{X}\cdot\mathscr{Y}$ is not associative; in
fact, eq.~\eqref{Lie-setb} measures the lack of associativity,
$\mathscr{X}\cdot(\mathscr{Y}\cdot\mathscr{Z})
-(\mathscr{X}\cdot\mathscr{Y})\cdot\mathscr{Z}
=\mathscr{Y}\cdot(\mathscr{X}\cdot\mathscr{Z})$ as in the standard
Lie algebra case. Indeed, for $n$=2 \eqref{Lie-n-Lie} reproduces the
JI written in the form $[X,[Y,Z]]-[[X,Y],Z]= [Y,[X,Z]]$, which
exhibits the obvious lack of associativity of the Lie bracket. For
$n>2$, eqs.~\eqref{Lie-setb} and \eqref{Lie-n-Lie} (and \eqref{gmod}
below) are a consequence of the characteristic identity that defines
the $n$-Lie algebra $\fG$, the FI.
\medskip

\subsection{Kasymov's criterion for semisimplicity of a FA}
\label{sec:Kas-crit}
{\ }

Kasymov's analogue of the Cartan-Killing metric for the case
of a FA $\fG$ is the $2(n-1)$-linear generalization
\begin{equation}
\label{KasCK}
k(\mathscr{X},\mathscr{Y}) = k(X_1,\dots ,
X_{n-1},Y_1,\dots ,Y_{n-1}) :=Tr(ad_{\mathscr{X}}ad_{\mathscr{Y}}) \quad, \quad
\mathscr{X},\mathscr{Y}\in \wedge^{n-1}\fG\; .
\end{equation}
Then, a FA is semisimple if it satisfies the
following generalization \cite{Kas:95a} of the Cartan criterion:

\begin{theorem}({\it Semisimplicty of a FA})
\label{th:Kas-crit}
{\ }

An $n$-Lie algebra is semisimple iff the Kasymov form above
is non-degenerate in the sense that
\begin{equation}
\label{intro7}
        k(Z, \mathfrak{G},{\mathop{\dots}\limits^{n-2}},\mathfrak{G},
        \mathfrak{G}, {\mathop{\dots}\limits^{n-1}}, \mathfrak{G})=0
        \ \Rightarrow \ Z=0 \; ,
\end{equation}
where the $2n-3$ arguments besides $Z$ are arbitrary elements of
$\mathfrak{G}$.
\end{theorem}
\begin{remark} ({\it Kasymov form and simple FA algebras})
\label{rem:Kas-simp}{\ }

One might have thought of another generalization of Cartan's
criterion by looking at the form $k$ in eq.~\eqref{KasCK}
as a bilinear form on $\wedge^{n-1}\mathfrak{G}$,
\begin{equation*}
 k:  \wedge^{n-1} \mathfrak{G} \times \wedge^{n-1} \mathfrak{G}
     \longrightarrow \mathbb{R} \; ,
\end{equation*}
to analyze the consequences of non-degeneracy in the usual sense,
$k(\mathscr{X},\mathscr{Y})=Tr(ad_{\mathscr{X}}ad_{\mathscr{Y}}) =
0\quad\forall\,\mathscr{Y}\in \wedge^{n-1} \mathfrak{G} \;
\Rightarrow \mathscr{X}=0$, perhaps thinking of extending to FAs the
semisimplicity proof that holds for Lie algebras with a
non-degenerate Cartan-Killing form. But we see immediately that this
does not work for $n\geq 3$. Any semisimple $n$-Lie algebra is the
direct sum  of its simple ideals \cite{Ling:93},
eq.~\eqref{sum-simple}. As a consequence, an $n$-bracket
$[\dots,X,\dots,Y,\dots]$ is zero whenever $X$ and $Y$ belong to
different simple ideals. Then, if $X_1,\dots ,X_{n-2}$ and $Y$ are
in different ideals, it follows that $ad_{\mathscr{X}}$ for
$\mathscr{X}=(X_1,\dots, X_{n-2}, Y)$ is the zero derivation, so
that $k(\mathscr{X},\mathscr{Y})=0$ for any $\mathscr{Y}$ without
$\mathscr{X}$ itself being zero ({\it cf.} Th.~\ref{th:Kas-crit}).
This cannot happen for an $n=2$ Filippov (Lie) algebra $\fg$, since
its fundamental objects are determined by a single element of $\fg$.

   Nevertheless, for {\it simple} $n$-Lie algebras $k$ is nondegenerate
as a bilinear metric on $\wedge^{n-1} \mathfrak{G}$. To show this,
we may use that a real simple $n$-Lie algebra is
\cite{Ling:93,Filippov} one of the FAs in eq.~\eqref{simp-n-FA}.
Since $k$ on $\wedge^{n-1}\mathfrak{G}$ is determined by its values
on a basis, we take $\mathscr{X}=(\textbf{e}_{a_1}, \dots ,
\textbf{e}_{a_{n-1}})$, $\mathscr{Y}=(\textbf{e}_{b_1}, \dots ,
\textbf{e}_{b_{n-1}})$. Using eq.~\eqref{xidot}, the action of
$ad_{\mathscr{X}} ad_{\mathscr{Y}}$ on the vector $\textbf{e}_c$ is
found to be
\begin{equation}
\label{simple1}
 ad_{\mathscr{X}} ad_{\mathscr{Y}} \,\textbf{e}_c = \sum_{d,g=1}^{n+1}
 \varepsilon_d \varepsilon_g \epsilon_{b_1 \dots b_{n-1} c}{}^d
 \epsilon_{a_1 \dots a_{n-1} d}{}^g \,\textbf{e}_g  \; ,
\end{equation}
so that the trace of the matrix
$ad_{\mathscr{X}} ad_{\mathscr{Y}}$ is given by
\begin{equation}
  k(\mathscr{X},\mathscr{Y}) =Tr(ad_ \mathscr{X}ad_\mathscr{Y})= \sum_{d,g=1}^{n+1}
 \varepsilon_d \varepsilon_g \epsilon_{b_1 \dots b_{n-1} g}{}^d
 \epsilon_{a_1 \dots a_{n-1} d}{}^g
 \equiv  k_{(a_1\dots a_{n-1})(b_1\dots b_{n-1})} \; .
\label{simple2}
\end{equation}
The numbers appearing in the $r.h.s.$ of (\ref{simple2}), seen as
a ${n+1 \choose n-1} \times {n+1 \choose n-1}$ matrix of
indices $(a_1\dots a_{n-1})(b_1\dots b_{n-1})$ characterized
by the fundamental objects above, determine a diagonal one
with non-zero diagonal elements. Indeed, given an index
determined by a certain set $(a_1 \dots a_{n-1})$, the antisymmetric
tensor $\epsilon_{a_1 \dots a_{n-1} d}{}^g$ fixes the remaining
$d$ and $g$ (and $\varepsilon_d \varepsilon_g $) so that
the other matrix index $(b_1, \dots, b_{n-1})$ has to
be given by a reordering of $(a_1 \dots a_{n-1})$. For
this reason the only non-zero elements of the
$k_{(a_1\dots a_{n-1})(b_1\dots b_{n-1})}$ matrix
are all its diagonal ones, and hence it is non-degenerate.
For $n=2$ one just finds, of course, that the Kasymov matrix
is proportional to $-\mathbb{I}_3$ and that $su(2)$ is
simple.
\end{remark}

\section{The Lie algebras associated to $n$-Lie algebras}
\label{sec:nLie-to-Lie}

\subsection{Preliminaries: the Lie algebra Lie$\,\fG$ associated
with a 3-Lie algebra $\fG$}
\label{3-to-Lie}{\ }

 Let $\fG$ be a 3-Lie algebra. Its three-bracket satisfies
the FI,
\begin{equation}
[X_1,X_2, [Y_1,Y_2,Z]] = [[X_1,X_2,Y_1],Y_2,Z]+[Y_1,
[X_1,X_2,Y_2],Z]+[Y_1,Y_2,[X_1,X_2,Z]] \; .
\end{equation}
Moving the last term to the {\it l.h.s.} and writing it in terms
of the adjoint map, $ad_{X_1,X_2}:Z \mapsto [X_1,X_2,Z]$, the
above equation reads
\begin{equation}
\label{actonz}
ad_{X_1,X_2} (ad_{Y_1,Y_2} \,Z)-ad_{Y_1,Y_2}
(ad_{X_1,X_2} \,Z) =
 ad_{([X_1,X_2,Y_1],Y_2)+ (Y_1,[X_1,X_2,Y_2])}\, Z
\end{equation}
or, equivalently (see Sec.~\ref{sec:fund-obj}),
\begin{equation}
\label{Lie-3-FA}
ad_{\mathscr{X}} (ad_{\mathscr{Y}} Z)-
ad_{\mathscr{Y}} (ad_{\mathscr{X}} Z) =
ad_{\mathscr{X}\cdot\mathscr{Y}}\, Z \qquad \forall
\mathscr{X},\mathscr{Y}\in\wedge^2\fG \,,\, Z\in\fG
\end{equation}
which, as already discussed (see eq.~\eqref{xy.skew}) is
$\mathscr{X}\leftrightarrow \mathscr{Y}$
skewsymmetric\footnote{For $n$=2, where $[X,Y]=XY-YX$, this equation
is the familiar $[[X,\quad] , [Y,\quad]]=[[X,Y], \quad]$, eq.~\eqref{adJI},
which is manifestly $X\leftrightarrow Y$ skewsymmetric in both sides.}
albeit $\mathscr{X}\cdot\mathscr{Y}\not=-\mathscr{Y}\cdot\mathscr{X}$
in the $r.h.s$ for $n>2$ .

Clearly, the commutator of two inner derivations is another
one. This may be more explicitly seen by rewritting
the above expressions as
\begin{equation}
\label{actonzb}
 [\,[X_1,X_2, \quad]\,,\,[Y_1,Y_2,\quad]\,] =
 [\,[X_1,X_2,Y_1],Y_2, \quad] + [Y_1,[X_1,X_2,Y_2] , \quad] \quad .
\end{equation}
Further, {\it the inner derivations of a 3-algebra $\fG$
determine an ordinary Lie algebra} since the Lie bracket of two
derivations $[\,[X_1, X_2\,,\quad]\,,\,[Y_1, Y_2\,,\quad]\,]$
satisfies the JI
\begin{equation}
\sum_{XYZ\, cycl.}[\,[X_1, X_2,\quad],[\,[Y_1, Y_2,\quad],
[Z_1, Z_2,\quad]\,]\,]=0 \quad ,
\end{equation}
as it is evident from Prop.~\ref{n-1 from n} above and
will be seen for general $n$ below.  Since
$ad:\wedge^2\fG\mapsto \mathrm{End}\,\fG$ may have
a non-trivial kernel, the map
$\mathscr{X}_{a_1 a_2}\in \wedge^2 \fG \rightarrow ad_{\mathscr{X}_{a_1 a_2}}$
is not injective.

The Lie algebra associated with a 3-Lie algebra will
become essential to define the symmetries of the BLG model
\cite{Ba-La:07b,Ba-La:07a,Gustav:08,Cher-Sa:08} in
Sec.~\ref{sec:BLG}.\\

\subsubsection{Coordinate expressions for $n$=3}
\label{sec:coord3alg}{\  }

Let $n=3$. The coordinates of the $\hbox{dim}\,\fG\times \hbox{dim}\,\fG$ matrices
$[X_{a_1},X_{a_2},\quad]\equiv \mathscr{X}_{a_1 a_2}\cdot
\equiv (\mathscr{X}_{a_1 a_2}) \equiv ad_{a_1 a_2} \in \hbox{End}\fG$  are given by
\begin{equation}
\label{3Liealelem}
 ad_{a_1 a_2}{}^l{}_k\equiv(\mathscr{X}_{a_1 a_2})^l{}_k =
 f_{a_1 a_2 k}{}^l \quad, \quad \mathscr{X}_{a_1 a_2}
 \cdot X_k=[X_{a_1},X_{a_2},X_k]= f_{a_1 a_2 k}{}^l\;X_l \quad .
\end{equation}
Then, the coordinates expression for eq.~\eqref{actonzb} reads
\begin{equation}
\label{nueva-sekew}
(\mathscr{X}_{a_1 a_2})^s{}_l\; (\mathscr{X}_{b_1 b_2})^l{}_k -
(\mathscr{X}_{b_1 b_2})^s{}_l \; (\mathscr{X}_{a_1 a_2})^l{}_k \equiv
f_{a_1 a_2 l}{}^s\; f_{b_1 b_2 k}{}^l - f_{b_1 b_2 l}{}^s \;f_{a_1 a_2 k}{}^l =
f_{a_1 a_2 b_1}{}^l\; f_{l b_2 k}{}^s + f_{a_1 a_2 b_2}{}^l\; f_{b_1 l k}{}^s \quad ,
\end{equation}
and the last equality (with $k=b_3$) reproduces the FI \eqref{coord3FI} as it should.

The above equation may be written in the form
\begin{equation*}
[(\mathscr{X}_{a_1 a_2}),(\mathscr{X}_{b_1 b_2})]^s{}_k= -f_{a_1 a_2 [b_1}{}^l\;f_{b_2]l k}{}^s
\end{equation*}
which means that we may express the above Lie commutators as \cite{Gustav:08}
\begin{equation}
\label{3Cs}
[(\mathscr{X}_{a_1 a_2}),(\mathscr{X}_{b_1 b_2})]^s{}_k=
\frac{1}{2}C_{a_1 a_2 \; b_1 b_2}{}^{c_1 c_2} \,(\mathscr{X}_{c_1 c_2})^s{}_k
\end{equation}
taking, for instance,
\begin{equation}
\label{3Csb}
C_{a_1 a_2 \; b_1 b_2}{}^{c_1 c_2} = f_{a_1 a_2[b_1}{}^{[c_1}\;\delta^{c_2]}_{b_2]}  \quad ,
\end{equation}
However, this does not mean that the above $C$'s are the structure
constants of Lie$\,\fG$ on two counts. First, although the $r.h.s.$
of expressions eq.~\eqref{nueva-sekew}-\eqref{3Cs} are $(a_1
a_2)\leftrightarrow (b_1 b_2)$ skewsymmetric as mandated by their
$l.h.s$, this does not necessarily imply that the constants in
eq.~\eqref{3Csb} retain this property since the sum over $(c_1 c_2)$
has been removed. One may, of course, write antisymmetric $C$'s in
eq.~\eqref{3Cs} by taking
\begin{equation}
\label{skew-3Cs}
C_{a_1 a_2 \; b_1 b_2}{}^{c_1 c_2}
= \frac{1}{2}\left(f_{a_1 a_2[b_1}{}^{[c_1}\;\delta^{c_2]}_{b_2]}
- (a\leftrightarrow b)\right)
\end{equation}
but, secondly, this is not sufficient for them to
be the structure constants of Lie$\,\fG$ since, in general,
the ($c_1,c_2$)-labeled matrices
$(\mathscr{X}_{c_1 c_2})$ are not a basis\footnote{Of course, the
$(a_1a_2)\leftrightarrow(b_1b_2)$ antisymmetry of the $C$s and
the JI hold on the matrices $(\mathscr{X}_{e_1 e_2})$,
\begin{eqnarray*}
  & & \qquad \qquad \left(C_{a_1 a_2 \; b_1 b_2}{}^{e_1 e_2}+
    C_{b_1 b_2 \; a_1 a_2}{}^{e_1 e_2}\right) (\mathscr{X}_{e_1 e_2})^s{}_l= 0
    \quad ,\nonumber \\
& &     \sum_{\textrm{cycl.} (a_1a_2),(b_1b_2),(c_1c_2)}
  \left(C_{a_1 a_2 \; b_1 b_2}{}^{d_1 d_2} C_{c_1 c_2 \; d_1 d_2}{}^{e_1
    e_2}\right) (\mathscr{X}_{e_1 e_2})^s{}_l= 0  \quad ,
\end{eqnarray*}
since this is what follows from \eqref{3Cs} and the JI in
End$\,\fG$, $\,\sum_{cycl.}[\,[(\mathscr{X}_{a_1
a_2}),(\mathscr{X}_{b_1 b_2})],(\mathscr{X}_{c_1 c_2})\,]=0$.
However, the $(\mathscr{X}_{e_1 e_2})$ cannot be removed in the
above equations.} of $\mathrm{Lie}\,\fG$.

When the 3-Lie algebra is {\it simple}, however,
the constants (\ref{3Csb}) are already skewsymmetric
in the lower indices,
\begin{equation*}
f_{a_1 a_2[b_1}{}^{[c_1}\;\delta^{c_2]}_{b_2]} =
-f_{b_1 b_2[a_1}{}^{[c_1}\;\delta^{c_2]}_{a_2]}
\end{equation*}
(as shown for arbitrary $n$ in Sec.~\ref{sec:sim-str} below) and,
further, define the structure constants of Lie$\,\fG$ (see
Th.~\eqref{th:simple-SO}).

\begin{example} ({\it The Lie algebra associated with the euclidean FA $A_4$})
\label{ex:LieA4}

The metric FA $A_4$ ($f_{a_1 a_2 a_3}{}^l =-\epsilon_{a_1 a_2 a_3}{}^l$)
is simple, and the $C$s in \eqref{3Csb} define the structure constants of
Lie$A_4$. To identify this algebra easily we may use eq.~\eqref{3Cs}
to derive the commutation relations for the dual
$\tilde{M}^{a_1 a_2}=\frac{1}{2} \epsilon^{a_1 a_2 b_1 b_2}
(\mathscr{X}_{b_1 b_2})$ generators acting on the $A_4$
 vector space. Using eq.~\eqref{eps-1}, the commutation
relations become
\begin{equation}
\label{orthog}
[\tilde{M}^{a_1 a_2}, \tilde{M}^{b_1 b_2}]=
 -\delta^{a_1 b_2}\tilde{M}^{a_2 b_1} -\delta^{a_2 b_1}\tilde{M}^{a_1 b_2}
  +\delta^{a_1 b_1}\tilde{M}^{a_2 b_2} +\delta^{a_2 b_2}\tilde{M}^{a_1
  b_1} \quad,
\end{equation}
which are immediately recognized as those of the (semisimple)
$so(4)=so(3)\oplus so(3)$ algebra.
\end{example}
This was already evident from the second equation in eq.~\eqref{3Liealelem}
for $A_4$ which, taking the dual and with all indices down,
gives the familiar action of $so(4)$ on a vector $\be_k \in \mathbb{R}^4$,
\begin{equation}
\label{so4mom}
\tilde{M}_{a_1 a_2}\cdot \be_k= -(\delta_{a_1 k}\be_{a_2}-\delta_{k a_2}\be_{a_1}) \; .
\end{equation}
Of course, this does not mean that the 3-bracket
$[\be_{a_1},\be_{a_2},\be_{a_3}]$ for $A_4$ may be given by the
$r.h.s.$ of the above equation, which is antisymmetric in its first
two ($a_1,a_2$) indices only. Three-brackets that are not
necessarily antisymmetric may define 3-Leibniz algebras
(Sec.~\ref{sec:n-Leibniz}). We shall see on
Sec.~\ref{sec:triple+3Leib} that the $r.h.s.$ of eq.~\eqref{so4mom}
does indeed define a particular example of 3-Leibniz algebra and,
more specifically, the {\it metric Lie-triple system} of
Ex.~\ref{so4-R4-triple} .

\subsubsection{Lie$\,\fG$-invariant tensors associated
with a simple $n=3$ metric FA}
\label{inv-ten}
{\ }

It was seen in eq.~\eqref{str-cons-inv} that the structure constants of a
metric FA determine an invariant, fully antisymmetric tensor
on $\fG$ of rank $n+1$. For $n=3$, they also provide a {\it symmetric}
invariant tensor for Lie$\,\fG$ with coordinates
\begin{equation}
\label{2nd-inv}
k^{(2)}_{(a_1 a_2)\;(b_1 b_2)} = k^{(2)} (\mathscr{X}_{a_1 a_2},\mathscr{X}_{b_1 b_2})
\equiv k^{(2)}((\be_{a_1},\be_{a_2}), (\be_{b_1},\be_{b_2}) ) :=
  \langle [\be_{a_1},\be_{a_2},\be_{b_1}],\be_{b_2} \rangle
  = f_{a_1 a_2 b_1 b_2} \quad ,
\end{equation}
which is obviously symmetric under the $(a_1 a_2)\leftrightarrow
(b_1 b_2)$ exchange. Similarly, the familiar {\it Killing form} $k$
of Lie{$\,\fG$ is given by eq.~\eqref{KasCK}
\begin{equation}
\label{assoc-Casim}
 k^{(1)}_{a_1 a_2 b_1 b_2}= k^{(1)} (\mathscr{X}_{a_1 a_2} , \mathscr{Y}_{b_1 b_2}):=
Tr (ad_{\mathscr{X}_{a_1 a_2}} ad_{\mathscr{Y}_{b_1 b_2}}) = f_{a_1
a_2 s}{}^l \, f_{b_1 b_2 l}{}^s \quad ,
\end{equation}
and may be seen to be proportional to
\begin{equation*}
\kappa^{(1)}_{(a_1 a_2) (b_1 b_2)}=
C_{a_1 a_2 \; c_1 c_2}{}^{d_1 d_2} \, C_{b_1 b_2 \; d_1 d_2}{}^{c_1 c_2} \quad .
\end{equation*}

   It is not difficult to check the Lie$\,\fG$-invariance of these
tensors in the formalism of fundamental objects. This reads
\begin{equation}
\label{Kas-inv-fund}
\mathscr{Z}\cdot k(\mathscr{X},\mathscr{Y})=
k(\mathscr{Z}\cdot \mathscr{X}, \mathscr{Y}) +
  k(\mathscr{X}, \mathscr{Z}\cdot \mathscr{Y}) = 0    \; .
\end{equation}
Indeed, with
$\mathscr{Z}=\mathscr{Z}_{c_1 c_2}\equiv(\be_{c_1},\be_{c_2})\,,\,
\mathscr{X}=\mathscr{X}_{a_1 a_2}\,,\,
\mathscr{Y}=\mathscr{Y}_{b_1 b_2}$,
the $r.h.s.$ of the above expression for $k^{(2)}$ gives,
using \eqref{dot-compo},
\begin{eqnarray*}
  & & k^{(2)} \left( \left( {f_{c_1c_2a_1}}^d \be_d , \be_{a_2}\right) +
\left( \be_{a_1} , {f_{c_1c_2a_2}}^d \be_d\right), (\be_{b_1},
\be_{b_2})\right) \nonumber\\
  & & \qquad \qquad + \,
 k^{(2)} \left( (\be_{a_1}, \be_{a_2}), \left( {f_{c_1c_2b_1}}^d \be_d ,
\be_{b_2}\right) +
\left( \be_{b_1} , {f_{c_1c_2b_2}}^d \be_d\right)\right) = 0 \; ,
\nonumber
\end{eqnarray*}
which indeed is zero since it yields eq.~\eqref{str-cons-inv} for $n=3$,
\begin{equation}
   {f_{c_1c_2a_1}}^d f_{da_2b_1b_2} + {f_{c_1c_2a_2}}^d f_{a_1db_1b_2}
  + {f_{c_1c_2b_1}}^d f_{a_1a_2db_2} +
  {f_{c_1c_2b_2}}^d f_{a_1a_2b_1d} = 0 \; .
\end{equation}
Similarly, $k^{(1)}$ also satisfies eq.~\eqref{Kas-inv-fund} since,
using the $ad$ representation property (eq.~\ref{Lie-3-FA} or
Th.~\ref{th:n-Lie->Lie} below)
\begin{eqnarray}
\label{CasInv}
& &  k^{(1)} (\mathscr{Z} \cdot \mathscr{X} , \mathscr{Y}) +
 k^{(1)}(\mathscr{X}, \mathscr{Z} \cdot \mathscr{Y}) =
Tr(ad_{\mathscr{Z}\cdot \mathscr{X}} ad_\mathscr{Y})
+ Tr(ad_\mathscr{X} ad_{\mathscr{Z}\cdot \mathscr{Y} })
 \nonumber\\
& & \quad\quad = Tr((ad_\mathscr{Z} ad_\mathscr{X}- ad_\mathscr{X}
ad_\mathscr{Z}) ad_\mathscr{Y}) + Tr(ad_\mathscr{X}
(ad_\mathscr{Z} ad_\mathscr{Y} - ad_\mathscr{Y}ad_\mathscr{Z})) =0 \; .
\end{eqnarray}

For the euclidean $A_4$ FA, $f_{a_1 a_2 a_3 a_4}= \epsilon_{a_1 a_2
a_3 a_4}$ and the metric $k^{(1)}$ in \eqref{assoc-Casim} becomes
the $so(4)$ Cartan-Killing metric of coordinates
\begin{equation}
\label{CK-so4}
k^{(1)}_{a_1 a_2 b_1 b_2}= -(\delta_{a_1 b_1}\delta_{a_2 b_2}-\delta_{b_1 a_2}\delta_{a_1 b_2}) \; ,
\end{equation}
which is negative definite. The metric $k^{(2)}$ has signature
(3,3), and will play an important r\^ole in the Chern-Simons term of
the BLG model in Sec.~\ref{sec:CSterm}.
\medskip

\subsection{The Lie algebra associated to an $n$-Lie algebra}
\label{sec:nFA-toLie}{\ }

The case of the three-algebra in Sec.~\ref{3-to-Lie} is readily
extended to an $n$-Lie algebra using
the properties of the fundamental objects and Lemma~\ref{le:deriv-lemma}.
The result may be re-stated as the following

\begin{theorem}
\label{th:n-Lie->Lie} ({Lie algebra associated with a
 $n$-Lie algebra}) {\ }

The inner derivations $ad_{\mathscr{X}}$ define
\cite{Da-Ku:97,Da-Tak:97} an ordinary Lie algebra for the bracket
\begin{equation}
\label{Lie-n-Lie-b}
 ad_{\mathscr{X}}ad_{\mathscr{Y}}-
ad_{\mathscr{Y}}ad_{\mathscr{X}} =
[ad_{\mathscr{X}}\,,\,ad_{\mathscr{Y}}]=
ad_{\mathscr{X}\cdot\mathscr{Y}} \quad \mathrm{or} \quad
 [[\mathscr{X}\,, -]\,,\,[\mathscr{Y}\,,-]]=
 [\mathscr{X}\cdot\mathscr{Y}\,,-]
 \quad .
\end{equation}
The subalgebra of End$\,\fG$
defined by the Lie bracket \eqref{Lie-n-Lie-b} is the Lie algebra
Lie$\,\fG= \mathrm{InDer}\,\fG$ associated with the FA $\fG$.
\begin{proof}
The JI,
\begin{equation}
\label{fdtal-Lie-JI} [ad_{\mathscr{X}}\,,\,[ad_{\mathscr{Y}}\,,\,
ad_{\mathscr{Z}}] +[ad_{\mathscr{Y}}\,,\,[ad_{\mathscr{Z}}\,,\,
ad_{\mathscr{X}}]+ [ad_{\mathscr{Z}}\,,\,[ad_{\mathscr{X}}\,,\,
ad_{\mathscr{Y}}]=0 \quad,
\end{equation}
is of course satisfied for any elements in End$\,\fG$; what we have to check
is the consistency with the Lie bracket defined above. This is so
because the $l.h.s.$ of the above equation is
\begin{equation}
\begin{aligned}
=[ad_{\mathscr{X}}, ad_{\mathscr{Y}\cdot\mathscr{Z}}] &+
[ad_{\mathscr{Y}},ad_{\mathscr{Z}\cdot\mathscr{X}}]+
[ad_{\mathscr{Z}},ad_{\mathscr{X}\cdot\mathscr{Y}}]\\
=ad_{\mathscr{X}\cdot(\mathscr{Y}\cdot\mathscr{Z})} &+
ad_{\mathscr{Y}\cdot(\mathscr{Z}\cdot\mathscr{X})}+
ad_{\mathscr{Z}\cdot (\mathscr{X}\cdot\mathscr{Y})} \; ,
\end{aligned}
\end{equation}
which is zero by eq.~\eqref{Lie-setb}.
\end{proof}
\end{theorem}

When $\fG$ is simple $ad$ is injective. If $\fG$
is only semisiple, all fundamental objects with
two entries in two different simple components
of $\fG$ determine the zero derivation and $ad$
is no longer injective (see Rem.~\ref{rem:Kas-simp}).

\subsubsection{The Lie algebras associated to the
simple FAs}
\label{sec:sim-str}
{\ }

In general, chosen a basis of $\fG$ the commutation relations for
different elements $ad_\mathscr{X}$ may be written as
\begin{equation}
\label{nCs}
[ad_{\mathscr{X}_{a_1\dots a_{n-1}}},ad_{\mathscr{X}_{b_1\dots b_{n-1}}}]=
\frac{1}{(n-1)!}
C_{a_1\,\ldots\,a_{n-1} \,b_1\,\ldots\, b_{n-1}}{}^{c_1 \, \ldots \, c_{n-1}}\,
ad_{\mathscr{X}_{c_1\dots c_{n-1}}}
\end{equation}
with {\it e.g.} antisymmetric $C$s  given by ({\it cf.} eq.~\eqref{skew-3Cs})
\begin{equation}
\label{skew-nCs}
C_{a_1\,\ldots\,a_{n-1} \,b_1\,\ldots\, b_{n-1}}{}^{c_1 \, \ldots \, c_{n-1}} =
\frac{1}{(n-2)! \,2} \left( f_{{a_1}\,\ldots\, {a_{n-1}}\, {[b_1}}{}^{[c_1}\delta_{b_2}^{c_2} \ldots
\delta_{b_{n-1}]}^{c_{n-1}]} - (a\leftrightarrow b) \right) \quad .
\end{equation}

When $\fG$ is simple (Th.~\ref{th:FAsimple}), the first half
of the $C$s above is already antisymmetric and
provides the structure constants of Lie$\,\fG$.
To check their antisymmetry explicitly, we write
\begin{eqnarray}
\label{nSCrewritten1}
   {C_{a_1\dots a_{n-1} b_1\dots b_{n-1}}}^{c_1\dots c_{n-1}}
 &\propto & {\epsilon_{a_1\dots a_{n-1} [b_1}}^d \delta_d^{[c_1}
 \delta_{b_2}^{c_2} \dots \delta_{b_{n-1}]}^{c_{n-1}]} \nonumber\\
 &\propto&
 {\epsilon_{a_1\dots a_{n-1} [b_1}}^d \epsilon_{b_2\dots b_{n-1}]d e_1 e_2}
 \epsilon^{c_1\dots c_{n-1} e_1e_2}\ ,
\end{eqnarray}
(note that in the first line above the lower index $d$ is
unaffected by the antisymmetrization
imposed by the square bracket, which acts on the ($n-1$) indices $c$
and $b$ only). We now use a Schouten-type identity obtained by
antisymmetrizing the ($n+2$) indices ($b_1\dots,b_{n-1},d,e_1,e_2$)
in the last term above to obtain
\begin{equation}
\label{SCrewritten2}
 {\epsilon_{a_1\dots a_{n-1}
 [b_1}}^d \epsilon_{b_2\dots b_{n-1}]d e_1e_2}
 \epsilon^{c_1\dots c_{n-1} e_1e_2}
 - (n-2)! 2 {\epsilon_{a_1\dots a_{n-1}e_1}}^d
 \epsilon_{b_2\dots b_{n-1} d b_1e_2}
 \epsilon^{c_1\dots c_{n-1}e_1e_2}=0 \; .
\end{equation}
This means that the structure constants are proportional to the
last term. Therefore,
\begin{equation}
\label{SCrewritten3}
 {C_{a_1\dots a_{n-1} b_1\dots b_{n-1}}}^{c_1\dots c_{n-1}}
 \propto \epsilon^{c_1\dots
 c_{n-1}e_1e_2} {\epsilon_{e_1a_1\dots a_{n-1}}}^d
 \epsilon_{e_2b_1 \dots b_{n-1}d} \ ,
\end{equation}
which is clearly antisymmetric under the interchange
of the  $a$ and $b$ indices.
\medskip

   For $n>3$ one proceeds as for the $A_4$ case in Ex.~\ref{ex:LieA4}.
For instance, for the euclidean $A_5$, $f_{a_1 a_2 a_3 a_4}{}^{a_5}=
\epsilon_{a_1 a_2 a_3 a_4}{}^{a_5}$ and $(M_{a_1 a_2 a_3})^s{}_k=
f_{a_1 a_2 a_3 k}{}^s$. Then one finds that the dual generators
${\tilde M}^{a_1 a_2}=\frac{1}{3!}\epsilon^{a_1 a_2 b_1 b_2 b_3}M_{b_1 b_2 b_3}$
($M_{b_1 b_2 b_3}=\frac{1}{2}\epsilon_{b_1 b_2 b_3 a_1 a_2}{\tilde M}^{a_1 a_2}$) reproduce
exactly the commutation relations for the ${5 \choose 2}=10$ skewsymmetric matrices
that are the generators of the orthogonal algebra, eq.~\eqref{orthog}.
Thus, the Lie algebra associated with $A_5$ is $so(5)$.

This was to be expected: since the $A_{n+1}$ algebras are simple,
the algebra Lie$\,A_{n+1}$ of inner derivations coincides
with the Lie algebra of the group of automorphisms (Th.~\ref{simple-inder}),
and therefore Lie$\,A_{n+1}$ is the algebra of the $SO(n+1)$ group.
Indeed, in the general case the matrices $f_{a_1\dots a_n}{}^{a_{n+1}}$, $a_i=1\dots,n+1$,
determine the ${n+1\choose n-1}={{n+1}\choose 2}$
generators of $so(n+1)$ algebra; they are obviously rotations, since
the ${{n+1}\choose 2}$ ($n+1$)-dimensional matrices are antisymmetric,
$f_{a_1\dots a_{n-1} b c}= -f_{a_1\dots a_{n-1} c b }$.
The more familiar $so(n+1)$ commutators follow
immediately from eqs.~\eqref{nCs} and \eqref{SCrewritten3} by
 moving to the dual generators $\tilde{M}^{a_1 a_2}$.

As a result, the following theorem follows:

\begin{theorem} ({\it Lie algebras associated with the simple
 euclidean $A_{n+1}$ algebras})
 \label{th:simple-SO}
 {\ }

The Lie algebra associated with the euclidean $A_{n+1}$ is the algebra
$so(n+1)$ of its inner derivations.
\end{theorem}
Since the $A_{n+1}$ algebras are simple, all their derivations
are inner (Th.~\ref{ssFAandInDev}) and determine semisimple Lie algebras
(actually, simple for $n>3$). Thus, it is not surprising that there is
an analogue of the Whitehead's lemma for FAs (Thms.~\ref{th:s-s-FA-ext}
and \ref{th:rigFAsim} below) which, accordingly, holds true \cite{Az-Iz:09}
for all $n$-Lie algebras, $n\geq 2$.

\subsection{Representations of Filippov algebras in the sense of Kasymov}
\label{sec:repFA}{\ }

  To motivate the definition below, let us note first that the expression
\begin{equation}
\label{adrepr1}
[ad_{X_1 \dots X_{n-1}}\,,\,ad_{Y_1 \dots Y_{n-1}}] =
\sum_{i=1}^{n-1}\, ad_{Y_1 \dots ad_{X_1,\dots,X_{n-1}}Y_i \dots Y_{n-1}}
\end{equation}
reproduces eq.~\eqref{Lie-n-Lie} acting on $Z\in\fG$, and that
\begin{equation}
\label{adrepr2}
ad_{X_1 \dots X_{n-2} [Y_1,\dots,Y_n]} =
\sum_{i=1}^{n}\,(-1)^{n-i}\, (ad_{Y_1 \dots \widehat{Y_i} \dots Y_n}) \,
(ad_{X_1 \dots X_{n-2} Y_i})
\end{equation}
reproduces eq.~\eqref{eq:FI} acting on $X_{n-1}$. The above
properties -the FI- satisfied by the $ad_\mathscr{X}$ may be extended
by replacing the representation space $\fG$ by a vector
space $V$ and $ad$ by a general $\rho$ subject to analogous conditions
to guarantee the preservation of the FA structure as dictated by the FI.
This leads to

\begin{definition}
({\it Representations of Filippov algebras} \cite{Kas:87,Kas:98};
see also \cite{Ling:93})
\label{FA-repres}{\ }

Let $V$ be a vector space. A representation $\rho$ of a FA on $V$ is
a multilinear map $\rho :\wedge^{n-1}\fG \rightarrow \hbox{End}V$, $\rho
: \mathscr{X} \rightarrow \rho(X_1,\dots, X_{n-1})$  such that
\begin{equation}
\label{FAro}
\begin{aligned}
\qquad  [\rho(\mathscr{X})\,,\,\rho(\mathscr{Y})] & =  \rho(\mathscr{X}\cdot\mathscr{Y}) \quad, \cr
\rho(X_1,\dots, X_{n-2},[Y_1,\dots,Y_n])   & =
\sum_{i=1}^{n} \,(-1)^{n-i}\,\rho(Y_1,\dots,\widehat{Y_i},\dots,Y_n)\rho(X_1,\dots,X_{n-2},Y_i) \quad .
\end{aligned}
\end{equation}
The representation space $V$ is said to be a (left) $\fG$-module. Note
that the arguments of $\rho$ are fundamental objects and not elements
in $\fG$ so that, strictly speaking, $\rho$ is not representing
$\fG$ itself.
\end{definition}

The subspace $\hbox{ker} \rho=\{Z\in\fG|
\rho(Z,\fG,\mathop{\dots}\limits^{n-2},\fG)=0\}\subset\fG$ is the
kernel of the representation $\rho$; it follows that it is an ideal
of $\fG$. When $\hbox{ker}\,\rho=0$ (resp. $\fG$) the representation
is {\it faithful} (resp.~{\it trivial}). In the intermediate cases,
$\rho$ is a faithful representation of the quotient FA
$\fG/\hbox{ker}\,\rho$ \cite{Kas:87}. When $\rho$ is the adjoint map
and the kernel of $ad$ is defined {\it as above} for $\rho=ad$, it
coincides with the centre $Z(\fG)$ of the FA since, if
$ad(Z,\fG,\mathop{\dots}\limits^{n-2},\fG)$ is the null element in
End$\,\fG$, it follows that
$[Z,\fG,\mathop{\dots}\limits^{n-1},\fG]=0$, which determines the
centre $Z(\fG)$ (Def.~\ref{var-defs}).
\medskip

 The above defining properties are also readily obtained
 (see \cite{Ling:93}) by imposing an $n$-Lie
algebra structure on the vector space $\fG\oplus V$ with the
condition that $V$ be an abelian ideal,
$[\fG,\mathop{\dots}\limits^{n-1},\fG,V]\subset V$,
$[\fG,\mathop{\dots}\limits^{n-2},\fG,V,V]=0$. Indeed, the FI
applied to $[X_1,\dots,X_{n-1},[Y_1,\dots,Y_{n-1},v]]$
 with the notation $[X_1,\dots,X_{n-1},v]=\rho(X_1,\dots,X_{n-1})\cdot v$  gives
 \begin{equation*}
 [X_1,\dots,X_{n-1},\,[Y_1,\dots,Y_{n-1},v]]-[Y_1,\dots,Y_{n-1},\,[X_1,\dots,X_{n-1},v]] =
 \rho(\mathscr{X}\cdot \mathscr{Y}) \cdot v \quad,
 \end{equation*}
which gives the first equality in eq.~\eqref{FAro}, and the second
one follows from eq.~\eqref{eq:FI} rewritten in the form
\begin{equation*}
[X_1,\dots,X_{n-2},[Y_1,\dots,Y_n],v]=
\sum_{i=1}^{n}\,(-1)^{n-i}\,[Y_1,\dots,\widehat{Y_i},\dots,Y_n,\,[X_1,\dots,X_{n-2},Y_i,v]]
\end{equation*}
and factoring out $v$. In fact, all manipulations are as before
where $[X_1,\dots,X_{n-1},v]$ = $\mathscr{X}\cdot v\;$ indicates
$\rho(X_1,\dots,X_{n-1})\cdot v=\rho(\mathscr{X}) \cdot v$ or
$\mathscr{X}\cdot v$ for short. For instance, on $v\in V$ the first
expression in eq.~\eqref{FAro} gives
\begin{equation}
\label{gmod}
\mathscr{X}\cdot (\mathscr{Y}\cdot v) -
\mathscr{Y}\cdot (\mathscr{X} \cdot v) =
 (\mathscr{X}\cdot\mathscr{Y}) \cdot v \qquad  \forall v\in V \quad ,
\end{equation}
where the dot indicates the $\rho$-action. For $V=\fG$, $v=Z$ this
is eq.~\eqref{Lie-n-Lie} and corresponds to $\rho=ad$; its matrix
representation is given in eq~\eqref{FA-adj-matr}. Note that,
although $V$ is called a $\fG$-module, $\rho$ above is {\it not} a
map of $\fG$ in End$\,V$; $V$ might have been called as well a
`$\rho(\wedge^{n-1}\fG)$-module' under the action $\rho$ of the
fundamental objects $\mathscr{X}$ on $V$, consistent with the FA
structure of $\fG$ through the expressions above. It is only for
$n=2$ that the map $\rho:\fg\rightarrow\mathrm{End}\,V$ is a
homomorphism of (Lie) algebras.

The {\it coadjoint representation} is obtained by taking $V=\fG^*$,
$coad:\wedge^{n-1}\fG\rightarrow \hbox{End}\,\fG^*$. $ad$ and $coad$
are related in the usual manner,
\begin{equation*}
coad_\mathscr{X}:\nu\in \fG^*\mapsto coad_\mathscr{X}\cdot\nu \in
\fG^* \;,\quad (coad_\mathscr{X}\cdot\nu)(Y)= - \nu
(ad_\mathscr{X}\cdot Y) \quad \forall \, \nu\in\fG^* \,,\,
\mathscr{X}\in\wedge^{n-1}\fG\,,\,Y\in\fG \; .
\end{equation*}
\medskip

 Summarizing, the discussion in this section
has exhibited two important aspects of Filippov algebras:

1) Due to Th.~\ref{th:n-Lie->Lie}, any $n$-Lie algebra $\fG$ has
an {\it associated ordinary Lie algebra} Lie$\,\fG$ = InDer$\,\fG$
defined through the fundamental objects $\mathscr{X}$ and the
commutator of eq.~\eqref{Lie-n-Lie-b} of adjoint endomorphisms
$ad_{\mathscr{X}}$ of $\fG$. This follows from the obvious fact that
the composition of endomorphisms in End$\,V$ has a Lie rather
than a FA structure. As a `representation' (in the sense of
the previous definition) of the FA on itself, Lie$\,\fG$
is consistent with the FI for $\fG$
(see eqs.~\eqref{Lie-n-Lie-b},~\eqref{Lie-n-Lie}). When the carrier
vector space is an arbitrary one $V$, the representation $\rho$ on
$V$ induced by the fundamental objects is that of
Def.~\ref{FA-repres}.

2) The fundamental r\^ole (hence their name) of the objects
$\mathscr{X}\in \wedge^{n-1}\fG$, characterized by ($n-1$) elements
of the $n$-Lie algebra $\fG$. Their relevance will be apparent again
when discussing the FA cohomology in Sec~\ref{sec:FA-coho}. For
$n=2$ the {\it fundamental objects} $\mathscr{X}\in \wedge^{n-1}\fG$
and the {\it elements} $X\in \fG$ of a FA are one and the
same object, but for $n>2$ they emerge as separate entities.

\section{$n$-Leibniz algebras}
\label{sec:n-Leibniz}
{\ }

There are occasions in which the skewsymmetry of the FA bracket
(eq.~\eqref{n-ad}) is not demanded, but the FI still holds. This is
the case of the $n$-Leibniz algebras $\mathfrak{L}$ \cite{Da-Tak:97,
Cas-Lod-Pir:02} (see also \cite{Rot:05}), which generalize the
Leibniz algebras $\mathscr{L}$ of Sec.~\ref{sec:Loday} to the
$n$-ary case.

\begin{definition} ({\it $n$-Leibniz algebras and their fundamental objects})
\label{def:n-Leib}{\  }

 An $n$-Leibniz algebra $\mathfrak{L}$ is a vector space
 endowed with an $n$-linear application
$ [\;,\;\mathop{\cdots}\limits^n \;,\;]:\,
\mathfrak{L}\times\mathop{\cdots}\limits^n \times
\mathfrak{L}\rightarrow \mathfrak{L} \,$, the Leibniz $n$-bracket, such
that the derivation property \eqref{n-der} is satisfied. For $n$=2,
$\mathfrak{L}$ reduces to the ordinary Loday/Leibniz algebra $\mathscr{L}$
\cite{Lod:93,Lod-Pir:93} of Sec. \ref{sec:Loday}.
\end{definition}
An $n$-Leibniz algebra that satisfies
\begin{equation*}
[X_1 ,\dots,X_i,\dots,X_j,\dots,X_n]=0 \qquad \forall \;X_i=X_j \quad 1\leq i,j \leq n \quad ,
\end{equation*}
has an anticommutative $n$-bracket and is also a FA algebra.\\

It also proves convenient to introduce {\it fundamental objects
$\mathscr{X}$ for $n$-Leibniz algebras}. This is simply done by relaxing the
skewsymmetry condition $(\mathscr{X}\in \wedge^{n-1}\fG$ for a FA)
 since the Leibniz $n$-bracket is not
antisymmetric. Accordingly, the {\it fundamental objects of an $n$-Leibniz algebra} $\mathfrak{L}$ are now defined as elements
$(X_{a_1},\dots,X_{a_{n-1}})\equiv \mathscr{X}_{a_1\dots a_{n-1}}$,
where no anticommutativity is now implied. Thus $\mathscr{X}\in
\otimes^{n-1}\mathfrak{L}\;$ for an $n$-Leibniz algebra, but again
$\mathscr{X}\cdot\equiv[X_1,X_2,\dots,X_{n-1},\quad]\in
\hbox{End}\,\mathfrak{L}$ defines an inner derivation of the Leibniz
$n$-bracket as a result of the FI and
$ad_{\mathscr{X}}\,Z\equiv\mathscr{X}\cdot
Z:=[X_1,\dots,X_{n-1},Z]\in\mathfrak{L}$. We shall use the same
notation $\mathscr{X}=(X_1,\dots,X_{n-1})$ for the fundamental
objects of both Filippov and Leibniz $n$-algebras when there is no
risk of confusion, without making explicit that the fundamental
objects of a FA $\fG$ are skewsymmetric in their arguments and not
necessarily so for a LA $\mathfrak{L}\,$:
$\mathscr{X}\in\wedge^{n-1}\fG\,,\, \otimes^{n-1}\fL$, respectively,
for a $\mathscr{X}$ of $\fG$, $\fL$.
\medskip

   The representations of Leibniz algebras were reviewed in
Sec.~\ref{sec:Leib-defs}; see also \cite{Lod-Pir:93},
where the construction of the universal enveloping algebra of
a $\mathscr{L}$ as well as the proof of a
Poincar\'e-Birkhoff-Witt theorem was given. Those of $n$-Leibniz
algebras $\mathfrak{L}$ were  considered in \cite{Ca-In-La:07},
where a PBW-type theorem for the universal enveloping algebras
of finite-dimensional $n$-Leibniz algebras was established.

\subsection{Leibniz algebra  $\mathscr{L}$
associated with an $n$-Leibniz algebra $\mathfrak{L}$}
\label{sec:n-Leib}
{\ }

An $n$-{\it Leibniz algebra}  $\mathfrak{L}$ is the non-antisymmetric
analogue of an $n$-Lie algebra $\fG$: the Leibniz $n$-bracket in $\mathfrak{L}$
satifies the FI but it is not necessarily skewsymmetric (when it is,
the $n$-Leibniz algebra becomes an $n$-Lie algebra). The composition
of two fundamental objects of $\mathfrak{L}$ is defined again by
eq.~\eqref{dot-compo} in which the bracket that appears there is now
the $n$-bracket in $\mathfrak{L}$. Although in
Sec.~\ref{sec:fund-obj} we had the FAs in mind, to obtain
properties \eqref{Lie-setb} and \eqref{gmod} from the composition of
the fundamental objects in eq.~\eqref{dot-compo}, only a vector
space closed under an $n$-bracket $\;[X_1,\dots,X_n]$ satisfying the
FI was required; neither the antisymmetry of the fundamental objects
$\mathscr{X}$ nor that of the $n$-bracket were assumed.
Consequently, as far as the proof of eqs.~\eqref{Lie-setb} and
\eqref{gmod} is concerned, nothing changes if the fundamental
objects $\mathscr{X}$ are the possibly non-skewsymmetric ones of an $n$-Leibniz algebra $\mathfrak{L}$.
\medskip

Let $\mathfrak{L}$ be an $n$-Leibniz algebra. Then,
eq.~\eqref{Lie-setb} also holds for $\mathfrak{L}$, and using the notation
$\mathscr{X}\cdot\mathscr{Y}\equiv[\mathscr{X}\,,\,\mathscr{Y}]$ it
takes the form
\begin{equation}
\label{FA-Leib} [\mathscr{X}\,,\,[\mathscr{Y}\,,\, \mathscr{Z}]]
=[[\mathscr{X}\,,\,\mathscr{Y}]\,,\,\mathscr{Z}]+
[\mathscr{Y}\,,\, [\mathscr{X}\,,\,\mathscr{Z}]]\quad, \quad
\mathscr{X}\,,\,\mathscr{Y}\,,\, \mathscr{Z} \in
\otimes^{n-1}\mathfrak{L} \quad,
\end{equation}
where $[\mathscr{X}\,,\,\mathscr{Y}]\not=
-[\mathscr{Y}\,,\,\mathscr{X}]\,$ is a non-antisymmetric
two-bracket. Eq.~\eqref{FA-Leib} constitutes a particular example
of the derivation property $[X,[Y,Z]]=[[X,Y],Z]+[Y,[X,Z]]$ that
defines an ordinary Leibniz algebra structure
\cite{Lod:93,Lod-Pir:93} (Sec.~\ref{sec:Loday}) for a
non-antisymmetric bilinear bracket $[\; ,\;]$ in which the two
entries in $[\;,\;]$ are fundamental objects $\mathscr{X}\in
\otimes^{n-1}\mathfrak{L}$. Hence, given an $n$-Leibniz algebra
$\mathfrak{L}$, the
linear space of the fundamental objects endowed with the dot
operation \eqref{dot-compo} which defines a non-antisymmetric
two-bracket,  $[\mathscr{X},\mathscr{Y}]\equiv
\mathscr{X}\cdot\mathscr{Y}$, becomes an $n=2$ Leibniz algebra
\cite{Da-Tak:97, Da-Ku:97}. This is the {\it ordinary Leibniz
algebra $\mathscr{L}$ associated with an $n$-Leibniz algebra}
$\mathfrak{L}$.

When the fundamental objects are those of a FA, $\mathscr{X},
\mathscr{Y} \in \wedge^{n-1}\fG$, and the $n$-bracket
involved in the definition of $\mathscr{X}\cdot\mathscr{Y}$ is
therefore that of an $n$-Lie algebra, the resulting $\mathscr{L}$
is the {\it Leibniz algebra associated with the Filippov algebra $\fG$}. For
$n$=2, the Leibniz algebra associated to the FA is in fact an
ordinary Lie algebra since the FA bracket is skewsymmetric
and the FA itself is an ordinary Lie algebra.

  Additionally, the inner endomorphisms $ad_\mathscr{X}: Z =[X_1,\dots,X_{n-1},Z]$
of an $n$-Leibniz algebra $\mathfrak{L}$ generate a {\it Lie} algebra
under the commutator, since the validity of eq.~\eqref{Lie-n-Lie}
depends only on the FI. Thus, in spite of relaxing the full
anticommutativity of the $n$-bracket, Lie$\,\mathfrak{L}$ is
obtained as usual. This is relevant to define the gauge
transformations in BLG-type models (Sec.~\ref{sec:BLGmod}) that use
3-algebras algebras with brackets that are not fully antisymmetric
(hence, 3-Leibniz algebras) as in \cite{Ba-La:08,Che-Do-Sa:08}.\\

\section{Lie-triple systems}
\label{sec:triple}
{\ }

    Since triple systems $\fT$ constitute a particular example of 3-Leibniz
algebras, we recall their definition here. Lie (and Jordan)
triple systems have been the subject of extensive study in mathematics
(see {\it e.g.} \cite{Jacob:49, Jacob:51, Lis:52,Yama:57};
see also \cite{Ber:00}) and in physics, as {\it e.g.} in connection with
parastatistics or the Yang-Baxter equation \cite{Okub:94, Okubo:93,Oku-Kam:96a}.
Further triple (and supertriple) systems generalizations may be found
in \cite{Oku-Kam:00, Oku-Kam:96b, Oku:03} and references therein.
For the algebra of inner derivations of $\fT$ see, in particular,
\cite{Lis:52,Yama:57}.

\begin{definition} ({\it Lie-triple systems})
\label{def:triple}

  A Lie-triple system is a vector space $\fT$ plus a trilinear map
map $\fT\times\fT\times\fT\rightarrow\fT$
satisfying the properties\\
a) $[X,Y,X]=-[Y,X,Z]$\\
b)
$[X_1,X_2,[Y_1,Y_2,Y_3]]=[[X_1,X_2,Y_1],Y_2,Y_3]]+[Y_1,[X_1,X_2,Y_2],Y_3]]
+[Y_1,Y_2,[X_1,X_2,Y_3]]$\\
c) $[X,Y,Z]+[Y,Z,X]+[Z,X,Y]=0\quad \forall\, X_i,Y_i,Z \in
\fT$ \; .
\end{definition}
Thus, Lie triple systems are 3-Leibniz algebras with brackets
skewsymmetric in their first two arguments which, in addition,
satisfy the cyclic property (c). When the space $\fT$ is
$Z_2$-graded, a suitable addition of signs in the above expression
accounting for the $Z_2$-grading leads to the definition of
super-triple systems \cite{Oku:03,Oku:95, Oku-Kam:96a}. There are
also {\it anti}-Lie triple systems, defined by putting a plus sign
in condition a) above. An elementary example of Lie-triple system is
that of a Lie algebra $\fg$; defining the three-bracket
\cite{Jacob:49} by $[X,Y,Z]:=[[X,Y],Z]$, all properties above are
fulfilled (the last one  being simply the JI for $\fg$). In fact,
any associative algebra $\mathfrak{A}$ is a Lie-triple system for
the bracket $[x,y,z]=[[xy]z]$ where $[xy]=xy-yx$,
$\;x,y,z\in\mathfrak{A}\,$.
\medskip

\subsection{Lie-triple systems and 3-Leibniz algebras}
\label{sec:triple+3Leib}
{\ }

  To see how certain 3-Leibniz algebras and Lie-triple systems
are related, let us consider here briefly the approach of
\cite{deMed-JMF-Men-Rit:08}, which is based in a construction of
metric 3-Leibniz algebras due to Faulkner \cite{Faulk:73}; we refer to
\cite{deMed-JMF-Men-Rit:08} for details and to \cite{Oku:03, Oku-Kam:96a}.
To make things simpler, consider the euclidean $A_4$. Clearly, $ad :\wedge^2\fG\rightarrow
so(4)$ is one-to-one, and the $so(4)$ $ad$ algebra is metric with
respect to the scalar product
\begin{equation*}
(ad_{a_1 a_2}, ad_{b_1 b_2})= <[\be_{a_1},\be_{a_2},\be_{b_1}],\be_{b_2}>=
(ad_{b_1 b_2}, ad_{a_1 a_2}) \quad, \quad a=1,2,3,4\quad
\end{equation*}
{\it i.e.}, with respect the metric $k^{(2)}$ in \eqref{2nd-inv} for
$\epsilon_{a_1 a_2 b_1 b_2}$.
We already know that this scalar product is not degenerate,
which also follows trivially from the above expression, since
$(ad_{a_1 a_2}, ad_{b_1 b_2})=0\;\forall b_1\, b_2\,$
implies that $[\be_{a_1},\be_{a_2},\quad]$ has to be the trivial endomorphism.
Since $(ad_{a_1 a_2},ad_{a_1 a_2})=0$ obviously, the algebra has a
basis of null vectors, and thus the inner endomorphisms algebra
InDer$A_4=so(4)$ is even-dimensional (which of course we knew)
and the scalar product ( , ) above has signature (3,3). Further,
since $so(4)$ preserves the scalar product $<\;,\;>$ in $A_4$,
\begin{equation*}
 <[\be_{a_1},\be_{a_2},\be_{b_1}],\be_{b_2}>=
 - <\be_{b_1},[\be_{a_1},\be_{a_2},\be_{b_2}]>  \;  .
\end{equation*}

  In general, it follows \cite{deMed-JMF-Men-Rit:08}
that a real metric 3-Lie algebra $\fG$ gives rise to a bilinear
map $\wedge^2\fG\rightarrow \mathrm{InDer}\fG \subset so(\mathrm{dim}\fG)$
and that the Lie algebra of the $ad$ derivations is metric.
Thus, a real metric 3-Lie algebra determines a metric vector
space ($\fG$ itself) and an inner derivations algebra
endowed with a non-degenerate, symmetric scalar product.

Reciprocally, given a real metric algebra
$\fg\subset so(\mathrm{dim}V)$ where $V$ is a
metric vector space carrying a faithful representation $L$ of $\fg$,
it is possible to construct a metric 3-Leibniz algebra.
Given an orthonormal basis $\{\be_a\}$ of $V$, these 3-Leibniz
structures are characterized by a three bracket defined by
\begin{equation}
\label{def-tres}
[\be_{a_1},\be_{a_1},\be_{a_1}]:=L_{a_1 a_2}\cdot \be_{a_3}  \; .
\end{equation}
which then satisfies the following properties:

a) {\it scalar product preservation}
\begin{equation}
\label{metricity-t}
<[\be_{a_1},\be_{a_2},\be_{b_1}],\be_{b_2}>
+ <\be_{b_1},[\be_{a_1},\be_{a_2},\be_{b_2}]>=0 \; ,
\end{equation}

b) {\it symmetry condition}
\begin{equation}
\label{sym-cond}
<[\be_{a_1},\be_{a_2},\be_{b_1}],\be_{b_2}>=
<[\be_{b_1},\be_{b_2},\be_{a_1}],\be_{a_2}>
\end{equation}
It follows from eqs.~\eqref{metricity-t}, \eqref{sym-cond}
that the 3-bracket is antisymmetric in the first two entries,
\begin{equation*}
[\be_{a_1},\be_{a_2},\be_{a_3}]= -[\be_{a_2},\be_{a_1},\be_{a_3}]
\end{equation*}
and, using the composition properties of $L_{a_1 a_2}$, that it
satisfies the FI.
\medskip

  Real 3-Leibniz algebras satisfying eqs.~\eqref{metricity-t}, \eqref{sym-cond}
and the FI were introduced in \cite{Cher-Sa:08,Che-Do-Sa:08} in the context
of the BLG model (see Sec.~\ref{sec:BLG-phys}) and called
{\it generalized 3-Lie algebras}. Thus, these algebras
are in one-to-one correspondence with the pairs ($\fg\,,\,V$)
where $\fg$ is a metric subalgebra of $so(\mathrm{dim}V)$
and $V$ a faithful $\fg$-module \cite{deMed-JMF-Men-Rit:08} (we refer
also to \cite{deMed-JMF-Men-Rit:08} for the case of the hermitian
(complex) BL algebras in \cite{Ba-La:08}, and to
\cite{Chow-Muk:09,Palm:09} in connection with Jordan-triple systems).
The above symmetry properties imply that the structure constants of the
generalized 3-Lie algebras satisfy the relations \cite{Cher-Sa:08}

\begin{equation*}
f_{a_1 a_2 b_1 b_2}= - f_{a_2 a_1 b_1 b_2}= f_{b_1 b_2 a_1 a_2}=-f_{a_1 a_2 b_2 b_1}\; .
\end{equation*}

  Since the above 3-bracket is a map
  $[\;,\;,\,\;]:\wedge^2 V \otimes V \rightarrow V$, it has the
general symmetry of $\yng(1,1)\otimes\yng(1)=\yng(1,1,1)\oplus\yng(2,1)$
(the second term is not irreducible since it contains a trace).
It then follows \cite{deMed-JMF-Men-Rit:08} that the above ternary
algebras have two special subcases associated with the above decomposition:
the first part, for which the 3-bracket is totally antisymmetric,
corresponds to an ordinary  FA ($V=\fG$); the second one determines a
bracket with a mixed symmetry. In this last case, it follows that the
3-bracket of the 3-Leibniz algebra also has the additional property c)
in Def.~\ref{sec:triple} above and hence defines a metric
three-Lie system ($V=\fT$). In general, however, the 3-Leibniz algebra
associated with the pair ($\fg\,,\,V$) is neither a 3-Lie
algebra nor a Lie-triple system.

\begin{example} ({\it The Lie-triple system associated with} ($so(4)\,,\,\mathbb{R}^4$))
\label{so4-R4-triple}

Consider \cite{Oku:03,deMed-JMF-Men-Rit:08} $so(4)$ acting
faithfully on the euclidean $\mathbb{R}^4$ with basis
$\{\be_a\}$ as in eq.~\eqref{so4mom}. By eq.~\eqref{def-tres}
above, $M_{a_1a_2}\cdot \be_{a_4}$ defines the 3-bracket
\begin{equation}
\label{triple-ex}
[\be_{a_1},\be_{a_2},\be_{a_3}]= -(\delta_{a_1 a_3}\be_{a_2}-\delta_{a_3 a_2}\be_{a_1})
\end{equation}
This bracket corresponds now to the $so(4)$ scalar product given by
\begin{equation}
(M_{a_1 a_2},M_{b_1 b_2})= <[\be_{a_1},\be_{a_2},\be_{b_1}],\be_{b_2}>=
-(\delta_{a_1 b_1}\delta_{a_2 b_2}-\delta_{b_1 a_2}\delta_{a_1 b_2}) \; ,
\end{equation}
which is the $so(4)$ Cartan-Killing metric in eq.~\eqref{CK-so4}.
The bracket \eqref{triple-ex} satisfies the cyclic property c) in
Def.~\ref{def:triple} and thus defines a Lie-triple system which, in
fact, appeared very long ago (see \cite{Jacob:49}).
\end{example}

Lie-triple systems are important in the theory of symmetric spaces
(see {\it e.g.} \cite{Hel:78}), the reason being that $\fg\oplus V$ may be given
a metric Lie algebra structure (see further \cite{Yama:57} for triple
systems as tangent spaces to totally geodesic spaces; recall that,
in contrast, general Filippov and Leibniz algebras do not
have a `linear approximation' interpretation). The canonical procedure to construct
Lie (and Lie super-) algebras from Lie- (and anti-Lie-) triple systems
is detailed in \cite{Oku:03}, where the Lie algebra $so(N+1)$
(from $\fg=so(N)$, dim$\,V=N$) and the $osp(1,N)$ superalgebra
(from $\fg=sp(N)$) are obtained. For the $so(4)$, dim$A_4$=4 example
above, the {\it embedding Lie algebra} is thus $so(5)$ and the
Lie triple system corresponds to the symmetric space
$SO(5)/SO(4)\sim S^4$; we refer to \cite{Oku:03,deMed-JMF-Men-Rit:08}
for further details and references.

\section{Cohomology and homology for Filippov algebras}
\label{sec:FA-coho} {\ }

We discuss now the cohomology and homology for FA's. We shall see
that there is more than one cohomology complex relevant in
applications. In analogy with Lie algebras, where the central
extensions and the infinitesimal deformations are characterized,
respectively, by the Lie algebra cohomology groups
$H^2_0(\fg,\mathbb{R})$ and $H^2_{ad}(\fg,\fg)$ for the trivial and
$ad$ representations, the cohomology groups relevant in the central
extensions and infinitesimal deformations of FA will be given by
$H^1_0(\fG,\mathbb{R})$ and $H^1_{ad}(\fG,\fG)$, albeit defined on
cohomology complexes that are not the `natural' analogues of the Lie
algebra ones, as will be shown below.
\medskip

\subsection{Introduction: 3-Lie algebras and central extensions}
\label{3-alg-central}{\ }

To motivate the general cohomology problem and the definitions that
will follow in the general case, let us consider first the simple
problem of extending centrally a 3-Lie algebra $\fG$ in close
analogy with the Lie algebra case discussed in
Sec.~\ref{ex:central-Lie}. The fundamental objects are determined by
two elements of $\fG$, $\mathscr{X}=(X_1,X_2)$.  The existence of
such a central extension, denoted $\tilde{\fG}$, implies that the
original three-bracket in $\fG$ may be modified by the presence of a
term in the additional central generator $\Xi$, so that the
commutation relations of $\tilde{\fG}$ read
\begin{equation}
\label{3-central}
 [\tilde{X}_a,\tilde{X}_b,\tilde{X}_b]= f_{abc}{}^d \tilde{X}_d + \alpha(X_a,X_b,X_c)
 \, \Xi \quad ,\quad [\tilde{X}_a,\tilde{X}_b,\Xi]=0
 \quad \forall \tilde{X}_a \in \tilde{\fG}\;
 ,
\end{equation}
where $\alpha$ is a (necessarily antisymmetric) trilinear map
$\alpha:(X_a,X_b, X_b) \mapsto \alpha (X_a,X_b, X_b)\in \mathbb{R}$.
The fact that (\ref{3-central}) is a 3-Lie algebra means that
$\alpha(X,Y,Z)$ must be such that the FI holds for the extended FA
$\tilde{\fG}$. This clearly constrains the possible $\alpha(X,Y,Z)$
in (\ref{3-central}) to those that are consistent with the structure
of the extended 3-Lie algebra.

This may be formalized in terms of 3-Lie algebra cohomology in the
following way. First, let us call {\it one}-cochains $\alpha$ to the
$\mathbb{R}$-valued antisymmetric trilinear maps that appear in
\eqref{3-central},
\begin{equation}
\label{1-coch} \alpha^1(Y_1,Y_2,Z)\equiv \alpha^1(\mathscr{X},Z) \;
,
\end{equation}
because they depend on {\it one} fundamental object (plus an element
of $\fG$). Thus, we may think of $\alpha^1$ as an element in
$\wedge^2 \fG^* \wedge \fG^*$ (this is of course $\wedge^3 \fG^*$,
but it proves convenient here to use this split notation).
We may now introduce a coboundary
operator $\delta$ that takes $p$-cochains to ($p+1$)-cochains  by
increasing the number of fundamental objects in the argument of the
cochain by one. In the present case, this means that the action on
$\alpha$ gives a two-cochain, a five-linear element of $\wedge^2
\fG^* \otimes \wedge^2 \fG^*\wedge \fG^*$. Let us define
specifically the action of $\delta$ on the one-cochain $\alpha^1$ by
\begin{equation}
\label{cob-1-coch}
\begin{aligned}
(\delta\alpha^1)(\mathscr{X},\mathscr{Y}, Z):= &
-\alpha^1(\mathscr{X}\cdot\mathscr{Y}, Z) -\alpha^1(\mathscr{Y},
\mathscr{X}\cdot Z)+ \alpha^1(\mathscr{X}, \mathscr{Y}\cdot Z)\\
= & -\alpha^1([X_1,X_2,Y_1],Y_2, Z) - \alpha^1(Y_1,[X_1,X_2,Y_2],Z)+\\
& -\alpha^1(Y_1,Y_2,[X_1,X_2, Z]) +\alpha^1(X_1,X_2,[Y_1,Y_2, Z])\;.
\end{aligned}
\end{equation}
A one-cochain on the three-algebra $\fG$ will be a one-cocycle if the
above expression is zero, in which case
\begin{equation}
\label{one-coc}
\alpha^1(X_1,X_2,[Y_1,Y_2, Z])=
\alpha^1([X_1,X_2,Y_1],Y_2, Z)+ \alpha^1(Y_1,[X_1,X_2, Y_2],Z)
+\alpha^1(Y_1,Y_2,[X_1,X_2,Z]) \; .
\end{equation}
The structure of this expression (or of eq.~\eqref{cob-1-coch}) is
easily justified: it follows by imposing the FI on
\eqref{3-central}. Thus, the one-cochain in \eqref{1-coch} that
defines a central extension must be a one-cocycle, since this
guarantees that the extended algebra (\ref{3-central}) satisfies the
FI. Hence, $\alpha$ defines a 3-Lie algebra central extension if
$\delta\alpha=0$.

As for ordinary Lie algebras, the above central extension may be
trivial. This means that it is possible to make a redefinition of
the basis of $\tilde{\fG}$ that removes the central element
$\,\Xi\,$ from the {\it r.h.s.} of (\ref{3-central}) so that $\Xi$
appears in $\tilde{\fG}$ as a fully independent addition,
$\tilde{\fG}=\fG \oplus \mathbb{R}$. Such a redefinition is possible
when the one-cocycle $\alpha$ is trivial (it is a one-coboundary),
which means that it is generated by a zero-cochain {\it i.e.},  by a
linear map on $\fG$, $\beta : X\mapsto \beta(X)\in \mathbb{R} $.
Then, $\alpha^1(X,Y,Z)=(\delta \beta)(X,Y,Z)=-\beta([X,Y,Z])$. The
last equality in this expression follows from the general definition
of coboundary operator for $\fG$ algebras to be given in the next
section, but at this stage it may be justified by the fact that,
using \eqref{cob-1-coch}, any one-cochain thus defined is a
(trivial) one-cocycle by virtue of the FI (as it should since
$\delta$ must be nilpotent) and, further, because a central
extension is actually a trivial one when the one-cocycle is a
one-coboundary as defined above. In fact, if $\alpha=\delta\beta$,
\begin{equation}
\alpha (\mathscr{X},Z)=
(\delta\beta)(\mathscr{X},Z)=-\beta(\mathscr{X}\cdot
Z)=-\beta([X,Y,Z]) \; ,
\end{equation}
and it is sufficient to define the new basis generators
$\tilde{X}'_d$ of $\tilde{\fG}$ by $\tilde{X}'_d= \tilde{X}_d -
\beta(X_d)\,\Xi$ to obtain
\begin{equation}
[\tilde{X}'_a,\tilde{X}'_b,\tilde{X}'_c]= f_{abc}{}^d\tilde{X}_d
-\beta([X_a,X_b,X_c])=f_{abc}{}^d(\tilde{X}_d-\beta(X_d))=
f_{abc}{}^d \tilde{X'}_d \quad ,\quad [\tilde{X}'_a,\tilde{X}'_b,\Xi]=0 \quad,
\end{equation}
which exhibits explicitly the triviality of the extension,
$\tilde{\fG}=\fG\oplus \mathbb{R}$.

    Note that, in naming the order of the three-algebra
cohomology cocycles, we are counting the number of fundamental
objects that they contain. Thus, a 3-Lie algebra $p$-cocycle has
$p.(3-1)+1$ elements of $\fG$ as arguments\footnote{Thus, an
ordinary Lie algebra $\fg$ cohomology (Def.~\ref{coho-for-repr})
{\it two}-cocycle corresponds to a $one$-cocycle from the present
$n$-Lie algebra $\fG$ point of view since for $n=2$ the fundamental
objects contain just an element of the algebra. For a Lie algebra
$\fg$, $\mathscr{X}=X$ and $\Omega(\mathscr{X},Y)=\Omega(X,Y)$.} and
a general $n$-Lie algebra $p$-cocycle contains $p(n-1)+1$ arguments,
since the fundamental objects have ($n-1$) elements of $\fG$ as
(antisymmetric) arguments.

\begin{example} ({\it The $n=3$ Nambu-Heisenberg-Weyl algebra}) \cite{Nambu:73}
\label{N-H-W}
{\ }

Consider an abelian $n=3$ FA of dimension $3N$, with basis
determined by three subsets of $N$ generators each, ($X_a,Y_a,Z_a$),
$a=1,\dots,N$. Since it is abelian,
\begin{equation*}
    [X_a,Y_b,Z_c]=0 \quad\mathrm{and,\,of\, course,}\quad
    [X,X,Y]=0\quad\mathrm{etc}\quad .
\end{equation*}
Let $\alpha^1$ be the one-cochain defined by
$\alpha^1(X_a,Y_b,Z_c)=\kappa$ for $a=b=c$ and
$\alpha^1(X_a,Y_b,Z_c)=0$ otherwise; $\alpha^1$ also gives zero when
two or more entries come from the same basis subset,
$\alpha^1(X_a,X_b,Z_c)=0$ etc., and is, of course, antisymmetric in
its three entries. The values of the constant $\kappa\in \mathbb{R}$
characterize distinct cohomology classes, but, since any non-zero
value will lead to extensions that are isomorphic as 3-Lie algebras,
we may take $\kappa=1$ (or $\hbar$). Eq.~\eqref{one-coc} is
obviously satisfied, and thus $\alpha^1$ is the one-cocycle that
defines the Nambu-Heisenberg-Weyl 3-Lie algebra
\begin{equation}
\label{N-H-W-al}
[X_a,Y_b,Z_c]=\mathbb{I} \quad \mathrm{for}\; a=b=c
\quad \mathrm{and\; =0 \quad in\; all\; other\; cases}\;
\end{equation}
as {\it e.g.}, $[X_a,X_b,Z_c]=0$.

It is worth noting that this central extension of the above abelian
$3N$-dimensional 3-Lie algebra is realized by the bracket
$\sum_{a=1}^{N}\frac{\partial(f^1,f^2,f^3)}{\partial(x^a,y^a,z^a)}$
\cite{Nambu:73} of a triplet of coordinates ($x^a,y^b,z^c$), since
$\{x^a,y^b,z^c\}=1$ for $a=b=c$ and $\{x^a,y^b,z^c\}=0$ otherwise,
etc. This constitutes the $n=3$ counterpart of the PB realization of
the classical abelian algebra of dynamical variables $[q^a,p_b]=0$,
which produces the Heisenberg-Weyl Lie algebra
$\{q^a,p_b\}=\delta^a_b$
($\,[\hat{q^a},\hat{p_b}]=i\hbar\delta^a_b\mathbb{I}$ upon Dirac
quantization). It is well known that Poisson brackets realize
`projective' (central) extensions of Lie algebras (see
\cite{Pa-Pro:66}). In this simple $n=3$ example, the Nambu bracket
realization of the above abelian 3-Lie algebra leads to the
Nambu-Heisenberg-Weyl 3-algebra of eq.~\eqref{N-H-W-al} although, of
course, an interpretation of the Nambu bracket algebra in terms of
transformations would associate these to pairs $\{x_a,y_b,\;\}$ etc.
corresponding to antisymmetric fundamental objects and not to
single algebra elements. This, we note in passing, is also
reflected in the way one-cocycles depend on the elements of the
($n=3$) FA as shown below and explains why we wrote
$\alpha^1\in\wedge^2\fG^*\wedge\fG^*$ rather than
$\alpha^1\in\wedge^3\fG^*$.
\end{example}

 Let us now consider the general case.

\subsection{FA cohomology complex adapted to the
central extension problem}{\ }

This is the cohomology complex for the trivial action that follows
from the previous discussion. We shall take here $\fG$-valued cochains
for later convenience; since the action is trivial, this
will not modify the form of the action of the coboundary
operator on the corresponding $p$-cochains.

\begin{definition}(FA {\it cohomology $p$-cochains})

A $p$-{\it cochain} is a linear map
 $\alpha\in C^p(\fG, \fG) =\mathrm{Hom}(\wedge^{(n-1)}\fG\otimes
\mathop{\cdots}\limits^p\otimes\wedge^{(n-1)}\fG\wedge\fG\,,\,\fG)$
\begin{equation}
\label{p-coch}
 \alpha :(\mathscr{X}_1,\dots,\mathscr{X}_p,Z) \mapsto
\alpha(\mathscr{X}_1,\dots,\mathscr{X}_p,Z)\equiv \alpha
(X^1_1,\dots,X^{n-1}_1, \dots ,X^1_p,\dots ,X^{n-1}_p, Z) \; .
\end{equation}
Thus, $\alpha$ takes $p$ fundamental objects $\mathscr{X}$ and an
element of $\fG$ as arguments; in all, ($p(n-1)+1$) elements of
the $n$-algebra $\fG$.\\
\end{definition}

\begin{definition} ({\it Coboundary operator $\delta$ for the trivial action})
\label{FA-cob-op}{\ }

The coboundary operator for the trivial action of $\fG$ is the map
$\delta: C^p(\fG, \fG) \rightarrow C^{p+1}(\fG, \fG)$ defined (see
\cite{AIP-B:97}) by its action on a $p$-cochain $\alpha\in C^p$ by
\begin{equation}
\label{n-cobop}
\begin{aligned} (\delta\alpha)
 (\mathscr{X}_1,\dots,\mathscr{X}_{p+1}, Z) & = \\
\sum_{1\leq i<j}^{p+1} (-1)^i \,& \alpha(\mathscr{X}_1,\dots,
\hat{\mathscr{X}}_i,\dots,\mathscr{X}_i\cdot\mathscr{X}_j,\dots,\mathscr{X}_{p+1},Z)\\
+  \sum_{i=1}^{p+1} (-1)^i\,  & \alpha (\mathscr{X}_1,\dots,
\hat{\mathscr{X}}_i,\dots,\mathscr{X}_{p+1}, \mathscr{X}_i \cdot
Z) \; ,
\end{aligned}
\end{equation}
{\it cf.} eq.~\eqref{Lie-cob-ad}.
\end{definition}
To check that $\delta^2\equiv 0$, it is sufficient to recall that
eq.~\eqref{Lie-setb} and $\mathscr{X} \cdot (\mathscr{Y} \cdot Z)-
\mathscr{Y} \cdot (\mathscr{X} \cdot Z)= (\mathscr{X} \cdot
\mathscr{Y}) \cdot Z$ hold by virtue of the FI. Let us illustrate
the nilpotency of $\delta$ with a couple of examples.
\medskip

  Let $\fG$ be a 3-Lie algebra. A zero-cochain is given by a map
$\alpha^0 : Z \mapsto \alpha^0(Z)\in \fG$, and the action of
$\delta$ on $\alpha^0$ defines the one-cochain $\delta\alpha^0$
given by
\begin{equation}
\label{FA-one-cob}
(\delta\alpha^0)(\mathscr{X},Z)= -\alpha^0(\mathscr{X}\cdot Z)=
-\alpha^0([X_1,X_2,Z]) \quad ,
\end{equation}
which is fully antisymmetric in its three arguments $X_1,X_2,Z$ as
any {\it one}-cochain should be.  Similarly, a one-cochain
$\alpha^1$ is a $\fG$-valued map $\alpha^1\in
\mathrm{Hom}(\wedge^2\fG\wedge\fG\,,\,\fG)\,$,
$\alpha^1:(\mathscr{X},Z) \mapsto \alpha(\mathscr{X},Z)\in \fG$. The
action of the coboundary operator on $\alpha^1$ is the two-cochain
$\delta\alpha^1$ defined like in eq.~\eqref{cob-1-coch}. The
one-cochain $\alpha^1$ is a one-cocycle if the two-cochain
$\delta\alpha^1=0$, in which case its coordinates satisfy
\begin{equation}
\label{1-cocy.cond3} (\delta\alpha)_{a_1 a_2 b_1 b_ 2 k}{}^s=f_{b_1
b_2 k}{}^l\alpha_{a_1a_2 l}{}^s- f_{a_1 a_2 b_1}{}^l\alpha_{l b_2
k}{}^s - f_{a_1 a_2 b_2}{}^l\alpha_{b_1 l k}{}^s - f_{a_1 a_2
k}{}^l\alpha_{b_1 b_2 l}{}^s=0  \; .
\end{equation}
 This is necessarily the case for any coboundary
$\alpha^1=\delta\alpha^0$ since $(\delta\alpha^0)(X_1,X_2,Z)=
\alpha^0([X_1,X_2,Z])$ and then eq.~\eqref{cob-1-coch} is zero
$\forall\, \alpha^0$ by linearity and the FI.

   We may compare this with the Lie $n=2$ case, for which
$\mathscr{X}, \mathscr{Y} = X,Y$ and eq.~\eqref{cob-1-coch} produces
the three-coboundary given by the fourth line in
eq.~\eqref{Lie-coh-ex} ($\rho=0$). A Lie algebra one-cochain
$\Omega^1$ for the trivial action generates the two-coboundary
$(\delta \Omega^1)(X_1,X_2)=-\Omega^1([X_1,X_2])$, and then
$(\delta^2 \Omega^1)(X_1,X_2,X_3)=0$ implies (or is guaranteed by)
the JI. The two-coboundary generated by the $\fg$-valued cochain
(trivial action) defined by $\Omega^1 (X_i)=-X_i$, has for
coordinates the $\fg$ structure constants, $(\delta\Omega)_{i
j}{}^k= C_{ij}{}^k$. Similarly, for $n=3$, a zero cochain $\alpha^0$
generates the one-coboundary $\delta\alpha^0$ in
eq.~\eqref{FA-one-cob} and  $\delta^2\alpha^1=0$ $\forall\,
\alpha^0$ is guaranteed by the FI as shown above. If we now take the
$\fG$-valued one-cochain of coordinates $\alpha^1_{a_1a_2
l}{}^s=f_{a_1 a_2 l}{}^s$, we see that $\alpha^1$ is a one-cocycle
for the trivial action since $(\delta \alpha^1)_{a_1 a_2 b_1 b_2
k}{}^s=0$, as given by eq.~\eqref{1-cocy.cond3}, is the FI.
Actually, $\alpha^1$ is the one-coboundary $\alpha^1=\delta\alpha^0$
generated by the zero cochain defined by $\alpha^0 (X_a)=-X_a$ by
eq.~\eqref{FA-one-cob}. The corresponding Lie algebra expressions,
written in terms of the canonical form $\theta$ on a Lie group $G$,
and within the present labelling of the cochain orders, correspond
to $\alpha^0=-\theta$, the one-coboundary to
$\delta\alpha^0=-d\theta=\theta\wedge\theta$ (eq.~\eqref{3canonMC})
and $d^2\theta=0$ follows by the JI.
\medskip

  As a second example, consider a FA cohomology two-cochain,
$\alpha^2(\mathscr{X}_1,\mathscr{X}_2,Z)$. Writing all terms,
$\delta\alpha^2$ is the three-cochain given by
\begin{equation*}
 \begin{aligned}
 (\delta\alpha^2)(\mathscr{X}_1, & \mathscr{X}_2,\mathscr{X}_3,Z)  =
-\alpha^2(\mathscr{X}_1\cdot\mathscr{X}_2,\mathscr{X}_3,Z) -
\alpha^2(\mathscr{X}_2,\mathscr{X}_1\cdot \mathscr{X}_3,Z) \\
&\qquad\qquad\quad
+\alpha^2(\mathscr{X}_1,\mathscr{X}_2\cdot\mathscr{X}_3,Z) -
\alpha^2(\mathscr{X}_2,\mathscr{X}_3,\mathscr{X}_1\cdot Z)\\
& \qquad\qquad\quad
+\alpha^2(\mathscr{X}_1,\mathscr{X}_3,\mathscr{X}_2\cdot Z)
-\alpha^2(\mathscr{X}_1,\mathscr{X}_2,\mathscr{X}_3\cdot Z)\quad .
\end{aligned}
\end{equation*}
We may now check that $\delta^2=0$ by assuming that $\alpha^2$ is
actually a two-coboundary, $\alpha^2=\delta\alpha^1$. Then, the
above expression becomes
\begin{equation*}
\begin{aligned}
& = \alpha^1((\mathscr{X}_1\cdot\mathscr{X}_2)\cdot\mathscr{X}_3,
Z)+ \alpha^1(\mathscr{X}_3, (\mathscr{X}_1\cdot\mathscr{X}_2)\cdot
Z)
-\alpha^1(\mathscr{X}_1\cdot\mathscr{X}_2,\mathscr{X}_3\cdot Z)\\
+ & \;
\alpha^1(\mathscr{X}_2\cdot(\mathscr{X}_1\cdot\mathscr{X}_3),Z)
 + \alpha^1(\mathscr{X}_1\cdot\mathscr{X}_3,\mathscr{X}_2\cdot Z)
- \alpha^1(\mathscr{X}_2,(\mathscr{X}_1\cdot\mathscr{X}_3)\cdot Z)\\
- &\; \alpha^1(\mathscr{X}_1\cdot(\mathscr{X}_2\cdot \mathscr{X}_3),
Z)
 - \alpha^1(\mathscr{X}_2\cdot\mathscr{X}_3,\mathscr{X}_1\cdot Z)
+ \alpha^1(\mathscr{X}_1,(\mathscr{X}_2\cdot\mathscr{X}_3)\cdot Z))\\
 +\; &\alpha^1(\mathscr{X}_2\cdot\mathscr{X}_3,\mathscr{X}_1\cdot Z)
 + \alpha^1(\mathscr{X}_3, \mathscr{X}_2\cdot(\mathscr{X}_1\cdot Z))
- \alpha^1(\mathscr{X}_2,\mathscr{X}_3(\cdot\mathscr{X}_1\cdot Z))\\
-  &\alpha^1(\mathscr{X}_1\cdot\mathscr{X}_3,\mathscr{X}_2\cdot Z)
 - \alpha^1(\mathscr{X}_3, \mathscr{X}_1\cdot(\mathscr{X}_2\cdot Z))
+ \alpha^1(\mathscr{X}_1,\mathscr{X}_3\cdot(\mathscr{X}_2\cdot Z))\\
+ &\; \alpha^1(\mathscr{X}_1\cdot\mathscr{X}_2, \mathscr{X}_3\cdot
Z)
 + \alpha^1(\mathscr{X}_2,\mathscr{X}_1\cdot(\mathscr{X}_3\cdot Z))
-\alpha^1(\mathscr{X}_1,\mathscr{X}_2\cdot(\mathscr{X}_3\cdot
Z))=0\; ,
\end{aligned}
\end{equation*}
since all terms cancel out. Indeed, three pairs of them cancel each
other directly, and the remaining terms can be collected in four
groups of three sharing a common argument, $\mathscr{X}$ or $Z$.
Then, each group is seen to add up to zero on account of
eqs.~\eqref{Lie-setb}, \eqref{Lie-n-Lie}, both a result of
the FI. Higher orders proceed similarly.
\medskip

   Thus, the cohomology complex $(C^\bullet(\fG),\delta)$ for
the trivial action is the relevant one for the central extension
problem of FAs. The following analogue of the Whitehead lemma for
Lie algebras, that we state here without proof, holds
\cite{Az-Iz:09}:

\begin{theorem} (Semisimple FAs and central extensions)
\label{th:s-s-FA-ext}
{\ }

Let $\fG$ be a semisimple FA. Then all its central extensions are trivial.
\end{theorem}

\begin{remark}
{\ }

  The above coboundary operator does not allow for
a formulation of the $n$-Lie algebra cohomology {\it \`a la}
Chevalley-Eilenberg for $n>2$ (Sec.~\ref{sec:CEco}), since the
existence of a group-like manifold associated with FAs, allowing for
the `localization' process \cite{Che-Eil:48} of {\it e.g.} the
linear maps $\alpha^0$ so that they become $\fG$-valued covariant
vector {\it fields} on such a manifold, is an open question. In
fact, as we shall see below, the operator that allows us to write
MC-like equations for a FA is not an exterior differential, as it
was the case for the generalized MC equations of the GLAs
$\mathcal{G}$. Further, for $n\geq 3$ it is Lie$\,\fG$-invariance
(and not `$\fG$-invariance') what is defined, a fact related to
the previous problem. Finally, there is no matrix realization for
the elements of $\fG$ in general allowing us to construct invariant
polynomials along the pattern that gives them in the Lie algebra
case; in fact, as soon as one speaks of invariance, it is Lie$\,\fG$
rather than $\fG$ that appears, since the notion of invariance is
associated to a {\it Lie} group of transformations.
\end{remark}

\begin{subsection}{MC-like equations for FAs}
\label{sec:MClikeFA}
{\ }

The discussion in Rem.~\ref{re:MCcoho} prompts us to see whether it
is possible to define MC-like equations for $n>2$ FAs as we did for
the GLAs in eqs.~\eqref{completef}, \eqref{completeg}. The key
difference is that, in contrast with the Lie algebra cohomology, the
$p$-cochains for the FA cohomology are no longer antisymmetric in
all their $\fG$ arguments for $p\geq 2$ and, further, there is no
underlying group manifold. Consider linear $\fG$-valued maps
$\omega$ on $\fG$, $\omega:\fG\rightarrow\fG$, $\omega=\omega^a\circ
X_a$, as zero-cochains for the cohomology complex of
Def.~\ref{FA-cob-op} (in the $n=2$ Lie case, $\omega$ would be the
canonical one-form $\theta$ on the group manifold). Then, the action
of the coboundary operator in eq.~\eqref{n-cobop} on
$\omega\in C^0(\fG,\fG)$ gives
$\delta\omega(\mathscr{X},Z)=-\omega([X_{a_1},\dots,X_{a_{n-1}},Z])$.

   Let us now define the {\it $n$-bracket of $\fG$-valued maps} in
terms of the FA by
\begin{equation}
\label{FA-onef-bra}
[\omega,\dots , \omega]:=
\omega^{a_1}\wedge\dots\wedge\omega^{a_n}\circ [X_{a_1},\dots,X_{a_n}] \quad i.e. \quad
[\omega,\dots , \omega]^c= f_{a_1\dots a_n}{}^c \omega^{a_1}\wedge\dots\wedge\omega^{a_n} \; ,
\end{equation}
since $[\omega,\dots ,\omega]=[\omega,\dots ,\omega]^c\,X_c$.
Then, in analogy with the MC eqs.~\eqref{3canonMC}, we may write in
terms of the coboundary operator $\delta$
\begin{equation}
\label{MCforFAs} \delta\omega=
-\frac{1}{n!}[\omega, \dots,\omega] \qquad \hbox{or}
\qquad (\delta\omega)^c=-\frac{1}{n!}f_{a_1\dots a_n}{}^c \,
\omega^{a_1}\wedge\omega^{a_2}\dots \wedge\omega^{a_n} \; .
\end{equation}
 Note that in the FA case we cannot write
$[\omega, \mathop{\dots}\limits^n , \omega] \propto \omega\wedge
\mathop{\dots}\limits^n \wedge \omega$ since the $\wedge$ product of
the $\omega$'s would generate a multibracket
$[X_{a_1},\dots,X_{a_n}]$, given by the full antisymmetrization of
the products $X_{a_1},\dots,X_{a_n}$, rather than a FA $n$-bracket.
Eqs.~\eqref{MCforFAs} constitute the {\it MC equations for a FA}
$\fG$ as defined by the FA cohomology coboundary operator $\delta$.
For $n=2$, of course, $\delta \rightarrow d$, $\omega\rightarrow
\theta$, $d\theta=-\frac{1}{2} [\theta,\theta]= -\theta\wedge\theta
\in\mathrm{Hom}(\wedge^2\fg,\fg)\,$ and the MC-like equations for
$\fG$ above become the MC ones of a Lie algebra $\fg$.

$\delta\omega$ is clearly a $\fG$-valued one-coboundary.
The expression stating that $\delta^2\omega=0$, as it follows from
Def.~\ref{FA-cob-op} or eq.~\eqref{cob-1-coch}, may be formally
rewritten in the form
\begin{eqnarray}
\label{del-2-ome}
\nonumber
[\omega,\mathop{\dots}\limits^{n-1},\omega,
[\omega, \mathop{\dots}\limits^n,\omega]]  -
[\omega, [\omega, \mathop{\dots}\limits^n,\omega],
\omega, \mathop{\dots}\limits^{n-2},\omega] - \qquad \\  \,
[\omega,\omega [\omega, \mathop{\dots}\limits^n,\omega ],
\omega, \mathop{\dots}\limits^{n-3},\omega] - \dots  -
[\omega,\mathop{\dots}\limits^{n-1} ,\omega,
[\omega, \mathop{\dots}\limits^n ,\omega]]=0   \; ,
\end{eqnarray}
where {\it e.g.}, the second double bracket means
\begin{equation}
\nonumber
[\omega, [\omega, \mathop{\dots}\limits^n,\omega],
\omega, \mathop{\dots}\limits^{n-2},\omega] \propto
\omega^{b_1}\omega^{a_1}\dots \omega^{a_{n-1}}\omega^{b_2}\dots
\omega^{b_n}[X_{b_1},[X_{a_1},\dots,X_{a_{n-1}},X_{b_2}],X_{b_3},\dots,X_{b_n}] \;.
\end{equation}
Obviously, the wedge product for the $\omega$'s {\it cannot} be
used above since there is antisymmetry only for the groups of $\omega$'s involving
the appropriate indices ({\it cf.} eq.~\eqref{FA-onef-bra}). The symmetry
properties of the $\fG$ arguments of any two-cochain follow from its definition,
$\alpha^2\in \mathrm{Hom}(\wedge^{n-1}\fG\otimes\wedge^{n-1}\fG\wedge\fG,\fG)$;
the FI, unlike the GJI, does not involve a full antisymmetrization.

  If we now define the action of
$[\omega,\mathop{\dots}\limits^{n-1},\omega,\,\cdot\,]$ by
\begin{equation}
[\omega,\mathop{\dots}\limits^{n-1},\omega,\,\cdot\,]:\omega \mapsto
[\omega, \dots,\omega] \; .
\end{equation}
we see that eq.~\eqref{del-2-ome} states that this operator
acts on the $n$-bracket of $\fG$-valued elements of $\fG^*$,
$[\omega, \mathop{\dots}\limits^n ,\omega]$, as a
{\it derivation}, and that this property follows
directly from the FI/nilpotency of $\delta$.
\medskip

  For $n=3$, the MC-like equation \eqref{MCforFAs} for FAs above
has the structure of the Basu-Harvey \cite{Bas-Har:05} equation.
This has appeared \cite{Ba-La:06,Ber:07,Ba-La:07b,Ho-Hou-Ma:08,Bo-Ta-Zab:08}
in the description of M-branes and in the BLG model and has its $n=2$ MC
precedent in the Nahm equation \cite{Nahm:80}.
\end{subsection}

\subsection{Dual FA homology complex for the trivial action}
\label{FA-homol} {\ }

The homology operator for $n$-algebras was introduced by Takhtajan
\cite{Tak:95, Da-Tak:97} in the context of the Nambu-Lie algebras
\cite{Tak:93} to be described in Sec. \ref{sec:NP}. We shall
consider here (see \cite{AIP-B:97}) the homology dual to the
cohomology complex of Def. \ref{FA-cob-op} adapted to the central
extension problem. Let the $p$-chains $C_p$ be defined as
$(\mathscr{X}_1,\dots,\mathscr{X}_p, Z)\in \wedge^{n-1}\fG
\otimes\mathop{\cdots}\limits^p\otimes\wedge^{n-1}\fG\wedge\fG$.

\begin{definition} ({\it n-algebra homology operator}){\ }

    On a $C_1$ chain, the homology operator $\partial:C_1\to
C_0\equiv\fG$ is defined by $\partial:(\mathscr{X},Z) \mapsto
\partial(\mathscr{X},Z) \equiv \partial (X_1,\dots,X_{n-1},Z)
 := [X_1,\dots,X_{n-1},Z]$. In
general, $\partial : C_p \rightarrow C_{p-1}$ is given by
\begin{equation}
\begin{aligned}
\partial(\mathscr{X}_1,\dots,\mathscr{X}_p,Z)=&
\displaystyle \sum\limits_{1\le i<j}^{p} (-1)^i
(\mathscr{X}_1,\dots,\widehat{\mathscr{X}_i},\dots,
\mathscr{X}_i\cdot \mathscr{X}_j,\dots ,\mathscr{X}_{p},Z) \cr &
\displaystyle +\sum\limits_{i=1}^{p} (-1)^i
(\mathscr{X}_1,\dots,\widehat{\mathscr{X}_i},\dots,
\mathscr{X}_p,\mathscr{X}_i\cdot Z)\quad.
\end{aligned}
\end{equation}
\end{definition}

It is not difficult to check that $\partial$ is the dual of
$\delta$. Using the definition of $\delta$ in \eqref{n-cobop}, we
find for a ($p-1$)-cochain $\alpha$
\begin{equation}
\begin{aligned}
\alpha ( \partial (\mathscr{X}_1,\dots ,\mathscr{X}_p,Z) ) & =
\sum_{1\leq i<j}^{p} (-1)^i \,
\alpha(\mathscr{X}_1,\dots,\widehat{\mathscr{X}_i},\dots,
\mathscr{X}_{j-1},\mathscr{X}_i \cdot
\mathscr{X}_j,\dots,\mathscr{X}_p, Z) \cr & +
 \sum_{i=1}^{p} (-1)^i
 \alpha(\mathscr{X}_1,\dots,\widehat{\mathscr{X}_i},\dots,\mathscr{X}_p,\mathscr{X}_i
 \cdot Z) \quad ,
\end{aligned}
\end{equation}
 which shows that, on $(\mathscr{X}_1,\dots,\mathscr{X}_p,Z)$,
\begin{equation}
\alpha \, \partial = \delta \, \alpha \quad .
\end{equation}

\subsection{Other FA cohomology complexes}{\ }

  We may also consider cochains of the type
$\alpha(\mathscr{X}_1,\dots,\mathscr{X}_p)$, depending on $p$
fundamental objects or on $p(n-1)$ elements of $\fG$ {\it i.e.},
taking arguments in $(\wedge^{n-1}\fG)\otimes \mathop{\cdots}
\limits^p \otimes (\wedge^{n-1}\fG)$. It is also possible to have
cochains taking the same arguments but that are valued on a FA left
module $V$ in the sense of eq.~\eqref{gmod}, elements of
$\hbox{Hom}(\otimes^p(\wedge^{n-1}\fG), V)$. In this case, there is
a FA cohomology complex that closely mimetizes the Lie algebra
cohomology one when the fundamental object arguments are replaced by
Lie algebra arguments. The coboundary operator for this $\fG$-module
cohomology complex is defined \cite{Gau:96} by
\begin{equation}
\label{cob-p(n-1)}
\begin{aligned}
 \delta \alpha (\mathscr{X}_1,\dots, \mathscr{X}_{p+1})=
 \sum_{i=1}^{p+1} & (-1)^{i+1} \,  \mathscr{X}_i \cdot
  \alpha(\mathscr{X}_1,\dots,\widehat{\mathscr{X}_i},\dots,\mathscr{X}_{p+1})  \\
   +\sum_{i,j\atop 1\leq i<j}^{p+1} & (-1)^i
 \alpha(\mathscr{X}_1,\dots,\widehat{\mathscr{X}_i},\dots,
 \mathscr{X}_i\cdot\mathscr{X}_j, \dots,\mathscr{X}_{p+1}) \; ,
 \end{aligned}
\end{equation}
where the dot in the first line above indicates the action of the
fundamental object on the $\fG$-module where the cochains takes
values (eqs.~\eqref{Lie-n-Lie}, \eqref{gmod}).

The proof of the nilpotency of $\delta$ above follows the same
pattern as that of the Lie algebra cohomology operator $s$
(eqs.~\eqref{Lie-cob},\eqref{Lie-cob-ad}). The key issue is that all
terms in $\delta(\delta\alpha^p)$ may be collected in groups of two
types, with three terms each group. The groups of the first type
share a $p$-cochain with equal arguments, on which two fundamental objects act in the form
$\mathscr{X}\cdot(\mathscr{Y}\cdot\alpha^p)-\mathscr{Y}\cdot(\mathscr{X}\cdot\alpha^p)
-(\mathscr{X}\cdot\mathscr{Y})\cdot\alpha^p =0$, and thus add up
to zero because $\alpha^p$ takes values on a $\fG$-module,
eq.~\eqref{gmod}.  Due to $p$-linearity, the groups of the second type
add up to $\alpha^p$ taking values on equal fundamental objects
but for one, which turns out to be the addition of
$(\mathscr{X}\cdot\mathscr{Y})\cdot\mathscr{Z} +
\mathscr{Y}\cdot(\mathscr{X}\cdot\mathscr{Z}) -
\mathscr{X}\cdot(\mathscr{Y}\cdot\mathscr{Z})=0$, again zero
 by \eqref{Lie-setb}.
\medskip

  When the action is non-trivial, other FA
 cohomology complexes are possible, as the one that is
relevant for deformations of FA algebras to described in
Sec.~\ref{sec:deform-FA} below. There, both the left and right
actions enter, as was seen in Sec.~\ref{extLeib} and will be found
again in Sec.~\ref{coho-for-def}. In fact, the cohomology of Leibniz
algebras underlies that of the FAs, as will be discussed in
Sec.~\ref{n-Leib-coho}.

\subsection{Deformations of Filippov algebras}
\label{sec:deform-FA}{\ }

We now look at the cohomology complex that governs the deformations
of FAs, in close parallel with the Lie algebra deformations
considered in Sec.~\ref{sec:def-Lie}.

\subsubsection{Infinitesimal deformations of FAs}
\label{ssec:ifi-def}{\ }

  Let us begin with the simple example of a
3-Lie algebra; the results, when written in terms of the
fundamental objects, will apply immediately to any $n$-Lie
algebra. Proceeding as in the case of the Lie algebra in
Sec.~\ref{sec:def-Lie}, an infinitesimal deformation requires the
existence of a deformed three-bracket $[X_1,X_2,Z]_t$,
\begin{equation}
\label{def-3-bra}
 [X_1,X_2,Z]_t = [X_1,X_2,Z] + t \alpha (X_1,X_2,Z) + O(t^2) \quad ,
 \end{equation}
where $\alpha(X_1,X_2,Z)$ is a three-linear skewsymmetric
$\fG$-valued map $\alpha: \wedge^2 \fG \wedge \fG \rightarrow
\fG\,$, $\,\alpha:(\mathscr{X},Z) \mapsto\alpha(\mathscr{X},Z)$,
depending on one fundamental object plus an element of $\fG$. The
bracket $[X_1,X_2,Z]_t$ satisfies the deformed FI, which in terms
of the fundamental objects reads
\begin{equation}
\label{def-3JI}
 [\mathscr{X}, (\mathscr{Y}\cdot Z)_t]_t=[(\mathscr{X}\cdot\mathscr{Y})_t,
 Z]_t + [\mathscr{Y},(\mathscr{X}\cdot Z)_t]_t \quad,
\end{equation}
where the subindex $t$ indicates that the $n$-brackets
$(\mathscr{X}\cdot Z)_t\,,\, (\mathscr{Y}\cdot Z)_t$ and those that
appear in the definition of $\mathscr{X}\cdot \mathscr{Y}$
(eq.~\eqref{dot-compo}) are deformed as in eq.~\eqref{def-3-bra}. At
first order, this gives the condition
\begin{equation}
\label{def-1-coch}
\begin{aligned}
 ad_{\mathscr{X}}& \alpha(\mathscr{Y}, Z) -ad_{\mathscr{Y}}\alpha(\mathscr{X}, Z)
  -[\alpha(\mathscr{X},Y_1),Y_2, Z] -[Y_1,\alpha(\mathscr{X},Y_2),Z] \\
 -\alpha & (\mathscr{X}\cdot\mathscr{Y}, Z) -\alpha(\mathscr{Y}, \mathscr{X}\cdot Z)
 + \alpha(\mathscr{X}, \mathscr{Y}\cdot Z) = 0 \quad .
 \end{aligned}
\end{equation}

Looking at the square brackets above we may introduce
$\alpha(\mathscr{X},\quad)\cdot \mathscr{Y}$ as the fundamental
object defined by
\begin{equation}
\label{dosf}
\begin{aligned}
 \alpha(\mathscr{X},\quad)\cdot \mathscr{Y} :=&
 (\alpha(\mathscr{X},\quad)\cdot Y_1,\,Y_2) \,+\, (Y_1,\,\alpha(\mathscr{X},\quad)\cdot Y_2)\\
 =& (\alpha(\mathscr{X},Y_1), Y_2) \,+\, (Y_1,\alpha(\mathscr{X},Y_2)) \qquad
[\,\alpha(\mathscr{X},\quad)\cdot Y_i  :=\alpha(\mathscr{X},Y_i)\,]  \quad .
\end{aligned}
\end{equation}
Then, the condition that $\alpha(\mathscr{X},Z)$ in
\eqref{def-3-bra} defines a  FA bracket may be written as ({\it
cf.} eq.~\eqref{cob-1-coch})
\begin{equation}
\label{def-1-coch-b}
\begin{aligned}
(\delta\alpha)(\mathscr{X}, \mathscr{Y}, Z) = ad_{\mathscr{X}}&
\alpha(\mathscr{Y}, Z) -ad_{\mathscr{Y}}\alpha(\mathscr{X}, Z)
  -(\alpha(\mathscr{X}, \quad )\cdot \mathscr{Y})\cdot Z \\
 -\alpha (\mathscr{X} & \cdot\mathscr{Y}, Z)  -\alpha(\mathscr{Y}, \mathscr{X}\cdot Z)
 + \alpha(\mathscr{X}, \mathscr{Y}\cdot Z) = 0 \quad ,
 \end{aligned}
\end{equation}
which may be simplified by writing
$ad_{\mathscr{X}}\alpha(\mathscr{Y},Z)=
\mathscr{X}\cdot\alpha(\mathscr{Y},Z)$, etc. (for the trivial
action, $\delta\alpha$ reproduces eq.~\eqref{cob-1-coch}). The
identification of this expression with $(\delta\alpha)(\mathscr{X},
\mathscr{Y}, Z)$, as written in its $l.h.s.$, indicates that
$\alpha(\mathscr{X}, Z)$ has to be a one-cocycle
$\alpha\equiv\alpha^1$ for the $\fG$-valued cohomology that will be
defined below\footnote{We note that for the $n=2$ case, where the
fundamental objects are simply the elements of the FA algebra,
$(\alpha(X, \quad )\cdot Y)\cdot Z =\alpha(X,Y)\cdot Z\;$ and
$\,X\cdot Y=[X,Y]$. Then, eq.~\eqref{def-1-coch-b} may be written as
\begin{equation*}
\begin{aligned}
&(\delta\alpha)(X, Y, Z) = X\cdot
\alpha(Y, Z) -Y\cdot\alpha(X, Z)
  -\alpha(X,Y)\cdot Z \\
 &\qquad \qquad \qquad -\alpha ([X, Y], Z)  -\alpha(Y, [X, Z])
 + \alpha(X, [Y,Z])=0 \quad ,
 \end{aligned}
\end{equation*}
which reproduces eq.~\eqref{Leib-2-cocy-b}, which is
the two-cocyle condition for the Leibniz algebra cohomology .
Further, since an $n=2$ FA is actually an ordinary Lie algebra,
the third term $-\alpha(X,Y)\cdot Z$ may be rewritten as
$Z\cdot \alpha(X,Y)$ and the above expression corresponds to
the second equation in \eqref{Lie-coh-ex}.}.

   Using now the ordinary Lie algebra experience, we may define a
zero-cocycle as a $\fG$-valued map $\alpha^0 :\fG \rightarrow
\fG\;$, $\alpha^0: Z\rightarrow \alpha^0(Z)$, which does not
contain any fundamental object as an argument. It will generate a
one-coboundary by
\begin{equation}
\label{gval-one-cob}
 (\delta\alpha^0)(\mathscr{X},Z)= \mathscr{X}\cdot \alpha^0(Z)
 -\alpha^0(\mathscr{X}\cdot Z)
 + (\alpha^0(\quad)\cdot \mathscr{X})\cdot Z \quad ,
\end{equation}
where the last term is $(\alpha^0(\quad)\cdot \mathscr{X})\cdot Z
=[\alpha^0(X_1),X_2,Z]+[X_1,\alpha^0(X_1),Z]$. Indeed, a
calculation using eqs.~\eqref{def-1-coch-b} and \eqref{Lie-setb}
shows that
$(\delta(\delta\alpha^0))(\mathscr{X},\mathscr{Y},Z)\equiv 0$.
   To check that \eqref{gval-one-cob} is a sensible definition, we
now look at the expression of the one-cocycle of an infinitesimal
deformation when the deformation is actually trivial. Let
$X_i'=X_i -t\alpha^0(X_i)$ the redefinition of the basis of the
apparently deformed FA which removes the $t\alpha(\mathscr{X},Z)$
term in \eqref{def-3-bra}. Then (we may assume directly that we
are dealing with an $n$-Lie algebra) the new bracket for the primed
generators reads, to order $t$,
\begin{equation}
\begin{aligned}
(\mathscr{X}' & \cdot Z')_t \equiv [X'_1,\dots, X'_{n-1}, Z']_t =
[X_1-t \alpha^0(X_1),\dots, X_{n-1}-t\alpha^0(X_{n-1}),\,
Z-t\alpha^0(Z)]_t\\ & = [X_1,\dots, X_{n-1}, Z]_t -t\,
(\;\sum_{i=1}^{n-1}[X_1,\dots,X_{i-1},\alpha^0(X_i),X_{i+1},\dots,X_{n-1},Z]
 +[X_1,\dots,X_{n-1},\alpha^0(Z)]\,) \\
 & =[X_1,\dots, X_{n-1}, Z] -t\alpha^0(\mathscr{X}\cdot
 Z) +t\, (\alpha^1(\mathscr{X},Z)+ \alpha^0(\mathscr{X}\cdot
 Z)-(\alpha^0(\quad)\cdot \mathscr{X})\cdot Z -\mathscr{X}\cdot
 \alpha^0(Z)\, )\\
 & =[X_1,\dots, X_{n-1}, Z]' + t\,(\alpha^1(\mathscr{X},Z)
 -(\delta\alpha^0)(\mathscr{X},Z)\;) \quad,
\end{aligned}
\end{equation}
where in the third line we have added and subtracted
$t\alpha^0(\mathscr{X}\cdot Z)= t\alpha^0([X_1,\dots,X_{n-1},Z])$
and used definition \eqref{gval-one-cob} in the fourth, the form of
which now appears justified. We see that the $t$ term above
disappears, and with it the infinitesimal deformation, if $\alpha^1$
is the one-coboundary $\alpha^1=\delta \alpha^0$. Thus, the
infinitesimal deformations of $n$-Lie algebras are governed by the
first\footnote{As mentioned, the order $p$ is defined by the number
of fundamental objects contained in the $p$-cochain. Within this
terminology, the infinitesimal deformations of a Lie algebra $\fg$
are governed by the $\fg$-valued one-cocycles for the adjoint
action, $H^1_{ad}$, where the action of $\mathscr{X}$ is defined as
above and in Def. \ref{deform-defb} below.} cohomology group
$H^1_{ad}$ for \eqref{def-1-coch-b} where $ad$ refers generically to
the action of the fundamental objects on the cochains (the general
cocycle condition will be given in Sec. \ref{coho-for-def} below); a
FA algebra is {\it stable}, or {\it rigid}, if $H^1_{ad}=0$ . As
defined here, the $p$-cochains have $p$ fundamental objects plus an
element of $\fG$ as their arguments; in all, $p(n-1)+1$ elements of
$\fG$.

 Comparing eqs.~\eqref{gval-one-cob}, \eqref{def-1-coch-b} with the
equivalent formulae for previous cohomology complexes, we see that
the coboundary operator producing these expressions (and therefore
adapted to the deformation problem), leads to a cohomology complex
different from {\it e.g.}, that defined by eq.~\eqref{cob-p(n-1)}.
This is because the cochain spaces $C^p$ for \eqref{cob-p(n-1)} are
 $\hbox{Hom}(\otimes^p(\wedge^{n-1}\fG),\fG)$ and thus
automatically unsuitable for deformations since (as it was also the
case for the central extensions cohomology) the cochains defined
here naturally include a solitary element of $\fG$ as an argument.
As a result, the situation for the deformation of FAs is different from
that for the Lie algebras, since the relevant cohomology leads to
\eqref{def-1-coch-b} rather than to the corresponding one-cocycle
expression from \eqref{cob-p(n-1)}, which is the cohomology complex
that follows the pattern of the ordinary Lie algebra cohomology
coboundary operator $s$ in eq.~\eqref{Lie-cob}. The key difference
is the definition of the action of the fundamental object in the
third terms in the right hand sides of eqs.~\eqref{def-1-coch-b} and
\eqref{gval-one-cob}. In all, this reflects the different dependence
of the cochains on the fundamental objects and the single element of
$\fG$ on the one hand and, on the other, of the actions induced by
these on the cochains.
\medskip

\subsubsection{Higher order deformations of FAs}{\ }

Let us now look at the problem of extending an infinitesimal
deformation to second order. Consider
\begin{equation}
[X_1,\dots,X_{n-1},Z]_t\equiv (\mathscr{X}\cdot Z)_t =
(\mathscr{X}\cdot Z) + t\alpha_1(\mathscr{X}, Z)
+t^2\alpha_2(\mathscr{X}, Z) + O(t^3) \quad,
\end{equation}
where the subindices $1,2$ refer to the deformation order; note that
$\alpha_1,\alpha_2\in \hbox{Hom}(\wedge^{n-1}\fG\wedge \fG, \fG)$
{\it i.e.}, both are $\fG$-valued one-cochains  and, further,
$\alpha_1$ is already assumed to be a one-cocyle. We now ask
ourselves whether the infinitesimal deformation $\alpha_1$ can be
expanded to a second order one determined by a certain $\alpha_2$.
Again, the condition is that the FI (eqs.~\eqref{eq:FI},
\eqref{def-3JI} for the deformed FA), that we shall now write in the
form
\begin{equation}
(\mathscr{X}\cdot(\mathscr{Y}\cdot Z)_t)_t-
((\mathscr{X}\cdot\mathscr{Y})_t\cdot Z)_t -
(\mathscr{Y}\cdot(\mathscr{X}\cdot Z)_t)_t =0 \quad ,
\end{equation}
be satisfied up to order $O(t^3)$. Since $\alpha_1$ is an
infinitesimal deformation, the terms of order $t$ are zero since
$\delta \alpha_1=0$. The terms of order $t^2$ give the condition
\begin{equation}
\label{FA-def-order2}
\begin{aligned}
\alpha_1(\mathscr{X},\alpha_1(\mathscr{Y},Z))-
\alpha_1(\alpha_1(\mathscr{X},\quad)\cdot\mathscr{Y}\,,\,Z)
-\alpha_1(\mathscr{Y},\alpha_1(\mathscr{X},Z))  \qquad\\
+ \mathscr{X}\cdot \alpha_2 (\mathscr{Y},Z)
-\alpha_2(\mathscr{X}\cdot\mathscr{Y},Z)
-\alpha_2(\mathscr{Y},\mathscr{X}\cdot Z) \qquad \\
-\mathscr{Y}\cdot\alpha_2(\mathscr{X},Z)+\alpha_2(\mathscr{X},\mathscr{Y}\cdot
Z) - (\alpha_2(\mathscr{X},\quad)\cdot \mathscr{Y})\cdot Z \qquad \\
\equiv  \gamma(\mathscr{X},\mathscr{Y},Z) +(
\delta\alpha_2)(\mathscr{X},\mathscr{Y},Z)=0 \quad \qquad \qquad
\end{aligned}
\end{equation}
using the one cocycle condition \eqref{def-1-coch-b}.
The above expression implies $\delta\gamma=0$ and thus
$\gamma(\mathscr{X},\mathscr{Y},Z)$ has to be a two-cocycle.
The two-cochain $\gamma$ in (\ref{FA-def-order2}) is
\begin{equation}
\label{deformations-b1}
    \gamma(\mathscr{X},\mathscr{Y},Z) := \alpha_1(\mathscr{X}, \alpha_1(\mathscr{Y},Z))
    -\alpha_1(\alpha_1(\mathscr{X},\ )\cdot \mathscr{Y},Z)-\alpha_1(\mathscr{Y},
    \alpha_1(\mathscr{X},Z))\; .
\end{equation}
The action of the coboundary operator $\delta$
on $\gamma$ is the case $p=2$ in eq.~\eqref{cob-op-for-defs} below.
A non completely trivial calculation shows that
\begin{equation*}
\label{deformations-b9}
    \delta\gamma(\mathscr{X}_1,\mathscr{X}_2,\mathscr{X}_3,Z) =
     \sum_{a=1}^{n-1} \delta\alpha_1((X_{(2)1},\dots, \alpha_1 (\mathscr{X}_1,
     X_{(2)a}),\dots, X_{(2)n-1}),\mathscr{X}_3,Z) = 0  \quad ,
\end{equation*}
so that $\gamma$ is indeed a two-cocycle since $\delta\alpha_1=0$.

Thus, if $\gamma$ is a two-coboundary generated by a one-cochain $\alpha'$,
$\gamma=\delta\alpha'$, it is sufficient to take the second order
one-cochain as $\alpha_2=-\alpha'$ to have the above condition
fulfilled as in the Lie algebra case (Sec.~\ref{sec:high-ord-def}).
Proceeding in this way, we may conclude that there is no obstruction
if $H^2_{ad}=0$, where the second cohomology group is defined with
respect to the cohomology complex defined in the next section.

\subsection{Cohomology complex for deformations
of $n$-Lie algebras}{\ }
\label{coho-for-def}

The above discussion leads us naturally to introducing the ingredients
of the deformation theory of FAs.

\begin{definition} ($\fG$-{\it valued p-cochains}){\ }

A $p$-cochain for the deformation cohomology is a map $\alpha^p:
\wedge^{(n-1)}\fG\otimes
\mathop{\cdots}\limits^p\otimes\wedge^{(n-1)}\fG\wedge\fG
\rightarrow \fG$ so that
\begin{equation}
\label{fg-p-coch}
 \alpha :(\mathscr{X}_1,\dots,\mathscr{X}_p,Z) \mapsto
\alpha(\mathscr{X}_1,\dots,\mathscr{X}_p,Z)\equiv \alpha
(X^1_1,\dots,X^{n-1}_1, \dots ,X^1_p,\dots, X^{n-1}_p, Z) \in \fG\;.
\end{equation}
\end{definition}
We shall refer generically to the space of the above $p$-cochains
as $C^p_{ad}$, since the (left) action of the
fundamental objects on the $\fG$-valued $\alpha^p$ is given by
\begin{equation*}
\mathscr{X}\cdot \alpha^p(\mathscr{X}_1,\dots,\mathscr{X}_p,Z)=
[X_1,\dots,X_{n-1},\,\alpha^p(\mathscr{X}_1,\dots,\mathscr{X}_p,Z)]\quad.
\end{equation*}

To define the action of the coboundary operator, we need first the
equivalent of \eqref{dosf} for an $n$-Lie algebra. This is given by
the sum of fundamental objects
\begin{equation}
\label{p-coc-ftal}
 \alpha^p(\mathscr{X}_1,\dots,\mathscr{X}_p\,,\quad)\cdot
 \mathscr{Y}=
 \sum_{i=1}^{n-1}
 (Y_1,\dots,\alpha^p(\mathscr{X}_1,\dots,\mathscr{X}_p,Y_i),\dots,Y_{n-1})
 \quad ,
\end{equation}
since, by definition,
$\alpha^p(\mathscr{X}_1,\dots,\mathscr{X}_p\,,Y_i)\in \fG$.

The previous discussion has allowed us to define the action of the
coboundary operator on zero-cochains $\alpha^0$ and one-cochains
$\alpha^1$ by eqs.~\eqref{gval-one-cob}, \eqref{def-1-coch-b}
respectively. The natural generalization of the above expressions
to the $p$-cochain case leads us to the deformation cohomology
complex below.

\begin{definition}({\it Coboundary operator for the deformation
cohomology}){\ }

\label{deform-defb}
The coboundary operator $\delta:C^p_{ad}\rightarrow C^{p+1}_{ad}$ is given by
\begin{equation}
\label{cob-op-for-defs}
\begin{aligned}
(\delta\alpha^p) &
(\mathscr{X}_1,\dots,\mathscr{X}_p,\mathscr{X}_{p+1},Z)= \\
&\sum_{1\leq j<k}^{p+1} (-1)^j
\alpha^p(\mathscr{X}_1,\dots,\widehat{\mathscr{X}_j},\dots,\mathscr{X}_{k-1},
\mathscr{X}_j\cdot
\mathscr{X}_k,\mathscr{X}_{k+1},\dots,\mathscr{X}_{p+1}, Z)\\
+& \sum_{j=1}^{p+1} (-1)^j \alpha^p
(\mathscr{X}_1,\dots,\widehat{\mathscr{X}_j},\dots,\mathscr{X}_{p+1},\mathscr{X}_j\cdot
 Z) \\
 + &\sum_{j=1}^{p+1} (-1)^{j+1} \mathscr{X}_j \cdot
 \alpha^p(\mathscr{X}_1,\dots,\widehat{\mathscr{X}_j}\dots,\mathscr{X}_{p+1},
 Z) \\
 +&(-1)^{p}
 (\alpha^p(\mathscr{X}_1,\dots,\mathscr{X}_p\,,\quad)\cdot
 \mathscr{X}_{p+1})\cdot Z \quad ,
\end{aligned}
\end{equation}
\end{definition}
\noindent
where the last term is defined by eq.~\eqref{p-coc-ftal}.
For $p=0,1$ this reproduces eqs.~\eqref{gval-one-cob}
and the first equality in \eqref{def-1-coch-b}, with
which one may check that
$(\delta^2\alpha^0)(\mathscr{X},\mathscr{Y},Z)=0$, already an
altogether non-trivial calculation. Explicitly, on a one-cochain
$\alpha^1$ (see eq.~\eqref{def-1-coch-b}),
\begin{equation*}
\begin{aligned}
(\delta\alpha^1)(\mathscr{X}, \mathscr{Y}, Z) = \mathscr{X} \cdot &
\alpha^1(\mathscr{Y}, Z) - \mathscr{Y}\cdot\alpha^1(\mathscr{X}, Z)
  -(\alpha^1(\mathscr{X}, \quad )\cdot \mathscr{Y})\cdot Z \\
 -\alpha^1(\mathscr{X} & \cdot\mathscr{Y}, Z)  -\alpha^1(\mathscr{Y}, \mathscr{X}\cdot Z)
 + \alpha^1(\mathscr{X}, \mathscr{Y}\cdot Z)  \quad ,
 \end{aligned}
\end{equation*}
an expression (set equal to zero)
repeatedly used in the previous calculation
for the one-cocycle defining the first
order deformation of a FA.
\\

It may be seen that the cohomology complex defined above is
essentially equivalent to the one introduced by Gautheron in
an important paper \cite{Gau:96}. There, he was the first to
consider the deformation cohomology for Nambu algebras, a
particular case of FAs (Ex.~\ref{ex:Nam-FA} and Sec.~\ref{sec:NP});
see also \cite{Da-Tak:97, Rot:05}. By defining
$\fG$-valued cocycles $Z^p_{ad}$ and coboundaries $B^p_{ad}$
in the usual manner, the previous results may be formulated
as the following

\begin{theorem} ({\it Deformations of FA algebras})
\label{th:FAdef}
{\ }

Let $\fG$ be a FA. The first cohomology group
$H^1_{ad}$ for the above coboundary operator governs the
infinitesimal deformations of $\fG$. If $H^1_{ad}=0$, the FA is
{\it rigid}. If it admits an infinitesimal deformation, the
obstructions to move to higher order result from $H^2_{ad}\not=0$.
\end{theorem}{\ }

In analogy with the Lie algebra case, the following theorem \cite{Az-Iz:09}
holds

\begin{theorem} (Semisimple FAs and deformations)
\label{th:rigFAsim}
{\ }

Semisimple $n$-Lie algebras are rigid.
\end{theorem}

This theorem and Th.~\ref{th:s-s-FA-ext} constitute the extension of
the standard Whitehead's lemma for Lie algebras (see \cite{Jac:79})
to $n$-Lie algebras. Since Lie algebras are $n=2$ FAs, this result
\cite{Az-Iz:09} states that Whitehead's lemma holds true for any
semisimple Filippov algebra, $n\geq 2$. The simple FAs inherit, so
to speak, the rigidity of their Lie ancestors $so(3)$ and $so(1,2)$.
\medskip

\section{$n$-Leibniz algebra cohomology. LA {\it vs.} FA cohomology}
\label{n-Leib-coho}
{\ }

 As discussed in Sec.~\ref{sec:n-Leib}, the composition
$\mathscr{X}\cdot\mathscr{Y}$ of fundamental objects,
rewritten as $[\mathscr{X},\mathscr{Y}]$, has the properties
of a Leibniz algebra commutator. Further, since the skewsymmetry
of the FA $n$-bracket is not needed for the nilpotency of $\delta$
in, say, Def.~\ref{FA-cob-op}, $\delta$ also defines a coboundary
operator for the cohomology of an $n$-Leibniz algebra
$\mathfrak{L}$, where now the $p$-cochains may
considered as elements $\alpha\in
(\otimes^{(n-1)}\mathfrak{L}^*)\otimes
\mathop{\cdots}\limits^p\otimes
(\otimes^{(n-1)}\mathfrak{L}^*)\otimes \mathfrak{L}^*=
\otimes^{p(n-1)+1}\mathfrak{L}^*=\hbox{Hom}(\otimes^{p(n-1)+1}\mathfrak{L}\,,
\,\mathbb{R}$). The key ingredient that guarantees the nilpotency
of the coboundary operator is the FI, which implies
eqs.~\eqref{Lie-setb}, \eqref{gmod}, etc.  The only difference between
the $n$-Lie algebra and $n$-Leibniz algebra cohomology complexes
for the trivial representation,
irrelevant for the nilpotency of $\delta$, is that for an $n$-Lie algebra $\mathscr{X}\in \wedge^{n-1}\fG$ and for a
 Leibniz algebra $\mathscr{X}\in \otimes^{n-1}\mathfrak{L}$.
Hence, with the appropriate changes in the definition of the
$p$-cochain spaces $C^p$, the coboundary operator $\delta$ of
Def.~\ref{FA-cob-op} defines both the corresponding cohomologies for
$n$-Lie $\fG$ and $n$-Leibniz $\mathfrak{L}$ algebras adapted to the
central extension problem, which correspond to the trivial action.

 Thus, the FA cohomologies defined by the coboundary operators \eqref{n-cobop}
(and the corresponding homology) constitute simply the application
of the $n$-Leibniz algebra coboundary operators that define the
corresponding $n$-Leibniz algebra cohomology complexes to the
$n$-Lie algebra cohomology. Analogously, an $n$-Leibniz algebra
cohomology $p$-cochain corresponding to the FA cohomology complex
defined by \eqref{cob-p(n-1)} is an element of
$\hbox{Hom}(\otimes^{p(n-1)}\mathfrak{L},V)$, where $V$ is the
corresponding $\mathfrak{L}$-module. In fact, $n$-Leibniz algebras
largely underlie the structural cohomological properties of the FAs.
\medskip

Similar considerations apply to the $n$-Leibniz algebra
cohomology adapted to the deformation problem already considered
for the FAs. To define the appropriate $n$-Leibniz algebra
cohomology complex it is sufficient to take the $p$-cochains as
$\mathfrak{L}$-valued elements in
$(\otimes^{n-1}\mathfrak{L}^*)\mathop{\cdots}\limits^p(\otimes^{n-1}\mathfrak{L}^*)\otimes\mathfrak{L}^*$,
{\it i.e.}, as elements of
$\hbox{Hom}(\otimes^{p(n-1)+1}\mathfrak{L}, \mathfrak{L})$.
Indeed, for the $n=2$ case, and reverting to the notation where $p$ indicates
the number of algebra elements on which $\alpha^p$ takes arguments,
$\alpha^p\in C^p(\mathscr{L},\mathscr{L})=\hbox{Hom}(\otimes^p \mathscr{L},\mathscr{L})$,
Def.~\ref{deform-defb} leads to
\begin{equation}
\label{deform-defb-n2}
\begin{aligned}
(s \alpha^p) &
(X_1,\dots,X_p,X_{p+1})= \\
&\sum_{1\leq j<k}^{p+1} (-1)^j
\alpha^p(X_1,\dots,\widehat{X_j},\dots,X_{k-1},
[X_j, X_k], X_{k+1},\dots,X_{p+1})\\
 + &\sum_{j=1}^{p} (-1)^{j+1} X_j \cdot
 \alpha^p(X_1,\dots,\widehat{X_j}\dots,X_{p+1}) \\
 +&(-1)^{p+1} \alpha^p(X_1,\dots, X_p) \cdot X_{p+1} \quad ,
\end{aligned}
\end{equation}
which coincides with the coboundary operator for the Leibniz algebra
cohomology complex ($C^\bullet(\mathscr{L},\mathscr{L}),s$) of
eq.~\eqref{Leib-s} (for Leibniz algebra homology, see
\cite{Pir:94}).

Alternatively, the FA cohomology complex would
follow from that for the $n$-Leibniz algebras
by demanding that the $n$-Leibniz bracket be
skewsymmetric so that it becomes the $n$-bracket of a FA.
\medskip

We conclude with a comment on deformations. The proof of the
Whitehead Lemma for FAs \cite{Az-Iz:09}, Ths.~\ref{th:s-s-FA-ext}
and \ref{th:rigFAsim}, relies  on the skewsymmetry of the
$n$-bracket of FAs and it will not hold when  the full antisymmetry
is relaxed. Thus, one might expect having a richer deformation
structure  for $n$-Leibniz algebras and even for Leibniz-type
deformations of FAs viewed as Leibniz algebras. This has been
observed already for the $n=2$ case \cite{Fia-Man:08} by looking at
Leibniz deformations of a Lie algebra and, further, a specific Leibniz
deformation of the euclidean 3-Lie algebra has been found
\cite{JMF:08}. Thus, the next natural step is to look {\it e.g.} at
$n$-Leibniz deformations of simple $n$-Lie algebras to see whether
this opens more possibilities. It has been shown \cite{Az-Iz:10} that
for $n$-Leibniz deformations with brackets that keep the
antisymmetry in their first $n-1$ arguments and thus have
antisymmetric fundamental objects, rigidity still holds for any
$n>3$.
\medskip

\section{$n$-ary Poisson structures}
\label{sec:higherPoisson}{\ }

 In the previous sections we have  mainly reviewed
 two possible ways of generalizing the Lie algebra structure to
algebras endowed with brackets with more than two entries, and
studied various aspects of the two generalizations. It is not
surprising that, as far as the Poisson bracket (PB) and the standard
Poisson structure (PS, see {\it e.g.} \cite{Wei:83} for a classic paper
and \cite{Vai:94} for a review) share the basic properties of Lie
algebras, similar $n$-ary generalizations of the ordinary PS should
also exist. We devote this section to two fully antisymmetric
extensions of the standard PS for brackets with more than two entries.

The first generalization of the PS to $n>2$
\cite{AzPePB:96a,AzPePB:96b}, termed {\it generalized Poisson
structure} (GPS), is naturally defined for brackets with an even
number of entries and parallels the properties of higher order
generalized Lie algebras, GLA, of Sec.~\ref{sec:GLA}. Thus, its
characteristic identity is a higher order Poisson bracket
version of the GJI satisfied by these higher order Lie algebras.
The second one, but earlier in time, is the Nambu-Poisson (N-P) structure
\cite{Nambu:73, Sah-Val:92, Sa-Va:93, Tak:93} (sometimes called
$n$-Poisson, see \cite{Mar-Vi-Vin:97}); its characteristic identity is the
N-P bracket Filippov identity satisfied by the Filippov algebras. We
review below both generalizations; a comparison between both
structures, GPS and N-P, may be found in \cite{AIPB:97} and in the
table in Sec.~\ref{GPSvsN-P}. Further discussion on GP and N-P
structures and related topics may be found {\it e.g.} in
\cite{Mic-Vin:96,Mar-Vi-Vin:97,Iba.Leo.Mar.Die:97,Iba.Leo.Mar:97,Nak:98,Vai:99,Mi-Va:00}
and references therein; references on more physical aspects are
given in Sec. \ref{GPSvsN-P} below. It is worth mentioning that it
is possible to drop the full antisymmetry in the definition of the N-P
brackets, in analogy with the situation of $n$-Leibniz algebras with
respect to FAs; this leads to the notion of Nambu-Leibniz brackets
and Nambu-Leibniz (or Nambu-Loday) algebras  \cite{Grab-Mar:01}, but
these will not be discussed here.

Besides having a fully antisymmetric $n$-ary bracket satisfing the
corresponding characteristic identity, both $n>2$ extensions of the
standard Poisson structure have an additional property that has not
been required for GLAs and $n$-Lie algebras: they satisfy
Leibniz's rule\footnote{Since $[XY,Z]=X[Y,Z]+[X,Z]Y$ within the
enveloping algebra $\mathcal{U}(\fg)$, one may apply similar
factorizations for the GLA $\mathcal{G}$ multibracket resolutions in
terms of two-brackets. Thus, since the composition of elements
in the multibrackets \eqref{multic} is associative but not commutative,
these GLAs do not satisfy a property similar to b) in
eq.~\eqref{GPSfc} as it is readily seen from the four-bracket
resolution in eq.~\eqref{n=4-resol}. In contrast, the infinite-dimensional
GLAs given by the GPS in Sec.~\ref{sec:GPS} incorporate Leibniz's
rule in their definition. Similarly, the FAs defined by the Nambu
bracket do satisfy Leibniz's rule; other FAs, as Ex.~\ref{ex:FA-tr},
do not. Leibniz's rule is a strong condition: other ternary -say-
structures, such as those given in terms of the associator
in Sec.~\ref{asso-FA}, do not respect it either as already pointed
out by Nambu \cite{Nambu:73}.}. The skewsymmetry of the $n$-ary
Poisson bracket and the requirement of the Leibniz's rule are
tantamount to saying that the two corresponding $n$-ary PS
generalizations may be defined through
a skewsymmetric multivector field. {\ }

\subsection{Standard Poisson structures: a short summary}
\label{def:PS} {\ }

\begin{definition}
\label{def:PSabc} ({\it Poisson structure}) (PS) {\ }

Let $M$ be a manifold and $\mathscr{F}(M)$ the ring of smooth
functions on $M$. A Poisson bracket on $\mathscr{F}(M)$ is a
bilinear mapping $\{\cdot,\cdot\} :\mathscr{F}(M)\times
\mathscr{F}(M)\rightarrow \mathscr{F}(M)$ such that, for any three
functions $f,g,h\in \mathscr{F}(M)$, satisfies
\begin{description}
\item[a)] Skewsymmetry
\begin{equation}
\{ f,g\}=-\{ g,f\}\quad,
\end{equation}
\item[b)] Leibniz's rule,
\begin{equation}
\{ f,gh\}=g\{ f,h\}+\{ f,g\}h\quad,
\end{equation}
\item[c)] Jacobi identity
\begin{equation}
\frac{1}{2}\hbox{Alt} \{ f,\{ g, h\}\}= \{ f,\{ g, h\}\} + \{ g,\{
h, f\}\} + \{ h,\{ f, g\}\} =0
 \quad.
\label{PSji}
\end{equation}
\end{description}
\end{definition}

The JI eq.~\eqref{PSji} may be read as stating that the full
skewsymmetrization of the double PB is zero or as an expression of
the derivation property,
\begin{equation}
\label{HPBder}
\{H,\{g,h\}\}=\{\{H,g\},h\}+\{g,\{H,h\}\}\quad ,
\end{equation}
which is behind the evolution of a mechanical system with
hamiltonian $H$.

   Let  $\{ x^i\}$ be local coordinates on the manifold $M$ and $f,g\in
\mathscr{F}(M)$. Conditions a), b) and c) mean that it is possible
to characterize a specific Poisson structure by defining the PB in
terms of some specific functions $\omega^{ij}(x)$. This is done by
defining the PB by
\begin{equation}
    \{ f(x),g(x)\}=\omega^{ij}\partial_i f\partial_j g  \; ,
\end{equation}
which clearly guarantees that Leibniz's rule b) is satisfied, with
$\omega^{ij}(x)$ subject to the conditions
\begin{equation}
\label{GPSb}
\quad \omega^{ij}=-\omega^{ji} \quad,\quad
   \omega^{jk}\partial_k \omega^{lm}+\omega^{lk}\partial_k \omega^{mj}
   +\omega^{mk}\partial_k \omega^{jl}=0\quad ,
\end{equation}
which take care of a) and c). Since the JI is fully antisymmetric in
$f,g,h$, the {\it differential condition} in eq.~\eqref{GPSb} does
not contain any second order derivatives.

   A PB on $\mathscr{F}(M)$ defines a Poisson structure (PS) on $M$,
usually denoted by the pair ($M,\omega$). It is possible to define a
PS by means of a bivector field or {\it Poisson bivector}
\begin{equation}
\label{GPSc}
    \Lambda=\frac{1}{2}\omega^{jk}\partial_j\wedge\partial_k \quad ,
\end{equation}
which takes care of a) and b) above and where
$\omega^{ij}(x)$ satisfies eq.~\eqref{GPSb} so that the PB given by
\begin{equation}
\label{GPSd}
   \{ f,g\}=\Lambda(df,dg) \quad .
\end{equation}
does satisfy the JI.

  Alternatively, a two-vector $\Lambda$ defines \cite{Lich:77} a
PS {\it i.e.}, it is a Poisson bivector, if it has a vanishing
Schouten-Nijenhuis bracket\footnote{The definition and properties of
the Schouten-Nijenhuis bracket are given in  Appendix 2.}
\cite{Sch:40,Nij:55} with itself,
\begin{equation}
\label{GPSe}
[\Lambda,\Lambda]=0 \quad ,
\end{equation}
since this condition reproduces \eqref{GPSb}. Since a Poisson
structure on a manifold $M$ is defined by a Poisson bivector
$\Lambda$, it is denoted ($\Lambda,M$).
\medskip

\begin{definition} ({\it Linear Lie-Poisson structure}) {\ }
\label{Lie-Poisson}

 When the manifold $M$ is the vector space dual to that
of a finite Lie algebra $\fg$, there always exists a PS. It is
obtained by defining the fundamental Poisson bracket $\{ x_i,x_j \}$
(where $\{ x_i\}$ are coordinates on $\fg^*$). Since $\fg\sim
(\fg^*)^*$, we may think of $\fg$ as a subspace of the ring of
smooth functions $\mathscr{F}(\fg^*)$. Then, the Lie algebra
commutation relations
\begin{equation}
\label{GPSfb}
\{ x_i,x_j\}=C^k_{ij}x_k
\end{equation}
define, by assuming b) above, a mapping $\mathscr{F}(\fg^*)\times
\mathscr{F}(\fg^*) \rightarrow \mathscr{F}(\fg^*)$ associated with
the two-vector
\begin{equation}
\label{L-P-bivec}
 \Lambda=\frac{1}{2}C^k_{ij}x_k
\frac{\partial}{\partial x_i}\wedge\frac{\partial}{\partial x_i}.
\end{equation}
This bivector defines a PS since condition (\ref{GPSb}) (or
(\ref{GPSe})) is simply the JI for the structure constants of $\fg$.
The resulting PS is called a {\it Lie-Poisson} structure.
\end{definition}

\medskip
\begin{definition}
\label{compatible} (\emph{Compatible Poisson structures}) {\ }

Two Poisson bivectors $\Lambda, \Lambda'$ are called
 {\it compatible} if the SNB among themselves is zero,
\begin{equation}
[\Lambda, \Lambda']=0\quad. \label{Iii}
\end{equation}
\end{definition}
\noindent
The compatibility condition is equivalent to requiring that any
linear combination $\lambda\Lambda +\mu \Lambda'$ of two
Poisson bivectors be a Poisson bivector.

\begin{lemma} ({\it Poisson cohomology}
\label{Poisson-coho}{\ }

A Poisson bivector $\Lambda$ defines a coboundary operator
$\delta_\Lambda$ acting on multivectors $A$ by $\delta_{\Lambda}: A
\mapsto [\Lambda\,,\,A]$, where the bracket is the SNB. The
resulting cohomology is called Poisson cohomology \cite{Lich:77}
(see also \cite{Koszul:85}).
\begin{proof}
It is sufficent to notice that the operator $\delta_{\Lambda}$ is
nilpotent since, on account of eq.~\eqref{GPSe} and eq.~\eqref{IIvi}
in the Appendix 2, $\delta_{\Lambda}^2 A =
[\Lambda\,,\,[\Lambda\,,\,A]]=0$ on any $A$ (see further \cite{Koszul:85}).
\end{proof}
\end{lemma}

To conclude this section, we would like to recall that the
Lie (Sec.~\ref{sec:Lie-review}), Loday/Leibniz (Sec.~\ref{sec:Loday})
and Poisson brackets are just important examples of a larger variety
of brackets with two entries. For an analysis of various related
`two'-structures see {\it e.g.} \cite{Kos:96} and references therein.

We now move to discuss two specific generalizations of the PS, the
{\it generalized Poisson structures} \cite{AzPePB:96a,AzPePB:96b}
and the {\it Nambu-Poisson structures} \cite{Tak:93}.

\subsection{Generalized Poisson structures}
\label{sec:GPS} {\ }

Since eq.~\eqref{GPSe} for the ordinary PS is simply the JI for a
Lie algebra $\fg$, it is natural to introduce
\cite{AzPePB:96a,AzPePB:96b} in the even case a generalization of
the PS by means of the following

\begin{definition} ({\it Generalized Poisson Structure} (GPS))
\label{def:GPS}{\ }

Let $n$ be even. A GPS is defined \cite{AzPePB:96a,AzPePB:96b} by an $n$-linear mapping $\{\;,\dots,\;\} :\mathscr{F}(M)\times
\mathop{\cdots}\limits^n \times\mathscr{F}(M) \rightarrow
\mathscr{F}(M)$, the {\it generalized Poisson bracket} (GPS) of
functions on the manifold $M$, satisfying the properties
\begin{equation}
\label{GPSfc}
\begin{array}{l}
\mbox{a)}\ \{ f_1,\dots,f_i,\dots,f_j,\dots, f_n\}=-\{
f_1,\dots,f_j,\dots,f_i, \dots, f_n\}\quad
\mbox{(skewsymmetry)}\; ,
\\
\mbox{b)}\ \{f_1,\dots,f_{n-1},gh\}=g\{f_1,\dots,f_{n-1},h\}+
\{f_1,\dots,f_{n-1},g\}h\quad \mbox{(Leibniz's rule)}\; .
\end{array}
\end{equation}
plus the GJI,
\begin{equation}
\label{GJIforGPB}
\begin{aligned}
  \; \hbox{c)} \; \hbox{Alt}\{ f_1,\dots , f_{2s-1}& ,
\{ f_{2s},\dots,f_{4s-1}\} \}\equiv\\
  & \sum_{\sigma\in
S_{4s-1}}(-1)^{\pi(\sigma)}\{ f_{\sigma(1)},\dots,
    f_{\sigma(2s-1)},\{
f_{\sigma(2s)},\dots,f_{\sigma(4s-1)}\} \}
     =0 \quad,
\end{aligned}
\end{equation}
which obviously corresponds to eq.~\eqref{GJI}. In fact, the
GPS are an example of infinite-dimensional higher order
algebras.
\end{definition}

The geometrical nature of the above definition is exhibited
by the following

\begin{lemma}({\it GPS multivectors} or {\it GPS tensors}){\ }

An $n=2s$ even multivector $\Lambda_{(2s)}$,
\begin{equation}
 \label{GPSf}
     \Lambda_{(2s)}=\frac{1}{(2s)!}\omega_{i_1\dots i_{2s}}\partial^{i_1}\wedge\dots\wedge
   \partial^{i_{2s}} \; ,
\end{equation}
defines a GPS $\;(\Lambda_{(2s)},M)$ {\it i.e.}, is a GPS tensor,
iff it has zero SN bracket with itself,
\begin{equation}
\label{GPSk}
[\Lambda_{(2s)},\Lambda_{(2s)}]=0    \quad ,
\end{equation}
where $[\ ,\ ]$ again denotes the SNB. The generalized Poisson
bracket (GPB) is then given by \cite{AzPePB:96a,AzPePB:96b}
\begin{equation}
\label{GPSnbra}
\Lambda_{(2s)}(df_1,\dots,df_{2s})=\{f_1,\dots,f_{2s}\}=
\omega_{i_1\dots,i_{2s}}\partial^{i_1}f_1\dots \partial^{i_{2s}}f_{2s}\quad.
\end{equation}

\begin{proof}
First we notice that the Schouten-Nijenhuis bracket (Appendix 2)
of an odd $n$-multivector with itself $[\Lambda_{(n)},\Lambda_{(n)}]$
vanishes identically and hence the above condition would be empty for $n$ odd.

The fact that $\Lambda_{(2s)}$ is a multivector field automatically guarantees the
skewsymmetry of the GPB and the Leibniz rule: the tensorial character
of $\Lambda$ implies that it is linear on functions and then
$\Lambda_{(2s)}(d(fg),df_2,\dots,df_{2s})=g\Lambda_{(2s)}(df,df_2,\dots,df_{2s})+
f\Lambda_{(2s)}(dg,df_2,\dots,df_{2s})$. Thus, to see that $\Lambda_{(2s)}$ defines
a GPS is sufficient to check that the GJI is satisfied.
Written in terms of the coordinate functions of $\Lambda_{(2s)}$, eq.~(\ref{GPSk})
gives the {\it differential condition} for the coordinates of $\Lambda_{(2s)}$,
\begin{equation}
\label{GPSl}
\epsilon^{j_1\dots j_{4s-1}}_{i_1\dots i_{4s-1}}
\omega_{j_1\dots j_{2s-1}\sigma}\partial^\sigma \omega_{j_{2s}\dots j_{4s-1}}=0 \quad \hbox{or}
\quad \omega_{\sigma [j_1\dots j_{2s-1}}\partial^\sigma \omega_{j_{2s}\dots j_{4s-1}]}=0 \quad ,
\end{equation}
(clearly, for $s$=1 we are left with $\omega_{\sigma
[j_1}\partial^\sigma\omega_{j_2j_3]}=0$, eq.~\eqref{GPSb}).

Now, the GJI eq.~\eqref{GJIforGPB} implies, using eq.~\eqref{GPSnbra},
\begin{eqnarray}
 \label{GPSn}
& &\epsilon^{j_1\dots j_{4s-1}}_{i_1\dots i_{4s-1}} \{
f_{j_1},\dots,f_{j_{2s-1}},\omega_{l_{2s}\dots l_{4s-1}}
\partial^{l_{2s}} f_{j_{2s}}\dots \partial^{l_{4s-1}}f_{j_{4s-1}}\}
\nonumber\\
& &\quad\quad\quad = \epsilon^{j_1\dots j_{4s-1}}_{i_1\dots
i_{4s-1}} \omega_{l_1\dots l_{2s-1}\sigma}
\partial^{l_1}f_{j_1}\dots\partial^{l_{2s-1}}f_{j_{2s-1}}
(\partial^\sigma \omega_{l_{2s}\dots
l_{4s-1}}\partial^{l_{2s}}f_{j_{2s}}\dots
\partial^{l_{4s-1}}f_{j_{4s-1}}
\nonumber\\ & &\quad\quad\quad + 2s\,\omega_{l_{2s}\dots
l_{4s-1}}\partial^\sigma\partial^{l_{2s}}f_{j_{2s}}
\partial^{l_{2s+1}}f_{j_{2s+1}}\dots\partial^{l_{4s-1}}f_{j_{4s-1}})=0\quad ,
\end{eqnarray}
where the last summand groups $2s$ terms terms that become equal
after suitable index relabelling. The second term in the $r.h.s$
vanishes because the part multiplying
$\partial^\sigma\partial^{l_{2s}}f_{j_{2s}}$ is antisymmetric with
respect to the interchange $\sigma\leftrightarrow l_{2s}$. Thus,
the second order derivatives disappear automatically and we are left with
(\ref{GPSl}) because the functions $f_{j_1},\dots ,f_{j_{4s-1}}$
are arbitrary.
\end{proof}
\end{lemma}

\begin{lemma}
({\it Generalized Poisson cohomology}\cite{AzPePB:96b}){\ }

A a Poisson $2s$-vector $\Lambda_{(2s)}$ defines a generalized Poisson
cohomology, with coboundary operator acting on multivectors
$\delta_{\Lambda_{(2s)}}:\wedge^q(M)\rightarrow
\wedge^{q+(2s-1)}(M)\,$ by $\;\delta_{\Lambda_{(2s)}}:A \mapsto
[\Lambda_{(2s)}\,,\,A]\,$.
\begin{proof}
It suffices to extend Lemma ~\ref{Poisson-coho} to the GPS.
\end{proof}
\end{lemma}

\begin{remark}({\it Odd GPS})
\label{GPSodd}
{\ }

  GP structures, like GLAs, are naturally defined for $n$ even. Nevertheless, it
is also possible to define a GPS in the {\it odd} case
\cite{AIPB:97}. For a GPB with an odd number $n$ of arguments, the
second summand in the $r.h.s.$ of (\ref{GPSn}) does not vanish,
giving rise to another, algebraic condition that has to be satisfied
by the coordinates of the GPS multivector. An {\it algebraic
condition} will always be present for the N-P structures
\cite{Tak:93} below, see Lemma~\ref{N-P-cond-d-a}).

 The above constructions and GPS may also be extended to the
$\mathbb{Z}_2$-graded (`supersymmetric') case \cite{AzIzPePB:96},
including the definition of the {\it graded GPS tensors}, that may
be introduced through a suitable graded SN bracket
\cite{AzIzPePB:96} (see \cite{Sor-Sor:08} for another
construction).

The generalized Poisson cohomology and homology was further
considered in \cite{AIPB:97} (see also
\cite{Iba.Leo.Mar:97}).
\end{remark}

\subsection{Higher order or generalized linear Poisson structures}
\label{sec:GPS-linear}{\ }

It is easy to construct examples \cite{AzPePB:96b} of GPS (infinitely
many, in fact) in the linear case by extending the construction in
Def.~\ref{Lie-Poisson}. They are obtained by applying the earlier
construction of GLAs to the GPS. Let $\mathcal{G}$ be a simple Lie
algebra of rank $l$. We know from Table 1 and Theorem
\ref{th:HOsimple} that associated with $\mathcal{G}$ there are
$(l-1)$ simple higher order Lie algebras. Their structure constants
define a GPB $\{\cdot,\mathop{\cdots}\limits^{2m_l-2},\cdot\}:
\mathcal{G}^*\times \mathop{\cdots}\limits^{2m_l-2}
\times\mathcal{G}^*\rightarrow \mathcal{G}^*$ by
\begin{equation}
       \{ x_{i_1},\dots, x_{i_{2m_l-2}}\}={\Omega_{i_1\dots
      i_{2m_l-2}}}^\sigma_\cdot x_\sigma\quad ,    \label{GPSo}
\end{equation}
where $\Omega$ is the $(2m_l-1)$-cocycle for the Lie algebra
cohomology of $\mathcal{G}$ associated with an invariant polynomial
of rank $m_l$.

If one now computes the GJI (\ref{GPSl}) for
$\omega_{i_1\dots i_{2m_l-2}}= {\Omega_{i_1\dots
i_{2m_l-2}}}^\sigma_\cdot x_\sigma$, or, alternatively, the SNB
$[\Lambda,\Lambda]$ for the ($2m-1$)-vector
\begin{equation}
\label{lambdalin}
 \Lambda=\frac{1}{(2m-2)!}{\Omega_{i_1\dots
i_{2m-2}}}^\sigma_\cdot x_\sigma
\partial^{i_1}\wedge\dots \wedge \partial^{i_{2m-2}}\quad,
\end{equation}
one finds that $[\Lambda,\Lambda]=0$ due to the GJI for the higher
order structure constants $C=\Omega$ given in \eqref{multig} and
$\Lambda$ above is called a {\it linear GPS tensor}. This means that
all these simple higher order Lie algebras associated with a simple $\fg$ define
linear GP structures (in fact, and more generally, it is also true
that the SNB of two ($2m$-2)- and ($2m'$-2)-multivectors $\Lambda$, $\Lambda'$
constructed from two ($2m$-1)-, ($2m'$-1)-cocycles $\Omega,\Omega'$
as in eq.~\eqref{lambdalin} is zero, $[\Lambda,\Lambda']=0$:
\cite{AzPePB:96b}, Th.~8.3). Conversely, given a linear
GPS with fundamental GPB (\ref{GPSo}), there is a higher order
Lie algebra $\mathcal{G}$ associated with it.
\medskip

It is clear from the discussion in \cite{AzPePB:96a,AzPePB:96b,AzBu:96},
however, that not only the basic, primitive invariants explicitly considered
in Sec.~\ref{sec:sec8} (and that led to the simple GLAs classification of
Th.~\ref{th:HOsimple}) can be used to generate linear GPS (and GLAs). In
fact, it is easy to check (using more general odd antisymmetric invariants,
the constructions in \cite{AzPePB:96a,AzPePB:96b,AzBu:96}
and the properties in eqs.~\eqref{IIii},\eqref{IIiia} of the SN bracket)
that these odd polynomials also determine GLAs and linear GPS.
This has also been pointed out independently in \cite{Pinc-Ushi:05}
using a different approach based on general super-Poisson
methods.
\medskip

We conclude this section by noting that the Jacobi structures
considered by Lichnerowicz, which are defined using Jacobi brackets
(a generalization of the standard Poisson bracket in which Leibniz's
rule is replaced by a weaker condition), may be similarly extended to the
$n$-ary case leading to {\it generalized Jacobi structures}
\cite{Pe-Bu:97}.

Let us turn now to the generalization of the PS
that has a FA structure.

\subsection{Nambu-Poisson (N-P) structures}
\label{sec:NP}{\ }

As we saw in Sec.~\ref{sec:Nambu-FA}, Nambu considered
\cite{Nambu:73} already in 1973 the possibility of extending the
Poisson brackets of standard Hamiltonian mechanics to brackets of
three functions defined by the Nambu Jacobian\footnote{When the
quark model of Gell-Mann and Zweig was proposed (1964), the Fermi
statistics for quarks did not mix well with the expected symmetry
properties of the ground state wavefunction of the $\Delta^{++}$
particle; alternatives were searched for at the time, as
Greenberg paraquarks (1964) and Han-Nambu quarks (1965). As is well
known, the final answer -QCD- was proposed by Gell-Mann, Fritzsch
and Leutwyler in 1971-73; in 1973, Nambu and Han were also analyzing
the various existing proposals \cite{Nam-Han:73}. It seems that
exploring new dynamical alternatives in this context was one of the
motivations underlying Nambu mechanics \cite{Nambu:73}.}. Clearly,
the Nambu bracket may be generalized further to a Nambu-Poisson
(N-P) one allowing for an arbitrary number of entries.

 The first two properties of the N-P bracket, skewsymmetry and
Leibniz's rule, are shared with the standard PS and the GPS, and are
again automatically guaranteed if the new bracket is defined in
local coordinates $\{ x_i\}$ on $M$ in terms of an $n$-vector
\begin{equation}
 \label{GPSfbis}
\Lambda_{(n)}=\frac{1}{n!}\eta_{i_1\dots i_n}(x)\,\partial^{i_1}\wedge\dots\wedge
\partial^{i_n}       \quad ,
\end{equation}
so that, as in (\ref{GPSnbra}), the N-P bracket is given by
\begin{equation}
 \label{GPSg}
 \{f_1,\dots,f_n\}=\Lambda_{(n)}(df_1,\dots,df_n)\quad .
\end{equation}

The key difference is the condition that generalizes the JI and that
completes the definition of the N-P structures; this restricts the
allowed $\eta_{i_1\dots i_n}(x)$ above so that $\Lambda_{(n)}$ is a
N-P tensor (Lemma \ref{N-P-cond-d-a} below). The identity behind the
$n=3$ generalized mechanics, which Nambu did not write in his paper
\cite{Nambu:73}, is the `five-point identity' of Sahoo and
Valsakumar \cite{Sah-Val:92, Sa-Va:93}. They introduced this
relation (which is the FI of the Nambu FA $\fN$) as a necessary
consistency requirement for the time evolution of Nambu mechanics.
In the general case of N-P brackets with $n$ entries the
corresponding Filippov `($2n-1$)-point identity' was written by
Takhtajan \cite{Tak:93}, who studied it in detail and called it the
{\it fundamental identity}. This is simply the FI of the
infinite-dimensional $n$-Lie algebra $\fN$ in
Sec.~\ref{sec:Nambu-FA}, already anticipated by Filippov for the
$n$-bracket of functions defined by the Jacobian \cite{Filippov,
Fil:98}. Thus, a general N-P structure is given by the following

\begin{definition}({\it Nambu-Poisson structures})\cite{Tak:93}
{\ }

A N-P structure is defined  by a N-P bracket with $n$-entries that
satisfies the conditions of eq.~\eqref{GPSfc} plus the FI
\begin{eqnarray}
\label{FIP}
\hbox{c)}&  &  \{ f_1,\dots,f_{n-1},\{ g_1,\dots,g_n\}\}=\{\{
f_1,\dots,f_{n-1},g_1\},
     g_2,\dots,g_n\}\nonumber\\
& &\quad +\{g_1,\{ f_1,\dots,f_{n-1},g_2\},g_3,\dots,g_n\}+\dots+
   \{ g_1\dots,g_{n-1},\{ f_1,\dots,f_{n-1},g_n\}\} \quad ,
\label{GPSh}
\end{eqnarray}
for the Filippov algebra $\fN$ defined by the N-P $n$-bracket. To
each N-P bracket satisfying the conditions of the definition
corresponds a N-P $n$-tensor such that eq.~\eqref{GPSg} is
satisfied. The manifold $M$ is called a {\it Nambu-Poisson manifold}
($\Lambda_{(n)}, M$).
\end{definition}

N-P structures \cite{Tak:93} have been discussed in many papers; see
 \cite{Gau:96} and {\it e.g.} \cite{Gau:96b, Mic-Vin:96,Da-Tak:97,AIPB:97,
Mar-Vi-Vin:97,Iba.Leo.Mar.Die:97,Iba.Leo.Mar:97,Nak:98,Vai:99,Mi-Va:00,Gra-Mar:00,Rot:05b}
for further information.
\medskip

  The N-P $n$-bracket, eq.~\eqref{GPSg}, establishes a
linear correspondence between ($n-1$)-forms and vector fields,
\begin{equation*}
df_1 \wedge\dots\wedge df_{n-1} \mapsto X_{f_1,\dots,f_{n-1}} \quad ,
\end{equation*}
defined by
\begin{equation*}
dg(X_{f_1,\dots,f_{n-1}})=X_{f_1,\dots,f_{n-1}}\,.\,g
=\Lambda_{(n)}(df_1,\dots,df_{n-1},dg)=\{f_1,\dots,f_{n-1},g\}\quad.
\end{equation*}
The hamiltonian vector field $X_{f_1,\dots,f_{n-1}}$
(see Ex.~\ref{NPcomp}) is indeed a vector field as a consequence of
the tensorial character of $\Lambda_{(n)}$. Its physical
significance is clear: in N-P mechanics, the time evolution of
a dynamical magnitude $F$ is determined through ($n-1$)
`Hamiltonians' $H_1,\dots, H_{n-1}$ by
\begin{equation}
\label{GPSi}
{\dot F}:= X_{H_1,\dots,H_{n-1}}\,.\,F=
\{ H_1,\dots, H_{n-1},F\}\quad  \forall F \quad.
\end{equation}
This expression is the evident dynamical counterpart
(eq.~\eqref{n-ad}) of the $ad_\mathscr{X}$ derivation
of the $n$-FAs: due to the FI, $X_{H_1,\dots,H_{n-1}}$
is a derivation of the N-P bracket: the
fundamental objects of $\fN$ define inner derivations
$ad_{H_1\dots H_{n-1}}$ of the N-P algebra $\fN$ since
the FI implies \cite{Tak:93}
\begin{equation}
\label{moreFI}
   \frac{d}{dt}\{ f_1,\dots,f_n\}=\{ {\dot f}_1,\dots,f_n\}+
   \{f, {\dot f}_2,\dots,f_n\} \dots +
  \{ f_1,\dots,{\dot f}_n\}\quad .
\end{equation}
Thus, the FI guarantees that the bracket of any $n$ constants
of the motion is itself a constant of the motion.

It follows that a N-P structure is defined by a multivector $\Lambda_{(n)}$ such that
the hamiltonian vector fields are a derivation of the N-P
$n$-bracket algebra. It is also easy to see that such a N-P tensor is
invariant under the action of a hamiltonian vector field,
$L_X \Lambda_{(n)}=0$. This follows from the FI since
\begin{equation*}
\begin{aligned}
X_{f_1,\dots,f_{n-1}}\,.\,\{g_1,\dots,g_n\}
=&L_{X_{f_1,\dots,f_{n-1}}}(\Lambda_{(n)}(dg_1,\dots,dg_n))\\
=&(L_{X_{f_1,\dots,f_{n-1}}}\Lambda_{(n)})(dg_1,\dots,dg_n) +
\sum_{i=1}^{n}\Lambda_{(n)}(dg_1,\dots,L_{X_{f_1,\dots,f_{n-1}}}dg_i,\dots,dg_n)\\
=&(L_{X_{f_1,\dots,f_{n-1}}}\Lambda_{(n)})(dg_1,\dots,dg_n)+
\sum_{i=1}^n \{g_1,\dots,\{f_,\dots,f_{n-1},g_i\},\dots,g_n\} \quad ,
\end{aligned}
\end{equation*}
using in the second line that the Lie derivative commutes with
contractions and in the third that $L_{X_{f_1,\dots,f_{n-1}}}\,dg=d
L_{X_{f_1,\dots,f_{n-1}}}\,g= d\{f_1,\dots,f_{n-1},g\}$. Thus, we
have the following

\begin{proposition}
\label{Ham-inv} {\ }

 The hamiltonian vector fields determine
infinitesimal automorphisms of the N-P structure
($\Lambda_{(n)},M$).
\end{proposition}

\begin{lemma}({\it Conditions that define a N-P tensor})\cite{Tak:93}
\label{N-P-cond-d-a} {\ }

Eqs.~\eqref{GPSfbis}, \eqref{GPSg} in eq.~\eqref{GPSh}
determine\footnote{There is a trivial misprint (a term obviously
missing) in eq. (5) of \cite{Tak:93}, accounted for in
\cite{Cha-Tak:95}, eq. (10).} two conditions  that the local
coordinates $\eta_{i_1\dots i_n}$ of an $n$-vector $\Lambda_{(n)}$
have to satisfy to define a N-P structure {\it i.e.}, to be a N-P
tensor. The first one is the {\it differential condition},
\begin{equation}
\label{difcond} \eta_{i_1\dots
i_{n-1}\rho}\partial^\rho\eta_{j_1\dots j_n}
-\frac{1}{(n-1)!}\epsilon^{l_1\dots l_n}_{j_1\dots j_n}
(\partial^\rho\eta_{i_1\dots i_{n-1} l_1})\eta_{\rho l_2\dots l_n}
=0 \quad.
\end{equation}
The second is the {\it algebraic condition}. It reads
\begin{equation}
\label{alcond}
\Sigma + P(\Sigma)=0\quad,
\end{equation}
where $\Sigma$ is the ($n+n$)-tensor
\begin{equation}
\label{Sigmadef}
\begin{array}{rl}
\Sigma_{i_1\dots i_n j_1\dots j_n} =&\eta_{i_1\dots
i_{n}}\eta_{j_1\dots j_{n}} -\eta_{i_1\dots i_{n-1}j_1}\eta_{i_n
j_2 \dots j_{n}} -\eta_{i_1\dots i_{n-1}j_2}\eta_{j_1 i_n j_3\dots
j_{n}}
\\[0.2cm]
-& \eta_{i_1\dots i_{n-1}j_3}\eta_{j_1 j_2 i_n j_4 \dots j_{n}}-
\dots -\eta_{i_1\dots i_{n-1}j_n}\eta_{ j_1j_2 \dots j_{n-1} i_n}
\;,
\end{array}
\end{equation}
and $P$ is the permutation operator that interchanges its $i_1$
and $j_1$ indices.
\end{lemma}
\noindent

  Eq. (\ref{Sigmadef}) may we rewritten as
\begin{equation*}
\Sigma_{i_1\dots i_n j_1\dots j_n} =
{1\over n!} \epsilon^{l_1\dots l_{n+1}}_{i_n j_1\dots j_n}
\eta_{i_1\dots i_{n-1} l_1}\eta_{l_2\dots l_{n+1}}\quad .
\label{ntensor}
\end{equation*}
Eq.~\eqref{alcond} follows\footnote{The presence of the algebraic
condition implies that a constant antisymmetric tensor, although it
satisfies automatically eq.~\eqref{difcond}, is not necessarily a
N-P tensor since it still has to satisfy eq~\eqref{alcond}.
Also, the direct sum of N-P tensors is not a N-P tensor since
it is not decomposable (see Lemma~\ref{le:N-P-dec} below).} from
requiring the vanishing of the second derivatives in \eqref{GPSh},
which now do not vanish automatically as for the GPS. Of course, for
a standard ($n=2$) PS eq.~\eqref{difcond} reproduces
eq.~\eqref{GPSb} and  eq.~\eqref{alcond} is absent. Thus, the FI is
much more constraining than the JI and, as a result, N-P $n\geq 3$
structures are more rigid that the standard $n=2$ Poisson ones (see
Lemma \ref{le:N-P-dec} below). Conditions \eqref{difcond} and
\eqref{alcond} play for the Nambu-Poisson structures the r\^ole of
eq.~\eqref{GPSl} for the GPS, which follows from the
geometrical requirement that the $2s$-vector that defines a GPS has
a vanishing SN bracket with itself \cite{AzPePB:96a,AzPePB:96b}.
\medskip

 The algebraic condition imposes severe restrictions on the
potential N-P tensors for $n>2$ (for $n=2$, eq.~\eqref{alcond} is
trivial). It turns out \cite{Ale.Guh:96,Gau:96,Mar-Vi-Vin:97} (see
also \cite{DFST:96, Nak:98, Vai:99}) that (contrarily to initial
expectations \cite{Tak:93} but confirming a later conjecture
\cite{Cha-Tak:95}), condition \eqref{alcond} implies that the
Nambu-Poisson $n$-vector $\Lambda_{(n)}$ in (\ref{GPSfbis}) is {\it
decomposable}. This means that $\Lambda_{(n)}$ can be written as the
exterior product of vector fields, a result that Takhtajan has
traced to follow from the Weiztenb\"ock condition (see
\cite{Wei:23}, p.~116) in the theory of invariants. Specifically,
the N-P tensors satisfy the following

\begin{lemma} ({\it Decomposability of the N-P tensors})
\label{le:N-P-dec}

If $\Lambda_{(n)}$ is a N-P tensor of order $n\geq 3$ on a manifold $M$,
there are local coordinates on $U\subset M$,
$\{x^1,\dots,x^n, x^{n+1},\dots, x^m\}$, such that for $x\in U$,
\begin{equation*}
\Lambda_{(n)} =
\frac{\partial}{\partial x^1}\wedge\frac{\partial}{\partial x^2}
\dots \wedge \frac{\partial}{\partial x^n} \quad,
\end{equation*}
and reciprocally.
\end{lemma}
This is the form of the canonical the N-P multivector in Euclidean
space $\mathbb{R}^n$ that defines the N-P bracket of $n$
$f_i(x_1,\dots,x_n)\in  \mathscr{F}(M)$
\begin{equation}
\label{Nam-n-bra}
\Lambda_{(n)}=\frac{1}{n!}\epsilon_{i_1\dots i_n}^{1\,\dots\,n}
\partial^{i_1}\wedge\dots\wedge\partial^{i_n}\quad,\quad
\{f_1,\dots,f_n\}=\Lambda_{(n)}(df_1,\dots,df_n)=\left|
\frac{\partial{(f_1,\dots,f_n)}}{\partial(x_{i_1},\dots,x_{i_n})}\right|
\;,
\end{equation}
expressed in terms of the Jacobian (Ex.~\ref{ex:Nam-FA}). Further,
any N-P bracket such as the one given in eq.~\eqref{Nam-vol} may be
written in the canonical Nambu Jacobian form above using a suitable
local coordinate system (see \cite{Vai:99}). The net result is that,
as mentioned, $n\geq 3$ N-P structures are extremely rigid; the
canonicity of the Nambu bracket above parallels the uniqueness of
the simple FAs $A_{n+1}$. This is not surprising on account of the
relation among the decomposability of N-P tensors and the Pl\"ucker
conditions \cite{Mi-Va:00}, a connection that was used \cite{Pap:08}
to show the unicity of the $n=3$ $A_4$ FA in the context of the
Bagger-Lambert model of Sec.~\ref{sec:BLG}. In fact, one may say
that the FAs $\fN$ determined by the Nambu $n$-brackets above
constitute the infinite-dimensional counterparts of the simple
$n$-Lie algebras $A_{n+1}$, a point that will find a
physical application in Sec.~\ref{sec:BLG+NB}. It is thus not
surprising that, already in his original paper \cite {Filippov},
Filippov considered the Jacobian FAs (see further \cite{Fil:98}).

\begin{remark} ({\it Subordinated N-P structures})
\label{N-P-subord} {\ }

 Following Prop.~\ref{n-1 from n} it follows that,
given a Nambu $n$-bracket, one may construct another one of order
($n-1$) by fixing one element in the original N-P $n$-bracket {\it
i.e.}, by defining
\begin{equation*}
\{f_1,\dots,f_{n-1}\}=\{f_1,\dots,f_{n-1},g\} \quad (\hbox{$g$
fixed}) \quad ,
\end{equation*}
since the Nambu ($n-1$)-bracket in the $l.h.s$ will fulfil the FI as
in Prop.~\ref{n-1 from n}. As a result, a N-P tensor that defines a
N-P structure ($\Lambda_{(n)},M$) of order $n$ induces a family of
($n-1$) N-P tensors $\Lambda_{(n-1)}$ that define ($n-1$)-N-P
brackets by the relation above so that ($\Lambda_{(n-1)},M$) is a
N-P structure. Thus, as in the finite-dimensional case, if $\fN$ is
an $n$-FA, the above construction defines a $\fN'$ Nambu subordinated
($n-1$)-FA.
\end{remark}

\begin{definition} ({\it Linear Nambu-Poisson structures})\cite{Tak:93}

A N-P tensor whose components are linear in $x_i$ ({\it cf.} Def.~\ref{Lie-Poisson}),
$\eta_{i_1 \dots i_n}(x)= \eta_{i_1\dots i_{i_n}}{}^j \, x_j$,
is called a {\it linear N-P tensor} and defines a {\it linear N-P
structure}. The corresponding bracket is given by
\begin{equation}
\label{lin-N-P-str}
\{x_{i_1},\dots,x_{i_n}\}=\eta_{i_1\dots i_n}{}^j \, x_j \; .
\end{equation}

The linear N-P structures of eq.~\eqref{lin-N-P-str}
play for FAs the r\^ole of the linear (or Lie-) Poisson ones
for Lie algebras. Any linear N-P structure of order $n$
defined by the linear $n$-N-P tensor
$\Lambda_{(n)}= f_{a_1\dots a_{n}}{}^b\,x_b\,\partial_{a_1}\wedge\cdots \wedge \partial_{a_n}\;$
induces an $n$-Lie algebra structure on the dual $(\mathbb{R}^m)^*$.
The converse, however, may not hold, since a linear $n$-vector $\Lambda_{(n)}=
 f_{a_1\dots a_{n}}{}^b\,x_b\,\frac{\partial}{\partial x_{a_1}}\wedge\cdots
\wedge \frac{\partial}{\partial x_{a_n}}\,$, where the $f_{a_1\dots a_{n}}{}^b$
are the structure constants of an $n$-Lie algebra, may give rise to a non-decomposable
tensor. In other words, although the differential condition \eqref{difcond}
for the $n$-tensor coordinates $\eta_{a_1\dots a_n}=f_{a_1\dots a_{n}}{}^b\,x_b$
becomes the FI of the FA (see eq.~\eqref{FIultrashort-b})
and is therefore satisfied, the algebraic condition may not hold,
in which case $\Lambda_{(n)}$ does not define a N-P structure and
therefore is not a N-P tensor.
\end{definition}

  The skewsymmetric tensor $\eta_{i_1 \dots i_n}(x)=
\epsilon_{i_1\dots i_n}{}^{i_{n+1}} x_{i_{n+1}}$ is a linear N-P
tensor \cite{Tak:93,Cha-Tak:95}. Accordingly, it defines the {\it
linear Poisson structure}
\begin{equation*}
\label{NP-A-linear}
 \{x_1,\dots, x_{i_{n+1}}\}= \epsilon_{i_1 \dots
i_n}{}^{i_{n+1}} x_{i_{n+1}} \quad ,
\end{equation*}
which (see Sec.~\ref{sec:simple-n-Lie}) is associated with
$A_{n+1}$. Clearly, and in contrast with the higher order linear
Poisson structures, only for $n=2$ the above N-P structure
corresponds to a linear {\it Lie}-Poisson structure, since it is
only for the standard Poisson case case that $\epsilon_{i j}{}^k$
are the {\it structure constants} of a Lie algebra, $su(2)$ ({\it
cf.} eq.~\eqref{GPSo}). Since the linear Poisson structures of
Def.~\ref{Lie-Poisson} are called Lie-Poisson structures, the linear
N-P structures above might be called as well Filippov-Nambu-Poisson
structures.

Further analysis of linear N-P structures is given in
\cite{Nak:98,Du-Zu:99,Vai:99} and refs. therein.
\medskip

  It is  possible to construct Nambu-Poisson $n$-tensors on
Lie groups $G$ (in fact, left-invariant N-P tensors) by using the
LI vector fields that generate a $(n\geq 3)$-dimensional subalgebra $\fh$ of
the Lie algebra $\fg$ of $G$. This is done by setting
$\Lambda_{(n)}=X_1\wedge\dots\wedge X_n$, where $\{X_i\}$ is a basis of
$\fh$; $\Lambda_{(n)}$ is then a LI N-P $n$-tensor \cite{Nak:98b}.
In fact, there is a one-to-one correspondence between the set of the
LI N-P $n$-tensors (up to a constant) on $G$ and the set of
$n$-dimensional subalgebras $\fh\subset\fg$. We shall not discuss
this further and refer instead to \cite{Nak:98b} for details.
\medskip

  There is one question that remains to be answered: the possible
connection between the {\it even} order GP and N-P structures.
Writing the FI in the form of eq.~\eqref{eq:FI}
\begin{equation*}
\begin{aligned}
\{f_{i_1},\dots,f_{i_{n-1}},  \{f_{i_n},\dots,f_{i_{2n-1}}
 \}\} = & \\
 =
\{\{f_{i_1},\dots,f_{i_{n-1}},f_{i_n}\}, f_{i_{n+1}}, \dots
f_{i_{2n-1}}\}  + \dots + & \{f_{i_n}, \dots,
f_{i_{2n-2}},\{f_{i_1},\dots
f_{i_{n-1}},f_{i_{2n-1}} \}\} = \\
(-1)^{n-1}\{f_{i_{n+1}}\dots,f_{i_{2n-1}},\{f_{i_1},\dots,f_{i_{n-1}},f_{i_n}\}\}
 + (-1)^{n-2} & \{f_{i_n},f_{i_{n+2}},\dots,f_{i_{2n-1}},\{f_{i_1},\dots,f_{i_{n-1}},f_{i_{n+1}}\}\}
\\ +\dots + (-1)^{n-n}
\{f_{i_n},\dots,f_{i_{2n-2}},\{f_{i_1},\dots,f_{i_{n-1}},f_{i_{2n-1}}\}\}
\end{aligned}
\end{equation*}
and contracting the first and last terms of the equality with
$\epsilon^{i_1\dots i_{2n-1}}$ it follows that
\begin{equation*}
\epsilon^{i_1\dots i_{2n-1}}
\{f_{i_1},\dots,f_{i_{n-1}},\{f_{i_n},\dots,f_{i_{2n-1}}\}\}= n
(-1)^{n-1} \epsilon^{i_1\dots i_{2n-1}}
\{f_{i_1},\dots,f_{i_{n-1}},\{f_{i_n},\dots , f_{i_{2n-1}}\}\}
\quad,
\end{equation*}
from which the GJI in eq.~\eqref{GJIforGPB} follows (one may also
look at the form of the FI and the GJI, for instance in
eqs.~\eqref{FIshort-coor}, \eqref{GJIcoord}). Thus, we may state the
(expected) following result

\begin{lemma}
\label{le:NPimpGPS}
{\ }

Every Nambu-Poisson structure of even order is also a generalized
Poisson structure, but the converse does not hold.
\end{lemma}
\noindent In fact, this Lemma also follows trivially from the
decomposability of the N-P tensors, Lemma.~\ref{le:N-P-dec},
which automatically implies zero SN bracket $[\Lambda,\Lambda]=0$.
\medskip

\subsection{Brief remarks on the quantization of higher order
Poisson structures} \label{GPSvsN-P} {\ }

Setting aside the intrinsic mathematical interest of the two $n$-ary
generalizations of the standard Poisson structure above discussed, a
first question is to find examples of physical mechanical systems
that might be described by them. It is fair to say that there are
not too many, particularly for the GPS since the GJI does not
reflect a derivation property of the multibracket. The reader
interested in finding discussions of mechanical systems described by
$n$-ary Poisson structures may look {\it e.g.}, at \cite{Nambu:73,
Hira:77, Sah-Val:92, Tak:93, Cha:95,
Cu-Za:02,Guha:02,Nutku:03,Za-Cu:04,AzPePB:96b} and references
therein. In field theory, an $n=3$ infinite-dimensional $\fN$ algebra
has recently appeared in the context of M-theory, as it will be
described in Sec.~\ref{sec:BLG+NB}.

The antisymmetry of the standard Poisson bracket is shared by the
higher order N-P structures of Sec.~\ref{sec:NP} and the GPS of
Sec.~\ref{sec:GPS}. As pointed out by Nambu himself \cite{Nambu:73},
the antisymmetry property is necessary to have hamiltonians that are
constants of the motion since the time evolution of a dynamical
quantity $F$ is governed by $\dot{F} = \{H_1,H_2,F \}$ or, in
general \cite{Tak:93}, by ($n-1$) hamiltonians, eq.~\eqref{GPSi}.
This is also the case for a mechanical system described by a GPS
\cite{AzPePB:96a, AzPePB:96b}; clearly, such a time evolution
implies that all the hamiltonian functions are constants of the
motion. The derivation property of the N-P bracket encoded in the FI
makes the N-P bracket specially suitable for the differential
equation describing the evolution of a dynamical quantity, while the
lack of this property for the GPS, governed by the GJI, makes less
obvious its application to mechanical systems\footnote{Nevertheless,
although the GPB of constants of the motion is not a constant of the
motion in general, a weaker result exists for any set of functions
$f_1,\dots,f_q$, $q>2s$, such that the functions in
($H_1,\dots,H_{2s-1}, f_1,\dots,f_{2s-1}$) are in {\it involution}
(see \cite{AzPePB:96b}, Th. 6.2). Under the restricting conditions
of this theorem, one also has that the FI \eqref{moreFI} is
satisfied; see \cite{AzPePB:96b} for further discussion.}. Another
aspect of the $n$-ary structures is the Liouville theorem; both N-P
mechanics and the linear GPS structures have an $n$-dimensional
analogue (see \cite{AIPB:97}).

    Let us conclude the discussion of $n$-ary Poisson structures
with a few words on quantization, a word often used too loosely,
at least from a physical point of view. For
the purposes of this review, the case of quantizing standard Poisson
structures may be considered `solved' by using {\it e.g.} the Dirac
approach even in the presence of constraints although, as already
mentioned, more sophisticated approaches to quantization exist.
There was also a `generalized quantum dynamics' \cite{Adl-Wu:94}
which, in principle, may lead directly to the equations of motion of
the operators without going through the quantization of the
classical theory.  When moving to higher order, however, the
difficulties appear already at a basic level. Indeed, the
quantization of the Nambu-Poisson mechanics is fraught with serious
difficulties, especially if the word `quantization' is understood
physically {\it i.e.}, as a general procedure that, starting from the
classical dynamics of a physical system (as described by N-P mechanics),
gives a quantum one that reproduces all the structural properties
of the original classical system in the $\hbar\rightarrow 0$ limit.
This implies, besides the skewsymmetry of the N-P $n$-bracket,
Leibniz's rule and the FI.

There is a simple argument against any
elementary quantization of N-P mechanics which tries to keep the
standard correspondence among dynamical quantities, their associated
quantum operators and all the structural relations satisfied by
their classical counterparts  through the N-P brackets. It is
physically natural to assume that the quantum operators
$\mathscr{O}_i$ corresponding to the different classical dynamical
quantities $O_i$ are associative. Then, it follows that a commutator
$[\mathscr{O}_1,\dots,\mathscr{O}_{2s}]$ defined by the
antisymmetrized sum of their products, as in eq.~\eqref{multic},
will naturally lead to the GJI rather than to the FI. This problem
is already underlined by the difficulty in finding matrix
realizations of FAs (see Lemma.~\ref{le:FA-Cliff} and
Ex.~\ref{ex:FA-tr} for just two examples). The possibility of
considering higher order ({\it e.g.} cubic) matrices has also been
explored \cite{Awa-Li-Mi-Yo:99,Kaw:02,Kaw:03,Ho-Hou-Ma:08} but again
there are difficulties to mimic all properties of Nambu brackets.
Further, for multibrackets of odd order $n$ we saw in
eq.~\eqref{multiep} that the zero in the $r.h.s.$ of the even GJI is
replaced by a larger multibracket
$[\mathscr{O}_1,\dots,\mathscr{O}_{2s-1}]$. In any case,
multibrackets defined by the antisymmetric products of associative
operators $\mathscr{O}$ as in eq.~\eqref{multic} will lead to an
identity which is {\it outside} the N-P algebraic scheme. As a
result, finding a simple procedure of quantizing Nambu-Poisson
mechanics with associative operators that keeps all the three
properties of the N-P structure, and especially the FI \eqref{eq:FI}
that regulates the time evolution of the system, is a problem likely
without solution, although sophisticated quantization methods have
been proposed\footnote{The deformation quantization (star product)
approach to Nambu mechanics was investigated in \cite{DFST:96}. It
was observed there that this approach does not provide a
straightforward solution to the quantization problem of
Nambu-Poisson structures in general. This led the authors to propose
a peculiar modification of the deformation quantization they termed
`Zariski quantization' (see \cite{Minic:99} in connection with
M-theory).

   For the deformation quantization approach and $*$-products,
see \cite{BFFLS:78a,BFFLS:78b,Stern:98} plus the pioneering work by Berezin
\cite{Ber:75, Ber:74}; see further \cite{Sta:97}. Recent work
on quantized Nambu-Poisson structures (in terms of non-commutative geometries)
has been done in \cite{DeBe-Sae-Sza:10}. For a sourcebook
on deformation quantization see \cite{Fai-Cur-Zac:05}. }. This
inherent difficulty for quantizing Nambu-Poisson brackets while
preserving their defining properties at the same time {\it i.e.}, of
having a correspondence between the classical and quantized versions
of the theory, has been pointed out repeatedly (see, {\it e.g.}
\cite{Estab:73,Tak:93, Sa-Va:94, DFST:96, Awa-Li-Mi-Yo:99,
Sakakibara:00, Xiong:00}), starting with Nambu himself
\cite{Nambu:73}. In this respect (but setting aside the question of
the time evolution), the even GPS satisfy an identity, the GJI, more
amenable to a quantum version.

  We collect in Table 2 below the main properties of GP and N-P
structures (see \cite{AIP-B:97})

\begin{table} [h!]
\centering
\begin{tabular}{lccc}
& PS & N-P & GPS (even order)
\\
\hline \hline Characteristic identity (CI): & eq.~\eqref{PSji} (JI)
& eq. (\ref{GPSh}) (FI) & eq. (\ref{GJIforGPB}) (GJI)
\\
Defining conditions: & eq. (\ref{GPSb}) & Eqs.
(\ref{difcond}),(\ref{alcond}) & eq. (\ref{GPSl})
\\
Liouville theorem: & Yes & Yes & Yes
\\
Poisson theorem: & Yes & Yes & No (in general)
\\
\begin{minipage}[b]{5.0cm}
CI realization in terms of associative operators:
\end{minipage}
& Yes & No (in general)& Yes
\end{tabular}
\vspace{8pt}
\\
{Table 2. Some properties of the Poisson, Nambu-Poisson and generalized
Poisson structures.}
\end{table}
{\ }

  The situation may be summarized by saying that the associativity
of the quantum operators, which implies the GJI that is more
suitable for quantization, is not compatible with the derivation
property of the N-P bracket; this last leads to the FI which, in
turn, is inconvenient for a quantum version. Such a compatibility
only exists for the standard Poisson structures to which both
schemes reduce for $n$=2. In his paper, Nambu stated \cite{Nambu:73}
that `quantum theory is pretty much unique although its classical
analog may not be'. But one might as well take the point of view
(see also \cite{Mu-Sud:76}) that classical mechanics is pretty
unique too if the term `dynamical system' is restricted to its
physical rather than to a mathematical meaning, under which a system
of differential equations may be judged sufficient to describe the
`dynamics' of a `system'. In any case, it seems clear that the
quantization of higher order Poisson brackets requires renouncing to
some of the standard steps towards quantum mechanics. Of course, all
quantization procedures tend to hide the obvious fact that Nature is
quantum {\it ab initio} and that, accordingly, the emphasis should
rather be on finding the classical limit of quantum descriptions if
these were readily available; the insistence on quantization schemes
just reflects the fact that, unfortunately, this is not so.

   Having said this, it is worth going back to the $n$ even case,
for which the GJI holds for associative operators. It is
possible to avoid the loss of correspondence between the classical
and the quantum versions of the theory if  one considers that
the Filippov identity is not the `fundamental' one for
the quantization of Nambu-Poisson \cite{Tak:93} or Nambu's
generalized hamiltonian systems. The relative virtues
of the Nambu bracket given by the Jacobian determinant
\eqref{Nam-n-bra} and that the GJI \eqref{GJI} may
be combined if one accepts that the identity that must be
satisfied by the quantized bracket is, as argued above, the GJI
that follows from the full antisymmetrization.
Further, since the above Jacobian determinant is given in terms of the
Levi-Civita tensor, it is possible to solve the $n=2s$ even Nambu bracket
in terms of ordinary Poisson brackets much in the same way a
multibracket with $2s$ entries is reducible to sums of products of
$s$ ordinary two-brackets\footnote{In general, the expansion of
the even $n=2p$ multibracket clearly mimics the expression of the Pfaffian of an
antisymmetric $n\times n$ matrix $A$, Pf$(A_{ij})=\frac{1}{2^p\,p!}
\epsilon^{i_1\dots i_{2p}}_{1,\,\dots,\,2p}A_{i_1 i_2}\cdots A_{i_{2p-1}i_{2p}}$,
if the non-commutativity of the different two-brackets is ignored, which
reduces the number of terms in expansions such as
eq.~\eqref{n=4-resol}  by $\frac{1}{p!}$. PB's conmmute, and thus
the above PB resolution is in obvious correspondence with the Pfaffian
(as also noted by K. Bering, see \cite{Curt-Zach:03}).}. In this way,
one may adopt -in fact, following Nambu\footnote{See
 eqs. (33a) and (33b) in \cite{Nambu:73} for $n$=3, which are the
same as eqs.~\eqref{NQB-3}; Nambu already appreciated the especial
difficulties of the odd case.}- that the quantum Nambu bracket is
the fully antisymmetric higher order commutator and that it is the
GJI of the multibracket, not the FI, the property that is relevant
in quantization.
\medskip

In fact, maximally superintegrable systems are all
describable classically by both Hamiltonian, and, equivalently,
Nambu mechanics \cite{Cu-Za:02,Curt-Zach:02,Za-Cu:04}. Thus, their Hamiltonian
quantization provides a check  for their
Nambu quantization. This point of view, which in a way takes the best
of the two worlds, has been consistently advocated by Curtright
and Zachos \cite{Cu-Za:02, Za-Cu:03a, Cu-Za:03b}. In it, the
connection between the quantum bracket (the multibracket,
satisfying the GJI and thus suitable for associative quantum
operators) and the classical one (the Nambu bracket given by the
Nambu Jacobian, eq.~\eqref{Nam-n-bra}) becomes evident. For
instance, since in the classical $\hbar\rightarrow 0$ limit
$\frac{1}{i\hbar}[\mathscr{O}_1,\mathscr{O}_2] \rightarrow
\{O_1,O_2\}$, for $n=4$ we obtain from eq.~\eqref{n=4-resol} that,
in that limit,
\begin{equation}
\begin{aligned}
\label{4h}
 \frac{1}{2} \frac{1}{{(i\hbar)}^{\frac{n}{2}}}
 [\mathscr{O}_1,\mathscr{O}_2,\mathscr{O}_3,\mathscr{O}_4]
& \rightarrow \{O_1,O_2\}\{O_3,O_4\} -  \{O_1,O_3\}\{O_2,O_4\}
+\{O_1,O_4\}\{O_2,O_3\} \\
= & \{O_1,O_2,O_3,O_4\}
\end{aligned}
\end{equation}
by using the decomposition of the $n=4$ Jacobian into products of
ordinary Poisson brackets given by the resolution of the $n$=4
Levi-Civita symbol in terms of products of $n$=2
ones; notice that, by proceeding in this way, one is taking the Nambu
Jacobian as the {\it fundamental} property of Nambu mechanics.

The same correspondence clearly works in the higher order even\footnote{It
may be possible, in principle, to consider odd cases by embedding
the odd ($2s-1$)-quantum brackets into even $2n$ quantum brackets
\cite{Curt-Zach:03} to reduce their quantization to the even
case.} case, for which
\begin{equation}
\frac{1}{(n/2)!}\frac{1}{(i\hbar)^{\frac{n}{2}}}
[\mathscr{O}_1,\mathscr{O}_2,\dots,\mathscr{O}_n] \rightarrow
\{O_1,O_2,\dots,O_n\}
\end{equation}
in the classical limit. The first factorial appears because the
reduction of the $n=2s$ bracket to products of two-brackets contains
terms that become the same in the classical limit which replaces
commutators by Poisson brackets ({\it cf.} eqs.~\eqref{n=4-resol}
and \eqref{4h}) since the product of functions is commutative. We
shall not carry the discussion any further and refer to the papers
quoted in this section for details.

\section{The Bagger-Lambert-Gustavsson (BLG) model}
\label{sec:BLG}
{\ }

We now come to the last part of this review, the appearance of 3-Lie
and Nambu FA structures in brane theory. We shall restrict ourselves
to the original BLG proposal and to its Nambu bracket extension
because of their higher structural simplicity. Other approaches will
be mentioned, but perhaps it is fair to say that, at present, there
is not a completely satisfactory answer to the questions mentioned
below and later in Sec.~\ref{sec:BLG-phys}. In the remaining
sections the emphasis will be on the geometry of BLG-like models,
rather than on their physical contents, as a way to illustrate the
previous $n$-ary algebraic structures. Some other aspects, as {\it
e.g.} the possible deformations of BLG and related theories (see
\cite{Hoso-2Lee:08,Go-Ro-Raa-Ver:08,Cra-Her-Tur:09,Ak-Sae-Wo:09,Gustav:09,Chen-Ho-Tak:10}
and refs. therein), will not be considered here.
\medskip

  The strong coupling limit of the IIA superstring theory is a
$D=11$ one, M-theory, the low energy limit of which is $D=11$
supergravity \cite{Hu-To:94,Wit:95}. This admits a fully
32-supersymmetric solution with the geometry of
$AdS_4\times S^7$ and isometry group $OSp(8|4)$. To obtain
some insight into the structure of the elusive M-theory
it became important, due to the AdS/CFT correspondence
\cite{Mal:97,Wit:98-AdS} (see \cite{Aha-Gub-Mal-Oo-Oz:99,Kle:00,Nas:07}
for reviews) to construct the action of the superconformal gauge
theories that are AdS/CFT dual to M-theory on the
above background.
As discussed in \cite{Schw:04}, these theories were expected to be
worldvolume $d=3$ gauge theories coupled to massless matter with
$\mathcal{N}=8$ linearly realized supersymmetries and $OSp(8|4)$
superconformal symmetry, which is also the symmetry of the M-theory
solution. Thus, excluding the possibility of singlets, they were to
contain eight $d=3$ real scalar fields, coming from the eight
transverse coordinates of the M2-brane (8=11-[3 M2-worldvolume
coord.]) plus 16 real (off-shell) $d=3$ Goldstone fermions (the
other 16 being removed by $\kappa$-symmetry) and, since they had a
$\mathcal{N}=8$ supersymmetry, they would present a natural $SO(8)$
R-symmetry. Since this corresponds on-shell to 8 (bosonic) =16/2
(fermionic) degrees of freedom, there is no room left for any more
on-shell physical bosonic $d.o.f.$ Therefore, it was proposed that
the gauge fields should appear in the action through a Chern-Simons
term and that, accordingly, the theory should be a supersymmetric
extension of a gauge theory of Chern-Simons type. It seemed after
the analysis in \cite{Schw:04}, however, that in spite of the
theoretical grounds for the existence of such a theory, a
superconformal action with the required properties (and specially
the $\mathcal{N}$=8 superPoincaré invariance) could not exist. A few
years after this apparent `no-go' result, the ground-breaking work
of Bagger and Lambert (and Gustavsson) to be described below showed that a $d=3$
superconformal action with CS term and $\mathcal{N}=8$ supersymmetry
was possible after all, sparkling a great interest on the subject.

\subsection{Symmetry considerations and ingredients of the BLG model}
\label{sec:BLGmod}
{\ }

   The BLG model \cite{Ba-La:06,Ba-La:07a,Gus:07,Ba-La:07b,Gustav:08}
(see also \cite{Raam:08} and \cite{Lam-Rich:09} for recent work
involving the background gauge fields of $D=11$ supergravity) is a
three-dimensional maximally supersymmetric superconformal gauge
theory aimed, along the lines above, to describing the very low
energy effective worldvolume theory of a system of $N$ coincident
M2-branes in the $D=11$ spacetime of M-theory. Although the original
goals were not reached in the form they were initially
expected, the $A_4$ BLG model provided the first successful example
of an interacting $d=3$ gauge theory with $\mathcal{N}=8$ linearly
realized supersymmetries ({\it i.e.}, $d=3$ maximally
supersymmetric) and with superconformal symmetry $OSp(4|8)$.
Further, the Noether currents associated with the BLG lagrangian
generate \cite{Pas:08} the $d=3$, $\mathcal{N}=8$ superPoincar\'e or
`M2'-algebra with central charges that had already been discussed on
general grounds \cite{Be-Go-To:97},
\begin{equation}
\label{M2alg}
\{Q_\alpha^p,Q_\beta^q\}=-2 (\gamma^\mu\gamma^0)_{\alpha\beta}\delta^{pq}P_\mu+
\epsilon_{\alpha\beta}Z^{[pq]}+(\gamma^\mu\gamma^0)_{\alpha\beta}Z_\mu^{(pq)}
\quad (\mu=0,1,2\;;\;\alpha,\beta=1,2\;;\;p,q,=1,\dots,8) \;,
\end{equation}
where the symmetric central charge is traceless, $\delta_{pq}Z_\mu^{(pq)}=0$.
This algebra has an obvious $SO(8)$ automorphism group under which the eight
$d=3$ two-component Majorana supercharges $Q^q$ form a chiral $Spin(8)$ spinor.
The $\binom{17}{2}$ degrees of freedom of the $l.h.s.$ of eq.~\eqref{M2alg}
split as $136= 3+ \binom{8}{2}+ 3(\binom{9}{2}-1)=3+28+3\times 35$.
Thus, the worldvolume zero- (one)-form $Z^{[pq]}\, (Z_\mu^{(pq)})$ transforms under
the {\bf 28} ({\bf 35}$^+$) representation of $SO(8)$. In transverse space,
these two central charges may also be understood, respectively, as a two- and
a self-dual four-form ($\frac{1}{2}\binom{8}{4}=35$).

  The Bagger-Lambert construction was originally based on a
three Filippov algebra structure, the properties of which were
actually  rediscovered from the physical requirements needed to
build the model. Bagger and Lambert where led to a {\it three}-Lie
algebra since they wanted to recover the Basu-Harvey equation
\cite{Bas-Har:05} which may be formulated in terms of a
three-bracket (see eq.~\eqref{BLM13}). Thus, the BLG theory appeared
to provide an application of the Filippov algebra structure
discussed in previous sections. We give below an outline of the
original BL action and its actual relation to the simple euclidean
$A_4$ 3-Lie algebra. We shall nevertheless keep occasionally a
generic $n=3$ FA notation when there is no need of identifying $\fG$
with $A_4$, for which $f_{abc}{}^d=-\epsilon_{abc}{}^d$.

In order to write certain terms of the worldvolume
lagrangian, including the kinetic ones, the $n=3$ FA was
required to be endowed with an invariant metric
$\left<\ ,\ \right>$, so that eq.~\eqref{sc-n-inv}
 (or \eqref{sc-n-inv-coor}) is satisfied. Further,
to avoid states with negative norm in the quantum theory, it was
assumed that the metric was  positive definite.  It turned out that
this determined completely the finite 3-Lie algebra \cite{Pap:08,
Ga-Gu:08,Gust:08a} to be the simple euclidean $A_4$ one; of course,
it is also possible to have a direct sum of multiple $A_4$ copies
and trivial one-dimensional algebras, as conjectured in
\cite{Ho-Hou-Ma:08,JMF-Pap:04}. There is, however, a simple argument
leading to $A_4$: if Lie$\,\fG$ has to be semisimple, something one
would require at least for a gauge group, $\fG$ has to be reductive
by Th.~\ref{th:redFAs}. Removing then a possible, uninteresting
centre, we are left with a semisimple $\fG$. Positive definiteness
-or the needed compactness of Lie$\,\fG$- leads then to a direct sum
of $A_4$ copies and simplicity restricts the 3-Lie algebra to the
euclidean $A_4$ as the only possibility.

   To look at the fields of the BLG theory,
let us assume that the M2-brane worldvolume coincides with the $D=11$
hyperplane parametrized by the $0,1,2$ spacetime coordinates; then,
the remaining $3,\dots,10$ coordinates are transverse to the
M2-brane. This splitting is preserved by a $SO(1,2)\times SO(8)$
subgroup of $SO(1,10)$. The fields describing a single M2-brane
depend on the $d=3$ Minkowski worldvolume membrane coordinates
$x^\mu$, $\mu=0,1,2$, and are given by two sets of `matter' fields,
bosonic and fermionic, plus additional gauge fields. The bosonic
fields describe the transverse fluctuations of the membrane and are
given by eight worldvolume scalar transverse coordinate  fields
$X^I(x)$  labelled by $I=3,\dots 10$. The 16 fermionic fields are
eight two-component worldvolume spinors that may be described in
terms of a $D=11$ 32-component Majorana spinor $\Psi(x)$ subject to
the condition $\Gamma^{012}\Psi =-\Psi$, where the $\Gamma$'s are
the gamma matrices of the eleven-dimensional spacetime $(\Gamma^\mu,
\Gamma^I$). Since the $D=11$ `$\gamma^5$' is the unity (there is no
chirality in odd dimensions), spinors that are $\Gamma^{012}$ chiral
have also a definite $Spin(8)$ chirality since this is determined by
$\Gamma^{34\dots 9(10)}$.

The M2-brane breaks half of the
supersymmetries, and the preserved ones are taken to be chiral,
$\Gamma^{012} \epsilon=\epsilon$. The antichiral $\Psi$ fields are
the goldstinos corresponding to the broken supersymmetries and the
$X^I$ are the Goldstone scalar fields that correspond to the eight
broken translations. The $\Psi$ fields have sixteen independent real
components; from the point of view of the $Spin (1,2)\times Spin(8)$
subgroup of $Spin(1,10)$, they are bidimensional $SO(1,2)$ spinors;
the $\mathcal{N}=8$ supersymmetries refer to this $d=3$ worldvolume
description (both the $SO(1,2)$ and the chiral $SO(8)$ spinorial
indices of $\Psi$ are omitted). Thus, the eight scalars $X^I(x)$ and
the eight fermionic spinors $\Psi(x)$ transform, respectively, under
the eight-dimensional vector and chiral spinor representations of
the R-symmetry group $SO(8)$ which preserves the M2 superalgebra
\eqref{M2alg}.

   In order to describe a stack of M2-branes,
Bagger and Lambert introduced bosonic and fermionic fields taking
values in a FA $\fG$ (which, as mentioned, turned out to be $A_4$ in
their first proposal). Thus, the  mater fields of the BLG-type
models $X^I(x)=X^{Ia}(x)\be_a$, $\Psi(x)=\Psi^a(x)\be_a$
($a=1,\dots,4=\hbox{dim}\,\fG$), where  $\{\be_a\}$ is a basis of
$\fG$, carry a dim$\,\fG$-dimensional representation of Lie$\,\fG$.
 As for the gauge fields $A_\mu^{ab}$, they are
 Lie$\,\fG$-valued worldvolume vector fields
with $A_\mu^{ab}(x)= -A_\mu^{ba}(x)$. The two indices $ab$ refer
(see Sec.~\ref{sec:coord3alg}) to those that determine the elements
of Lie$\,\fG$ through the fundamental objects of $\fG$. This means
that, assuming simplicity ($\fG=A_4$), we have a one-to-one
correspondence between fundamental objects
$\mathscr{X}\in\wedge^2\fG$ and elements $T_{ab}\in
\mathrm{Lie}\,\fG$, $(T_{ab})_c{}^d=f_{abc}{}^d$ so that
$A=A^{ab}T_{ab}$, $A_c^d=A^{ab}f_{abc}{}^d$. Thus, in spite of being
given through the structure constants $f_{abc}{}^d$ of a {\it
three}-Lie algebra, the vector fields $A_\mu$  (which may be seen as
$\wedge^2 \fG$-valued) are of course ordinary Lie algebra,
Lie$\,\fG$-valued gauge fields. Thus, the gauge fields are  in the
adjoint representation of the gauge group as usual (but not the
matter ones which, unlike in \cite{Schw:04}, take values in $\fG$
itself). The fact that Lie$\,\fG$ determines the gauge group
explains the r\^ole and the physical importance of the Lie algebra
associated with the FA $\fG$.

\subsection{The BLG lagrangian}
{\ }

  The BLG model is given by the worldvolume lagrangian density
\cite{Ba-La:07a,Ba-La:07b} (see also \cite{Gus:07, Be-Ta-Tho:08,Raam:08})
\begin{eqnarray}
\label{BLM2}
  \mathcal{L}_{BL} & = &
  -\frac{1}{2} \left< \mathcal{\mathcal{D}}_\mu X^I,\mathcal{\mathcal{D}}^\mu X^I \right>
  + \frac{i}{2} \left< \bar{\Psi}, \Gamma^\mu \mathcal{D}_\mu \Psi\right>
  \nonumber\\
   & & - g\,\frac{i}{4} \left< \left[ \bar{\Psi}, X^I, X^J\right],
   \Gamma_{IJ} \Psi \right> - \frac{g^2}{2\cdot 3!} \left< \left[
   X^I,X^J,X^K\right], \left[
   X^I,X^J,X^K\right] \right> \nonumber\\
  & & + \frac{1}{g}\, \mathcal{L}_{CS} \quad .
\end{eqnarray}
The corresponding action may be split, following the above three lines, as
\begin{equation}
\label{BLGaction}
I_{BLG}=\int d^3 x \,\mathcal{L}_{kin} +
\int d^3 x \,\mathcal{L}_{int} + \frac{1}{g} \int d^3 x\, \mathcal{L}_{CS} \; .
\end{equation}
The covariant derivative $\mathcal{D}_\mu$ above is defined for a generic
$\fG$-valued matter field $V=V^a\be_a$ by
\begin{equation}\label{BLM3}
    (\mathcal{D}_\mu V)^a= \partial_\mu V^a - {f_{cdb}}^a A_\mu^{cd} V^b\ ,
\end{equation}
and $\mathcal{L}_{CS}$ has the form
\begin{equation}\label{BLM4}
  \mathcal{L}_{CS}= \frac{1}{2} \epsilon^{\mu\nu\rho} \left(
  f_{abcd} A_\mu^{ab} \partial_\nu A_\rho^{cd} + \frac{2}{3}
  {f_{cda}}^g f_{efgb} A_\mu^{ab} A_\nu^{cd} A_\rho^{ef} \right) \; .
\end{equation}
The Chern-Simons term $\mathcal{L}_{CS}$ was called `twisted'
because it did not seem to have the standard CS expression (but see
Sec.~\ref{sec:CSterm} below).

The BLG action is scale-invariant
provided that the gauge fields have length dimension $A=L^{-1}$ and
the constant $g$ is dimensionless. Then, the kinetic terms for the
worldvolume matter fields are also scale-invariant with
$[X]=L^{-\frac{1}{2}}$ and $[\Psi]=L^{-1}$, the expected
dimensions for a $d=3$ theory with no dimensionful constants.
The dimension of $A$, consistent with its r\^ole as part of
a covariant derivative, would be unnatural for the kinetic term of a $d=3$ field
theory, but there is no such a term for the gauge field in
eq.~\eqref{BLM2}. It may be seen that these dimensions also fix the
form of the possible interaction terms in the lagrangian, which
cannot depend on any dimensional coupling constant. In spite of the
appearance of a Chern-Simons term in the lagrangian and that
standard CS terms are parity odd, the theory is parity invariant due
to the composite nature of the `twisted' $\mathcal{L}_{CS}$ term as
we shall see later.

\subsubsection{Gauge and supersymmetry transformations}
{\ }

    The BLG action is invariant under both gauge symmetry and
supersymmetry; it has $OSp(4|8)$ superconformal symmetry
\cite{Ba-Li-Schw:08}. The non-propagating gauge fields
$A^{ab}$ are needed for the closure of the supersymmetry algebra
transformations \cite{Ba-La:07a} which are given below but that will not
be discussed here. By standard arguments \cite{De-Ja-Te:82a,De-Ja-Te:82b},
the invariance of the quantum theory under `large'
gauge transformations implies the quantization
of the coefficient of the CS term \cite{Ba-La:07b}.
For the original $A_4$ BLG model, the geometry of the
$\mathcal{L}_{CS}$ term (sec.~\ref{sec:CSterm}) leads to
the quantization condition $\frac{1}{g}=\frac{k \hbar}{2\pi}$
where $k$ is an integer. As a result of this quantization,
the theory does not contain any continuous parameter, and thus it
must be conformally invariant \cite{Ba-La:07b} to all orders of
perturbation theory since there are no coupling constants to run.
Further, it is possible to redefine the three-bracket and the gauge
fields so that $g$ disappears from eq.~\eqref{BLM2}. In fact,
it had been known for some time that three-dimensional Chern-Simons
gauge theories were themselves conformally invariant
\cite{Che-Se-Wu:92,DC-Fr-He-Pi:99}, both the pure gauge
theories and those coupled to massless matter fields. The problem,
thus, was how to incorporate the extended supersymmetries
needed to give a dual description of M2-branes, and this is
what the gauge BLG model succeeded in doing.

  The gauge transformations of the different fields are
given by (they will be rewritten in a more geometrically
transparent way in Sec.~\ref{sec:CSterm})
\begin{eqnarray}
\label{BLM5}
    \delta X^{Ia} & = &\lambda^{cd} {f_{cdb}}^a X^{Ib} \nonumber \\
   \delta \Psi^{a} &=& \lambda^{cd} {f_{cdb}}^a \Psi^{b} \nonumber \\
   \delta ({f_{cdb}}^a A^{cd}_\mu ) & = & \partial_\mu ( {f_{cdb}}^a
   \lambda^{cd}) +
   2 {f_{cdb}}^a {f_{efg}}^{c} \lambda^{dg} A^{ef}_\mu
\end{eqnarray}
These transformations actually correspond to the Lie$\,\fG$ (Sec.
\ref{sec:CSterm}) gauge group algebra. For $\fG=A_4$,
Lie$\,A_4=su(2) \oplus su(2)$ (Ex.~\ref{ex:LieA4}); the gauge group
of the $\mathcal{N}=8$ BLG theory with positive definite
metric is thus $SU(2)\times SU(2)$. Further, for $A_4$, the
$f_{cdb}{}^a$ can be removed in the third line of eq.~\eqref{BLM5}
above and eq.~\eqref{BLM6} below.

The supersymmetry transformations are given by
\begin{eqnarray}
\label{BLM6}
    \delta_\epsilon X^I &=& i\bar{\epsilon} \Gamma^I
    \Psi\nonumber\\
    \delta_\epsilon \Psi &=& \mathcal{D}_\mu X^I \Gamma^\mu\Gamma^I \epsilon
    - \frac{g}{3!} [X^I,X^J,X^K] \Gamma^{IJK} \epsilon \nonumber \\
    \delta_\epsilon ( {f_{cdb}}^a A^{cd}_\mu) &=& i g {f_{cdb}}^a \bar{\epsilon} \Gamma_\mu
    \Gamma^I X^{Ic} \Psi^d  \; ,
\end{eqnarray}
where the supersymmetry parameter $\epsilon$ has the
standard dimensions $[\epsilon]=L^{\frac{1}{2}}$.

\subsubsection{Physical considerations and ternary algebras}
\label{sec:BLG-phys}
{\ }

   It was soon realized that the non-trivial gauge symmetry of the model,
Lie$A_4=SU(2)\times SU(2)$ (short, {\it e.g.}, of a $SU(N)$ one),
could not give rise to the moduli space of a stack of M2-branes (see
\cite{Ba-La:08,Raam:08}) as initially hoped for. In fact, it was
argued in \cite{La-To:08,Di-Mu-Pa-Raa:08} that the $A_4$ BLG model
at level $k=1$ describes two M2-branes propagating in a
$\mathbb{R}^8/Z_2$ orbifold background (for general $k$, on a
`M-fold'). Thus, the original BL model does not describe a number
$N>2$ of M2-branes, a fact that might have been expected
\cite{Papad:08} from the smallness of its gauge group.

There is, nevertheless, the possibility of relaxing the assumption of
positive definiteness of the metric. There are models with
Lorentzian signature on the 3-Lie algebra
\cite{Go-Mi-Ru:08,Be-RG-To-Ve:08,Ho-Im-Mat:08}, although these
immediately raise the question of the unitarity of the quantized
theory (see also \cite{deMe-Fi-M-E-Rit:09} for this point and
\cite{JMF-nLie:08} for $n$-Lie algebras with Lorentzian metric). The
Lorentzian FAs are obtained from a semisimple Lie algebra $\fg$ from
which the (fully antisymmetric) structure constants of the FA are
constructed. As for the Lorentzian theory in \cite{Go-Mi-Ru:08}, it
can be recast as an ordinary gauge theory, but the Lie algebra
associated with the FA is no longer semisimple, with its Levi factor
given by the original semisimple $\fg$; as a result, the theory has
features reminiscent of those of WZW models \cite{Nappi-Wi:93} and
gauge theories \cite{Tsey:95} based on non-semisimple Lie groups. It
is possible to remove the ghosts of the 3-Lie Lorentzian theories,
but the modification breaks the conformal invariance spontaneously
and reduces them to maximally supersymmetric $d=3$ YM theories
\cite{Ba-Lip-Sch:08a,Go-Rod-Raam:08}.

One may also relax the full
antisymmetry of the three-bracket\footnote{Other aspects of
three-Lie algebras have been discussed in
\cite{Gust:08a,Mor:08,Gra-Nil-Pet:08,Be-RG-To-Ve:08,Chen:09}; see
also \cite{Awa-Li-Mi-Yo:99} for early work.} as is the case of the
hermitian algebras \cite{Ba-La:08} to be mentioned below and that of
the real `generalized three-algebras' in \cite{Cher-Sa:08} (see
further \cite{Che-Do-Sa:08,Ak-Sae-Wo:09}); the gauge symmetries are
generated, of course, by ordinary Lie groups. These algebras lead to
superconformal field theories that accommodate a higher number of
M2-branes in exchange for a reduced amount of supersymmetry with
respect to the $A_4$ BLG model. Nevertheless, these and other
algebras (as in \cite{Yama:08}) will not be considered here,
although we refer to \cite{deMed-JMF-Men-Rit:08} for a discussion of
the Lie-algebraic structure of the hermitian \cite{Ba-La:08} and
real `generalized three-algebras' \cite{Cher-Sa:08}; see also
Sec.~\ref{sec:triple} for the real case. To conclude the above discussion
we will mention that, very recently, the algebras appearing in
BLG-type models have been discussed in the context of Jordan-triple
systems ({\it Lie}-triple systems were considered in
Sec.~\ref{sec:triple}); we just refer here to
\cite{Chow-Muk:09,Palm:09} for further information.

   Setting aside the above Filippov and related three-algebras
approach, it has been shown \cite{Aha-Be-Ja-Mal:08} that, giving up
the full $\mathcal{N}=8$ manifest supersymmetry, it is possible to
find $d=3$ superconformal  Chern-Simons theories with $U(N)\times
U(N)$ (and $SU(N)\times SU(N)$) gauge groups, thus providing a more
general description of the $N$ M2-brane system (see
\cite{Aha-Berg-Ja:08} and \cite{Kle-To:09} for further discussion of
the ABJM theory and \cite{Be-Bri-Sung-Ram:10} for a different,
alternative approach to BLG theory). The family of ABJM models \cite{Aha-Be-Ja-Mal:08}
have matter fields in the bi-fundamental representation of the gauge
group and a double set of gauge fields in the adjoint; they present
explicit $\mathcal{N}=6$ supersymmetry and, furthermore, the 3-Lie
algebra structure does not play any r\^ole at all in their
formulation. The $SU(2) \times SU(2)$ gauge symmetry and the
original BLG model appears as a particular $N=2$ example of the
sequence of theories in \cite{Aha-Be-Ja-Mal:08}, for which the
original $\mathcal{N}=6$ supersymmetry of twelve real supercharges
is enhanced to a $\mathcal{N}=8$ one with sixteen (see also
\cite{Ba-Li-Schw:08, Gu-Rey:09}) because the fundamental
representation of $SU(2)$ is equivalent to its conjugate. It is
stated in \cite{Aha-Be-Ja-Mal:08} that at levels $k=1$ and $k=2$ the
ABJM theories describe, respectively, the low energy limit of $N$
M2-branes in a flat and on a $\mathbb{R}^8/Z_2$ space and,
generically, on $\mathbb{R}^8/Z_k$ at level $k$ (see further
\cite{Lam-Con:10} for the equivalence of the $U(2)\times U(2)$ ABJM
theory at $k=1,2$ with the $\mathcal{N}=8$  BLG
one and \cite{Ba-Br:10} for the most general three-dimensional
$\mathcal{N}=5$ superconformal CS theories based on three algebras).

It was then shown by Bagger and Lambert \cite{Ba-La:08} that
it was possible to recover $\mathcal{N}=6$ models of the type
considered in \cite{Aha-Be-Ja-Mal:08} by making the three-algebra
complex, relaxing the full skewsymmetry of the three-bracket by
making it skewsymmetric and linear in the first two entries and
antilinear in the last and, in so doing, by moving effectively from
a 3-Lie algebra to a kind of complex (right) 3-Leibniz one
(Sec.~\ref{sec:n-Leibniz}). In this approach, the standard Filippov
identity changes to a `hermitean FI' and accordingly its expression
in terms of the structure constants due to the appearance of complex
conjugation. The resulting generalized BL model \cite{Ba-La:08}
presents $\mathcal{N}=6$ supersymmetry, $SU(4)$ R-symmetry and
$U(1)$ global symmetry. The extended worldvolume superalgebra of the
$\mathcal{N}=6$ BL model has been computed in \cite{Low:09}, where
for a particular choice of the three-bracket the  ABJM model
superalgebra with central charges has been derived.

  In order to describe a large  number $N$ or condensate
of M2-branes, it is also possible to move to the
infinite-dimensional three-algebra defined by the Nambu bracket.
This can be done in a natural way that implies, in particular,
making  the range of the FA indices of the BLG model infinite as we
will show in Sec.~\ref{sec:BLG+NB}.

\subsection{The BLG model and the Basu-Harvey (B-H) equation}
\label{sec:BLG+BH}
{\ }\\

The B-H equation \cite{Bas-Har:05} arises as a BPS condition
for the BLG theory; reproducing it was, in fact, one of its
motivations \cite{Ba-La:06,Ba-La:07b}. Indeed, the model given by
the lagrangian \eqref{BLM2} has a BPS solution that corresponds to
the $D=11$ M2-brane-M5-brane system as seen from M2, and the
condition that such a bosonic configuration satisfies is the
B-H equation \cite{Bas-Har:05}, a counterpart of the Nahm
equation \cite{Nahm:80} that appears in the D1-D3 system in $D=10$
(for generalizations of The B-H equation see \cite{Ber-Copl:05,Kri-Mac:08}
in the context of general calibrations, ref.~\cite{Bo-Ta-Zab:08}
for arbitrary $p$-algebras and \cite{Iu-Na-Sa-Ze:09} for the
B-H equations as `homotopy M-C equations').
This bosonic solution may be found in the usual way by requiring that it
preserves some supersymmetries. These configurations saturate a {\it
Bogomol'nyi}-type bound, which is a lower bound on the energy given
in terms of some charges that, in general, appear in the
supersymmetry algebra as central (when the Lorentz part is ignored).
The invariance of the BLG action under the supersymmetry
transformations \eqref{BLM6} leads \cite{Pas:08} to charges that
generate the $d=3$ $\mathcal{N}=8$ supersymmetry
algebra \eqref{M2alg}. The presence of central charges is due to the
fact that the Lagrangian density is only quasi-invariant (invariant
up to a worldvolume total derivative) under supersymmetry
transformations, and for that reason they are topological (see
\cite{Az-Ga-Iz-To:89}) {\it i.e.}, relevant in topologically
non-trivial situations. Some of the topological central charges of
the enlarged $d=3$, $\mathcal{N}=8$ superalgebra \cite{Be-Go-To:97}
appear in the Bogomol'nyi bound that corresponds to the M2-M5-brane
brane system.
\medskip

Let us see how the B-H equation appears when the preservation
of a fraction of the supersymmetries is required; actually, it will be seen
that the solutions of the B-H equation are $\frac{1}{2}$-BPS
solitons. The bosonic configurations searched for are of the type
\begin{equation}
\label{BLM7}
\begin{array}{cccccccc}
  M2: & 0 & 1 & 2 &   &  &   &   \\
  M5: & 0 & 1 &   & 3 & 4& 5 &  6 \quad ,
\end{array}
\end{equation}
which means that the membranes are extended along the 1,2
spatial directions and that the spatial directions of the
M5 brane worldvolume coincide with the (13456) hyperplane.
Thus, $x^2$ is the coordinate of the M2-brane
that is transverse to the M5-brane and $x^1$ parametrizes the M2-M5
intersection. In terms of the fields appearing in the BLG
lagrangian \eqref{BLM2}, the configuration that describes the
system \eqref{BLM7} has to have $X^3,X^4,X^5$ and $X^6$ as the only
non-constant coordinate fields $X^I$ because these are the fivebrane coordinates
transverse to the M2 worldvolume as seen in \eqref{BLM7}
so that $I=3,\dots,10 \rightarrow \mathcal{I}= 3,4,5,6$.
Also, these fields must depend only on $x^2$, which parametrizes
the `distance' of the M2 points to the M5 brane.
\medskip

The condition that a fraction of the supersymmetries is
preserved for a bosonic configuration,
$\Psi=0$, reduces to $\delta_\epsilon \Psi=0$ in
the second equation in \eqref{BLM6}, which gives
\begin{equation}
\label{BLM8}
 \partial_2 X^\mathcal{I} \Gamma^2 \Gamma^\mathcal{I} \epsilon -\frac{g}{6}
 [X^\mathcal{I},X^\mathcal{J},X^\mathcal{K}]
 \Gamma^{\mathcal{I}\mathcal{J}\mathcal{K}} \epsilon = 0 \quad,\quad
 \mathcal{I},\mathcal{J},\mathcal{K},\mathcal{L}=3,4,5,6 \quad,
\end{equation}
where the gauge field is also taken to be zero ($\mathcal{D}_2\rightarrow \partial_2$),
an {\it ansatz} the consistency of which will be shown below.
   The B-H equation arises as a condition that
the bosonic fields have to satisfy if the above equation has to
have a non-trivial $\epsilon\neq 0$ solution. To see this, we notice that
$\Gamma^{\mathcal{I}\mathcal{J}\mathcal{K}}$ can be written as
\begin{equation}
\label{BLM9}
 \Gamma^{\mathcal{I}\mathcal{J}\mathcal{K}}  =
 \Gamma^{3456}
 \epsilon^{\mathcal{I}\mathcal{J}\mathcal{K}\mathcal{L}} \Gamma^\mathcal{L} \ ,
\end{equation}
where $\epsilon^{\mathcal{I}\mathcal{J}\mathcal{K}\mathcal{L}}$ is
completely antisymmetric and defined by $\epsilon^{3456}=1$.
Inserting eq.~\eqref{BLM9} in eq.~\eqref{BLM8} we obtain
\begin{equation}
\label{BLM10}
 \partial_2 X^\mathcal{I}  \Gamma^\mathcal{I} \epsilon =\frac{g}{3!}
 \epsilon^{\mathcal{I}\mathcal{J}\mathcal{K}\mathcal{L}}
 [X^\mathcal{I},X^\mathcal{J},X^\mathcal{K}]
 \Gamma^{23456} \Gamma^\mathcal{L} \epsilon \equiv
  B^\mathcal{L} \Gamma^{23456} \Gamma^\mathcal{L}\epsilon  \; ,
\end{equation}
where the $\fG$-valued $B^\mathcal{L}$ above is
introduced to simplify the expressions below.
Since $\{\Gamma^\mathcal{I},\Gamma^\mathcal{J} \}
=2\delta^{\mathcal{I}\mathcal{J}}$, it is seen that $\Gamma^{23456}$
above commutes with  the $\Gamma^\mathcal{I}$ ($\mathcal{I}=3,4,5,6$) and
squares to the unit matrix. Now, we compute
\begin{eqnarray}
\label{BLM11}
\left<\partial_2 X^\mathcal{I} ,\partial_2 X^\mathcal{I}\right> \epsilon &=&
\left< \partial_2 X^\mathcal{I} , \partial_2 X^\mathcal{J} \right>
\Gamma^\mathcal{I} \Gamma^\mathcal{J} \epsilon = \left< \partial_2
X^\mathcal{I} \Gamma^\mathcal{I},
\partial_2 X^\mathcal{J} \Gamma^\mathcal{J} \epsilon \right>\nonumber\\
 &=& \left< \partial_2 X^\mathcal{I}
\Gamma^\mathcal{I}, \Gamma^{23456} B^\mathcal{J}
\Gamma^\mathcal{J} \epsilon\right> = \left< \partial_2 X^\mathcal{I},
B^\mathcal{J} \right> \Gamma^{23456} \Gamma^\mathcal{I}
\Gamma^\mathcal{J} \epsilon
\nonumber\\
&=&- \left< \partial_2 X^\mathcal{I}, B^\mathcal{J} \right>
\Gamma^{23456} \Gamma^\mathcal{J} \Gamma^\mathcal{I} \epsilon +
2\left< \partial_2 X^\mathcal{I}, B^\mathcal{I} \right> \Gamma^{23456}
\epsilon
\nonumber\\
&=& - \left< B^\mathcal{I}, B^\mathcal{I} \right> \epsilon +
2\left< \partial_2 X^\mathcal{I}, B^\mathcal{I} \right> \Gamma^{23456}
\epsilon \; ,
\end{eqnarray}
where eq.~\eqref{BLM10} has been used in the second and forth lines.
The above equality is equivalent to $\Gamma^{23456}\epsilon= \left(
\frac{\left<\partial_2 X^\mathcal{I} ,\partial_2
X^\mathcal{I}\right> +  \left< B^\mathcal{I}, B^\mathcal{I} \right>}
{2\left< \partial_2 X^\mathcal{I}, B^\mathcal{I} \right> } \right)
\epsilon$ and, since $(\Gamma^{23456})^2=1$,  implies that
$\Gamma^{23456}\epsilon =\pm \epsilon$ and
\begin{equation}
\label{BLM12}
   \left< \partial_2 X^\mathcal{I} \mp B^\mathcal{I},
   \partial_2 X^\mathcal{I} \mp B^\mathcal{J} \right> = 0 \; .
\end{equation}
Since the metric is positive definite, this means that the
configurations that preserve a fraction of
the supersymmetries (actually, half of the sixteen)
have to satisfy
\begin{equation}
\label{BLM13}
    \partial_2 X^\mathcal{I} -  \frac{g}{3!}
 \epsilon^{\mathcal{J}\mathcal{K}\mathcal{L}\mathcal{I}}
 [X^\mathcal{J},X^\mathcal{K},X^\mathcal{L}]= 0 \; ,
\end{equation}
selecting the upper sign.

We have still to check that the  bosonic configurations that
solve \eqref{BLM13} are actual solutions of the bosonic field equations
derived from the lagrangian \eqref{BLM2}. These are
given by
\begin{eqnarray}
\label{BLM14}
   & & \mathcal{D}^\mu \mathcal{D}_\mu X^I + \frac{g^2}{2} \left[ X^J,X^K,\left[
   X^I,X^J, X^K\right] \right] = 0 \; , \nonumber\\
   & & {f_{cdb}}^a F^{cd}_{\mu\nu} = - g\,{f_{cdb}}^a \epsilon_{\mu\nu\rho}
   X^{Jc} \mathcal{D}^\rho X^{Jd} \; ,
\end{eqnarray}
where $F^{cd}$ are the curvatures corresponding to the gauge
fields in \eqref{BLM5}, the form of which will be given
in Sec.~\ref{sec:CSterm} below. For the configurations that we are
considering, the above field equations become
\begin{eqnarray}
\label{BLM15} & & \partial^2 \partial_2 X^\mathcal{I} +
\frac{g^2}{2} \left[
X^\mathcal{J},X^\mathcal{K},\left[X^\mathcal{I},X^\mathcal{J},
X^\mathcal{K}\right] \right] = 0 \; ,
\nonumber\\
& & {f_{cdb}}^a F^{cd}_{01}= 0 = -g {f_{cdb}}^a  X^{\mathcal{J}c} \partial^2 X^{\mathcal{J}d}
\quad .
\end{eqnarray}
Inserting \eqref{BLM13} into the first equation of \eqref{BLM15} one
obtains an identity. We check now the consistency of setting $A=0$
in the covariant derivatives in eq.~\eqref{BLM8}, \eqref{BLM15}.
Since this implies $F=0$, it requires that the $r.h.s.$ of the last
equation above be also zero. Using \eqref{BLM13}, the $r.h.s.$
becomes
\begin{eqnarray}
\label{BLM16}
    &  & \frac{g}{3!} {f_{cdb}}^a
    X^{\mathcal{I}c}\epsilon ^{\mathcal{J}\mathcal{K}\mathcal{L}\mathcal{I}}
    X^{\mathcal{J}e} X^{\mathcal{K}f} X^{\mathcal{L}g} {f_{efg}}^d \nonumber\\
    & =&
    \frac{g}{3\cdot 4!} \epsilon ^{\mathcal{J}\mathcal{K}
    \mathcal{L}\mathcal{I}} X^{\mathcal{I}c} X^{\mathcal{J}e} X^{\mathcal{K}f}
    X^{\mathcal{L}g} {f_{[efg}}^d {f_{c]db}}^a = 0 \quad ,
\end{eqnarray}
which is indeed zero by the FI (eq.~\eqref{FIshort-coor}) (and, in
this particular point, any $\fG$ would do). Therefore, one may use
ordinary rather than covariant derivatives in eq.~\eqref{BLM13}.

   Equation \eqref{BLM13} is the B-H equation
\cite{Bas-Har:05} as it was written in \cite{Ba-La:06}. The original
B-H equation has the form
\begin{equation}
\label{BLM18}
   \partial_2 X^\mathcal{I} -  \frac{g}{3!}
 \epsilon^{\mathcal{J}\mathcal{K}\mathcal{L}\mathcal{I}}
 [G,X^\mathcal{J},X^\mathcal{K},X^\mathcal{L}]= 0\ ,
\end{equation}
where here $G,X^\mathcal{J},X^\mathcal{K},X^\mathcal{L}$, $G^2=1$,
are matrices and the {\it four} entries bracket is the {\it
multibracket} of a GLA {\it i.e.}, it is given by the complete
antisymmetrization of the products of its entries,
eq.~\eqref{multic}. The matrix $G$ (see \cite{Bas-Har:05} for its
definition and its connection with the construction of the fuzzy
3-spheres) has the property $\{G,X^\mathcal{I}\}=0$ and the
$X^\mathcal{I}$ may be realized as
$X^\mathcal{I}=X^{\mathcal{I}a}\Gamma_a$, as in
Th.~\ref{le:FA-Cliff}, where the two forms of the B-H equation, \eqref{BLM13} and \eqref{BLM18}, are related.

Equation \eqref{BLM13} is satisfied by the `fuzzy three-funnel' solution,
which is given by a fuzzy 3-sphere \cite{Bas-Har:05,Gur-Ram:00}
with a radius that increases as the worldvolume
coordinate $x^2$ approaches the value ($x^2=0$, say) where the
M5-brane is located with respect to M2, so that it matches with the 5-brane
worldvolume field theory solution that describes the M2-M5 brane
system from the M5-brane point of view and that is known to be
\cite{Ho-La-We:97} a  self-dual string soliton.
We refer to the literature (see \cite{Bas-Har:05} and
references therein) for details. As stated earlier,
the above supersymetric fuzzy funnel solution saturates a
Bogomol'nyi bound on the energy that involves the one-form
`central' charge of the supersymmetry algebra \cite{Be-Go-To:97,Pas:08}.
\medskip

To conclude this section, we note that eq.~\eqref{BLM13} has the
structure of the MC-like equations (Sec.~\ref{sec:MClikeFA}) that hold
for an arbitrary $n$-Lie algebra; thus, the B-H equation,
which describes M2-branes ending on M5 branes, are intimately tied
to the FA structure. This is also the case for their $n=2$
precedent, the Nahm equation \cite{Nahm:80}, which is a MC-type
equation for the D1-D3 system. The solution of the Nahm equation is
a `fuzzy two-funnel', a fuzzy two-sphere with a radius depending on
the transverse distance from the D1 brane that gets larger as it
approaches the D3 brane. In the light of other possibilities (see
\cite{Be-Gi-To:06}), one might think of extending the above $n=2,3$
pattern to higher $n$ generalizations \cite{Bo-Ta-Zab:08} which in
the present context would be tied to a higher MC-type equation for
the FA (see eq.~\eqref{MCforFAs}), but this will not be discussed
here any further (see also \cite{Iu-Na-Sa-Ze:09}).
Finally, we mention that both the Nahm and the
B-H equations have been interpreted in terms of
non-commutative geometry, where the realizations of the worldvolume
coordinate functions of the D3-brane and the M5-brane satisfy,
respectively, non-commuting relations $[X^i,X^j]=i\theta^{ij}$,
$[X^i,X^j,X^k]=i\theta^{ijk}$, where the $r.h.s.$ of the brackets
that describe the quantum geometry over the D3- and the M5-branes
are constant, completely antisymmetric matrices; see
\cite{Chu-Smith:09}.
\medskip

\subsection{The geometry of the $A_4$ BLG model CS term}
\label{sec:CSterm}
{\ }

 Let $\fG=A_4$ explicitly. Its three-bracket is
given by  $[\be_a,\be_b,\be_c] = -{\epsilon_{abc}}^d \be_d$,
(eq.~\eqref{3-basis}) where indices
are raised and lowered with the euclidean metric. We also
know that the Lie${\,\fG}$=Lie$A_4$ bracket has the generic form of
eq.~\eqref{3Cs} with structure constants given by eq.~\eqref{3Csb}
since $A_4$ is simple. Further, there is a one-to-one
correspondence between the fundamental objects $\mathscr{X}_{ab}$
of $A_4$ and the $A_4$ endomorphisms $\mathscr{X}_{ab}\cdot =
ad_{\mathscr{X}_{ab}}\,$ that determine Lie${\,A_4}=so(4)$.
Thus, for notational simplicity, we may use here
$\mathscr{X}_{ab}$ to denote the elements of $so(4)$ as well
(Ex.~\ref{ex:LieA4}),
the commutation relations of which are expressed in the familiar
form \eqref{orthog} using the dual basis.

 To introduce $SO(4)$ gauge fields, we start
 from the $so(4)$ Maurer-Cartan equations written
 for one-forms $\omega^{ab}$ dual to the $so(4)$ basis elements
 $\mathscr{X}_{ab}=-\mathscr{X}_{ba}$,
 \begin{equation}
 \label{iCS2}
     \omega^{a_1a_2}(\mathscr{X}_{b_1b_2})=
     \delta^{[a_1}_{b_1} \delta^{a_2]}_{b_2} \quad \mathrm{or} \quad
     \omega^A(\mathscr{X}_B) = \delta^A_B \ .
 \end{equation}
 The capital letters $A,B$ label the pairs $(a_1 a_2)\,,\, (b_1 b_2)$
 with $a_1<a_2\,$, $b_1<b_2$. Then,
 the MC equations \eqref{3MC} for Lie${\,A_4}$ read
 \begin{eqnarray}
 \label{iCS3}
    d\omega^{c_1c_2} & = & d\omega^C = -\frac{1}{2} {C_{AB}}^C
    \omega^A \wedge \omega^B \nonumber\\
    & = & -\frac{1}{2} \frac{1}{4} {C_{a_1a_2b_1b_2}}^{c_1c_2}
    \omega^{a_1a_2}\wedge \omega^{b_1b_2} \nonumber
    \; = \; \frac{1}{4} {\epsilon_{a_1a_2b_1}}^{[c_1}
    \delta^{c_2]}_{b_2} \omega^{a_1a_2}\wedge \omega^{b_1 b_2}\; .
\end{eqnarray}
The gauge fields $A^{c_1c_2}$ are the $so(4)$-valued connection
one-forms obtained by `softening' the MC forms $\omega^{c_1c_2}$
(all forms are assumed to be defined on the appropriate manifolds,
the coordinates of which are omitted). The Cartan structural equations
then provide the curvatures or field strengths $F$,
\begin{eqnarray}
\label{iCS4}
    F^{c_1c_2} & = & F^C = (dA+A\wedge A)^C = dA^C  +\frac{1}{2} {C_{AB}}^C
    A^A \wedge A^B \nonumber\\
    & = & dA^{c_1c_2}+\frac{1}{2} \frac{1}{4} {C_{a_1a_2b_1b_2}}^{c_1c_2}
    A^{a_1a_2}\wedge A^{b_1b_2} \nonumber\\
    & = &dA^{c_1c_2} -\frac{1}{4} {\epsilon_{a_1a_2b_1}}^{[c_1}
    \delta^{c_2]}_{b_2} A^{a_1a_2}\wedge A^{b_1b_2}\ .
\end{eqnarray}
These curvatures are covariant under the Lie${\,A_4}$ gauge
transformations of $A$, which are given by
\begin{eqnarray}
\label{iCS7}
   \delta_\lambda A &=& {\mathcal D}\lambda= d\lambda + [A,\lambda]\ ,\nonumber\\
   \delta_\lambda A^{c_1c_2} &=& d\lambda^{c_1c_2}+\frac{1}{4} {C_{a_1a_2b_1b_2}}^{c_1c_2}
    A^{a_1 a_2} \lambda^{b_1 b_2}\; .
\end{eqnarray}

The covariant derivatives of zero-forms such as the matter
fields, objects of the form $V=V^a\be_a$ where the basis $\{\be_a\}$
carries a representation of the gauge group,
are defined easily. In the BLG model the $V$'s denote
bosonic or fermionic worldvolume fields, $V=X^I\,,\,\Psi$.
Under Lie${\,A_4}$, $V$ transforms by
\begin{eqnarray}\label{iCS5}
    \delta V &=& \delta V^d \be_d = \frac{1}{2} (-\lambda^{ab}) [\be_a,\be_b,\be_c]
    V^c \nonumber\\
    &=& \frac{1}{2} (-\lambda^{ab}) \mathscr{X}_{ab} \cdot \be_c V^c
    \equiv (-\lambda) \cdot V \nonumber\\
    &=& \frac{1}{2} \lambda^{ab} {\epsilon_{abc}}^d V^c \be_d \quad \Rightarrow \quad
     \delta V^d= \frac{1}{2} \lambda^{ab} {\epsilon_{abc}}^d V^c  \; ,
\end{eqnarray}
where the dot is the usual adjoint action (see eq.~\eqref{xidot}),
and the $\lambda^{ab}$ parameter is $-2$ times the one that
appears in the BLG literature because there the $1/2$ factor is
not added. Then, the covariant derivative of the matter fields
is defined by
\begin{equation}\label{iCS6}
   {\mathcal D}V=dV+A\cdot V\ ,
\end{equation}
where $A=A^A \mathscr{X}_A= \frac{1}{2}A^{ab} \mathscr{X}_{ab}$.
Note that the gauge fields used here are minus twice the ones used
in the original BL and subsequent papers, because their covariant
derivative is defined with a minus sign and there is no compensating
$1/2$ factor in the BL definition of the components of the gauge
field.
\medskip

The above gauge fields (connection forms) and field strengths
(curvatures) may now be used to construct a Chern-Simons (CS)
three-form by the standard Chern-Weil theorem. First, an
invariant four-form $P(F)$ is introduced with the help of an
invariant symmetric tensor, $k_{AB}=k_{a_1a_2b_1b_2}=k_{BA}$,
\begin{equation}
\label{iCS8}
   P(F)= k_{AB} F^A\wedge F^B=\frac{1}{4}k_{a_1a_2b_1b_2} F^{a_1a_2}\wedge
   F^{b_1b_2} \equiv H \; ,
\end{equation}
denoted $H$ for short (not need of worrying here about
factors).  The Lie$\,{\fG}$-invariance of the
polynomial guarantees that the four-form $H$ is closed and, since
gauge free differential algebras are contractible, that
$H=d\Omega$. The {\it Chern-Simons three-form} $\Omega$ for
the symmetric invariant polynomial $k_{AB}$ is given by
\begin{equation}
\label{CSgral}
\Omega  = k_{AB}\left( A^A\wedge dA^B +\frac{2}{3} A^A\wedge (A\wedge A)^B\right) \quad ,
\nonumber
\end{equation}
so that
\begin{eqnarray}
\label{iCS9}
 \Omega & = & k_{AB}\left( A^A\wedge dA^B +\frac{1}{3} {C_{CD}}^B
 A^A\wedge A^C \wedge A^D\right) \nonumber\\
 &=& \frac{1}{4} k_{a_1a_2b_1b_2}\left( A^{a_1a_2}\wedge dA^{b_1b_2} + \frac{1}{3}
 \frac{1}{4} {C_{c_1c_2d_1d_2}}^{b_1b_2}
 A^{a_1a_2}\wedge A^{c_1c_2} \wedge A^{d_1d_2}\right)\ .
\end{eqnarray}

  We now look for rank two $SO(4)$-invariant polynomials; they were given in
Sec.~\ref{inv-ten}. The first and obvious one $k^{(1)}_{AB}$ is the
Killing metric (eq.~\eqref{assoc-Casim}):
\begin{equation}
\label{iCS10}
k^{(1)}_{a_1a_2b_1b_2} = \delta_{AB} =\delta_{b_1[a_1} \delta_{a_2]b_2} \;.
\end{equation}
Additionally, there is also the possibility of taking
(eq.~\eqref{2nd-inv})
\begin{equation}
\label{iCS11}
    k^{(2)}_{a_1a_2b_1b_2} = \epsilon_{a_1a_2b_1b_2} = k^{(2)}_{b_1 b_2 a_1 a_2 }\;.
\end{equation}
This independent metric exists because $so(4)$ is not simple
(see Table in Sec.~\ref{sec:compGtable}). The above $k^{(1)}$ and
$k^{(2)}$ are obviously invariant; to see it explicitly,
it suffices to check (eq.~\eqref{kinva}) with the
$C$'s for $A_4$ that
\begin{equation*}
\label{iCS12}
    C_{c_1c_2d_1d_2a_1a_2}= k_{a_1a_2b_1b_2}{C_{c_1c_2d_1d_2}}^{b_1b_2}
\end{equation*}
is antisymmetric under the interchange of $(a_1a_2)$ and
$(d_1d_2)$ (or of $(a_1a_2)$ and $(c_1c_2)$).

   The CS term that appears in the action of the
BLG model, eq.~\eqref{BLM4},  is the one obtained by using $k^{(2)}$
in eq.~\eqref{iCS11} \cite{Gustav:08} {\it i.e.}, the polynomial
that does not admit a $so(n)$ generalization\footnote{The
$A_4$-based BLG model which is considered here selects the Lie
algebra $so(4)$. It is easy to see why the polynomial \eqref{iCS11}
does not generalize for arbitrary $n$. The existence of a metric
determined by the fully skewsymmetric tensor of a
($n+1$)-dimensional space would require $2(n-1)=n+1$, hence $n=3$.
This is the obstruction which prevents moving from the original BLG
$SO(4)$ model to a $SO(n+1)$ one.}. Indeed, using the results of
Sec.~\ref{sec:coord3alg}, we obtain
\begin{equation}
\label{iCS15}
    \Omega^{(2)} = \frac{1}{4}\epsilon_{a_1a_2b_1b_2} \left( A^{a_1a_2}\wedge
    dA^{b_1b_2} - \frac{1}{3} {\epsilon_{c_1c_2d_1}}^{b_1} A^{a_1a_2}
    \wedge A^{c_1c_2} \wedge A^{d_1b_2}\right)\ ,
\end{equation}
which coincides, up to a global factor, with the CS term
\eqref{BLM4} of the BLG model \cite{Ba-La:07a,Ba-La:08}, once the
second term of (\ref{iCS15}) is multiplied by
the factor $-2$ that appears when moving to the gauge fields
commonly used in the BLG literature.

   Let us now see explicitly that
the CS terms for $k^{(1)}$ and $k^{(2)}$ are the sum and the
difference, respectively, of the CS terms for the two $so(3)$
components of $so(4)=so(3)\oplus so(3)$. Define
\begin{equation}\label{iCS16}
    \omega_\pm^1 = \omega^{14}\pm \omega^{23} \ ,\ \omega_\pm^2 =
    \omega^{24}\pm \omega^{31}\ ,\ \omega_\pm^3 = \omega^{34}\pm
    \omega^{12}  \qquad   \left( \omega_\pm^i = \omega^{i4} \pm
    \frac{1}{2}\epsilon_{iab4} \omega^{ab} \right) \; ,
\end{equation}
$i=1,2,3$. This new basis splits explicitly
$so(4)$ into its two (plus and minus) $so(3)$
components since, using (\ref{iCS3}), we find,
\begin{equation}\label{iCS17}
    d\omega^1_\pm = - \omega^2_\pm \wedge \omega^3_\pm \ ,\quad
  d\omega^2_\pm = - \omega^3_\pm \wedge \omega^1_\pm \ ,\quad
  d\omega^3_\pm = - \omega^3_\pm \wedge \omega^1_\pm \; ,
\end{equation}
which are the MC equations of two $so(3)$ copies in the
standard basis. In this basis the two $so(3)$ gauge
fields  $A^i_\pm$ and corresponding curvatures $F^i_\pm$ are given by
\begin{equation}
\label{iCS18}
    A_\pm^i = A^{i4} \pm
    \frac{1}{2}\epsilon_{iab4} A^{ab}\ ,\quad F_\pm^i = F^{i4} \pm
    \frac{1}{2}\epsilon_{iab4} F^{ab}  \; .
\end{equation}

   To show that $\Omega^{(1,2)}$,
the two CS forms obtained by using the invariant
polynomials $k^{(1,2)}$ respectively,
can be written in terms of two $so(3)$
CS three-forms for the $A^i_+$ and $A^i_-$ gauge
fields, we write
the invariant four-forms $H^{(1)}$ and $H^{(2)}$ in
terms of
\begin{equation}\label{iCS19}
    H_\pm = \sum^3_{i=1} F^i_\pm \wedge F^i_\pm = d\Omega_\pm\; .
\end{equation}
Consider first $H^{(1)}$. Using (\ref{iCS8}) and
(\ref{iCS10}) we obtain
\begin{eqnarray}\label{iCS20}
    H^{(1)} &=& \frac{1}{4} \delta_{b_1[a_1}
    \delta_{a_2]b_2} F^{a_1a_2}\wedge F^{b_1b_2} =
    \frac{1}{2} F^{a_1a_2}\wedge F_{a_1a_2} \nonumber\\
    &=&  F^{12}\wedge F_{12}+ F^{13}\wedge F_{13}+F^{14}\wedge F_{14} +
    F^{23}\wedge F_{23}+ F^{24}\wedge F_{24} + F^{34}\wedge F_{34}  \; .
\end{eqnarray}
Similarly, $H^{(2)}$ is given by
\begin{equation}\label{iCS21}
    H^{(2)}= \frac{1}{4} \epsilon_{a_1a_2a_3a_4} F^{a_1a_2}\wedge F^{a_3a_4} =
    2(F^{12}\wedge F^{34} + F^{13}\wedge F^{42} + F^{14}\wedge
    F^{23})\; .
\end{equation}
Now, computing the sum and the difference of $H_+$ and $H_-$
in eq~\eqref{iCS19},
\begin{eqnarray}
\label{iCS22}
H_+\pm H_- & = & (F^{14}+F^{23}) \wedge (F^{14}+F^{23}) +
(F^{24}+F^{31}) \wedge (F^{24}+F^{31}) \nonumber\\
& &+ (F^{34}+F^{12}) \wedge
(F^{34}+F^{12}) \nonumber\\
& & \pm (F^{14}-F^{23}) \wedge (F^{14}-F^{23}) \pm (F^{24}-F^{31})
\wedge (F^{24}-F^{31})\nonumber\\ & &\pm (F^{34}-F^{12}) \wedge
(F^{34}-F^{12}) \; ,
\end{eqnarray}
it is seen that for the plus sign the crossed terms cancel and the
squares add, giving
\begin{equation}\label{iCS23}
  \frac{1}{2} (H_+ + H_-) = H^{(1)} \; .
\end{equation}
For the minus sign it is instead the crossed terms that survive
so that
\begin{equation}
\label{iCS24}
  \frac{1}{2}(H_+ - H_-) = H^{(2)}\ .
\end{equation}

Therefore, the CS term for the Killing form $k^{(1)}$
is the sum of two ordinary $so(3)$ CS terms, whereas the `twisted'
CS term \eqref{iCS15} for $k^{(2)}$ is given
by their difference \cite{Be-Ta-Tho:08, Raam:08}. The relative minus
sign in $\Omega_+ - \Omega_-$ solves the problem of parity invariance
\cite{Ban-Li-Sch:08}, since the odd parity of the standard CS term is
compensated by requiring that parity interchanges the two
different $SU(2)$ gauge fields so that $H_+\leftrightarrow H_-$,
something that could not be done if the CS term were the single
one associated with a simple group.
{\ }

\section{A BLG model based on the Nambu bracket infinite-dimensional FA $\fN$}
\label{sec:BLG+NB}
{\ }

  So far we have described the BLG model having in mind that the
vector space of the FA $\fG$ is finite-dimensional and, actually,
the four-dimensional $A_4$. As we have seen, the $A_4$-based BLG
first model is equivalent to one based on a semisimple Lie group,
$SU(2)\times SU(2)$, and describes only two M2-branes. Although, as
mentioned in Sec.~\ref{sec:BLG-phys}, there have been other
proposals to overcome this limitation, the problem of finding an
action describing an arbitrary number $N$ of coincident M2-branes
cannot be considered fully closed (for instance, the ABJM proposal
\cite{Aha-Be-Ja-Mal:08} -which does not require an $n$-Lie algebra-
has only $\mathcal{N}=6$ manifest supersymmetry, see
Sec.~\ref{sec:BLGmod}).
 The positive definiteness of the metric of the 3-Lie algebra led
to the $A_4$, $SU(2)\times SU(2)$ BLG first model of the previous
section but, still keeping the positivity, there is one more option:
one may move to infinite-dimensional FAs.

  Let us then consider the model that results by replacing the previous
$A_4$ FA by the infinite-dimensional $n=3$ FA based on the Nambu
bracket \cite{Nambu:73} of functions (Ex.~\ref{NPcomp}).
This leads to a BLG-Nambu bracket (BLG-NB) model
describing the low energy limit of a `condensate' of $N\rightarrow
\infty$ M2-branes. To derive the BLG-NB action
\cite{Ho-Im-Mat-Shi:08,Ba-To:08} we shall take the original BLG one
in eq.~\eqref{BLM2} and replace the $A_4$ FA by the
infinite-dimensional Nambu algebra $\mathfrak{N}$ of functions on a
(compact) three-dimensional manifold $M_3$ that will be identified
with $S^3$ (the isometry group of which is $SO(4)$) for the BLG-NB
model, although we will often maintain a generic $M_3$ notation. The
transition from $A_4$ to $\fN$ will be achieved by replacing the
$A_4$ three-brackets by Nambu brackets, and the sums over the $a$
indices by sums over infinite, discrete indices $\mathfrak{a}$,
following as much as possible the finite-dimensional BLG $A_4$ case.
The use of the Nambu bracket in the context of the BLG model was
already mentioned in \cite{Ba-La:07b}, initiated in
\cite{Ho-Mat:08,Ho-Im-Mat-Shi:08} and studied in general in
\cite{Ba-To:08} (see also \cite{Ba-To:08a}). The novelty introduced
by the Nambu FA $\mathfrak{N}$ is that, in this case,
Lie$\,\mathfrak{N}$ turns out to be the infinite-dimensional Lie
algebra of the volume-preserving diffeormorphisms group of $M_3$,
{\it SDiff}$\,(M_3)$ ($S$ for `special' or volume preserving).

These (rigid) diffeomorphisms, when they are made
local by making them to depend on the $d=3$ spacetime variables,
become characterized by functions $\xi$ of six coordinates
($x^\mu, y^i$), where $x^\mu$ are the spacetime variables,
$\mu=1,2,3$, and $y^i, i=1,2,3$, are the coordinates of $y\in M_3$.
The fact that the matter fields take values in the infinite vector space
where the adjoint derivatives act, will imply that they have the
form $\phi(x,y)$, $\phi = (X,\Psi)$. It turns out that that the
BLG-NB Lagrangian provides a {\it gauge} theory of volume preserving
diffeormorphisms (of $M_3$) which has $\mathcal{N}=8$ supersymmetry
and superconformal invariance. It has been conjectured
\cite{Ba-To:08} that the $N\rightarrow \infty$ limit of the ABJM
model \cite{Aha-Be-Ja-Mal:08} might lead to the $\mathcal{N}=8$
supersymmetric BLG-NB one, in which case the ABJM model would be a
discretization of the Nambu bracket approach to be considered below.

  The study of branes as gauge theories of volume preserving
diffeormorphisms goes back to the work of Hoppe
\cite{Hop:82,Hop:88,Hoppe:96} who found that the full
diffeomorphisms group for membranes with the topology of a sphere
may be thought of as the $N\rightarrow \infty$ limit of $SU(N)$. The
case of the supermembranes was taken up in \cite{dW-Ho-Ni:88},
where it was shown that the light cone gauge-fixed action could be
considered equivalent to that of a super-Yang-Mills theory of
`symplectic' or area preserving diffeomorphisms of the membrane
surface. This is so because the $N\rightarrow \infty$ limit
\cite{Hop:88} of $SU(N)$ Yang-Mills theories may be considered as
leading to {\it SDiff}$\,(M_2)$ gauge theories in which the Lie
algebra commutator becomes the Poisson bracket of functions on $M_2$
\cite{Flo-Il-Tik:88}. This was extended to general $p>2$
super-$p$-branes in \cite{Be-Se-Ta-To:90}, where it was shown that
these may be considered as `exotic' gauge theories (but see
Sec.~\ref{BLG-NB-CS}), the gauge group being {\it SDiff}$\,(M_p)$.
It is then natural to look for volume preserving diffeomorphisms
gauge actions using general Nambu-Poisson brackets
(Sec.~\ref{sec:NP}): after all, the Poisson bracket is simply the
$n=2$ Nambu-Poisson one.

  But, before discussing the BLG-NB model, let us
go back to the Nambu algebra as an infinite-dimensional FA.

\subsection{The Nambu algebra $ \mathfrak{N} $ as an infinite-dimensional
3-Lie algebra} \label{Nambu-inf-3-FA}
{\ }

Let $M_3$ be a three-dimensional compact oriented manifold without
boundary, $M_3=S^3$. Let $y^i=(y^1,y^2,y^3)$ be local coordinates on
$M_3$, and let $\mu(y)= e(y) dy^1\wedge dy^2\wedge dy^3$ be the
volume form on $M_3$. A volume preserving diffeomorphism is defined
by a vector field $\xi(y)=\xi^i(y)\partial_i$ on $M_3$ such that
$L_\xi \,\mu=0$ where $L_\xi$ is the Lie derivative; this implies
that the components $\xi^i$ of $\xi\in${\it sdiff}$\,(M_3)$ satisfy
the condition $\partial_i (e\xi^i)=0$.
  By eq.~\eqref{Nam-vol}, the Nambu bracket of three
functions $\phi_1(y),\phi_2(y),\phi_3(y)\in \mathscr{F}(M_3)$ is
locally given by
\begin{equation}
\label{NBBLG1} \{ \phi_1(y), \phi_2(y), \phi_3(y) \} = e^{-1}(y)
\epsilon^{ijk} \partial_i \phi_1(y) \partial_j \phi_2(y)
\partial_k \phi_3(y) \ ,
\end{equation}
where $\epsilon^{ijk}$ is defined from $\epsilon^{123}=1$; the
scalar density $e$ is not necessarily derived from a metric ({\it
i.e.}, that $e=\sqrt{g}$), since no metric on $M_3$ is assumed.
However, the Nambu algebra itself is taken to be metric in the sense
of Sec.~\ref{sec:NB-metric-al}. For the antisymmetric symbol with
indices down we take $\epsilon_{ijk}$ given by $\epsilon_{123}=e(y)$
as in \cite{Ba-To:08}; then, $\epsilon^{i_1 i_2 i_3}\epsilon_{j_1
j_2 j_3}= e\,\epsilon^{i_1 i_2 i_3}_{j_1 j_2 j_3}\,$.

Since $M_3$ is compact we shall assume that there is a complete set
of functions $\mathbf{e}_{\mathfrak{a}}(y)$, where $\mathfrak{a}$
denotes a set of discrete indices, that is orthonormal with respect
to the metric \eqref{Nam-metric} such that, for a scalar function
$\phi \in \mathscr{F}(M_3)$,
\begin{equation}
\label{NBBLG2}
    \phi(y) = \sum_{\mathfrak{a}} \phi^{\mathfrak{a}}
    \mathbf{e}_{\mathfrak{a}}(y) \
    ,
\end{equation}
which is completed with
\begin{equation}
\label{NBBLG2b}
    < \mathbf{e}_\mathfrak{a}(y),\mathbf{e}_\mathfrak{b}(y)> =
    \delta_{\mathfrak{a}\mathfrak{b}} \ ,\quad \phi^\mathfrak{a} =
    < \phi(y), \mathbf{e}_\mathfrak{a}(y)> \ ,\quad
    <\phi_1(y),\phi_2(y)> = \int_{M_3}
    \mu(y) \phi_1(y) \phi_2(y) \ ,
\end{equation}
as in \eqref{Nam-metric}, together with the `resolution of the
identity' on $\mathscr{F}(M_3)$
\begin{equation}
\label{NBNEW1}
    \sum_\mathfrak{a} \mathbf{e}_\mathfrak{a}(y)
    \mathbf{e}_\mathfrak{a}(y') = \delta^3(y,y') \ ,\quad \int_{M_3}
    \mu(y) \phi(y) \delta^3(y,y') = \phi(y')  \ .
\end{equation}
We do not need specifying $\delta^3(y,y')$; it is sufficient to know
that it exists for a compact $M_3$ as $S^3$ (for instance, the
spherical harmonics satisfy $\sum^\infty_{l=1} \sum_{m=-l}^l
Y_{lm}(\theta,\phi) Y_{lm}(\theta',\phi')^* =
\delta^2(\theta,\theta';\phi,\phi') = \delta(\cos \theta- \cos
\theta') \delta(\phi-\phi')\,)$. For $S^3$, which is the $SU(2)$
group manifold, eq. \eqref{NBBLG2} just expresses the Peter-Weyl
theorem for the harmonic analysis on compact groups with $y\in S^3
\sim SU(2)$. For higher dimensional spheres $S^p,\;p>3$ (no longer
group manifolds), such an expression is related to general
hyperspherical harmonics expansions. The case of $S^3$ has been
discussed in \cite{Dow:90}; that {\it e.g.} of the spheres
$S^{3A-4}$, where $A$ is the mass number, appeared very long ago in
the $n$-body problem in nuclear physics (see {\it e.g.}
\cite{Lo-Gal:72}; see also \cite{Ramg:01} in the context of fuzzy
spheres).

 The linearity of the bracket implies, after
replacing the vectors $\phi$ by their expansions in
eq.~\eqref{NBBLG2}, that the above three-bracket is also given by
\begin{equation}
\label{NBBLG4}
    \{ \phi_1(y), \phi_2(y), \phi_3(y) \} = \sum_{\mathfrak{a}
    \mathfrak{b}\mathfrak{c}} \phi_1^\mathfrak{a}(y)
    \phi_2^\mathfrak{b}(y) \phi_3^\mathfrak{c}(y)   \{
    \mathbf{e}_\mathfrak{a}(y), \mathbf{e}_\mathfrak{b}(y),
    \mathbf{e}_\mathfrak{c}(y) \} \ .
\end{equation}
 The structure constants of this infinite-dimensional FA,
referred to the basis $\{\be_\mathfrak{a}\}$ are, by definition,
the coefficients ${f_{\mathfrak{a}\mathfrak{b}
\mathfrak{c}}}^\mathfrak{d}$ that appear in the expression
\begin{equation}
\label{NBBLG5}
    \{\mathbf{e}_\mathfrak{a}(y), \mathbf{e}_\mathfrak{b}(y),
    \mathbf{e}_\mathfrak{c}(y) \} = \sum_\mathfrak{d} {f_{
    \mathfrak{a} \mathfrak{b} \mathfrak{c}}}^\mathfrak{d}
    \mathbf{e}_\mathfrak{d}(y) \ .
\end{equation}
Therefore, using eqs.~\eqref{NBBLG2}, \eqref{NBBLG5} and
\eqref{NBBLG1} we find that the structure constants of the Nambu
algebra $\mathfrak{N}$ in the $\{\be_\mathfrak{a}\}$ basis are
given by
\begin{equation}
\label{NBBLG6}
    f_{\mathfrak{a} \mathfrak{b} \mathfrak{c} \mathfrak{d}} =
    < \{\mathbf{e}_\mathfrak{a}, \mathbf{e}_\mathfrak{b},
    \mathbf{e}_\mathfrak{c} \},\mathbf{e}_\mathfrak{d}>
    =\int_{M_3} \mu(y) e^{-1}(y) \epsilon^{ijk} \partial_i
    \mathbf{e}_\mathfrak{a}(y) \partial_j
    \mathbf{e}_\mathfrak{b}(y) \partial_k
    \mathbf{e}_\mathfrak{c}(y) \mathbf{e}_\mathfrak{d}(y)
\end{equation}
which exhibits an obvious ($\mathfrak{a}, \mathfrak{b},
\mathfrak{c}$) antisymmetry which becomes a full one
in ($\mathfrak{a}, \mathfrak{b},
\mathfrak{c}, \mathfrak{d}$) on account of
partial integration. Since the basis is orthonormal,
there is no distinction between the up and down
indices in $f$.

\subsection{The BLG-NB model}
{\ }

That some forms of $p$-brane actions can be formulated by using
Nambu-Poisson brackets has been know for some time
\cite{Hoppe:96,Hop:02}. This is so because the determinant of the
induced metric $g$ that appears in the original actions can be
rewritten in terms of Nambu $(p+1)$-brackets. Here we shall restrict
ourselves to the Nambu bracket realization of the BLG theory, which
may be viewed as the low energy limit
 of a condensate of nearly coincident M2-branes.
To construct its lagrangian in a direct way, we have to
extend the $A_4$ BLG one to the Nambu algebra $\fN$ case (see below
eq.~\eqref{Nam-n-bra}) which, in essence, means that
we have to look at the consequences of replacing the previous
finite range index $a=1,\dots,\mathrm{dim}\,\fG$ by the discrete
infinite  $\mathfrak{a}$. One might think of constructing the
lagrangian below using a ($n>$3)-dimensional manifold $M_3$ and
the corresponding $n$-Nambu-Poisson bracket, but there are
difficulties for $n\geq 4$ \cite{Ba-To:08}. From
the present point of view, $n=3$ is also
natural by an extension of the dimensional arguments of the
finite 3-Lie algebra case. As a result, the BLG-NB theory should
be an infinite-dimensional Nambu bracket algebra version of the BLG
model, in which the CS term would be constructed from the Nambu algebra
analogue of the invariant polynomial $k^{(2)}$
in eq.~\eqref{iCS11}.

  To show this, let us begin with the matter
fields appearing in the BLG-NB model. These depend on the
three-dimensional worldvolume Minkowski coordinates $x^\mu
=(x^0,x^1,x^2)$ as before but now include the
$\mathfrak{N}$-algebra basis index $\mathfrak{a}$. Therefore, as
vectors in $\mathfrak{N}$ depending on the worldvolume
coordinates, these fields have the coordinate expansions
\begin{equation}
\label{NBBLG7}
    X^I(x,y) = \sum_\mathfrak{a} X^{I\mathfrak{a}}(x)
    \mathbf{e}_\mathfrak{a}(y)\ , \quad \Psi^I(x,y) = \sum_\mathfrak{a}
    \Psi^{I\mathfrak{a}}(x)
    \mathbf{e}_\mathfrak{a}(y)\ ,
\end{equation}
in which the sum over the index $a$ for $A_4$ has been replaced by a
sum over the set of indices $\mathfrak{a}$ for $\mathfrak{N}$. The
remaining fields of the BLG-NB model are the vector spacetime
fields. Their components will be given by $A_\mu(x)^{\mathfrak{a}
\mathfrak{b}}= -A_\mu(x)^{\mathfrak{b} \mathfrak{a}}$, in parallel
with $A_\mu^{ab}(x)=-A_\mu^{ab}(x)$ for the finite case. This is
because the $A_\mu$ fields depend on the indices that characterize
the fundamental objects in $\fN\wedge \fN$, hence their double,
antisymmetric $\mathfrak{a} \mathfrak{b}$ indices. {\ }

\subsubsection{The BLG-NB lagrangian}
{\ }

   To construct the BLG-NB lagrangian \cite{Ba-To:08,Ho-Im-Mat-Shi:08}
an invariant metric on $\mathfrak{N}$ is needed. It was seen in
Sec.~\ref{sec:NB-metric-al} that the Nambu algebra may be made
metric. Using the scalar product \eqref{Nam-metric}  the BLG
lagrangian in eq.~\eqref{BLM2}, when the Filippov bracket is taken
to be the Nambu one, leads to
\begin{eqnarray}
\label{NBBLG7a}
  \mathcal{L}_{BL-NB}(x) & = &
   \int_{M_3} d^3y\, e(y)\left( -\frac{1}{2}\mathcal{\mathcal{D}}_\mu X^I(x,y)
  \mathcal{\mathcal{D}}^\mu X^I(x,y)
  + \frac{i}{2}  \bar{\Psi}(x,y) \Gamma^\mu \mathcal{D}_\mu \Psi (x,y)
  \right.
  \nonumber\\
   & & -g\, \frac{i}{4}  \left\{ \bar{\Psi}(x,y), X^I(x,y), X^J(x,y)\right\}
   \Gamma_{IJ} \Psi(x,y) \nonumber\\
   & &-\left. \frac{g^2}{2\cdot 3!}
   \left\{
   X^I(x,y),X^J(x,y),X^K(x,y)\right\} \left\{
   X^I(x,y),X^J(x,y),X^K(x,y)\right\} \right) \nonumber\\
  & & + \frac{1}{g}\,\mathcal{L}_{CS}\quad ;
\end{eqnarray}
the explicit form of $\mathcal{L}_{CS}$ will be derived below.

  The covariant derivative of the matter $\phi(x,y)$ fields in
eq.~\eqref{NBBLG7a} is given by ({\it cf.} eq.~\eqref{BLM3})
\begin{eqnarray}
\label{NBBLG7b}
   \mathcal{D}_\mu \phi(x,y) &=& \sum_\mathfrak{d} \mathcal{D}_\mu
   \phi^\mathfrak{d}(x) \mathbf{e}_\mathfrak{d}(y) \nonumber\\
   &=& \sum_\mathfrak{d}\left( \partial_\mu \phi^\mathfrak{d}(x) -
   \sum_{\mathfrak{a}\mathfrak{b}\mathfrak{c}} {f_{
    \mathfrak{a} \mathfrak{b} \mathfrak{c}}}^\mathfrak{d}
    A_\mu^{\mathfrak{a} \mathfrak{b}}(x) \phi^\mathfrak{c}(x)
    \right) \mathbf{e}_\mathfrak{d}(y) \; =\; \partial_\mu \phi(x,y)-\nonumber\\
    & -& \sum_{\mathfrak{a} \mathfrak{b}
    \mathfrak{c}\mathfrak{d}} \epsilon^{ijk} \int_{M_3} \mu(y')
    e^{-1}(y')
    \partial_{y'^i} \mathbf{e}_\mathfrak{a}(y') \partial_{y'^j}
    \mathbf{e}_\mathfrak{b}(y') \partial_{y'^k}
    \mathbf{e}_\mathfrak{c}(y') \mathbf{e}_\mathfrak{d}(y')
    A_\mu^{\mathfrak{a} \mathfrak{b}}(x)
    \phi^\mathfrak{c}(x) \mathbf{e}_\mathfrak{d}(y) \nonumber\\
    &=& \partial_\mu \phi(x,y)- e^{-1}(y) \epsilon^{ijk}
     \sum_{\mathfrak{a} \mathfrak{b}} \partial_i
     \mathbf{e}_\mathfrak{a}(y) \partial_j
     \mathbf{e}_\mathfrak{b}(y)A_\mu^{\mathfrak{a}
     \mathfrak{b}}(x) \partial_k \phi(x,y)\nonumber\\
  &\equiv & \partial_\mu \phi(x,y) + s^k_\mu(x,y) \partial_k
  \phi(x,y)\ ,
\end{eqnarray}
where we have used \eqref{NBNEW1} and  expression \eqref{NBBLG6}
for the structure constants of the Nambu algebra to compute the
term multiplying $\partial_k\phi$ above. This is given by
\begin{equation}
\label{NBBLG11n}
    s^k_\mu(x,y) = - e^{-1}(y) \epsilon^{ijk}
     \sum_{\mathfrak{a} \mathfrak{b}} \partial_i
     \mathbf{e}_\mathfrak{a}(y) \partial_j
     \mathbf{e}_\mathfrak{b}(y)A_\mu^{\mathfrak{a}
     \mathfrak{b}}(x)\ .
\end{equation}
The above expression may be equivalently written as
\begin{equation}
\label{NBBLG7c}
 s^k_\mu(x,y)= e^{-1}(y) \epsilon^{ijk} \partial_i A_{j\mu} \quad
 \textrm{with} \quad A_{j\mu}(x,y):= - \sum_{\mathfrak{a} \mathfrak{b}}
     \mathbf{e}_\mathfrak{a}(y) \partial_j
     \mathbf{e}_\mathfrak{b}(y)A_\mu^{\mathfrak{a}
     \mathfrak{b}}(x) \quad ,
\end{equation}
where $A_j$ is globally defined since $M_3=S^3$. Also,
\begin{equation}
\label{NBBLG7C}
    \sum_{\mathfrak{a} \mathfrak{b}} \partial_i
     \mathbf{e}_\mathfrak{a}(y) \partial_j
     \mathbf{e}_\mathfrak{b}(y)A_\mu^{\mathfrak{a}
     \mathfrak{b}}(x) = -\frac{1}{2} \epsilon_{ijk} s^k_\mu(x,y) \
     , \quad \partial_{[i} A_{j]\mu}(x,y)= \epsilon_{ijk} s^k_\mu(x,y)
\end{equation}
on account of the antisymmetry in $i,j$ of the
$l.h.s.$

   Summarizing, the covariant derivative of the matter fields is
given by
\begin{equation}
\mathcal{D}_\mu\,\phi=(\partial_\mu+s_\mu{}^i\partial_i)\phi \qquad
i.e., \qquad  \mathcal{D}:= d+s=d + s^i\partial_i \quad ,
\end{equation}
where $d$ acts on spacetime forms; the form of $\mathcal{D}$ is thus
a consequence of the form of the $\fN$ structure constants in
eq.~\eqref{NBBLG7b}. Thus, the spacetime one-form $s=
s^i\partial_i$, $s^i(x,y)=s^i_\mu(x,y)dx^\mu$, plays the r\^ole of a
gauge connection. Since its explicit expression follows from that of
$A_{i\mu}$ in eq.~\eqref{NBBLG7c}, the spacetime one-form $s^i(x,y)$ defines
$A_j(x,y)$ up to a total $\partial_i$ derivative: $A_j(x,y)$ and
$A'{}_j(x,y)= A_j(x,y) +\partial_j \alpha(x,y)$ for some one-form
$\alpha(x,y)$ lead to the same $s$. As the
true {\it gauge field} is $s_\mu^i(x,y)$, $A_{i\mu}$ plays the
r\^ole of a {\it pre-potential} \cite{Ba-To:08}.

  Eq.~\eqref{NBBLG7c} implies that the components
$e s_\mu{}^i$ of the spacetime one-form $es^i$ satisfy the
condition
\begin{equation}
\label{conn-closed}
\partial_i (e s_\mu{}^i)=0 \quad .
\end{equation}
Hence the gauge field  defines
vector fields $s_\mu{}^i\partial_i$ on $M_3$ that are (local) volume
preserving diffeomorphisms and $s=s_\mu dx^\mu$ is a {\it
sdiff}$\,(M_3)$-valued connection one-form. Since
two pre-potentials $A'^{\mathscr{a}\mathscr{b}}(x)$
and $A^{\mathscr{a}\mathscr{b}}(x)$ lead to the same $s^i_\mu(x,y)$
if they are `pre-gauge' related, the freedom in choosing
$A_i(x,y)$ corresponds to taking elements
$A^{\mathfrak{a}\mathfrak{b}}(x)$
in $\wedge^2 \fN$ that differ in one that belongs to the
kernel of the $ad$ mapping, since this difference
will produce a zero $s$.

   Let us now write the CS part of the BLG-NB action. This
is obtained by moving from the $\mathcal{L}_{CS}(x)$ in
\eqref{BLM4} to the Nambu case. Using the structure constants of
$\mathfrak{N}$, the resulting CS lagrangian is the spacetime
three-form $L_{CS}(x)=\mathcal{L}_{CS}(x) d^3x$ given by
\begin{eqnarray}
\label{NBBLG7d}
     L_{CS}(x) &=& \frac{1}{2} \sum_{\mathfrak{a} \mathfrak{b}
     \mathfrak{c} \mathfrak{d}} \left( {f_{\mathfrak{a} \mathfrak{b}
     \mathfrak{c}}}^\mathfrak{d} A^{\mathfrak{a} \mathfrak{b}}(x)
     \wedge d A^{\mathfrak{c} \mathfrak{d}}(x) + \frac{2}{3}
     \sum_{\mathfrak{e} \mathfrak{f} \mathfrak{g}} {f_{\mathfrak{c} \mathfrak{d}
     \mathfrak{a}}}^\mathfrak{g} f_{\mathfrak{e} \mathfrak{f} \mathfrak{g}
     \mathfrak{b}} A^{\mathfrak{a} \mathfrak{b}}(x) \wedge A^{\mathfrak{c}
     \mathfrak{d}}(x) \wedge A^{\mathfrak{e}
     \mathfrak{f}}(x)\right) \nonumber\\
     &=& -\frac{1}{2} \int_{M_3} \mu(y) \left( s^k(x,y)\wedge
     dA_k(x,y) - \frac{1}{3} \epsilon_{ijk} s^i(x,y) \wedge s^j(x,y)
     \wedge s^k(x,y)\right) \, ,
\end{eqnarray}
where eq.~\eqref{NBBLG7C} has been used and $A^{\mathfrak{a}
\mathfrak{b}} (x)=A^{\mathfrak{a} \mathfrak{b}}_\mu(x)dx^\mu$, and
$s$ are spacetime one-forms; recall that $\mu(y)=e(y)\,d^3y$ and
that the structure constants $f_{\mathfrak{a} \mathfrak{b}
\mathfrak{c} \mathfrak{d}}$ (eq.~\eqref{NBBLG6}) are fully
antisymmetric.

 Note that $L_{CS}$ is not entirely written in terms of the gauge field $s^i(x,y)$: it also
requires the pre-potential $A_i(x,y)$ defined by the last equation
in \eqref{NBBLG7c}. For this reason, the above $L_{CS}$ was called
`CS-like' in \cite{Ba-To:08}; its structure will be exhibited in
Sec.~\ref{BLG-NB-CS}. $L_{CS}$ is of course well defined because
(pre-gauge) related pre-potentials $A{'}_k, A_k$ lead to the same
$L_{CS}$ due to condition \eqref{conn-closed}.
{\ }

\subsubsection{Gauge and supersymmetry transformations of
the fields of the BLG-NB model}
{\ }

The gauge and supersymmetry transformations of the fields in the
BLG-NB action can be readily obtained by looking at
eqs.~\eqref{BLM5}, \eqref{BLM6}. First, we notice that the
derivations determined by Nambu brackets with one void entry as
expressed in terms of the fundamental objects of $\mathfrak{N}$ come
from coefficients $\lambda^{\mathfrak{a}
\mathfrak{b}}=-\lambda^{\mathfrak{b} \mathfrak{a}}$, where the
$\mathfrak{a}, \mathfrak{b}$ antisymmetry reflects that of the
fundamental objects of the Nambu algebra (as the $cd$ skewsymmetry
in eq.~\eqref{BLM5}). These (rigid) infinitesimal Lie algebra
transformations are made local  by making $\lambda^{\mathfrak{a}
\mathfrak{b}}$ to depend on the $d=3$ Minkowski coordinates,
$\lambda^{\mathfrak{a} \mathfrak{b}}\rightarrow\lambda^{\mathfrak{a}
\mathfrak{b}}(x)$. Using the expression of the $\fN$ structure
constants, it is found that the gauge transformations are actually
determined by local functions $\xi(x,y)$ on the $M_3$ manifold. In
fact, we demonstrate below that the variations of the fields under
the gauge transformations of parameter $\lambda^{\mathfrak{a}
\mathfrak{b}}(x)$ are local {\it SDiff}$\,(M_3)$
transformations
\begin{eqnarray}
\label{NBBLG7f}
   \delta X^I(x,y) &=& -\xi^k(x,y) \partial_k X^I(x,y) \nonumber \\
   \delta \Psi(x,y) &=& -\xi^k(x,y) \partial_k \Psi(x,y)  \nonumber \\
   \delta s^i(x,y) &=& d\xi^i(x,y) - \xi^j(x,y) \partial_j s^i(x,y) +
   \partial_j \xi^i(x,y) s^j(x,y)\; ,
\end{eqnarray}
determined by $\xi^k(x,y)$.

   To see how $\xi(x,y)$ appears, let us compute
$\delta\phi$ to show how the original $(\mathfrak{a},
\mathfrak{b})$ dependence of $\lambda$ leads to a $y$ dependence
of $\xi$. This will also identify $\xi(y)$ as a volume preserving
diffeomorphism. The variation is given by the adjoint derivative
$\sum_{\mathfrak{a} \mathfrak{b}}\lambda^{\mathfrak{a}
\mathfrak{b}}(x) \{ \be_\mathfrak{a}(y), \be_\mathfrak{b}(y), \;\;
\}\,$; with $\phi=X,\Psi$ we obtain from eqs. \eqref{NBBLG2} and
\eqref{NBBLG6}
\begin{eqnarray}
  \label{NBBLG8}
    \delta \phi(x,y) &=& \sum_{\mathfrak{a} \mathfrak{b}}
    \lambda^{\mathfrak{a} \mathfrak{b}}(x) \{
    \be_\mathfrak{a}(y), \be_\mathfrak{b}(y) , \phi(x,y)
    \}
    = \sum_{\mathfrak{a} \mathfrak{b} \mathfrak{c} \mathfrak{d}}
    \lambda^{\mathfrak{a}
    \mathfrak{b}}(x) {f_{\mathfrak{a} \mathfrak{b}
     \mathfrak{c}}}^\mathfrak{d} \be_\mathfrak{d}(y)
     \phi^\mathfrak{c}(x) \nonumber\\
     &=& \sum_{\mathfrak{a} \mathfrak{b} \mathfrak{c} \mathfrak{d}}
     \lambda^{\mathfrak{a}
    \mathfrak{b}}(x) \int_{M_3} \mu(y') e^{-1}(y') \epsilon^{ijk} \partial_{y'^i}
    \be_\mathfrak{a}(y') \partial_{y'^j}
    \be_\mathfrak{b}(y') \partial_{y'^k}
    \be_\mathfrak{c}(y') \be_\mathfrak{d}(y')
    \be_\mathfrak{d}(y) \phi^\mathfrak{c}(x) \nonumber\\
    &=& e^{-1}(y) \epsilon^{ijk} \sum_{\mathfrak{a} \mathfrak{b}}
    \lambda^{\mathfrak{a} \mathfrak{b}}(x) \partial_i
    \be_\mathfrak{a}(y) \partial_j
    \be_\mathfrak{b}(y) \partial_k \phi(x,y) \ .
    \end{eqnarray}
Therefore,
\begin{equation}
\label{NBBLG8ad}
  \delta\phi(x,y)  = -  \xi^k(x,y) \partial_k \phi(x,y)  \ ,
\end{equation}
with
\begin{equation}
\label{NBBLG11}
 \xi^k(x,y)=-e^{-1}(y) \epsilon^{ijk}
\sum_{\mathfrak{a} \mathfrak{b}}
\partial_i \be_\mathfrak{a}(y)  \partial_j \be_\mathfrak{b}(y)
\lambda^{\mathfrak{a} \mathfrak{b}}(x) \quad \Rightarrow
\partial_k(e(y)\xi^k(x,y))=0
\end{equation}
({\it cf.} eq.~\eqref{NBBLG11n}) so that  $\xi\in${\it sdiff}$\,(M_3)$
and is local since
$\lambda^{\mathfrak{a} \mathfrak{b}}= \lambda^{\mathfrak{a}\mathfrak{b}}(x)$.

  The last expression in eq.~\eqref{NBBLG7f} for the gauge fields
follows similarly. The extension of the corresponding
formula in eq.~\eqref{BLM5} to the Nambu algebra case reads
\begin{equation}
\label{GV1}
   \delta \left( \sum_{\mathfrak{c} \mathfrak{d}} {f_{\mathfrak{c}
 \mathfrak{d}\mathfrak{b}}}^\mathfrak{a} A_\mu^{\mathfrak{c}
\mathfrak{d}}(x) \right) = \partial_\mu \left( \sum_{\mathfrak{c}
\mathfrak{d}} {f_{\mathfrak{c}
 \mathfrak{d}\mathfrak{b}}}^\mathfrak{a} \lambda^{\mathfrak{c}
\mathfrak{d}}(x) \right)
 + 2 \sum_{\mathfrak{c} \mathfrak{d}\mathfrak{e} \mathfrak{f}
\mathfrak{g}} {f_{\mathfrak{c}
 \mathfrak{d}\mathfrak{b}}}^\mathfrak{a}
{f_{\mathfrak{e}
 \mathfrak{f}\mathfrak{g}}}^\mathfrak{c} \lambda^{\mathfrak{d}
\mathfrak{g}}(x) A_\mu^{\mathfrak{e} \mathfrak{f}}(x) \ .
\end{equation}
Let us now compute the expression between brackets  in the
$l.h.s.$ of \eqref{GV1}. It is given by
\begin{eqnarray}
\label{GV2}
   & & \sum_{\mathfrak{c} \mathfrak{d}} \int_{M_3} \mu(y) e^{-1}(y) \epsilon^{ijk}
\partial_i \be_{\mathfrak{c}}(y)
\partial_j \be_{\mathfrak{d}}(y) \partial_k
\be_{\mathfrak{b}}(y) \be_{\mathfrak{a}}(y)
A_\mu^{\mathfrak{c} \mathfrak{d}}(x) \nonumber\\
 & &\quad\quad\quad  = -\int_{M_3} \mu(y)  s_\mu^k(x,y) \partial_k
\be_{\mathfrak{b}}(y) \be_{\mathfrak{a}}(y)
  = -\langle s_\mu^k(x) \partial_k \be_{\mathfrak{b}}\, ,
\, \be_{\mathfrak{a}} \rangle \quad ,
\end{eqnarray}
where we have used \eqref{NBBLG11n}. Similarly, the expression
between brackets appearing in the first term of the $r.h.s.$ of
\eqref{GV1} is given by
\begin{equation}
\label{GV3}
     -\langle \xi^k(x) \partial_k \be_{\mathfrak{b}},
\be_{\mathfrak{a}} \rangle \ .
\end{equation}

We now compute the second term of the $r.h.s.$ of \eqref{GV1}.
After integrating by parts with respect to $y^i$ in the integral
corresponding to ${f_{\mathfrak{c}
 \mathfrak{d}\mathfrak{b}}}^\mathfrak{a} $ and
using \eqref{NBNEW1}, it gives
\begin{equation}
\label{GV4}
 -2 \epsilon^{ijk} \epsilon^{rst} \int_{M_3} \mu(y) e^{-1}(y)
 \sum_{\mathfrak{d}\mathfrak{g}\mathfrak{e}\mathfrak{f}}\partial_j
\be_\mathfrak{d}(y) \partial_k \be_\mathfrak{b}(y)
\partial_i \be_\mathfrak{a}(y) \partial_r
\be_\mathfrak{e}(y) \partial_s \be_\mathfrak{f}(y)
\partial_t \be_\mathfrak{g}(y) \lambda^{\mathfrak{d}
\mathfrak{g}}(x) A_\mu^{\mathfrak{e} \mathfrak{f}}(x)\ .
\end{equation}
  Now we use \eqref{NBBLG11n} and
\begin{equation}
\label{GV8}
 \sum_{\mathfrak{a} \mathfrak{b}} \partial_i \be_\mathfrak{a}(y)
 \partial_j \be_\mathfrak{b}(y)
 \lambda^{\mathfrak{a}\mathfrak{b}}(x) = -\frac{1}{2}
 \epsilon_{ijk} \xi^k(x,y) \ ,
\end{equation}
which follows from eq.~\eqref{NBBLG11} since
$ \lambda^{\mathfrak{a}\mathfrak{b}}= -\lambda^{\mathfrak{b}\mathfrak{a}}$
(see the analogue for $A^{\mathfrak{a} \mathfrak{b}}_\mu(x)$ and
$s_\mu(x,y)$ in eq.~\eqref{NBBLG7C}). Then, \eqref{GV4} is equal to
\begin{eqnarray}
 \label{GV9}
   & &  \int_{M_3}  \mu(y) e^{-1}(y) \epsilon^{ijk} \epsilon_{jtl} \partial_i
   (e(y) \xi^l(x,y) s^t(x,y)) \partial_k \be_\mathfrak{b}(y)
   \be_\mathfrak{a}(y)\nonumber\\
   & & \quad\quad\quad = - \int_{M_3}  \mu(y) (\partial_i \xi^k(x,y)
   s_\mu^i(x,y)-\xi^i(x,y)\partial_i s_\mu^k(x,y) )  \partial_k
   \be_\mathfrak{b}(y)
   \be_\mathfrak{a}(y)
\nonumber\\
& & \quad\quad\quad = - \langle \, (\partial_i \xi^k(x)
   s_\mu^i(x)-\xi^i(x)\partial_i s_\mu^k(x) )  \partial_k
   \be_\mathfrak{b}\, , \,
   \be_\mathfrak{a} \, \rangle \ ,
\end{eqnarray}
using in the last equality that $\partial_i(e(y)s_\mu^i(x,y))=0\,$
and $\,\partial_i(e(y) \xi^i(x,y))=0$. Substituting
\eqref{GV2}, \eqref{GV3} and \eqref{GV9} in \eqref{GV1}, we finally
arrive at
\begin{equation}
\label{GV10}
   \langle\, \left( \delta s^k_\mu(x) - \left(\partial_\mu \xi^k(x) +
   \partial_i \xi^k(x) s_\mu^i(x) - \xi^i(x) \partial_i
   s^k_\mu(x)\right) \right) \partial_k \be_\mathfrak{b}\; ,\;
   \be_\mathfrak{a}\; \rangle = 0
\end{equation}
which, since $\mathfrak{a}$ and $\mathfrak{b}$ are arbitrary,
reproduces the gauge transformations $\delta s^i$ of the gauge
fields in eq.~\eqref{NBBLG7f}.

 Thus, the infinitesimal gauge transformations for
the matter $\phi=(X,\Psi)$ and gauge $s^i$ fields in
the BLG-NB model have all the expected standard form
\begin{equation}
\label{NBBLG13}
     \delta\phi= -L_\xi \phi  \quad ,\quad \delta s^i =d\xi + [s,\xi] \quad
     \mathrm{with}\quad \xi=\xi^i\partial_i\; ,\quad \partial_i(e\xi^i)=0 \quad \Rightarrow
     \xi\in \textit{sdiff}\,(M_3) \quad
\end{equation}
and $\delta (\mathcal{D}\phi)= -L_\xi(\mathcal{D}\phi)$ guarantees the
invariance of the action under local diffeomorphisms. We have thus
recovered the known fact \cite{Hop:88,Hoppe:96} that the adjoint
action or derivation defined by the Nambu three-bracket generates
local, $M_3$-volume preserving diffeomorphisms.  As a result, the
BLG-NB action defines a {\it SDiff}$\,(S^3$) gauge theory. The field
strength is given by
\begin{eqnarray}
\label{NBBLG13a}
F=\mathcal{D}s := ds+\frac{1}{2}[s,s] \quad  & , & F=F^i\partial_i \quad ,\nonumber \\
F^i(x,y) = d s^i(x,y) + s^j(x,y) \wedge \partial_j s^i (x,y) &
\mathrm{with}&
    \quad \partial_i (e(y) F^i(x,y))=0 \; ,
\end{eqnarray}
where $d$ acts on the spacetime functions. Under gauge
transformations, $\delta F= [F,\xi]$.

   Let us look again at the $\delta \phi$ in eq.~\eqref{NBBLG8}
to analyze further the gauge transformations. These are adjoint
transformations -derivations- given in terms of the fundamental
objects
$\mathscr{X}_{\be_\mathfrak{a}\be_\mathfrak{b}}=(\be_\mathfrak{a}(y),\be_\mathfrak{b}(y))$
(as in \eqref{(n-1)-ad-def}) with local parameters
$\lambda^{\mathfrak{a} \mathfrak{b}}(x)$, and determine elements of
the gauged Lie$\,\mathfrak{N}$ by eq.~\eqref{NBBLG11}. Given an
arbitrary $\xi^k(x,y)$ satisfying $\partial_k(e\xi^k)=0$ we may find
a $\lambda^{\mathfrak{a} \mathfrak{b}}(x)$ that generates it. The
fact that it is the functions $\xi^k(x,y)$ rather than the
$\lambda^{\mathfrak{a} \mathfrak{b}}(x)$ that determine Lie$\,\fN$
illustrates again that, for general FAs, different fundamental
objects may induce the same element in InDer$\fG$=Lie$\,\fG$ (see
\cite{Da-Tak:97} in the present Nambu algebra context). The elements
of Lie$\,\fN$ are obtained by taking the quotient by the kernel of
the $ad$ map, here characterized by (see \eqref{NBBLG8})
\begin{eqnarray}
\label{NBBLG12a}
 & & \sum_{\mathfrak{a} \mathfrak{b}}
\lambda^{\mathfrak{a} \mathfrak{b}}(x)
   (\be_\mathfrak{a},\be_\mathfrak{b}) \in \mathrm{ker}\,ad \; \Leftrightarrow  \;
   \sum_{\mathfrak{a} \mathfrak{b}}
\lambda^{\mathfrak{a} \mathfrak{b}}(x)
   \{\be_\mathfrak{a}(y),\be_\mathfrak{b}(y), \be_\mathfrak{c}(y)\}=0
   \quad \forall \be_\mathfrak{c}(y) \nonumber \\
& & \Leftrightarrow \quad  \xi^k=-e^{-1}(y) \epsilon^{ijk}
   \sum_{\mathfrak{a} \mathfrak{b}} \partial_i \be_\mathfrak{a}(y)
   \partial_j\be_\mathfrak{b}(y) \lambda^{\mathfrak{a} \mathfrak{b}}(x) = 0\; .
\end{eqnarray}
Thus, all elements of $\wedge^2\mathfrak{N}$ determined by
$\lambda^{\mathfrak{a} \mathfrak{b}}$s that satisfy the above
condition produce the trivial diffeomorphism of $M_3$ or zero element of
Lie$\,\mathfrak{N}$={\it sdiff}$\,(M_3)$.

  We conclude this section with the supersymmetry
transformations of the BLG-NB fields. These
may be similarly  found, and are given by
\begin{eqnarray}
\label{NBBLG7g}
    \delta_\epsilon X^I(x,y) &=& i \bar{\epsilon} \Gamma^I
    \Psi(x,y) \nonumber\\
   \delta_\epsilon \Psi(x,y) &=& {\mathcal{D}}_\mu X^I(x,y)
   \Gamma^\mu \Gamma^I \epsilon - \frac{g}{3!} \{ X^I(x,y),
   X^J(x,y), X^K(x,y) \} \Gamma^{IJK} \epsilon \nonumber \\
   \delta_\epsilon s^k_\mu(x,y) &=&  -i g\,e^{-1}(y) \epsilon^{ijk}
   \bar{\epsilon} \Gamma_\mu \Gamma^I \partial_i X^I(x,y)
   \partial_j \Psi(x,y)   \; ,
\end{eqnarray}
where again the last equality follows from a calculation similar
to that leading to  \eqref{GV10}.
{\ }

\subsubsection{Structure of the BLG-NB Chern-Simons term}
\label{BLG-NB-CS}
{\ }

  As mentioned, the $L_{CS}$ piece includes both the gauge
field $s^i$  and pre-potential $A_k$ one-forms
\cite{Ba-To:08}. This is also reflected in the four-form which is
obtained by taking the exterior differential of $L_{CS}$ in
\eqref{NBBLG7d}, which is given by
\begin{equation}
\label{NBBLG14}
  d L_{CS} = -\frac{1}{2} \int_{M_3} \mu(y) F^i(x,y) \wedge G_i(x,y) \ ,
\end{equation}
where $G_i$ and $F^i$ are given by
\begin{equation}
\label{NBBLG15}
  G_i(x,y) =
dA_i(x,y) -\frac{1}{2} \epsilon_{ijk} s^j(x,y) \wedge s^k(x,y)\ ,
\quad  F^i(x,y)=e^{-1}(y) \epsilon^{ijk} \partial_j G_k(x,y) \; ;
\end{equation}
$G_i$ is a `pre-field strength' two-form,
$G_i=dA_i+\frac{1}{2}s^j\wedge \partial_{[j}A_{i]}$. In spite of the mixed
appearance of the gauge field and its pre-potential in $L_{CS}$ as well
as that of the field and pre-field strengths in $dL_{CS}$, we show
below that the $L_{CS}$ term of the BLG-NB model may be constructed
using only the curvature $F$ and a suitable symmetric bilinear
(metric) on the relevant Lie algebra, Lie$\,\mathfrak{N}$={\it
sdiff}$\,(M_3$), much as the standard CS odd forms are obtained from
the finite Lie algebra invariant symmetric polynomials and the
curvature of the connection.

We saw in Sec.~\ref{sec:CSterm} that the CS term of the euclidean
$A_4$ BLG lagrangian was in fact an ordinary CS three-form obtained
form the polynomial in eq.~\eqref{iCS11} given
by the structure constants with all indices down, eq.~\eqref{2nd-inv}.
In the present infinite-dimensional situation, the analogue of $k^{(2)}$ in
\eqref{2nd-inv} is given by the  structure constants of the Nambu FA
$\fN$, namely
\begin{eqnarray}
\label{NBCS1b}
    k^{(2)}((\be_\mathfrak{a},\be_\mathfrak{b}),(\be_\mathfrak{c},\be_\mathfrak{d}))
    &=& f_{\mathfrak{a} \mathfrak{b} \mathfrak{c} \mathfrak{d}} =
      f_{\mathfrak{c} \mathfrak{d} \mathfrak{a} \mathfrak{b}} = \nonumber \\
      &=&  \epsilon^{ijk} \int_{M_3} \mu(y) e^{-1}(y) \partial_i \be_\mathfrak{a}(y)
      \partial_j \be_\mathfrak{b}(y)
      \partial_k  \be_\mathfrak{c}(y)  \be_\mathfrak{d}(y)   \; .
\end{eqnarray}
This metric is however degenerate and its kernel is given by the fundamental
objects of $\fN$ determined by $\lambda$'s that satisfy
\begin{equation}
k^{(2)} ( \sum_{\mathfrak{a} \mathfrak{b}} \lambda^{\mathfrak{a}
\mathfrak{b}}(x) (\be_\mathfrak{a} ,\be_\mathfrak{b})\,,\,
(\be_\mathfrak{c},\be_\mathfrak{d}) ) = 0 \quad \quad \forall \,
\be_\mathfrak{c},\be_\mathfrak{d} \quad .
\end{equation}

It may be checked that the above condition is equivalent to eq.
\eqref{NBBLG12a}, which determines the kernel of the $ad$ map as
defined there. Thus, taking the quotient of the space $\fN\wedge
\fN$ by ker$\,ad$, we conclude that $k^{(2)}$ is well defined on
{\it sdiff}$\,(M_3)$.
   Let us then introduce, using our $k^{(2)}$ polynomial, the
four-form $P(F)$
\begin{equation}
\label{NBCS2}
 P(F)= k^{(2)}(F,F) \; ,
\end{equation}
where we denote by $F$ the {\it sdiff}$\,(M_3)$-valued curvature
corresponding to $F^i(x,y)$ given by \eqref{NBBLG13a}. As argued,
the choice of the $\wedge^2\fN$ representative for $F$ does not
have any effect on $P(F)$ and we may use for $F(x,y)$ any
$\wedge^2\fN$-valued
\begin{equation}
\label{NBCS3}
         F(x,y)=\sum_{\mathfrak{a}
\mathfrak{b}} F^{\mathfrak{a} \mathfrak{b}}(x)\;
(\be_\mathfrak{a}(y),\be_\mathfrak{b}(y))
\end{equation}
 provided that $F^{\mathfrak{a} \mathfrak{b}}(x)$ gives rise to $F^i$,
\begin{equation}
\label{NBCS4}
    F^i(x,y) = - e^{-1} \epsilon^{ijk} \sum_{\mathfrak{a}
\mathfrak{b}} \partial_j \be_\mathfrak{a}(y) \partial_k
\be_\mathfrak{b}(y) F^{\mathfrak{a} \mathfrak{b}}(x) \; ,
\end{equation}
since then $F(x,y)\,\phi(x,y)=F^{\mathfrak{a} \mathfrak{b}}(x)
\sum_{\mathfrak{c}}\,
\{\be_\mathfrak{a}(y) ,\be_\mathfrak{b}(y) ,\be_\mathfrak{c}(y)\}\,
\phi^\mathfrak{c}(x) = F^i (x,y)\partial_i\phi(x,y)$
using eq.~\eqref{NBBLG6}. In this way,
we move from any $\wedge^2\fN$-valued representative $F(x,y)$
to $F(x,y)=F^i(x,y)\partial_i$, which is $\mathit{sdiff}(M_3)$-valued
since $\partial_j(e F^j)=0$.
Inserting \eqref{NBCS3} into \eqref{NBCS2} we compute $P(F)$ to be
\begin{eqnarray}
\label{NBCS5}
    P(F) &=& \sum_{\mathfrak{a} \mathfrak{b} \mathfrak{c}
    \mathfrak{d}} f_{\mathfrak{a} \mathfrak{b} \mathfrak{c}
    \mathfrak{d}} F^{\mathfrak{a} \mathfrak{b}}(x) \wedge
    F^{\mathfrak{c} \mathfrak{d}}(x)\nonumber \\
    &=& \epsilon^{ijk} \int_{M_3} \mu(y) e^{-1}(y) \sum_{\mathfrak{a}
    \mathfrak{b} \mathfrak{c}
    \mathfrak{d}} \partial_i
    \be_\mathfrak{a}(y) \partial_j \be_\mathfrak{b}(y) \partial_k
    \be_\mathfrak{c}(y) \be_\mathfrak{d}(y)
    F^{\mathfrak{a} \mathfrak{b}}(x) \wedge
    F^{\mathfrak{c} \mathfrak{d}}(x) \nonumber\\
 &=& - \int_{M_3} \mu(y) F^k(x,y) \wedge G_k(x,y)
\; ,
\end{eqnarray}
with (see eqs.~\eqref{NBBLG15}, \eqref{NBCS4})
$G_k(x,y) = - \sum_{\mathfrak{a} \mathfrak{b}} \be_\mathfrak{a}(y)
\partial_k \be_\mathfrak{b}(y) F^{\mathfrak{a} \mathfrak{b}}(x)$.
It then follows that $dL_{CS}$ in \eqref{NBBLG14} may be written as $dL_{CS}=
\frac{1}{2}P(F)$ and hence has the standard Chern-Weil expression for
the infinite $\mathit{sdiff}(M_3)$ algebra. Thus, the BLG-NP theory
is in this sense an ordinary (rather than `exotic', {\it cf.} \cite{Ba-To:08})
$\mathit{sdiff}(M_3)$ gauge theory: as in the $A_4$ case,
the three-form lagrangian $L_{CS}(x)$ on $d=3$ Minkowski space
given by eq.~\eqref{NBBLG7d} is obtained from an invariant
polynomial $P(F)$ on the curvature.

As the BLG CS term, the CS term of the BLG-NP model
is parity even \cite{Ba-To:08}, since the parity change in the
three-dimensional spacetime can be compensated
by a `parity flip' in the `internal' $M_3$ three-space.
The BLG-NB model, eq.~\eqref{NBBLG7a},
also leads to a BPS solution corresponding to field
 configurations $X^I(x^2,y)$ that determine a
three-sphere\footnote{The authors of \cite{Ba-To:08} assumed implicitly
that the maps $X$ of $M$ ($S^3$) on the unit 3-sphere had
degree one; one might think of the physical consequences of
having other degrees, see eq.~\eqref{degree8}.}
the radius of which goes to infinity as $x^2$ goes to zero. This is
consistent with the idea that the fuzzy sphere solution of the
B-H equation should become \cite{Ba-La:07b} a smooth one in
the $N \rightarrow\infty$ limit, in agreement with the fact that the
BLG-NP model, which is a $d=3$ superconformal $\mathcal{N}=8$
supersymmetric theory, describes a condensate of M2-branes
\cite{Ba-To:08}.

  A superfield formulation of the BLG-NP model based on a
set of eight scalar superfields constrained by a superembeding-like
equation has been given in \cite{Ban:08}. We shall not discuss this
nor comment on the possible connection between the BLG-NB model and
the large $N$ limit of the ABJM proposal \cite{Aha-Be-Ja-Mal:08},
and refer to \cite{Ba-To:08} instead. We conclude by mentioning that
the BLG-NB model has also been conjectured to describe a single
M5-brane in the strong (constant) three-form field of $D=11$
supergravity, an interpretation proposed in
\cite{Ho-Mat:08,Ho-Im-Mat-Shi:08} and further developed in
\cite{Pa-Sa-So-To:09} and \cite{Fur:10}.
\medskip

\section{$n$-ary structures: a brief physical outlook}
{\ }\\

We have described in this review the general structure of GLAs and
FAs (and other related structures) as well as some possible
applications in physics, with special emphasis in the Filippov
algebras. In general, the $n$-bracket of a FA is not defined through
a certain combination of associative products of its entries
and, in fact, this is the reason that makes it difficult to give
{\it e.g.}, matrix realizations of FAs. The simple $n>2$ FAs
are rather few, far less
than those of the $n=2$ Cartan classification: in fact, just one
type of FA for each $n>2$ if we ignore the signature of the
($n+1$)-dimensional metric real vector space on which the simple
$n$-Lie algebras may be constructed (or we consider complex simple
FAs).

A question that  immediately arises is whether there is a Filippov
`group-like' structure associated with the $n$-Lie algebra one
(besides the obvious Lie group associated with Lie$\,\fG$) {\it
i.e.}, whether there is an integrated version of FAs of which these
would be the linear approximation. The answer to this question is
unknown; in fact, as far as we are aware, the problem of finding a
Filippov `group-like' manifold associated to general FAs has not
even been discussed. The case of the ($n>2$)-Leibniz generalizations
should be even more difficult to tackle, since already for $n=2$ it
gives rise to the {\it coquecigrue} problem of Loday's algebras
mentioned in the main text. It is completely unclear whether such a
notion exists in general.

 For GLAs, the characteristic identity of the GLA bracket,
the GJI, is a necessary result of the associativity of the
composition of the elements in the (even) multibracket. We have seen
that there is an infinite number of examples which may be
constructed from the non-trivial cocycles for the cohomology  of the
simple compact Lie algebras. Although these simple
GLAs take advantage of the existence of an underlying
Lie group manifold, there is still room for other
examples. But, in general, one might argue that a consequence of the
present analysis is the `rigidity' of the ordinary Lie algebra
structure with respect to any possible $n$-ary generalizations.
These entail losing different, but essential, parts of the
properties associated with Lie algebras (as reflected by the Y-type
bifurcation that leads either to the GJI or to the FI when $n>2$)
and which the physical world seems to like.

It is well appreciated that mathematics is full of developments
prompted by or related to advances in physics. It may be that the
full physical usefulness of these higher order $n$-ary algebras
lays ahead. Nevertheless, we have already seen that, as far as
specific applications to generalized mechanical systems
are concerned (through their associated $n$-ary
General Poisson and Nambu-Poisson structures),
it is fair to say that there are not so many.
At the same time, the quantization of $n$-ary Poisson
structures is fraught with the difficulties discussed in this
review. This is not an isolated case; there have been other
mathematically interesting attempts to quantization, as {\it e.g.}
geometric quantization, which, albeit geometrically attractive,
have not met much success from the physical/practical point of view.
As for the applications of $n$-Lie algebras to problems in M-theory,
those associated with the BLG model have motivated the
present renewed interest in Filippov algebras. Some specific
types of 3-Leibniz algebras might be relevant here since,
as discussed, the full anticommutativity of the finite ($n$=3)-Lie
algebras is too restrictive.  Although, as already mentioned, there
exist alternatives to using three-algebras, it seems that these may nevertheless
provide a natural way to encode various desired symmetry
properties of the theory, at the root of which is the delicate
interplay between the three-algebras and their associated Lie
algebras. There is also, of course, the infinite-dimensional Nambu FA BLG-NP
version of the BLG model described in the previous section, which
retains full anticommutativity for its (Nambu) 3-bracket.

Having said that, it is worth recalling that the relation between
mathematics and physics is as deep as full of surprises, as
exhibited {\it e.g.} by the unexpected and recent application of
mathematical aspects of M-theory holography to condensed matter
physics (see \cite{Har:09,Horo:10} for reviews). Thus, and in spite
of the fact that the initial hopes have not been fulfilled, it
might still turn out that $n$-ary structures in general and Filippov
and Filippov-like algebras in particular have come to M-theory
physics to stay, although it is quite unclear at present
that this will be so.

But, in any case, $n$-ary algebraic structures are of course
interesting by themselves.
\medskip

\section{Appendix 1: two forms of the Filippov identity}
We present here the proof of the equivalence of the expressions
\eqref{FIshort}, \eqref{FIultrashort-a} of the FI for general $n$, using the
anticommuting ghosts $b^a$ and $c^a$ of eq.~\eqref{FIultrashort-a}.

Let us show first that the form \eqref{FIshort} of the FI
implies \eqref{FIultrashort-a}. Eq.~\eqref{FIshort}
may be rewritten as
\begin{eqnarray}
\label{FISpr1}
  & &\left[ X_{b_1},\dots ,  X_{b_{n-2}}, X_{b_{n-1}}, [X_{a_1}, \dots
 X_{a_n} ] \right]\nonumber\\
 & &\quad\quad \quad\quad = \sum^n_{l=1}
 \left[ X_{b_1},\dots ,  X_{b_{n-2}}, X_{a_l}, [X_{a_1}, \dots ,
 X_{a_{l-1}},X_{b_{n-1}}, X_{a_{l+1}}, \dots ,
 X_{a_n} ] \right]\ .
\end{eqnarray}
Now we contract this equation with $c^{b_1}\dots c^{b_k}
b^{b_{k+1}} \dots b^{b_{n-2}} b^{b_{n-1}} b^{a_1} \dots b^{a_k}
c^{a_{k+1}}\dots c^{a_n}$, for $k=0,\dots , n-2$. This results in
the following $n-1$ equations:
\begin{eqnarray}
 \label{FISpr2}
  & & [ C, \mathop{\dots}\limits^{k}, C, B,
   \mathop{\dots}\limits^{n-2-k}, B,B, [ B,
   \mathop{\dots}\limits^{k}, B, C, \mathop{\dots}\limits^{n-k}
   ,C] ]  = \nonumber\\
   & & \quad\quad\quad \quad -k [ C, \mathop{\dots}\limits^{k}, C, B,
   \mathop{\dots}\limits^{n-2-k}, B ,B ,[ B,
   \mathop{\dots}\limits^{k}, B, C,
   \mathop{\dots}\limits^{n-k}
   ,C] ] \nonumber\\
   & &\quad\quad\quad \quad -(n-k) [ C, \mathop{\dots}\limits^{k}, C, B,
   \mathop{\dots}\limits^{n-2-k}, B,C [ B,
   \mathop{\dots}\limits^{k}, B,B, C, \mathop{\dots}\limits^{n-k-1}
   ,C]] \ ,
 \end{eqnarray}
 or
\begin{equation}
\label{FISpr3}
  (k+1)R_k = -(n-k) R_{k+1} \; \quad (k=0,\dots, n-2)\ ,
\end{equation}
where we have introduced
\begin{equation}
\label{FISpr4}
  R_k \equiv [ C, \mathop{\dots}\limits^{k}, C, B,
   \mathop{\dots}\limits^{n-1-k}, B, [ B,
   \mathop{\dots}\limits^{k}, B, C, \mathop{\dots}\limits^{n-k}
   ,C]] \ , \quad (k=0,\dots ,n-1)\ .
\end{equation}
The recurrence \eqref{FISpr3} has the solution
\begin{equation}
\label{FISpr5}
  R_k = (-1)^k \frac{k! (n-k)!}{n!} R_0 \; \quad (k=1,\dots, n-1)
\end{equation}
which, taking $k=n-1$, implies $ R_0=n (-1)^{n-1} R_{n-1}$. But this
is the FI written in the form \eqref{FIultrashort-a} because $R_0$
is the bracket at its $l.h.s.$ and $(-1)^{n-1}R_{n-1}$
is the one at the $r.h.s.$ (for intance, for a 3-Lie algebra, it gives
$[B,B,[C,C,C]]=3 [C,C,[B,B,C]]$, eq.~\eqref{FIultrashort-a}).

Conversely, let us assume that the FI in the form
eq.~\eqref{FIultrashort-a} is satisfied. We use it on the set of
$n-1$ double brackets ($k=0,\dots n-2$)
\begin{equation}
\label{FISpr6}
  P_k \equiv
  [B,\mathop{\dots}\limits^{k},B,C,\mathop{\dots}\limits^{n-1-k},C,
  [B,\mathop{\dots}\limits^{n-2-k},B,
  C,\mathop{\dots}\limits^{k+2},C]] \ ,
\end{equation}
and obtain that each of the $n-1$ $P_k$ may be expressed as
\begin{eqnarray}
\label{FISpr7}
 P_k &=&  (n-2-k) [[
 B,\mathop{\dots}\limits^{k},B,C,\mathop{\dots}\limits^{n-1-k},
 C,B],
 B,\mathop{\dots}\limits^{n-3-k},B,C,\mathop{\dots}\limits^{k+2},
 C] \nonumber\\
 & & + (-1)^{(n-1)(n-2-k)} (k+2) [
 B,\mathop{\dots}\limits^{n-2-k}, B, [B,
 \mathop{\dots}\limits^{k}, B , C,\mathop{\dots}\limits^{n-1-k}
 , C], C,\mathop{\dots}\limits^{k+1}, C] \; .
 \end{eqnarray}
 Reordering the entries in the $r.h.s.$ so that parts in it are
identified with $P$s, this gives
\begin{equation}
\label{FISpr8}
  P_k =
  (-1)^{n-1} (n-2-k) P_{n-3-k} + (-1)^{n-1} (k+2) P_{n-2-k} \ ,
  \quad k=0,\dots , n-2 \;  .
\end{equation}
This system of equations actually implies that all $P_k$ vanish and,
in particular, that $P_{n-2}=0$, which is equivalent to
\eqref{FIshort}. To see it, first notice that the above equation
gives, for $k=n-2$,
\begin{equation}
\label{FISpr9}
  P_{n-2} =
    (-1)^{n-1} n P_0  \;   .
\end{equation}
Secondly, inserting the expressions of $P_{n-3-k}$ and $P_{n-2-k}$
in the $r.h.s.$ of \eqref{FISpr8} for $k=0,\dots,n-3$ leads
to the recurrence relation
\begin{equation}
\label{FISpr10}
  P_k = (n-2-k) [(k+1) P_k + (n-k-1) P_{k+1}] + (k+2) [kP_{k-1} +
  (n-k)P_k]   \;  .
\end{equation}
It can be checked that its solution is
\begin{equation}
\label{FISpr11}
  P_k = (-1)^k \frac{(k+2)! (n-2-k)!}{2(n-2)!}P_0 \quad ,\quad k=1,\dots
  n-2    \quad .
\end{equation}
For $k=2$ this gives $P_{n-2}= (-1)^{n-2} (\frac{n(n-1)}{2}) P_0$.
This relation, together with \eqref{FISpr9}, implies $P_k=0$ for all
$k$ as stated which, for $k=n-2$, reproduces the FI in the form of
eq.~\eqref{FIshort} as $P_{n-2}=0$. For instance, with $n=3$,
$P_1=0$ in eq.~\eqref{FISpr6} gives $[B,C,[C,C,C]]=0$.
\medskip

\section{Appendix 2: the Schouten--Nijenhuis bracket}
\label{SNbra}

Let $\wedge(M)=\bigoplus^n_{j=0}\wedge^j(M)\
(\wedge^0=\mathscr{F}(M) \,,\, n={\rm dim} M$), be the contravariant
exterior algebra of skewsymmetric contravariant ({\it i.e.},
tangent) tensor fields ({\it multivectors} or $j$-{\it vectors})
over a manifold $M$. Then the Schouten-Nijenhuis bracket (SNB) of
$A\in \wedge^p(M)$ and $B\in \wedge^q(M)$ is the unique (up to a
constant) extension of the Lie bracket of two vector fields on $M$
to a $\mathbb{R}$-bilinear mapping
$\wedge^p(M)\times\wedge^q(M)\rightarrow \wedge^{p+q-1}(M)$ in such
a way that $\wedge(M)$ becomes a graded superalgebra. We start by
defining the SNB for multivectors given by wedge products of vector
fields.

\medskip
\begin{definition}
\label{SN-vectors}{\ }

 Let $X_1,\ldots, X_p, Y_1,\ldots, Y_q$ be vector
fields over $M$. Then
\begin{equation}
\begin{aligned}
\label{numerari}
 [X_1\wedge...\wedge X_{p}\; ,& \;
Y_{1}\wedge \ldots \wedge Y_{q}] = \\ = \sum
(-1)^{t+s}X_{1}\wedge...\hat X_{s}...& \wedge X_{p} \wedge[X_{s}
,Y_{t}]\wedge Y_{1} \wedge ... \hat Y_{t} ... \wedge Y_{q}\quad,
\end{aligned}
\end{equation}
 where $[\ ,\ ]$ is the SNB and $\hat X$ indicates the
omission of $X$\footnote{It is not difficult to see intuitively
the origin of this generalization of the Lie algebra of vector
fields to multivectors. If $B=Y_1\wedge Y_2\dots\wedge Y_q$ is a
$q$-vector field, it is natural to define $[X_1\wedge X_2\dots
\wedge X_p\,,\,B]=\sum_{i=1}^{p}(-1)^{i+1} X_1\wedge\dots\wedge
\hat{X_i}\dots\wedge X_p \wedge [X_i\,,B]$ and then take $[X\,,\,
Y_1\wedge Y_2,\dots \wedge Y_q] =L_{X}(Y_1\wedge Y_2,\dots \wedge
Y_q)=
 \sum_{j=1}^{q} Y_1\wedge Y_2\wedge \dots Y_{j-1}
 \wedge[X\,,Y_j]\wedge Y_{j+1}\wedge\dots\wedge Y_q$,
 which leads to \eqref{numerari}.}.
 It is easy to check that \eqref{numerari} is equivalent to
 original definition \cite{Sch:40,Nij:55}.
\end{definition}

\begin{theorem} {\ }

Let $M=G$ be the group manifold of a Lie group, and let
the above vector fields $X\,,\,Y$ be LI (resp. RI) vector fields
on $G$. Then, the SNB of LI (resp. RI) skew multivector fields is
also LI (resp. RI).
\begin{proof} It suffices to recall that if $X$ is LI and $Z$
is the generator of the left translations, $L_Z X=[Z,X]=0$ by the
first equation in \eqref{LRzero}.
\end{proof}
\end{theorem}
\medskip
\begin{definition} (\emph{Schouten-Nijenhuis bracket})
\label{SNBdef} {\ }

Let $A\in \wedge^p(M)$ and $B\in
\wedge^q(M)$, $p,q\leq n$, be the $p$- and $q$-vectors given in a
local chart by
\begin{equation}
A(x)={1\over p!}{A^{i_1\ldots i_p}(x)} \partial_{i_1}
\wedge\ldots\wedge \partial_{i_p}\quad,\quad B(x)={1\over
q!}{B^{j_1\ldots j_q}(x)}
\partial_{j_1}\wedge\ldots\wedge \partial_{j_q}\quad.
\label{IIvii}
\end{equation}
The SNB of $A$ and $B$ is the skewsymmetric contravariant tensor
field $[A,B]\in \wedge^{p+q-1}(M)$
\begin{equation}
\begin{aligned}
[A,B]= &{1\over (p+q-1)!} [A,B]^{k_1\ldots k_{p+q-1}}
\partial_{k_1}\wedge\ldots\wedge \partial_{k_{p+q-1}}\quad,\\
[A,B]^{k_1\ldots k_{p+q-1}}=&{1\over (p-1)!q!} \epsilon^{k_1\ldots
k_{p+q-1}}_{i_1\ldots i_{p-1} j_1\ldots j_q} A^{\nu i_1\ldots
i_{p-1}}\partial_\nu B^{j_1\ldots j_q}\\ &+{(-1)^p\over p!(q-1)!}
\epsilon^{k_1\ldots k_{p+q-1}}_{i_1\ldots i_{p} j_1\ldots j_{q-1}}
B^{\nu j_1\ldots j_{q-1}}\partial_\nu A^{i_1\ldots i_p}\quad,
\label{VIIii}
\end{aligned}
\end{equation}
where $\epsilon$ is the usual Kronecker symbol in eq.~\eqref{defep}.
\end{definition}

The SNB is graded-commutative,
\begin{equation}
[A,B]=(-1)^{pq}\,[B,A]\quad. \label{IIv}
\end{equation}
As a result, the SNB is identically zero if $A=B$ are of odd order
(or even {\it degree}; ${\rm degree}(A)\equiv{\rm order}(A)-1$).
Since $[A,B]$ is a ($p+q-1$)-vector, the SNB is also zero if $p+q-1>\hbox{dim}\,M$
and, of course, when $A,B$ are constant multivectors,
$A^{i_1\dots i_p}\not=A^{i_1\dots i_p}(x)\,,\,B^{j_1\dots j_q}\not=B^{j_1\dots j_q}(x)$.

The SNB satisfies the graded Jacobi identity,
\begin{equation}
(-1)^{pr}\,[[A,B],C]+(-1)^{qp}\,[[B,C],A]+(-1)^{rq}\,[[C,A],B]=0\quad,
\label{IIvi}
\end{equation}
where $(p,q,r)$ denote the order of $(A,B,C)$ respectively (thus,
if $\Lambda$ is of even order and $[\Lambda,\Lambda]=0$ it follows from
 \eqref{IIvi} that $[\Lambda,[\Lambda,C]]=0$).

Let $A\wedge B\in \wedge^{p+q}(M)$,
\begin{equation}
\begin{aligned}
& A\wedge B ={1\over (p+q)!}(A\wedge B)^{i_1\ldots i_{p+q}}
\partial_{i_1}\wedge\ldots\wedge\partial_{i_{p+q}}\quad, \\
&(A\wedge B)^{i_1\ldots i_{p+q}}={1\over p!q!}
\epsilon^{i_1\ldots i_{p+q}}_{j_1\ldots j_{p+q}} A^{j_1\ldots
j_{p}}B^{j_{p+1}\ldots j_{p+q}}\quad, \label{wedproduct}
\end{aligned}
\end{equation}
and let $\alpha\in\wedge_{p+q-1}(M)$ be an arbitrary
$(p+q-1)$-form, $\alpha={1\over (p+q-1)!} \alpha_{{i_1}\ldots
i_{p+q-1}}dx^{i_1}\wedge\ldots\wedge dx^{i_{p+q-1}}$. Then, the
well known formula for one-forms and vector fields,
$d\omega(X,Y)=L_X\omega(Y)-L_Y\omega(X)-i_{[X,Y]}\omega\,,$
generalizes to
\begin{equation}
i_{A\wedge B}d\alpha=(-1)^{pq+q}i_A d(i_B \alpha)+(-1)^p i_B
d(i_A\alpha)- i_{[A,B]}\alpha \quad, \label{IIi}
\end{equation}
where the contraction $i_A\alpha$ is the $(q-1)$-form
\begin{equation}
i_A\alpha(\cdot)={1\over p!}\alpha(A,\cdot)\quad,\quad
i_A\alpha={1\over (q-1)!}{1\over p!}A^{i_1\ldots i_{p}}
\alpha_{{i_1}\ldots i_{p} j_1\ldots j_{q-1}}
dx^{j_1}\wedge\ldots\wedge dx^{j_{q-1}} \quad, \label{contraction}
\end{equation}
so that, on forms, $i_Bi_A=i_{A\wedge B}$. When $\alpha$ is {\it
closed}, eq.~\eqref{IIi} provides a definition of the SNB through
$i_{[A,B]}\alpha$.

{}From the definition of the SNB it follows that
\begin{equation}
[A,B\wedge C]=[A,B]\wedge C+(-1)^{(p-1)q}B\wedge [A,C]\quad,
\label{IIii}
\end{equation}
\begin{equation}
[A\wedge B,C]=(-1)^p A\wedge [B,C]+ (-1)^{rq}[A,C]\wedge B\quad.
\label{IIiia}
\end{equation}
In particular, for the case of the SNB among the wedge product of
two {\it vector} fields
\begin{equation}
[A\wedge B,X\wedge Y]= -A\wedge [B,X]\wedge Y +B\wedge [A,X]\wedge
Y -B\wedge [A,Y]\wedge X +A\wedge [B,Y]\wedge X \quad,
\label{IIiii}
\end{equation}
so that
\begin{equation}
[A\wedge B,A\wedge B]=-2 A\wedge B\wedge[A,B]\quad. \label{IIiv}
\end{equation}
For instance, if $\Lambda$ is given by $X=X_j\wedge \partial^j\,,\,
X_j=\frac{1}{2}C_{ij}^k\partial_i$ (see \eqref{L-P-bivec}), we may
apply \eqref{IIiii} to find that the condition
$[\Lambda,\Lambda]=0$ leads to the Jacobi identity.

\begin{remark}
\label{rem:SN} As mentioned, the SNB is the unique extension of
the usual Lie bracket of vector fields which makes a $Z_2$-graded
Lie algebra of the (graded-)commutative algebra of skewsymmetric
contravariant tensors: ${\rm degree}([A,B])={\rm degree}(A)+{\rm
degree}(B)$. In it, the adjoint action is a graded derivation with
respect to the wedge product \cite{Koszul:85} (see eq.~\eqref{IIii}).
To make this graded structure explicit, it is convenient to define
a new SNB, $[\ ,\ ]'$, which differs from the original one $[\ ,\
]$ by a factor $(-1)^{p+1}$ in the {\it l.h.s.} of
\eqref{numerari}, \eqref{VIIii}:
\begin{equation}
[A,B]':=(-1)^{p+1}[A,B]\quad. \label{nuevosnb}
\end{equation}
This definition modifies \eqref{IIv} to read
\begin{equation}
[A,B]'=-(-1)^{(p-1)(q-1)}[B,A]'\equiv -(-1)^{ab}[B,A]' \quad.
\label{IIvb}
\end{equation}
where $a={\rm degree}(A)=(p-1)$, etc. Similarly, \eqref{IIvi} is
replaced by
\begin{equation}
(-1)^{pr+q+1}[[A,B]',C]'+(-1)^{qp+r+1}[[B,C]',A]'+(-1)^{rq+p+1}[[C,A]',B']=0
\quad, \label{IIvib}
\end{equation}
which in terms of the degrees $(a,b,c)$ of $A,B,C$ adopts the
graded JI form
\begin{equation}
(-1)^{ac}[[A,B]',C]'+(-1)^{ba}[[B,C]',A]'+(-1)^{cb}[[C,A]',B']=0
\quad. \label{IIvic}
\end{equation}

In fact, the multivector algebra with the exterior product and the
SNB is a Gerstenhaber algebra\footnote{ A Gerstenhaber algebra
\cite{Ger:63} is a ${\mathbb Z}$-graded vector space (with
homogeneous subspaces $\wedge^a$) with two bilinear multiplication
operators, $\cdot$ and $[\ ,\ ]$, with the following properties
($u\in \wedge^a$, $v\in \wedge^b$, $w\in \wedge^c$):
\begin{quote}
a) deg$(u\cdot v)=a+b$,
\\
b) deg$[u,v]=a+b-1$,
\\
c) $(u\cdot v)\cdot w=u\cdot (v\cdot w)$,
\\
d) $[u,v]=-(-1)^{(a-1)(b-1)}[v,u]$,
\\
e) $(-1)^{(a-1)(c-1)}[u,[v,w]]+(-1)^{(c-1)(b-1)}[w,[u,v]]+
(-1)^{(b-1)(a-1)}[v,[w,u]]=0$,
\\
f) $[u,v\cdot w]=[u,v]\cdot w+(-1)^{(a-1)b}v\cdot [u,w]\;$.
\end{quote} }, in which deg$(A)=p-1$ if $A\in \wedge^p$. Thus, the multivectors
of the form (\ref{completek}) form an abelian subalgebra of this
Gerstenhaber algebra, the commutativity (in the sense of the SNB)
being a consequence of (\ref{GJIcoord}).

The Schouten-Nijenhuis bracket definition of eq. (\ref{nuevosnb})
is used in \cite{Koszul:85, GraMaPe:93,CIMP:94,AlPe:97} and is more
adequate to stress the graded structure of the exterior algebra of
skew multivector fields; for instance, \eqref{IIvb} and
\eqref{IIvic} have the same form as in supersymmetry (see {\it
e.g.}, \cite{CoNeSt:74}). In Sec. \ref{sec:GPS}, however, we use
Def. \ref{SNBdef} for the SNB as in
\cite{Lich:77,Nij:55,BFFLS:78b,AzPePB:96b} and others.
\end{remark}

\section*{Acknowledgments}

We are grateful to I. Bandos, T. Curtright, J.M. Figueroa O'Farrill,
J.-L. Loday, A.J. Macfarlane, A.J. Mountain, J. Navarro-Faus,
A. Perelomov, J.C.  P\'erez-Bueno, M. Pic\'on,
D. Sorokin, J. Stasheff, P. K. Townsend and C. K. Zachos
for their collaboration on past joint work and/or for helpful
discussions, correspondence or references on parts of the material
covered by this review.

This work has been partially supported by the research grants
FIS2008-01980 and FIS2009-09002 from the Spanish MICINN and by
VA-013-C05 from the Junta de Castilla y Le\'on.

\providecommand{\href}[2]{#2}\begingroup\raggedright\endgroup

\bibliographystyle{utphys}
\bibliography{ReviewJAA}

\end{document}